\documentclass[phd,tocprelim]{cornell}
\usepackage{amsmath}
\usepackage{amsfonts}       
\usepackage{booktabs}       

\usepackage{bbm}

\usepackage{caption}
\usepackage{graphicx}

\usepackage{hyperref}       

\usepackage{makecell}
\usepackage{microtype}      
\usepackage{mdwlist}
\usepackage{nicefrac}       

\usepackage{mathtools}

\usepackage{subfigure}

\usepackage{thmtools} 
\usepackage{thm-restate}
\usepackage{times}
\usepackage{todonotes}

\usepackage{url}            

\renewcommand{\caption}[1]{\singlespacing\hangcaption{#1}\normalspacing}

\usepackage{lmodern}
\usepackage{xspace}
\usepackage{amsfonts,amsmath,amssymb,amsthm, pbox}

\usepackage{multirow}
\usepackage{dsfont} 

\usepackage{algorithmicx,algpseudocode, algorithm}

\usepackage[shortlabels]{enumitem}
\setitemize{noitemsep,topsep=0pt,parsep=0pt,partopsep=0pt}
\setenumerate{noitemsep,topsep=0pt,parsep=0pt,partopsep=0pt}

\makeatletter
\@ifundefined{theorem}{
  \theoremstyle{definition}
  \newtheorem{definition}{Definition}
  \theoremstyle{plain}
  \newtheorem{theorem}{Theorem}
  \newtheorem{corollary}{Corollary}
  \newtheorem{lemma}{Lemma}
  \newtheorem{claim}{Claim}
  \newtheorem{fact}{Fact}
  
  \theoremstyle{remark}
  \newtheorem{remark}{Remark}
  \newtheorem{question}{Question}
  
  \newtheorem{example}{Example}
  
}{}
\makeatother

\newcommand{\ignore}[1]{}

\newcommand{\II}{\mathbb{I}} 
\newcommand{\EE}{\mathbb{E}}

\newcommand{\NN}{\mathbb{N}}

\newcommand{\RR}{\mathbb{R}}

\newcommand{\expectation}[1]{\EE\left[#1\right]}
\newcommand{\variance}[1]{Var\left(#1\right)}

\def \cA     {{\cal A}}
\def \cB     {{\cal B}}
\def \cC     {{\cal C}}
\def \cD     {{\cal D}}
\def \cE     {{\cal E}}

\def \cG     {{\cal G}}
\def \cH     {{\cal H}}

\def \cL     {{\cal L}}

\def \cN     {{\cal N}}

\def \cP     {{\cal P}}
\def \cQ     {{\cal Q}}

\def \cS     {{\cal S}}

\def \cV     {{\cal V}}
\def \cW     {{\cal W}}
\def \cX     {{\cal X}}
\def \cY     {{\cal Y}}

\newcommand{\Var}{{\rm Var}}

\newcommand{\ie}{\textit{i.e.,}\xspace}  

\def \upto  {{,}\ldots{,}}

\def \ceil#1{{\lceil{#1}\rceil}}

\newcommand{\absv}[1]{\left|#1\right|}

\def \Paren#1{{\left({#1}\right)}}

\def \Brack#1{{\left[{#1}\right]}}

\newcommand{\ed}{:=}

\def \half    {{\frac12}}

\newcommand{\probof}[1]{\Pr\Paren{#1}}

\def\ignore#1{}

\newcommand{\bi}{\begin{itemize}}
\newcommand{\ei}{\end{itemize}}

\def\orpro{\mathop{\mathchoice
   {\vee\kern-.49em\raise.7ex\hbox{$\cdot$}\kern.4em}
   {\vee\kern-.45em\raise.63ex\hbox{$\cdot$}\kern.2em}
   {\vee\kern-.4em\raise.3ex\hbox{$\cdot$}\kern.1em}
   {\vee\kern-.35em\raise2.2ex\hbox{$\cdot$}\kern.1em}}\limits}

\def\andpro{\mathop{\mathchoice
 {\wedge\kern-.46em\lower.69ex\hbox{$\cdot$}\kern.3em}
 {\wedge\kern-.46em\lower.58ex\hbox{$\cdot$}\kern.25em}
 {\wedge\kern-.38em\lower.5ex\hbox{$\cdot$}\kern.1em}
 {\wedge\kern-.3em\lower.5ex\hbox{$\cdot$}\kern.1em}}\limits}

\def\simge{\mathrel{%
   \rlap{\raise 0.511ex \hbox{$>$}}{\lower 0.511ex \hbox{$\sim$}}}}

\def\simle{\mathrel{
   \rlap{\raise 0.511ex \hbox{$<$}}{\lower 0.511ex \hbox{$\sim$}}}}

\newcommand{\ab}{k}
\newcommand{\absvn}[1]{\left\|#1\right\|_1}
\newcommand{\absz}{k}
\newcommand{\alg}{\cA}
\newcommand{\argmin}[1]{\underset{#1}{\operatorname{arg}\,\operatorname{min}}\;}

\newcommand{\ber}{\cB}
\newcommand{\Bern}[1]{B(#1)}
\newcommand{\bino}[2]{{\rm Bin}(#1, #2)}

\newcommand{\chisquare}[2]{\chisquarerestr[]{#1}{#2}}
\newcommand{\chisquarerestr}[3][]{{\operatorname{d}^{#1}_{\chi^2}\!\left({#2 \mid\mid #3}\right)}}

\newcommand{\Dgk}{\Delta_{\ge\frac1\ab}}
\newcommand{\dims}{d}
\newcommand{\Dinfty}{D_\infty}
\newcommand{\dist}{\alpha}

\newcommand{\Dk}{\Delta_{\ab}}

\newcommand{\DKL}{D_{\mathrm{KL}}}
\newcommand{\dkl}[2]{d_{KL}(#1, #2)}

\newcommand{\dst}{\alpha}

\newcommand{\dtv}[2]{d_{TV}(#1, #2)}

\newcommand{\dP}{P}
\newcommand{\dQ}{Q}

\newcommand{\E}{\EE}

\newcommand{\ellone}[1]{||#1||_{\ell_1}}
\newcommand{\emp}{\hat{\p}_\ns}
\newcommand{\entp}[1]{H(#1)}
\newcommand{\entemp}{H(\emp)}
\newcommand{\eps}{\varepsilon}
\renewcommand{\epsilon}{\ve}

\newcommand{\expectationsub}[2]{\EE_{#1}\left[#2\right]}
\newcommand{\expectationf}[2]{\EE_{#2}\left[#1\right]}

\newcommand{\fp}{\beta}

\newcommand{\ham}[2]{d_{ham}(#1,#2)}
\newcommand{\he}{\hat{E}}

\newcommand{\id}{\mathbb{I}}

\newcommand{\kldist}[2]{D(#1\lVert #2)}

\newcommand{\Lap}{\mathrm{Lap}}

\newcommand{\MM}{\texttt{MM}}

\newcommand{\Mltsmb}[2]{M_{#1}(#2)}

\newcommand{\mP}{p}

\newcommand{\newzs}[1]{\textcolor{black}{#1}}
\newcommand{\newhz}[1]{\textcolor{black}{#1}}

 \newcommand{\newnewzs}[1]{{\color{black}#1}}
 
 \newcommand{\norm}[1]{\left\|#1\right\|_2}
 
 \newcommand{\normone}[1]{\left\lVert#1\right\rVert_1}
\newcommand{\norminf}[1]{\left\lVert#1\right\rVert_{\infty}}
\newcommand{\norminfone}[1]{\left\lVert#1\right\rVert_{\infty,1}}

 \newcommand{\norms}[1]{\left\|#1\right\|^2_2}

\newcommand{\ns}{n}
\newcommand{\nsmb}{N_{\smb}}

\newcommand{\p}{p}
\newcommand{\pemps}{\emp(\smb)}

\newcommand{\plugin}{\texttt{plug-in}}

\newcommand{\polylog}{\mathrm{polylog}}
\newcommand{\polyn}{\texttt{poly}}
\newcommand{\posa}{u}
\newcommand{\posb}{v}
\providecommand{\poly}{\operatorname*{poly}}
\newcommand{\pr}[2][]{\mathrm{Pr}\ifthenelse{\not\equal{}{#1}}{_{#1}}{}\!\left[#2\right]}

\newcommand{\prfi}[1]{\Phi_{#1}}

\newcommand{\proboff}[2]{\Pr_{#2}\Paren{#1}}
\newcommand{\probofsub}[2]{\Pr\nolimits_{#1}\Paren{#2}}
\newcommand{\psmb}{\p_\smb}

\newcommand{\q}{q}

\newcommand{\R}{\mathbb{R}}

\newcommand{\Sct}{S_{\tt CT}(\ab, \dist,\eps)}

\newcommand{\Sctnp}{S_{\tt CT}(\ab, \dist)}

\newcommand{\Sitnp}{S_{\tt IT}(\ab, \dist)}
\newcommand{\Sit}{S_{\tt IT}(\ab, \dist,\eps)}

\newcommand{\smb}{x}

\newcommand{\Sut}[2]{S_{\tt UT}(#1, #2,\eps)}

\newcommand{\Sutnp}[2]{S_{\tt UT}(#1, #2)}

\newcommand{\unif}{u}
\newcommand{\unifab}{\unif[{\ab}]}
\newcommand{\unifabs}[1]{\unif[{#1}]}

\def \ve {{\eps}} 

\newcommand{\xon}{x_1^\ns}
\newcommand{\Xon}{X_1^\ns}
\newcommand{\Xonn}{\widetilde{X}_1^{\ns}}
\newcommand{\Xot}{X_1^t}
\newcommand{\Xos}{X_1^s}
\newcommand{\Xosmo}{X_1^{s-1}}
\newcommand{\Xosp}{X_1^{'s}}

\newcommand{\yon}{y_1^\ns}
\newcommand{\Yon}{Y_1^\ns}
\newcommand{\Yonn}{\widetilde{Y}_1^{\ns}}

\newcommand{\Yosmo}{Y_1^{s-1}}
\newcommand{\Yos}{Y_1^s}
\newcommand{\Yosp}{Y_1^{'s}}
\newcommand{\Yot}{Y_1^t}

\newcommand{\Zon}{Z_1^\ns}

\title {Statistical Inference in the Differential Privacy Model}
\author {Huanyu Zhang}
\conferraldate {August}{2021}
\degreefield {Ph.D.}
\copyrightholder{Huanyu Zhang}
\copyrightyear{2021}

\begin{document}

\maketitle
\makecopyright

\begin{abstract}
In modern settings of data analysis, we may be running our algorithms on datasets that are sensitive in nature. However, classical machine learning and statistical algorithms were not designed with these risks in mind, and it has been demonstrated that they may reveal personal information. These concerns disincentivize individuals from providing their data, or even worse, encouraging intentionally providing fake data.
 
To assuage these concerns, we import the constraint of differential privacy to the statistical inference, considered by many to be the gold standard of data privacy. 
This thesis aims to quantify the cost of ensuring differential privacy, i.e., understanding how much additional data is required to perform data analysis with the constraint of differential privacy. Despite the maturity of the literature on differential privacy, there is still inadequate understanding in some of the most fundamental settings.

In particular, we make progress in the following problems:

\medskip
\begin{itemize}
\item What is the sample complexity of DP hypothesis testing?

\item Can we privately estimate distribution properties with a negligible cost?

\item What is the fundamental limit in private distribution estimation?

\item How can we design algorithms to privately estimate random graphs?

\item What is the trade-off between the sample complexity and the interactivity in private hypothesis selection?

\end{itemize}
\end{abstract}

\begin{biosketch}
Huanyu Zhang received his B.S. degree in Electronics Engineering from Peking University, Beijing, China, and the M.S. degree in Electrical and Computer Engineering from Cornell University, Ithaca, NY, USA, in 2016 and 2019, respectively.
His research interest lies broadly in machine learning and algorithms, especially in the areas of privacy-preserving data analysis. Specifically, he is interested in studying the tradeoffs between privacy and other resources when performing fundamental statistical tasks.
\end{biosketch}


\begin{acknowledgements}

I am extremely fortunate to have Jayadev Acharya as my Ph.D. advisor, and I want to express my deepest gratitude to him for his support and help throughout my Ph.D. study. 
He has guided and inspired me in many ways: 
not only have I acquired all my academic skills from him on how to find and solve problems, write technical papers, and give academic presentations, but most importantly, I have learned how to be an enthusiastic and reliable researcher like him. I will never forget how serious he is about research: every time we submit a paper to Arxiv, he is always polishing it again and again, until we both get satisfied. 
He
has been very supportive of my decisions and thoughts, and tolerant of my mistakes. 
I will never forget he stayed up the whole night and changed every word of the introduction when we submitted the paper ``Differentially Private Assouad, Fano, and Le cam''. 
All the achievements in my graduate studies could not happen without his countless help, and I can never thank him enough! It is my honor to be his student.

Next, I want to thank Aaron Wagner and Siddhartha Banerjee for serving as my committee members and providing valuable feedback. I have learned a lot from their lectures and discussing with them. Special thanks to Aaron Wagner for his advising. I always wish I could be as knowledgeable and energetic as him.

I thank Gautam Kamath for being an amazing mentor, and friend throughout. I will never forget the Friday afternoon we spent in Seattle, when we finally figured out the algorithm of the adversarial selection.

I would like to thank all of my co-authors so far in my academic career: Jayadev Acharya, Sivakanth Gopi, Meisam Hejazinia, Gautam Kamath, Janardhan Kulkarni, Ilya Mironov, Aleksandar Nikolov, Ziteng Sun, Di Wang and Zhiwei Steven Wu. It is amazing to have so many smart people around! 
In particular, I want to thank Janardhan Kulkarni,  Meisam Hejazinia, and Ilya Mironov for their mentoring during my internships.

I am extremely fortunate to work with a wonderful group of colleagues and labmates at Cornell: Sourbh Bhadane, Boshuang Huang, Yu Gan, Saravanan Kandasamy, Yuhan Liu, Ziteng Sun, Chao Wang, Xu Xiao, and Chengrun Yang. I will never forget the fun times we had in Ithaca. It is them who made my journey as a Ph.D. student memorable.
Special thanks to my best friend, Ziteng Sun. I really enjoy our conversation, both academic and non-academic. He is bright and insightful, and his sharp idea always lights up my day.

Last but most importantly, I want to express my gratitude to my family, for their unconditional support and love. Thanks for raising me up, educating me, and respecting all my thoughts and decisions. Special Thanks to my beloved girlfriend, Naipeng Lin, for her trust, caring and support throughout my Ph.D. studies. There is clearly not enough space to convey my love towards them.

\end{acknowledgements}

\contentspage
\tablelistpage
\figurelistpage

\normalspacing \setcounter{page}{1} \pagenumbering{arabic}
\pagestyle{cornell} \addtolength{\parskip}{0.5\baselineskip}

\chapter{Introduction}
Statistical estimation is one of the most classical statistical problems, which asks the following question: given samples from an unknown probabilistic model,  can we estimate  some property of the underlying model? For example, given samples from an unknown Gaussian distribution, is there an accurate way to estimate its mean?

This question has received great attention in statistics, which can be traced back to the early 20th century~\cite{NeymanP33, Fisher92, LehmannR06}. Unfortunately, the classical studies are not entirely aligned with the needs of modern data science. Specifically, there are two stringent challenges:

\noindent\textbf{Non-asymptotic regime.} Previously, this problem was mostly studied in the asymptotic regime when the number of samples $\ns\to\infty$, where the empirical distribution is a good approximate of the true distribution. However, this assumption is not always satisfied in modern data science. For example, in estimating gene mutations, it is quite common that the domain size is extremely large, and data is ``scarce'' in comparison to the size of the domain. Carrying out statistical inference in such setting can be totally inaccurate, if the problem is studied in the asymptotic assumption.

Motivated by these new scenarios, there has been recently a lot of work from the computer science, information theory, and statistics community on various statistical inference problems in the non-asymptotic (small-sample) regime, where the domain size $k$ could be potentially larger than $\ns$ (see, e.g.,~\cite{BatuFRSW00, BatuFFKRW01, GoldreichR00, Paninski08, OrlitskySW16}). 
Instead of characterizing the asymptotic performance, the goal is to characterize the sample complexity, which is the minimum number of samples necessary (or equivalently, the minimax risk), as a function of the domain size $k$, and the other parameters.

\noindent\textbf{Privacy.} 
There is another challenge rising from modern data science: in several estimation tasks, individual samples have sensitive information that must be protected. This is particularly of concern in applications such as healthcare, finance, geo-location, etc. For example, in medical studies, the data may contain individuals' health records and whether they carry some disease which bears a social stigma. Alternatively, navigation apps suggest routes based on aggregate positions of individuals, which contains delicate information including users' residence data. However, classical statistical inference algorithms were not designed with this issue in mind, and it has been demonstrated that they may reveal personal information~\cite{HomerSRDTMPSNC08}. These concerns disincentivize individuals from providing their data, or even worse, encouraging intentionally providing fake data.

To preserve the privacy of sensitive data, a common strategy in practice is anonymization, which simply removes the sensitive information from the original data, such as name, race, and social security number. However, this strategy is far from enough, where the sensitive information can still be learned after that information is anonymized. One of the most famous examples is that in the late 2000s, Netflix ran a competition to develop a better film recommendation algorithm. To drive the competition, they released an ``anonymized''  version of the dataset that had removed obvious identifying information. Unfortunately, this scheme turned out to be insufficient. In~\cite{NarayananS08}, it was shown that when paired with a small amount of additional information, the released dataset could be used to re-identify specific users, and even predict their political affiliation.

Another strategy is to mediate data access through a trusted interface, which only answers specific queries from data analysts. However, it remains a non-trivial task to design such a system that can protect privacy. For example, a natural question is what kind of query is allowed, and what kind of query is prohibited. Clearly, the queries should not be allowed which target  specific persons (e.g., ``Does Bob smoke?''). However, even if every single query does not do so, a combination of queries can still be used to detect individual information (e.g., ``what is the average salary of the dataset?'' and ``what is the average salary of the people who are not Ph.D. students?''). 

A possible solution to alleviate privacy issues is to design algorithms with privacy guarantees, where a natural problem to ask is how to define privacy.  
In 1977,  statistician Tore Dalenius proposed the following definition of data privacy~\cite{Dalenius77}: the attacker should know nothing new about each individual after the analysis. However,  the following example~\cite{Dwork08} shows that it is impossible to be achieved. We suppose revealing one's exact height is a privacy violation, and there is a dataset containing the heights of all the people in China. An adversary who has access to the dataset, and the prior that ``Alice's is two inches taller than an average Chinese'' easily learns Alice's height. Note that Alice is not even included in the dataset! This example shows that it is impossible to fully protect  privacy with no constraint of the attacker's prior knowledge.

Here we consider an alternative way to define privacy: the attacker should learn virtually nothing more about an individual than they would learn if that person’s record were absent from the dataset, which is exactly the motivation of differential privacy (DP)~\cite{DworkMNS06}. Informally, DP requires that the outputs of the algorithm are indistinguishable for two neighbouring datasets, which differ at exactly one record. In other words, it is impossible for an adversary to infer whether a specific individual is involved, thus protecting the privacy of the data providers. DP allows statistical inference while preserving the privacy of the individual samples, which has become one of the most {popular notions of privacy}~\cite{DworkMNS06, WassermanZ10, DworkRV10, BlumLR13, McSherryT07, DworkR14, KairouzOV17}. It has been adopted by the US Census Bureau for the 2020 census and several large technology companies, including Google, Apple, and Microsoft~\cite{ErlingssonPK14, AppleDP17, DingKY17}.


Combining these two challenges, we want to answer the following question:
\begin{center}
\textbf{What is the fundamental limit and how to design algorithms in DP statistical inference?}
\end{center}

Specifically, the objective of this thesis is to quantify the cost of ensuring differential privacy in statistical inference, i.e., investigating how the sample complexities change with the constraint of differential privacy. 

\section{Contributions}

In this section, we outline the contributions of the thesis.

\begin{enumerate}

\item In Chapter~~\ref{cha:toolbox}, we firstly establish a toolbox of proving lower bounds in DP statistical inference. Le Cam's method, Fano's inequality, and Assouad's lemma are three widely used techniques to prove lower bounds for statistical estimation tasks. In Chapter~\ref{cha:toolbox}, we propose their analogues under DP. Our new tools are simple, easy to apply and we use them to establish sample complexity lower bounds in several statistical inference tasks, which are illustrated in later chapters.

\item  In Chapter~\ref{cha:testing}, we study the problems of DP identity testing (goodness of fit), and closeness testing (two sample test) of distributions over $\ab$ elements, which are both fundamental problems in statistical inference. We derive upper and lower bounds on the sample complexity of both problems under DP. For both problems, our results are tight up to constant factors for all parameter ranges. The lower bounds are established through our private Le Cam's method.

\item In Chapter~\ref{cha:ins}, we develop DP methods for estimating various distributional properties.
Specifically, we prove almost-tight bounds on the sample complexity for this problem for several functionals of interest, including support size, support coverage, and entropy. We show that the cost of privacy is negligible in a variety of settings, both theoretically and experimentally. We establish the lower bounds by applying our private Le Cam's method.

\item In Chapter~\ref{cha:estimation}, we move to another important problem in statistical inference -- distribution estimation. We establish the optimal sample complexity of DP discrete distribution estimation under total variation distance and $\ell_2$ distance, and we provide lower bounds for several other distribution classes, including product distributions and Gaussian mixtures that are tight up to logarithmic factors. Our lower bounds can be viewed as applications of the private Fano's inequality, and the private Assouad's lemma.

\item In Chapter~\ref{cha:MRF},  we focus on estimating a more complicated class of distributions -- random graphs. Specifically, we consider the problem of learning Markov Random Fields (including the prototypical example, the Ising model) under the constraint of
 DP.  Our learning goals include both \emph{structure learning}, where we try to estimate the underlying graph structure of the model, as well as the harder goal of \emph{parameter learning}, in which we additionally estimate the parameter on each edge.  We provide algorithms and lower bounds for both problems under a variety of privacy constraints -- namely pure, concentrated, and approximate differential privacy. While non-privately, both learning goals enjoy roughly the same complexity, we show that this is not the case under differential privacy.  
As a result, we show that the privacy constraint imposes a strong separation between these two learning problems in the high-dimensional data regime. 
 
 \item In Chapter~\ref{cha:phs},  we initiate the study of hypothesis selection under local differential privacy, which can be viewed as a generalization of the classic problem of multi-way hypothesis testing. Absent privacy constraints, this problem requires $O(\log k)$ samples, where $k$ is the size of the hypothesis class. We first show that the constraint of local differential privacy incurs an exponential increase in cost: any algorithm for this problem requires at least $\Omega(k)$ samples. Second, for the special case of multi-way hypothesis testing, we provide a non-interactive algorithm which nearly matches this bound, requiring $\tilde O(k)$ samples. Finally, we provide sequentially interactive algorithms for the general case, requiring $\tilde O(k)$ samples and only $O(\log \log k)$ rounds of interactivity. Our algorithms for the general case are achieved through a reduction to maximum selection with adversarial comparators. For this problem, we provide a family of algorithms which are near-optimal in the trade-off between the error and the interaction.

\end{enumerate}
\section{Organization and Bibliographic Information}
In Section~\ref{sec:main_prelim}, we introduce notations, preliminaries and problem formulations that will be used in the rest of the thesis.

Chapter~\ref{cha:toolbox} establishes a toolbox of proving DP lower bounds for statistical inference tasks. This is based on the paper ``Differentially Private Assouad, Fano, and Le cam,'' which is joint work with Jayadev Acharya and Ziteng Sun, and appeared in the 31st International Conference on Algorithmic Learning Theory~\cite{AcharyaSZ21}.

Chapter~\ref{cha:testing} studies the problem of identity testing and closeness testing of discrete distributions in DP. Part of the chapter is based on the paper ``Differentially Private Testing of Identity and Closeness of Discrete Distributions,'' which is joint work with Jayadev Acharya and Ziteng Sun, and appeared in the Proceedings of the 32nd International Conference on Neural Information Processing Systems~\cite{AcharyaSZ18}.

Chapter~\ref{cha:ins} investigates the problem of privately estimation several distribution properties. This is based on the paper ``INSPECTRE: Privately Estimating the Unseen,'' which is joint work with Jayadev Acharya, Gautam Kamath, and Ziteng Sun, and appeared in the Proceedings of the 35th International Conference on Machine Learning~\cite{AcharyaKSZ18}. 

Chapter~\ref{cha:estimation} focuses on the problem of private distribution estimation. This is based on the paper ``Differentially Private Assouad, Fano, and Le cam,''  which is joint work with Jayadev Acharya and Ziteng Sun, and appeared in the 31st International Conference on Algorithmic Learning Theory~\cite{AcharyaSZ21}.

Chapter~\ref{cha:MRF} describes results of privately estimating Markov random fields. This is based on the paper ``Privately Learning Markov Random Fields,'' which is joint work with Gautam Kamath, Janardhan Kulkarni, and Zhiwei Steven Wu, and appeared in the Proceedings of the 37th International Conference on Machine Learning~\cite{ZhangKKW20}.

Chapter~\ref{cha:phs} considers the problem of private hypothesis selection in distributed setting. This is based on the paper ``Locally Private Hypothesis Selection,'' which is joint work with Sivakanth Gopi, Gautam Kamath, Janardhan Kulkarni, Aleksandar Nikolov, and Zhiwei Steven Wu, and appeared in the 33rd Annual Conference on Learning Theory~\cite{GopiKKNWZ20}.

Other papers by the author over his PhD studies, but not in this thesis include~\cite{AcharyaSZ18a, WangZGX21,ZhangMH21,KamathLZ21,AcharyaSZ21-2}.
\section{Preliminaries and Notation}
\label{sec:main_prelim}
Let $\Delta_k$ be the class of all discrete distributions over a domain of size $\ab$, which wlog is assumed to be $[\ab]\ed\{1\upto\ab\}$. Let $X\sim\p$ denote that the random variable $X$ has distribution $\p$. We denote length-$\ns$ samples $X_1\upto X_\ns$ by $\Xon$. For $\smb\in[\ab]$, let $\psmb$ or $\p(\smb)$ be the probability of observing element $\smb$ under $\p$. The choice will be clear from the context. For $A\subseteq[\ab]$, let $\p(A) = \sum_{\smb\in A} \psmb$.  We use $\Mltsmb{\smb}{\Xon}$ to denote the number of times $\smb$ appears in $\Xon$. 

\subsection{Privacy Preliminaries}

We first introduce the definition of $(\eps,\delta)$-differential privacy.
\begin{definition}
A randomized algorithm $\cA$ on a set $\cX^\ns\to \cS$ is said to be $(\eps,\delta)$-differentially private if for any $S\subset \text{range}(\cA)$, and all pairs of $\xon$, and $\yon$ with $\ham{\xon}{\yon}\le 1$ \newhz{such that $\probof{\cA(\xon)\in S}\le e^{\eps}\cdot\probof{\cA(\yon)\in S}+\delta$.}
\end{definition}
The case when $\delta=0$ is called \emph{pure differential privacy}.  For simplicity, we denote pure differential privacy as $\eps$-differential privacy ($\eps$-DP). 

Next we show two properties of differential privacy. The first is the post-processing inequality, which says that differential privacy is immune to post-processing. 
\begin{lemma}
Let  $\cA$ be a randomized algorithm  $\cA: \cX^\ns\to \cS$ which is $(\eps,\delta)$-differentially private, and let $f:\cS \rightarrow \cS^\prime$ be an arbitrary randomized mapping. $f \circ \cA$ is $(\eps,\delta)$-differentially private.
\end{lemma}

\begin{proof}
We prove the proposition for a deterministic function $f$. The result then follows because any randomized mapping can be decomposed into a convex combination of deterministic functions, and a convex combination of differentially private mechanisms is differentially private.

For any neighbouring datasets $\xon$ and $\yon$, and a fixed event $T \subset \cS^\prime$. Let $R = \{r \in \cS: f(r) \in T\}$. We then have
\begin{align*}
\probof{f(\cA(\xon))\in T} &= \probof{\cA(\xon)\in R} \\
&\le e^\eps \cdot  \probof{\cA(\yon)\in R} +\delta \\
&= e^\eps \cdot  \probof{ f(\cA(\yon))\in T} +\delta. 
\end{align*}
\end{proof}

Then we introduce another important property of differential privacy -- group property. Informally, it describes the relationship between the outputs of a DP algorithm on two datasets, which differ in at most $t$ records.
\newhz{
\begin{lemma} \label{lem:group}
	Let $\alg$ be a $(\eps,\delta)$-DP algorithm, then for sequences $\xon$, and $\yon$ with $\ham{\xon}{\yon} \le t$, and $\forall S \subset \text{range}(\cA)$, 
$		\probof{\cA(\xon)\in S}\le e^{t\eps}\cdot\probof{\cA(\yon)\in S}+ \delta t e^{\eps (t-1)}$. 
\end{lemma}}

\begin{proof}
Let $\ham{\xon}{\yon}= \hat D$. When $\hat D = 0$ or $1$, the lemma is trivially true.  Then, suppose $\hat D \ge 2$ we can find $\hat D-1$ sequences  $z^\ns_1, \ldots, z^\ns_{\hat D-1}$ over $\cX^\ns$ with $\ham{\xon}{z^\ns_1} =1, \ham{z^\ns_{\hat D - 1} }{\yon} = 1$ and $\ham{z^\ns_i}{z^\ns_{i+1}} = 1$ for $i \in \{1,2,...,\hat D - 2\}$. 
Hence, by the condition of $(\eps,\delta)$ - differential privacy,
\begin{align}
	\probof{\alg\Paren{\xon}=\q} & \leq  e^{\eps} \probof{\alg\Paren{z^\ns_{1}}=\q} + \delta 
	\leq e^{\eps} (e^{\eps} \probof{\alg\Paren{z^\ns_{2}}=\q} + \delta) + \delta \le \ldots \nonumber\\
	&\leq  e^{\hat D\eps} \probof{\alg\Paren{\yon}=\q}+ \delta\cdot\sum_{i = 0}^{\hat D - 1} e^{i\eps} 
	\leq  e^{\hat D\eps} \probof{\alg\Paren{\yon}=\q}  + \delta \hat D e^{(\hat D-1) \eps } \nonumber \\ 
	&\leq  e^{t\eps} \probof{\alg\Paren{\yon}=\q}  + \delta t e^{\eps(t-1)}. \nonumber
\end{align}
\end{proof}

In their original paper,~\cite{DworkMNS06} provides a famous scheme for achieving differential privacy, known as the Laplace mechanism. We first define the sensitivity, and then introduce their scheme. Informally, the sensitivity measures the maximum divergence between the outputs of the algorithm on two neighboring datasets.
\begin{definition}
\label{def:sensitivity} The \emph{sensitivity} of $f:\cX^{\ns}\rightarrow \RR$ is
\begin{align*}
	\Delta_{\ns,f}\ed\max_{\ham{\xon}{\yon}\le1} \absv{f(\xon)-f(\yon)}.
\end{align*}
\end{definition}

Next we introduce the Laplace mechanism. This method adds Laplace noise to a non-private output in order to make it private.
 \begin{lemma}[Laplace mechanism~\cite{DworkMNS06}]
 \label{lem:laplace_mechanism}
 For any $\eps \ge 0$, and $f: \cX^\ns \rightarrow \mathbb{R}$, $\cA(\Xon) = f(\Xon)+ Laplace \Paren{\frac{\Delta_{\ns,f}}{\eps}}$ satisfies $(\eps,0)$-DP.
 \end{lemma}
 
 \subsection{Problem Formulation}
 \label{sec:preliminary_problem_formulation}
In this section, we formally define the framework of DP statistical inference.

\noindent\textbf{Setting.} Let $\cP$ be \emph{{any}} collection of distributions over $\cX^\ns$, {where $\ns$ denotes the number of samples.}\footnote{In the general setting, we are not assuming i.i.d. distribution over $\cX^\ns$, although we will specialize to this case later.} {Let $\theta:\cP\to\Theta$ be a parameter of the distribution that we want to estimate}. {Let $\ell:\Theta\times\Theta\to\RR_+$ be a pseudo-metric which is the loss function for estimating $\theta$.
\medskip

\noindent\textbf{Statistical inference.} The risk of an estimator $\hat{\theta}:\cX^\ns\to\Theta$ under loss $\ell$ is $\max_{\p \in \cP} \EE_{\Xon \sim \p}\left[\ell(\hat \theta (\Xon),\theta(p))\right]$, {the worst case expected loss of $\hat\theta$ over $\cP$. Note that $\Xon\in\cX^\ns$, since $\p$ is a distribution over $\cX^\ns$. The minimax risk of estimation under $\ell$ for the class $\cP$ is}
\begin{align*}
R(\cP, \ell) := \min_{\hat \theta}\ \max_{\p \in \cP}\ \EE_{\Xon \sim \p}\left[\ell(\hat \theta (\Xon),\theta(p))\right].
\end{align*}

In this thesis, we study the the minimax risk under differentially private protocols, which is given by restricting $\hat\theta$ to be differentially private. For $(\eps, \delta)$-DP, we study the following minimax risk:
\begin{align}
R(\cP, \ell, \eps, \delta) := \min_{\hat \theta \text{ is }(\eps,\delta)\text{-DP}}\ \max_{\p \in \cP}\ \EE_{\Xon \sim \p}\left[\ell(\hat \theta (\Xon),\theta(p))\right].\label{eqn:minmax-risk}
\end{align}
For $\delta=0$, the above minimax risk under $\eps$-DP is denoted as $R(\cP, \ell, \eps)$. 

There are several specific problems we explore in the thesis.

\begin{enumerate}

\item \noindent\textbf{Hypothesis testing.} {Let $\cP_1\subset \cP$, and $\cP_2\subset \cP$ be two disjoint subsets of distributions denoting the two hypothesis classes. Let $\Theta=\{1,2\}$, such that for $\p\in\cP_i$, let $\theta(p)=i$. For a test $\hat\theta:\cX^n\to \{1,2\}$, and $\ell(\theta, \theta')=\II\{\theta\ne\theta'\}=|\theta-\theta'|$,} {the error probability is the worst case risk under this loss function:}
\begin{align*}
P_e(\hat\theta, \cP_1, \cP_2):= \max_i \max_{\p\in \cP_i}\probofsub{ 
\newnewzs{\Xon\sim \p}}{\hat\theta(\Xon)\ne i } = \max_i \max_{\p\in 
\cP_i}\EE_{\Xon\sim \p} 
\left[{|\hat\theta(\Xon)-\theta(p)}|\right].
\end{align*}
We explore this problem in Chapter~\ref{cha:testing}.

\item \noindent\textbf{Distribution property estimation.} Let $\cQ$ be a collection of distributions over $\cX$, and for this $\cQ$, let $\cP=\cQ^{\ns}:=\{q^\ns:q\in\cQ\}$ be the collection of $n$-fold distributions over $\cX^\ns$ induced by i.i.d. draws from a distribution over $\cQ$. Let $f: \cQ \rightarrow \mathbb{R}$ be the property of interest, and the parameter space be $\Theta=\{ f(q): q \in \cQ  \}$. For a tester $\hat{\theta}: \cX^\ns \rightarrow \mathbb{R}$, the loss function is defined as $\ell (\theta, \hat{\theta}) = |\theta - \hat{\theta}|$. Let $\alpha>0$ be a fixed parameter. The sample complexity,  $S_{\tt PE}(\cQ,\alpha, \eps, \delta)$ is the smallest number of samples $n$ to make $R(\cQ^\ns, \ell, \eps, \delta)\le \alpha$, \ie
\[
S_{\tt PE}(\cQ,\alpha, \eps, \delta)= \min \{n: R(\cQ^\ns, \ell, \eps, \delta)\le \alpha\}.
\]
When $\delta = 0$, we denote the sample complexity by $S_{\tt PE}(\cQ, \alpha, \eps)$. 

We explore this problem in Chapter~\ref{cha:ins}.

\item \noindent\textbf{Distribution estimation.} Let $\cQ$ be a collection of distributions over $\cX$, and for this $\cQ$, let $\cP=\cQ^{\ns}:=\{q^\ns:q\in\cQ\}$ be the collection of $n$-fold distributions over $\cX^\ns$ induced by i.i.d. draws from a distribution over $\cQ$. The parameter space is $\Theta=\cQ$, where $\theta(\q^n)=\q$, and $\ell$ is a distance measure between distributions in $\cQ$. 
Let $\alpha>0$ be a fixed parameter. The sample complexity,  $S_{\tt DE}(\cQ, \ell, \alpha, \eps, \delta)$ is the smallest number of samples $n$ to make $R(\cQ^\ns, \ell, \eps, \delta)\le \alpha$, \ie
\[
S_{\tt DE}(\cQ, \ell, \alpha, \eps, \delta)= \min \{n: R(\cQ^\ns, \ell, \eps, \delta)\le \alpha\}.
\]
When $\delta = 0$, we denote the sample complexity by $S_{\tt DE}(\cQ, \ell, \alpha, \eps)$. 

This problem is explored in Chapter~\ref{cha:estimation}, and~\ref{cha:MRF}.

\end{enumerate}
 
 \subsection{Measures of Distance}

We firstly introduce several measures of distance between distributions, which are heavily used in this thesis.

\begin{definition}
The \emph{total variation} distance between distributions $\p$ and $\q$ over $[\ab]$ is 
\[
\dtv{\p}{\q} \ed \sup_{A\subset[\ab]}\{ \p(A) - \q(A)\} = \frac12\|\p-\q\|_1.
\]
Note that this is equivalent to the half of the $\ell_1$ distance between $\p$ and $\q$.
\end{definition}

\begin{definition}
The $KL$ divergence between distributions $\p$, and $\q$ over $[\ab]$ is 
\[
D_{KL}\Paren{\p,\q}:=\sum_{\smb\in [\ab]}p(\smb)\log \frac{p(\smb)}{q(\smb)}.
\]
This definition uses the convention that $0\log 0 = 0$.
\end{definition}

\begin{definition}
The $\chi^2$-distance (or chi-squared distance) between $\p$, and $\q$ over $[\ab]$ is 
\[
D_{\chi^2}\Paren{\p,\q}:=\sum_{\smb\in [\ab]} \frac{\Paren{p(\smb)-q(\smb)}^2}{q(\smb)}.
\]
\end{definition}

\begin{definition}
The $\ell_2$-distance between distributions $\p$, and $\q$ over $[\ab]$ is 
\[
D_{\ell_2}\Paren{\p,\q}:=\sqrt{ \sum_{\smb \in [\ab]} \Paren{\p(\smb)-\q(\smb)}^2} = \|\p-\q\|_2.
\]
\end{definition}

We have the following relationships between these distance measures. The following lemma, which is known as Pinsker's inequality, reveals the relationship between the total variation distance, the $KL$ divergence, and the $\chi^2$-distance.

\begin{lemma}[Pinsker's inequality]
\label{lem:distance_relationship}
Let  $\p$ and $\q$ be distributions over $[\ab]$,
\[
2 \cdot \dtv{\p}{\q} \le  \sqrt {D_{KL} \Paren{\p,\q} }\le \sqrt{ D_{\chi^2}\Paren{\p,\q}}.
\]
\end{lemma}

The next lemma, which follows from Cauchy-Schwarz, tells the relationship between the total variation distance and the $\ell_2$-distance.
\begin{lemma}
\label{lem:distance_relationship2}
Let  $\p$ and $\q$ be distributions over $[\ab]$,
\[
D_{\ell_2}\Paren{\p,\q} \le 2 \cdot \dtv{\p}{\q} \le  \sqrt{k} D_{\ell_2}\Paren{\p,\q}.
\]
\end{lemma}

Finally, we introduce the Hamming distance, which measures the distance between two sequences of samples.

\begin{definition}
The \emph{Hamming distance} between two sequences $\Xon$ and $\Yon$ is 
$\ham{\Xon}{\Yon} \ed \sum_{i=1}^{\ns} \II\{{X_i\ne Y_i}\},$
 the number of positions where $\Xon$, and $\Yon$ differ. 
\end{definition}

\chapter{A Toolbox of Proving Lower Bounds}
\label{cha:toolbox}
\section{Introduction}

Statistical estimation tasks are often characterized by the optimal trade-off between the sample size and estimation error. Generally speaking, there are two steps in establishing tight sample complexity bounds: An information-theoretic lower bound on sample complexity and an algorithmic upper bound that achieves it.
Several works
have developed general tools to obtain the lower bounds (e.g.,~\cite{LeCam73, Assouad83, IbragimovK13, BickelR88, Devroye87, HanV94, CoverT06, ScarlettC19}, and references therein), {and three prominent techniques are Le Cam's method, Fano's inequality, and Assouad's lemma.} Le Cam's method is used to establish lower bounds for hypothesis testing and functional estimation. Fano's inequality, and Assouad's lemma prove {lower bounds} for multiple hypothesis testing problems and can be applied to parameter estimation tasks such as estimating distributions. An exposition of {these three methods and their connections} is presented in~\cite{Yu97}\footnote{The title of~\cite{Yu97}, ``Assouad, Fano, and Le Cam'' is the inspiration for our title.}. 

In this chapter, we propose their analogs under differential privacy, which will be frequently used to establish lower bounds in later chapters. We firstly introduce the following observation, which is the motivation of our new results.

\subsection{An Observation}

{
We recall the definition of coupling.
\begin{definition}
A \emph{coupling} between distributions $\p_1$ and $\p_2$ over 
$\cX^\ns$ is a joint distribution $(\Xon,\Yon)$ over $\cX^\ns \times\cX^\ns$ 
whose marginals satisfy $\Xon \sim \p_1$ and $\Yon \sim 
\p_2$\footnote{\newnewzs{We use the term coupling to refer to both the random 
variable $(\Xon,\Yon)$ and the joint distribution.}}. 
\end{definition}

We remark that coupling can be viewed as a randomized function $f: \cX^\ns\to\cX^\ns$ such that if $\Xon\sim\p$, then $\Yon=f(\Xon)\sim\q$. Note that $\Xon$, and $\Yon$ are not necessarily independent. 
\begin{example}
\label{exm:coin}
Let $\Bern{b_1}$ be Bernoulli distributions with bias $b_1$. Let $\p_1$, and $\p_2$ be distributions over $\{0,1\}^\ns$ obtained by $\ns$ \emph{i.i.d.} samples from $\Bern{0.5}$, and $\Bern{0.5+\dist}$ respectively, with $0<\dist<0.5$. Let $\Xon$ be distributed according to $\p$. A sequence $\Yon$ can be generated as follows: If $X_i=1$, then $Y_i=1$. If $X_i=0$, we flip another coin with bias $2\dist$, and let $Y_i$ be the output of this coin. Repeat the process independently for each $i$, such that the $Y_i$'s are all independent of each other. Then $\probof{Y_i=1} = 0.5 +(1-0.5) \cdot 2\dist = 0.5+\dist$, and $\Yon$ are distributed according to $\p_2$.
\end{example}

Our lower bounds are based on the following observation. If there is a coupling $(\Xon,\Yon)$ between distributions $\p_1$ and $\p_2$ over $\cX^\ns$ with $\expectation{\ham{\Xon}{\Yon}}=D$, then a draw from $\p_1$ can be converted to a draw from $\p_2$ by changing $D$ coordinates in expectation. 
\newzs{By the group property of differential privacy (Lemma~\ref{lem:group}), roughly speaking, for 
any $(\eps,\delta)$-DP estimator $\hat{\theta}$, it must satisfy $$\forall 
S\subseteq \Theta, \probofsub{ 
		\newnewzs{\Xon\sim \p_1}}{\hat\theta(\Xon) \in S} \le e^{D\eps}\cdot 
		\probofsub{ 
		\newnewzs{\Yon\sim \p_2}}{\hat\theta(\Yon) \in S}+ \delta D e^{\eps(D-1)}.$$
	Hence, if there exists an algorithm that distinguishes between $\p_1$ and $\p_2$ reliably, $D$ must be large, i.e., $D = \Omega\Paren{\frac1{\eps+\delta}}$.}

\subsection{Organization}

In Section~\ref{sec:dp-lecam},~\ref{sec:dp-fano}, and~\ref{sec:dp-assouad} we state the privatized versions of Le Cam, Fano, and Assouad's method respectively. In Section~\ref{sec:related}, we compare them with the classic lower bound tools in differential privacy. Finally, we present their proofs in Section~\ref{sec:proofs}.

\section{DP Le Cam's Method}
\label{sec:dp-lecam}
\newhz{
Le Cam's method (Lemma 1 of~\cite{Yu97}) is widely used to prove lower bounds for composite hypothesis testings such as uniformity testing~\cite{Paninski08}, density estimation~\cite{Yu97, ray2016cam}, and estimating functionals of distributions~\cite{JiaoVHW15, WuY16, polyanskiy2019dualizing}.}


We use the expected Hamming distance between couplings of distributions in the two classes to obtain the following extension of Le Cam's method with $(\eps, \delta)$-DP. For the hypothesis testing problem described above, let $\text{co}(\cP_i)$ be the convex hull of distributions in $\cP_i$, which are also families of distributions over $\cX^\ns$.

\begin{restatable}[$(\eps, \delta)$-DP Le Cam's method]{theorem}{dplecam}
	\label{thm:le_cam}
	Let $\p_1\in\text{co}(\cP_1)$ and $\p_2\in\text{co}(\cP_2)$. Let $(\Xon,\Yon)$ be a coupling between $\p_1$ and $\p_2$ with $D=\expectation{\ham{\Xon}{\Yon}}$. Then for $\eps\ge0, \delta\ge0$, any $(\eps,\delta)$-differentially private  hypothesis testing algorithm $\hat{\theta}$ must satisfy 
	\begin{align}
	P_e(\hat{\theta}, \cP_1, \cP_2)  \ge  \frac12 \max \left\{  1-\dtv{p_1}{p_2}, 0.9 e^{-10 \eps D} - 10 D \delta\right\},\label{eqn:le-cam}
	\end{align}
	{where $\dtv{p_1}{p_2} := \sup_{A \subseteq \cX^\ns} 
	\Paren{p_1(A)-p_2(A)}=\frac12\|\p_1 -  \p_2\|_1$ is the total 
	variation (TV) distance of $p_1$ and $p_2$.}
\end{restatable}
The first term here is the original Le Cam's result~\cite{LeCam73, LeCam86, Yu97, Canonne15} and the second term is a lower bound on the additional error due to privacy. \newzs{Note that the second term increases when $D$ decreases. Choosing $\p_1, \p_2$ with small $D$ makes the RHS of~\eqref{eqn:le-cam} large, hence giving better testing lower bounds.} {A similar result (Theorem 1 in~\cite{AcharyaSZ18})}, along with a suitable coupling was used in~\cite{AcharyaSZ18} to obtain the optimal sample {complexity of testing discrete distributions}. We defer the proof of this theorem to Section~\ref{sec:le}.


\section{DP Fano's Inequality}
\label{sec:dp-fano}
\newzs{Theorem~\ref{thm:le_cam} (DP Le Cam's method) characterizes lower bounds for binary hypothesis testing. In estimation problems with multiple parameters, it is common to reduce the problem to a multi-way hypothesis testing problem. The following theorem, proved in Section~\ref{sec:fano_pro}, provides a lower bound on the risk of multi-way hypothesis testing under $\eps$-DP.}

\begin{restatable} [$\eps$-DP Fano's inequality] {theorem}{dpfano}
	\label{thm:dp_fano}
	Let $\cV=\{\p_1, \p_2,...,\p_M\}\subseteq \cP$ such that {for all $i\ne j$,} 
	\begin{enumerate}[label=(\alph*)]
		\item  $\ell \Paren{\theta(\p_i),\theta(\p_j)} \ge \alpha$,
		\item $D_{KL} \Paren{\p_i,\p_j} \le \beta$,
		\item \label{item:dham} there exists a coupling $(\Xon, \Yon)$ between $\p_i$ and $\p_j$ such that $\expectation{\ham{\Xon}{\Yon}} \le D$, then
	\end{enumerate}
	\begin{align} \label{eqn:fano_result}
		R(\cP, \ell, \eps) \ge \max \Bigg\{ & \frac{\alpha}{2} \left(1 - \frac{\beta + \log 2}{\log M}\right), 0.4\alpha \min\left\{1, \frac{M}{e^{10\eps D}}\right\} \Bigg\}. 
	\end{align}
\end{restatable}

\newzs{Under pure DP constraints, Theorem~\ref{thm:dp_fano} extends Theorem~\ref{thm:le_cam} to the multiple hypothesis case.} Non-private Fano's inequality (e.g., Lemma 3 of~\cite{Yu97}) requires only conditions $(a)$ and $(b)$ and provides the first term of the risk bound above. {Now, if we consider the second term, which is the additional cost due to privacy, we would require $\exp(10\eps D)\ge M$, \ie $D \ge {\log M}/{(10\eps)}$ to achieve a risk less than $0.4\alpha$.} Therefore, for reliable estimation, the expected Hamming distance between any pair of distributions cannot be too small. 

Theorem~\ref{thm:dp_fano} ($\eps$-DP Fano's inequality) can also be seen as a probabilistic generalization of the classic packing lower bound~\cite{Vadhan17}. The packing argument, with its roots in database theory, considers inputs to be deterministic datasets whose pairwise Hamming distances are bounded with probability one, while Theorem~\ref{thm:dp_fano} considers randomly generated datasets whose Hamming distances are bounded in expectation. This difference makes Theorem~\ref{thm:dp_fano} better suited for proving lower bounds for statistical estimation problems. We discuss this difference in details in Section~\ref{sec:related}.

\smallskip
\noindent\textbf{Remark.} {Theorem~\ref{thm:dp_fano} is a bound on the risk for pure differential privacy ($\delta=0$). Our proof extends to $(\eps,\delta)$-DP for \newhz{$\delta = O \Paren{\frac1M}$},  which is not sufficient to establish meaningful bounds since in most problems $M$ will be chosen to be exponential in the problem parameters. To circumvent this difficulty, in the next section we provide a private analogue of Assouad's method, which also works for $(\eps,\delta)$-DP.}




\section{DP Assouad's Method}
\label{sec:dp-assouad}
{Our next result is a private version of Assouad's lemma (Lemma 2 of~\cite{Yu97}, and~\cite{Assouad83}). Recall that $\cP$ is a set of distributions over $\cX^\ns$. Let $\cV\subseteq\cP$ be} a set of distributions indexed by the hypercube $ \cE_{\ab} := \{\pm 1\}^{\ab}$, and the loss  $\ell$ is such that
\begin{align}
\forall u,v \in \cE_{\ab}, \ell(\theta(\p_u), \theta(\p_v)) \ge 2\tau \cdot \sum_{i=1}^{\ab} \mathbb{I}\Paren{u_i\neq v_i}.
\label{eqn:assouad-loss}
\end{align}

{Assouad's method provides a lower bound on the estimation risk for distributions in $\cV$, which is a lower bound for $\cP$. For each coordinate $i\in[k]$, consider the following mixture distributions obtained by averaging over all distributions with a fixed value at the $i$th coordinate,}
\[
	\p_{+i}  =  \frac{2}{|\cE_{k}|} \sum_{e \in \cE_{k}: e_i= + 1}\p_e, ~~~ \p_{-i}  =  \frac{2}{|\cE_{k}|} \sum_{e \in \cE_{k}: e_i= -1}\p_e.
\]
Assouad's lemma provides a lower bound on the risk by using~\eqref{eqn:assouad-loss} and considering the problem of distinguishing $\p_{+i}$ and $\p_{-i}$. 
{Analogously, we prove the following privatized version of Assouad's lemma by considering the minimax risk of a private hypothesis testing $\phi:\cX^\ns\to\{-1,+1\}$ between $\p_{+i}$ and $\p_{-i}$. The detailed proof is in Section~\ref{sec:private_assouad}.
}
\begin{restatable} [DP Assouad's method] {theorem} {dpassouad}
	\label{thm:assouad} 
{$\forall i \in [k]$, let $\phi_i:\cX^\ns\to\{-1,+1\}$ be a binary classifier.
	\begin{align*}
	R(\cP, \ell, \eps, \delta) \ge \frac{\tau}{2} \cdot \sum_{i=1}^{\ab} \min_{\phi_i \text{ is $(\eps, \delta)$-DP}}   ( \probofsub{ \Xon\sim \p_{+i} }{ \phi_i(\Xon) \neq 1 } 
	+ \probofsub{\Xon\sim \p_{-i} }{\phi_i(\Xon) \neq -1 } ). \nonumber
	\end{align*}
	Moreover, if $\forall i\in [\ab]$, there exists a coupling $(\Xon,\Yon)$ between $\p_{+i}$ and $\p_{-i}$ with $\expectation{\ham {\Xon} {\Yon}} \le D$,
	\begin{equation} \label{eqn:assouad}
	R(\cP, \ell, \eps, \delta) \ge \frac{\ab \tau}{2} \cdot \Paren{0.9 e^{-10 \eps D} - 10 D \delta}. 
	\end{equation} }
\end{restatable}

\noindent The first bound is the classic Assouad's Lemma and~\eqref{eqn:assouad} is the loss due to privacy constraints. Once again note that~\eqref{eqn:assouad} grows with decreasing $D$. \newzs{Compared to Theorem~\ref{thm:dp_fano} (DP Fano's inequality), Theorem~\ref{thm:assouad} works under $(\eps, \delta)$-DP, which is a less stringent privacy notion.}

\section{Related and Prior Work}
\label{sec:related}
Several methods have been proposed in the literature to prove lower bounds {under DP constraints}. These include packing argument~\cite{HardtT10,Vadhan17}, fingerprinting~\cite{BunNSV15, SteinkeU17a, SteinkeU15, BunSU17, BunUV18, KamathLSU18} and coupling based arguments~\cite{AcharyaSZ18, KarwaV18}.

\smallskip
\noindent\textit{Binary Testing and Coupling.} Coupling based arguments have been recently used to prove lower bounds for binary hypothesis testing, including the independent works of~\cite{AcharyaSZ18, KarwaV18}.~\cite{AcharyaSZ18} establishes a very similar result to Theorem~\ref{thm:le_cam} and uses it to obtain lower bounds for a composite hypothesis testing problem on discrete distributions.~\cite{KarwaV18} proves a similar result for simple hypothesis testing and uses it to lower bound the sample complexity of estimating the mean of a one-dimensional Gaussian distribution. 
For both papers, the coupling argument implies that it is hard to 
differentially privately distinguish between two distributions, supposing 
there exists a coupling with small expected Hamming distance. This 
method can be viewed as another form of private Le Cam's method 
(Theorem~\ref{thm:le_cam}) and it can only be applied where binary 
hypothesis testing is involved. \newnewzs{\cite{barber2014privacy} also 
uses a private version of Le 
Cam's method to prove lower bounds for differentially private mean 
estimation. However, instead of the expected Hamming distance between 
any couplings, their method only depends on the TV distance between the 
distributions, which corresponds to the naive independent coupling.}
\newzs{\cite{CanonneKMSU19} uses coupling bounds in~\cite{AcharyaSZ18} to derive instance-optimal bounds for simple binary hypothesis testing under pure DP. They consider a coupling only for symbols whose likelihood ratio between the two hypothesis distributions is large, which results in better bounds for certain instances. The argument only considers pure DP and the case where samples are i.i.d. generated while Theorem~\ref{thm:le_cam} and~\cite{AcharyaSZ18} can handle approximate DP and arbitrary distributions (e.g. mixtures of i.i.d. distributions) .}
\smallskip

%
%

\noindent\textit{Pure DP Estimation and Packing.} Packing argument~\cite{HardtT10, Vadhan17} is a geometric approach to prove lower bounds for estimation under pure DP. We state a form of the packing bound below:
\begin{lemma}[Packing lower bound~\cite{Vadhan17}] \label{thm:packing}
	Let $\cV=\{x_1, x_2,...,x_M\}$ be a set of $M$ datasets over $\cX^n$.  For any pair of datasets $x_i$ and $x_j$, we have $\ham{x_i}{x_j} \le d$. Let $ \{S_{i}\}_{i \in [M]}$ be a collection of disjoint subsets of $\cS$. If there exists an $\eps$-DP algorithm $\cA : \cX^{\ns} \to \cS$ such that for every $i \in [M]$, $\probof{\cA(x_i) \in S_{i}}\ge {9}/{10}$, then
	\[
	\eps = \Omega \Paren{\frac{\log M}{d}}.
	\]
\end{lemma}

Our $\eps$-DP Fano's inequality (Theorem~\ref{thm:dp_fano}) can be viewed as a probabilistic packing argument which generalizes Lemma~\ref{thm:packing} to the case where $\cV$ consists of distributions over $\cX^n$ instead of deterministic datasets. The distances between distributions are measured in the minimum expected hamming distance between random datasets generated from a coupling between the distributions. Lemma~\ref{thm:packing} can be obtained from $\eps$-DP Fano's inequality by setting the distributions to be point masses over $\cX^n$. 

Note that $d$ in Lemma~\ref{thm:packing} is an upper bound on the worst-case Hamming distance while $D$ is a bound on the expected Hamming distance and therefore $D \le d$. In statistical applications where $D \ll d$, we can obtain stronger lower bounds by replacing $d$ with $D$. For example, in the $\ab$-ary distribution estimation problem, a naive application of the packing argument can only give a lower bound of $\ns = \Omega\Paren{{\ab \log \Paren{1/\dist}}/{\eps}}$ instead of the optimal $\ns = \Omega\Paren{{\ab}/{\dist \eps}}$ lower bound, where there is an exponential gap in the parameter $1/\dist$.

\smallskip

\noindent\textit{Approximate DP and Fingerprinting.} Fingerprinting~\cite{SteinkeU15, BunNSV15, DworkSSUV15, SteinkeU17a,  BunSU17, BunUV18, KamathLSU18, CaiWZ19} is a versatile lower bounding method for $(\eps,\delta)$-DP for $\delta=O(1/\ns)$. It has been used to prove lower bounds for several problems, including attribute mean estimation in databases~\cite{SteinkeU17a}, lower bounds on the number of online statistical queries~\cite{BunSU17}, and private selection problem~\cite{SteinkeU17b}.~\cite{KamathLSU18} uses fingerprinting to prove lower bounds on estimating Bernoulli product distributions and Gaussian distributions. We believe fingerprinting and DP Assouad's lemma are both powerful tools for proving lower bounds under approximate DP. In estimating Gaussian distributions, fingerprinting provides strong lower bounds under approximate DP, whereas private Assouad's method gives an additional polynomial blow-up compared to fingerprinting. 
However, for discrete distribution estimation, private Assouad's method provides tight lower bounds, and we do not know how to obtain such bounds from the fingerprinting lemma.



\cite{DuchiJW13} derives analogues of Le Cam, Assouad, and Fano in the local model of differential privacy, and uses them to establish lower bounds for several problems under local differential privacy.~\cite{AcharyaCT19, acharya2020interactive} proves lower bounds for various testing and estimation problems under local differential privacy using a notion of chi-squared contractions based on Le Cam's method and Fano's inequality. 

%
%
%
%
%
%
\section{Proof of Theorems}
\label{sec:proofs}
\subsection{Proof of DP Le Cam's Method (Theorem~\ref{thm:le_cam})}
\label{sec:le}
The proof technique is similar to the proof of coupling lemma in~\cite{AcharyaSZ18}. However, we directly characterize the error probability in Theorem~\ref{thm:le_cam}, which we then use to prove Theorem~\ref{thm:assouad} (DP Assouad's method).

\dplecam*

\begin{proof}
{From the definition of hypothesis testing,} 
	\[
	P_e(\hat \theta, \cP_1, \cP_2) \ge \frac12 \Paren{ \probofsub{ \Xon\sim \p_1 }{ \hat \theta(\Xon) \neq \p_1 } + \probofsub{\Xon\sim \p_2 }{\hat \theta(\Xon) \neq \p_2 }}.
	\]
	
{The first term in Theorem~\ref{thm:le_cam} follows from the classic Le Cam's method (Lemma 1 in~\cite{Yu97}). For the second term, let $(\Xon, \Yon)$ be distributed according to a coupling of $\p_1$ and $\p_2$ with $\expectation{\ham{\Xon}{\Yon}} \le D$. 
	By Markov's inequality, 
	$\probof{\ham{\Xon}{\Yon}>10D}<0.1$.
	Let $\xon$ and $\yon$ be the realization of $\Xon$ and $\Yon$.  $W := \{ (\xon,\yon) \in \cX^n \times \cX^n| \ham{\xon}{\yon} \le 10D \}$ be the set of pairs of realizations with Hamming distance at most $10D$.} Then we have 
	\begin{align}
	\probof{\alg\Paren{\Xon}=\p_2}  = & \sum_{\xon,\yon} \probof{\Xon = \xon, \Yon = \yon} \cdot \probof{\alg\Paren{\xon}=\p_2} \nonumber\\
	\ge &  \sum_{(\xon,\yon)\in W} \probof{\Xon = \xon, \Yon = \yon} \cdot \probof{\alg\Paren{\xon}=\p_2}.
	\end{align}
Let  $\beta_1=\probof{\alg\Paren{\Xon}=\p_2} $, so we have 
	$$ \sum_{(\xon,\yon)\in W} \probof{\Xon = \xon, \Yon = \yon}  \cdot\probof{\alg\Paren{\xon}=\p_2} \le \beta_1$$
	
	Next, we need the following group property of differential privacy.
	
	\begin{lemma} \label{lem:group}
		Let $\alg$ be a $(\eps,\delta)$-DP algorithm, then for sequences $\xon$, and $\yon$ with $\ham{\xon}{\yon} \le t$, and $\forall S$, 
		$		\probof{\hat \theta(\xon)\in S}\le e^{t\eps}\cdot\probof{\hat \theta(\yon)\in S}+ \delta t e^{\eps (t-1)}$. 
	\end{lemma}
	
	\noindent By Lemma~\ref{lem:group}, and $\probof{\ham{\Xon}{\Yon}>10D}<0.1$, let $\probof{\alg\Paren{\Yon}=\p_2}=1- \beta_2$,
	\begin{align}
	1- \beta_2=&  
	\sum_{(\xon,\yon)\in W} \probof{\xon,\yon}  \cdot \probof{\alg\Paren{\yon}=\p_2}+ 
	\sum_{(\xon,\yon)\notin W} \probof{\xon,\yon} \cdot \probof{\alg\Paren{\yon}=\p_2} \nonumber \\
	\leq& \sum_{(\xon,\yon)\in W} \probof{\xon,\yon}  \cdot \Paren{ e^{\eps \cdot 10D} \probof{\alg\Paren{\xon}=\p_2} + 10 D\delta \cdot e^{\eps \cdot 10(D-1)}} +0.1	 \nonumber \\
	\leq& \beta_1 \cdot e^{\eps\cdot 10 D} +  10 D\delta \cdot e^{\eps \cdot 10D} + 0.1. \nonumber
	\end{align}
	Similarly, we get
	\begin{align}
	1- \beta_1 \le  \beta_2 \cdot e^{\eps\cdot 10 D} +  10 D\delta \cdot e^{\eps \cdot 10D} + 0.1. \nonumber
	\end{align}
	
	Adding the two inequalities and rearranging terms,
	\[
	\beta_1 + \beta_2 \ge \frac{1.8 - 20D\delta e^{\eps \cdot 10D} }{1 + e^{\eps\cdot 10 D}} \ge 0.9 e^{-10 \eps D} - 10 D \delta.
	\]
\end{proof}

\subsection{Proof of Private Fano's Inequality (Theorem~\ref{thm:dp_fano})}
\label{sec:fano_pro}
In this section, we prove $\eps$-DP Fano's inequality (Theorem~\ref{thm:dp_fano}), restated below. 

\dpfano*

The proof is based on the observation that if you can change a sample from $\p_i$ to $\p_j$ by changing $D$ coordinates in expectation, then an algorithm that  algorithm that correctly outputs a sample as from $\p_i$ has to output $\p_j$ with probability roughly $e^{-\eps D}$. With a total of $M$ distributions in total, we show that the error probability is large as long as $\frac{M}{e^{\eps D}}$ is large.

\begin{proof}
{The first term in~\eqref{eqn:fano_result} follows from the non-private Fano's inequality (Lemma~3 in~\cite{Yu97}).}
{For an observation $\Xon\in\cX^\ns$,} 
\[
	\hat{\p}(\Xon) := \arg \min_{\p \in \cV} \ell \Paren{\theta(\p),  \hat \theta(\Xon)},
\]
{be the distribution in $\cP$ closest in parameters to an $\eps$-DP estimate $\hat \theta(\Xon)$. Therefore, $\hat{\p}(\Xon)$ is also $\eps$-DP. By the triangle inequality,
	\[
	\ell\Paren{\theta(\hat \p), \theta(\p)} \le \ell \Paren{\theta(\hat \p), \hat \theta(\Xon)}  + \ell \Paren{ \theta(\p), \hat \theta(\Xon)} \le 2\ell \Paren{\theta(\p), \hat \theta(\Xon)}.
	\]}
	Hence,
	\begin{align} 
	\max_{\p \in \cP} \EE_{\Xon \sim \p} \left[\ell(\hat \theta (\Xon),\theta (\p)) \right]  
	\ge \max_{\p \in \cV} \EE_{\Xon \sim \p} \left[\ell(\hat \theta (\Xon),\theta (\p)) \right]
	&\ge \frac{1}{2}\max_{\p \in \cV} \EE_{\Xon \sim \p} \left[\ell(\theta(\hat \p) ,\theta (\p)) \right]  \nonumber \\
	& \ge \max_{\p \in \cV} \frac{\alpha}{2}\probofsub{\Xon \sim \p}{\hat{\p} (\Xon) \neq \p} \nonumber \\
	&\ge \frac{\alpha}{2M}\sum_{\p \in \cV} \probofsub{\Xon \sim \p}{\hat{\p} (\Xon) \neq \p}.
	\label{eqn:reduction}
	\end{align}
	Let $\beta_i = \probofsub{\Xon \sim \p_i}{\hat{\p} (\Xon) \neq \p_i}$ be the probability that $\hat p(\Xon)\ne\p_i$ when the underlying distribution generating $\Xon$ is $\p_i$. For $\p_i, \p_j \in \cV$, let $(\Xon,\Yon)$ be the coupling in condition $(c)$. By Markov's inequality $\probof{\ham{\Xon}{\Yon}>10D}<1/{10}.$

{Similar to the proof of Theorem~\ref{thm:le_cam} in the previous section, let $W := \{ (\xon,\yon) | \ham{\xon}{\yon} \le 10D \}$ and $\probof{\xon,\yon}  = \probof{\Xon = \xon, \Yon = \yon}$ . Then
	\begin{align}
	1- \beta_j = \probof{\hat{\p}\Paren{\Yon} =\p_j}  \leq  \sum_{(\xon, \yon)\in W} \probof{\xon,\yon}  \cdot \probof{\hat{\p} \Paren{\yon}  = \p_j} + \sum_{(\xon, \yon)\notin W} \probof{\xon,\yon}. \nonumber
	\end{align}}
Therefore, 
	\[ \sum_{(\xon, \yon)\in W} \probof{\xon,\yon}  \cdot \probof{\hat{\p} \Paren{\yon}  = \p_j} \ge 0.9 - \beta_j.\] Then, we have
	\begin{align}
	\probof{\hat{\p} \Paren{\Xon}  = \p_j} & \ge  \sum_{(\xon,\yon)\in W} \probof{\xon, \yon} \cdot \probof{\hat{\p} \Paren{\xon}  = \p_j}\nonumber \\ 
	&\ge \sum_{(\xon,\yon)\in W} \probof{\xon, \yon} e^{-10 \eps D} \probof{\hat{\p} \Paren{\yon}  = \p_j} \label{eqn:step-dp} \\
	& \ge ( 0.9 - \beta_j)  e^{-10 \eps D},\nonumber
	\end{align}
	where~\eqref{eqn:step-dp} uses that $\hat p$ is $\eps$-DP and $\ham{\xon}{\yon} \le 10D$. Similarly, for all $j' \neq i$, 
	\[
	\probof{\hat{\p} \Paren{\Xon}  = \p_{j'}} \ge ( 0.9 - \beta_{j'})  e^{-10 \eps D}.
	\]
	Summing over $j'\ne i$, we obtain
	\begin{align}
	\beta_i &= \sum_{j' \neq i} \probof{\hat{\p} \Paren{\Xon} = \p_{j'}}  \ge \Paren{ 0.9(M-1) - \sum_{j' \neq i}\beta_{j'} } e^{-10 \eps D}.\nonumber
	\end{align}
	Summing over $i \in [M]$,
	\[
	\sum_{i \in [M]} \beta_i \ge  \Paren{ 0.9 M (M-1) - (M - 1)\sum_{i \in [M]} \beta_{i} } e^{-10 \eps D}.
	\]
	Rearranging the terms
	\[
	\sum_{i \in [M]} \beta_i \ge \frac{0.9 M (M - 1)}{M - 1 + e^{10 \eps D}} \ge 0.8M \min\left\{1,  \frac{M}{e^{10 \eps D}}\right\}.
	\]
	Combining this with~\eqref{eqn:reduction} completes the proof.
\end{proof}

\subsection{{Proof of Private Assouad's Method (Theorem~\ref{thm:assouad})}}
\label{sec:private_assouad}
\newzs{We restate the theorem below and the notions are the same as defined in Section~\ref{sec:dp-assouad}.}
\dpassouad*
\begin{proof}
{The first part is from the non-private Assouad's lemma, which we include here for completeness. Let $\p\in\cV \subset \cP$ and $\Xon\sim\p$. For an estimator $\hat \theta(\Xon)$, consider an estimator $\he(\Xon) = \arg \min_{e \in \cE_{\ab}} \ell \Paren{\hat \theta(\Xon), \theta(\p_e)}$.} Then, by the triangle inequality,
		\[
	\ell \Paren{ \theta(\p_{\he}), \theta(\p) } \le \ell \Paren{ \hat \theta, \theta(\p_{\he})}  + \ell \Paren{ \hat \theta, \theta(\p) } \le 2 \ell \Paren{ \hat \theta, \theta(\p) }.
	\]
	Hence,
	\begin{align} \label{eqn:reduce_index}
	R(\cV, \ell, \eps, \delta) &= \min_{\hat \theta \text{ is }(\eps,\delta)-DP}\ \max_{\p \in \cV}\ \EE_{\Xon \sim \p}\left[\ell(\hat \theta (\Xon),\theta(p))\right] \\
	&\ge \frac{1}{2} \min_{\he \text{ is }(\eps,\delta)-DP}\ \max_{\p \in \cV}\ \EE_{\Xon \sim \p}\left[\ell(\theta(\p_{\he(\Xon)}),\theta(p))\right].
	\end{align}
	
\noindent For any $(\eps, \delta)$-DP index estimator $\he$, and by~\eqref{eqn:assouad-loss},
	\begin{align}
	\max_{\p \in \cV }\expectationsub{\Xon \sim \p} {\ell(\theta(\p_{\he}), \theta(\p))} \geq \frac1{|\cE_{\ab}|} \sum_{ e \in \cE_{k}} \expectationsub{\Xon \sim \p_{e}} {\ell(\theta(\p_{\he}),\theta(\p_{e}) )} 
	\geq  \frac{2\tau}{|\cE_{\ab}|}  \sum_{i=1}^{k}  \sum_{ e \in \cE_{k}} \probof{\he_i \neq {e}_i| E = e }. \nonumber 
	\end{align}
	%
	
\noindent For each $i$, we divide $ \cE_{k} =  \{\pm1\}^{k}$ into two sets according to the value of $i$-th position,
	\begin{align}
	\max_{\p \in \cV }\expectationsub{\Xon \sim \p} {\ell(\theta(\p_{\he}), \theta(\p))} & \geq  \frac{2\tau}{|\cE_{\ab} |} \sum_{i=1}^{k}  \Brack {\sum_{e:e_i=1} \probof{ \he_i\neq 1| E=e }+\sum_{e:e_i=-1} \probof{\he_i \neq -1  | E = e} }\nonumber\\
	&= \tau \cdot \sum_{i=1}^{\ab}  \Paren{ \probofsub{ \Xon\sim \p_{+i} }{ \he_i\neq 1 } + \probofsub{\Xon\sim \p_{-i} }{\he_i\neq -1 }}  \nonumber \\
	&\ge \tau \cdot \sum_{i=1}^{\ab} \min_{ \phi_i: \phi_i \text{ is DP}}  \Paren{ \probofsub{ \Xon\sim \p_{+i} }{ \phi_i(\Xon) \neq 1 } + \probofsub{\Xon\sim \p_{-i} }{\phi_i(\Xon) \neq -1 }}.  \nonumber
	\end{align}
	Combining with~\eqref{eqn:reduce_index}, we have
	\[
	R(\cP, \ell, \eps, \delta)  \ge R(\cV, \ell, \eps, \delta) \ge  \frac{\tau}{2} \cdot \sum_{i=1}^{\ab} \min_{ \phi_i: \phi_i \text{ is DP}}  \Paren{ \probofsub{ \Xon\sim \p_{+i} }{ \phi_i(\Xon) \neq 1 } + \probofsub{\Xon\sim \p_{-i} }{\phi_i(\Xon) \neq -1 }},
	\]
	proving the first part.

	For the second part. Note that for each $i \in [\ab]$, the summand above is the error probability of hypothesis testing between the mixture distributions $\p_{+i}$ and $\p_{-i}$. Hence, using Theorem~\ref{thm:le_cam}, 
	\[
	R(\cP, \ell, \eps, \delta) \ge \frac{\ab \tau}{2} \cdot \Paren{0.9 e^{-10 \eps D} - 10 D \delta}.
	\]
\end{proof}

\chapter{Private Identity Testing and Closeness Testing of Discrete Distributions}
\label{cha:testing}
\section{Introduction}


Testing whether observed data conforms to an underlying model is a fundamental scientific problem. 
In the past two decades, there has been a lot of work from the computer science, information theory, and statistics community on various distribution testing problems in the non-asymptotic (small-sample) regime, where the domain size $\ab$ could be potentially larger than $\ns$. Here the goal is to characterize the minimum number of samples necessary (sample complexity) as a function of the domain size $k$, and the other parameters.

Meanwhile, preserving the privacy of individuals who contribute to the data samples has emerged as one of the critical challenges in statistical inference. Hypothesis testing is definitely one of the most compelling examples of the need for private statistics, since it forms the lifeblood of the scientific method, 
and involves huge amounts of highly sensitive data. Without a properly designed mechanism, statistical processing might divulge sensitive information about the data. 

A natural and interesting question when designing a differentially private algorithm is to understand how the data requirement grows to ensure privacy, along with the same accuracy. In this chapter, we study the sample size requirements for differentially private discrete distribution testing. 


\subsection{Results and Techniques}
We consider two fundamental statistical tasks for testing distributions over $[\ab]$: (i) identity testing, where given sample access to an unknown distribution $\p$, and a known distribution $\q$, the goal is to decide whether $\p=\q$, or $\dtv{\p}{\q}\ge\dist$, and (ii) closeness testing, where given sample access to   unknown distributions $\p$, and $\q$, the goal is to decide whether $\p=\q$, or $\dtv{\p}{\q}\ge\dist$. (See Section~\ref{sec:testing_preliminaries} for precise statements of these problems).
Given differential privacy constraints $(\eps,\delta)$, we provide $(\eps,\delta)$-differentially private algorithms for both these tasks. For identity testing, our bounds are optimal up to constant factors for all ranges of $\ab, \alpha, \eps, \delta$, and for closeness testing the results are tight in the small sample regime where $\ns=O(\ab)$. Our upper bounds are based on various methods to privatize the previously known tests. A critical component is to design and analyze test statistic that have low sensitivity (see Definition~\ref{def:sensitivity}), in order to preserve privacy. 

\newhz{
We first state that any $(\eps+\delta,0)$-DP algorithm is also an $(\eps, \delta)$ algorithm.~\cite{CaiDK17} showed that for testing problems, any $(\eps, \delta)$ algorithm will also imply a $(\eps+ c\delta, 0)$-DP algorithm. Please refer to Lemma~\ref{lm:epsdelta} and Lemma~\ref{lm:epsdelta2} for more detail. Therefore, for all the problems, we simply consider $(\eps,0)$-DP algorithms \newhz{($\eps$-DP)}, and we can replace $\eps$ with $(\eps+\delta)$ in both the upper and lower bounds without loss of generality. 
}

We describe our other results below. A summary of the results is presented in Table~\ref{fig:badass:table}, which we now describe in detail.
\begin{enumerate}[leftmargin=*]
\item
{\bf Reduction from identity to uniformity.} We reduce the problem of $\eps$-DP identity testing of distributions over $[\ab]$ to $\eps$-DP uniformity testing over distributions over $[6\ab]$. Such a reduction, without privacy constraints was shown in~\cite{Goldreich16}, and we use their result to obtain a reduction that also preserves privacy, with at most a constant factor blow-up in the sample complexity. This result is given in Theorem~\ref{thm:unif-identity}.  

\item
{\bf Identity Testing.} It was recently shown that $O(\frac{\sqrt\ab}{\alpha^2})$~\cite{Paninski08, ValiantV14, DiakonikolasKN15a, AcharyaDK15} samples are necessary and sufficient for identity testing without privacy constraints. The statistic used in these papers are variants of chi-squared tests, which could have a high global sensitivity.

Given the reduction from identity to uniformity, it suffices to consider uniformity testing. We consider the test statistic studied by~\cite{DiakonikolasGPP17} which is simply the distance of the empirical distribution to the uniform distribution. This statistic also has a low sensitivity, and futhermore has the optimal sample complexity in all parameter ranges, without privacy constraints. In Theorem~\ref{thm:main-identity}, we state the optimal sample complexity of identity testing. The upper bounds are derived by privatizing the statistic in~\cite{DiakonikolasGPP17}.  For lower bound, we use our technique in Theorem~\ref{thm:coupling}. We design a coupling between the uniform distribution $\unifab$, and a mixture of distributions, which are all at distance $\alpha$ from $\unifab$ in total variation distance. In particular, we consider the mixture distribution used in~\cite{Paninski08}. Much of the technical details go into proving the existence of couplings with small expected Hamming distance.~\cite{CaiDK17} studied identity testing under pure differential privacy, and obtained an algorithm with complexity $O\Paren{ \frac{\sqrt \ab}{\dist^2}+ \frac{\sqrt {\ab\log\ab} } {\dist^{3/2} \eps}+ \frac {(k\log\ab)^{1/3} } {\dist^{5/3}\eps^{2/3}}}$. Our results improve their bounds significantly.

\item
{\bf Closeness Testing.}
Closeness testing problem was proposed by~\cite{BatuFRSW00}, and optimal bound of $\Theta\Paren{\max\{\frac{\ab^{2/3}}{\dist^{4/3}}, \frac{\sqrt\ab} {\dist^2}\}}$ was shown in~\cite{ChanDSS14}. In a recent work,~\cite{DiakonikolasGKPP20} proposed a statistic based on the empirical total variation distance, which we show has a small sensitivity. We privatize their algorithm to obtain the sample complexity bounds. Since closeness testing is a harder problem than identity testing, all the lower bounds from identity testing port over to closeness testing. In Theorem~\ref{thm:close-main}, we establish the optimal sample complexity bounds for DP closeness testing, which are tight up to constant factors in all parameter ranges.


\end{enumerate}

\begin{table}[htbp]
\centering
      \begin{tabular}{| c | c | }
      \hline
      {\bf Problem} & {\bf Sample Complexity Bounds}  \\ \hline
      {\bf Identity Testing} & {\bf Non-private : $\Theta\Paren{\frac{\sqrt \ab}{\alpha^2}} $ \cite{Paninski08}}\\
              & {\bf $\eps$-DP algorithms:} $O\Paren{ \frac{\sqrt \ab}{\alpha^2}+\frac{\sqrt {\ab\log k}}{\alpha^{3/2}\eps}}$~\cite{CaiDK17}\\
       & \newhz{$\Sit = \Theta\Paren{\frac{\sqrt \ab}{\alpha^2}+{\max\left\{\frac{\ab^{1/2}}{\alpha \eps^{1/2}}, \frac{\ab^{1/3}}{\alpha^{4/3}\eps^{2/3}}, \frac1{\dist \eps}\right\}}}$}[Theorem~\ref{thm:main-identity}]\\\hline
            {\bf Closeness Testing} & {\bf Non-private:} {$\Theta{ \left(\frac{\ab^{2/3}}{\dist^{4/3}}+ \frac{\ab^{1/2}}{\dist^2}\right) }$ \cite{ChanDVV14}}\\    
        & \newhz { {\bf $\eps$-DP algorithms:} }\\
              & $\Sct = \Theta\Paren{\frac{\ab^{2/3}}{\dist^{4/3}}+ \frac{\ab^{1/2}}{\dist^2}  +{{\max\left\{\frac{\ab^{1/2}}{\alpha \eps^{1/2}}, \frac{\ab^{1/3}}{\alpha^{4/3}\eps^{2/3}}, \frac1{\dist \eps}\right\}}} }$  [Theorem~\ref{thm:close-main}]\\ \hline
      \end{tabular}
    \caption{\label{fig:badass:table} \newhz{Summary of the sample complexity bounds for $\eps$-DP identity, and closeness testing. For $(\eps,\delta)$-DP algorithms, we can simply replace $\eps$ in the sample complexity by $(\eps+\delta)$.}}
\end{table}
\subsection{Related Work}

A number of papers have recently studied hypothesis testing problems under differential privacy guarantees~\cite{WangLK15, GaboardiLRV16, RogersK17}. Some works analyze the distribution of the test statistic in the asymptotic regime. The work most closely related to ours is~\cite{CaiDK17}, which studied identity testing in the finite sample regime. We mentioned their guarantees along with our results on identity testing in the previous section. 

There has been a line of research for statistical testing and estimation problems under the notion of \emph{local} differential privacy~\cite{WainwrightJD12, DuchiJW13,ErlingssonPK14, PastoreG16, KairouzBR16, WangHWNXYLQ16, YeB18, AcharyaSZ18a, Sheffet18, AcharyaCFT18}. These papers study some basic statistical problems and provide minimax lower bounds using Fano's inequality.~\cite{DiakonikolasHS15} studies structured distribution estimation under differential privacy.
Information theoretic approaches to data privacy have been studied recently using quantities like mutual information, and guessing probability to quantify privacy~\cite{Mir12, SankarRP13, CuffY16, WangYZ16, IssaW17}.

In a contemporaneous and independent work,~\cite{AliakbarpourDR18}, the authors study the same problems that we consider, and obtain the same upper bounds for the sparse case, when $\ns\le\ab$. They also provide experimental results to show the performance of the privatized algorithms. However, their results are sub-optimal for $\ns=\Omega(\ab)$ for identity testing, and they do not provide any lower bounds for the problems. Both~\cite{CaiDK17}, and~\cite{AliakbarpourDR18} consider only pure-differential privacy, which are a special case of our results.

\medskip
\subsection{Organization}
In Section~\ref{sec:testing_preliminaries}, we discuss the definitions and notations. 
Section~\ref{sec:identity} gives upper and lower bounds for identity testing, and closeness testing is studied in Section~\ref{sec:close:}. The proofs of some lemmas are given in Section~\ref{sec:testing_proof}.
\section{Preliminaries} \label{sec:testing_preliminaries}


\newhz{We first introduce the following lemmas, which state a relationship between $(\eps, \delta)$ and $\eps$-differential privacy in testing. We give a proof of Lemma~\ref{lm:epsdelta} below. And Lemma~\ref{lm:epsdelta2} follows from~\cite{CaiDK17}. }

\begin{lemma} \label{lm:epsdelta}
	Any $(\eps+\delta,0)$- differentially private algorithm is also $(\eps, \delta)$-differentially private.
\end{lemma}

\begin{proof}
	Suppose $\cA$ is a $(\eps+\delta)$-differentially private algorithm. Then for any $\Xon$ and $\Yon$ with $\ham{\Xon}{\Yon}\le 1$ and any $S\subset \text{range}(\cA)$, we have
	\begin{align*}
		\probof{\cA(\Xon)\in S} \le e^{\eps}\cdot\probof{\cA(\Yon)\in S} + (e^{\delta} - 1)\cdot e^{\eps}\probof{\cA(\Yon)\in S}. 
	\end{align*}
	If $e^{\eps}\cdot\probof{\cA(\Yon)\in S} > 1-\delta$, then $\probof{\cA(\Xon)\in S}\le1< e^{\eps}\cdot\probof{\cA(\Yon)\in S}+\delta$.
	Otherwise, $e^{\eps}\cdot\probof{\cA(\Yon)\in S} \leq 1 - \delta$. To prove $(e^{\delta}-1)\cdot e^{\eps}\cdot\probof{\cA(\Yon)\in S}<\delta$, it suffices to show $(e^{\delta} - 1)(1 - \delta) \le \delta$, which is equivalent to $e^{-\delta} \ge 1 - \delta$, completing the proof.
\end{proof}

\begin{lemma}
\label{lm:epsdelta2}
An $(\eps, \delta)$-DP algorithm for a testing problem can be converted to an $(\eps+c\delta, 0)$ algorithm for some constant $c>0$. 
\label{lem:delta-pure}
\end{lemma}
Combining these two results, it suffices to prove bounds for $(\eps,0)$-DP, and plug in $\eps$ with $(\eps+\delta)$ to obtain bounds that are tight up to constant factors for $(\eps,\delta)$-DP. 

For $x\in\RR$,
$\sigma(x) \ed \frac1{1+\exp(-x)} = \frac{\exp(x)}{1+\exp(x)}$ is the sigmoid function. The following properties follow from the definition of $\sigma$.
\begin{lemma}\label{lem:sig-cont}
\begin{enumerate}
	\item  For all $x, \gamma\in\RR$, 
$\exp(-\absv{\gamma})\le \frac{\sigma(x+\gamma)}{\sigma(x)} \le \exp(\absv{\gamma})$.
\item
Let $0<\eta<\frac12$. Suppose $x\ge \log\frac1\eta$. Then $\sigma(x)>1-\eta$.
\end{enumerate}
\end{lemma}
\paragraph{Identity Testing ({\tt IT}).} Given description of $\q\in\Delta_k$ over $[\ab]$, parameters $\alpha$, and $\ns$ independent samples $\Xon$ from unknown $\p\in\Delta_k$. $\cA$ is an $(\ab,\alpha)$-identity testing algorithm for $\q$, if
when $\p=\q$, $\cA$ outputs ``$\p=\q$'' with probability at least 0.9, and
when $\dtv{\p}{\q}\ge\alpha$, $\cA$ outputs ``$\p\ne\q$'' with probability at least 0.9.

\begin{definition}
The sample complexity of DP-identity testing, denoted $\Sit$, is the smallest $\ns$ for which there exists an $\eps$-DP algorithm $\cA$ that uses $\ns$ samples to achieve $(\ab,\dist)$-identity testing.
Without privacy concerns, $\Sitnp$ denotes the sample complexity. When $\q=\unifab$, the problem reduces to uniformity testing, and the sample complexity is denoted as $\Sut{\ab}{\dist}$.
\end{definition}
\noindent\textbf{Closeness Testing ({\tt CT}).} Given $\ns$ independent samples $\Xon$, and $\Yon$ from unknown distributions $\p$, and $\q$. An algorithm $\cA$ is an $(\ab,\alpha)$-closeness testing algorithm if 
when $\p=\q$, $\cA$ outputs $\p=\q$ with probability at least 0.9, and
when $\dtv{\p}{\q}\ge\alpha$, $\cA$ outputs $\p\ne\q$ with probability at least 0.9.
\begin{definition}
The sample complexity of DP-closeness testing, denoted $\Sct$, is the smallest $\ns$ for which there exists an $\eps$-DP algorithm $\cA$ that uses $\ns$ samples to achieve $(\ab,\dist)$-closeness testing. 
When privacy is not a concern, we denote the sample complexity of closeness testing as $\Sctnp$. 
\end{definition}
\section{Identity Testing}

 \label{sec:identity}

In this section, we prove the bounds for identity testing. Our main result is the following.

\begin{theorem}
\label{thm:main-identity}
%
%
\newhz{
\begin{center}
$\Sit = \Theta\Paren{\frac{k^{1/2}}{\dist^2}+\max\left\{\frac{\ab^{1/2}}{\alpha \eps^{1/2}}, \frac{\ab^{1/3}}{\alpha^{4/3} \eps^{2/3}}, \frac1{\dist\eps}\right\}}$.
\end{center}}
Or we can write it according to the parameter range,
\newhz{
\[
    \Sit =
    \begin{cases}
                                   \Theta\Paren{\frac{\sqrt k}{\alpha^2} + \frac{\ab^{1/2}}{\alpha \eps^{1/2}}}, & \text{if $k = \Omega \Paren{\frac1{\alpha^4}}$ and $k = \Omega \Paren{\frac1{\alpha^2 \eps}}$,} \\
              		  \Theta\Paren{\frac{\sqrt k}{\alpha^2} +\frac{\ab^{1/3}}{\alpha^{4/3} \eps^{2/3}}}, & \text{if $k = \Omega \Paren{\frac{\alpha}{\eps}}$ and $k = O \Paren{\frac1{\alpha^4}+\frac1{\alpha^2 \eps} }$,}\\
              		               \Theta\Paren{\frac{\sqrt k}{\alpha^2} + \frac1{\alpha \eps}}, & \text{if  $k = O \Paren{\frac{\alpha}{\eps}}$.}\\ 
    \end{cases}
\] }
\end{theorem}
Our bounds are tight up to constant factors in all parameters. To get the sample complexity for $(\eps,\delta)$-differential privacy, we can simply replace $\eps$ by $(\eps+\delta)$.

 In Theorem~\ref{thm:unif-identity} we will show a reduction from identity to uniformity testing under pure differential privacy. Using this, it will be enough to design algorithms for uniformity testing, which is done in Section~\ref{sec:upper-uniformity}.
 
 Moreover since uniformity testing is a special case of identity testing, any lower bound for uniformity will port over to identity, and we give such bounds in Section~\ref{sec:lower-identity}.

\subsection{Uniformity Testing Implies Identity Testing}
\label{sec:iden-to-unif}

The sample complexity of testing identity of any distribution is $O(\frac{\sqrt{\ab}}{\alpha^2})$, a bound that is tight for the uniform distribution. Recently~\cite{Goldreich16} proposed a scheme to reduce the problem of testing identity of distributions over $[\ab]$ for total variation distance $\alpha$ to the problem of testing uniformity over $[6k]$ with total variation parameter $\alpha/3$. In other words, they show that $\Sitnp\le \Sutnp{6\ab}{\dist/3}$. 
Building on~\cite{Goldreich16}, we prove that a similar bound also holds for differentially private algorithms.
\begin{theorem}
\label{thm:unif-identity}
$\Sit \le \Sut{6\ab}{\dist/3}.$
\end{theorem}

\begin{proof} We first briefly describe the essential components of the construction of~\cite{Goldreich16}. Given an explicit distribution $\q$ over $[\ab]$, there exists a randomized function $F_{\q}:[\ab]\to[6\ab]$  such that if $X\sim\q$, then $F_{\q}(X)\sim \unifabs{6\ab}$, and if $X\sim\p$ for a distribution with $\dtv{\p}{\q}\ge\dist$, then the distribution of $F_{\q}(X)$ has a total variation distance of at least $\dist/3$ from $\unifabs{6\ab}$. 
Given $s$ samples $X_1^s$ from a distribution $\p$ over $[\ab]$. Apply $F_\q$ independently to each of the $X_i$ to obtain a new sequence $Y_1^s = F_\q(X_1^s)\ed F_q(X_1)\ldots F_{\q}(X_s)$. Let $\cA$ be an algorithm that distinguishes $\unifabs{6\ab}$ from all distributions with total variation distance at least $\dist/3$ from it. Then consider the algorithm $\cA^{'}$ that outputs $\p=\q$ if $\cA$ outputs ``$\p = \unifabs{6k}$'', and outputs $\p\ne \q$ otherwise. This shows that without privacy constraints, $\Sitnp \le \Sutnp{6\ab}{\dist/3}$ (See~\cite{Goldreich16} for details). 

We now prove that if further $\cA$ was an $\eps$-DP algorithm, then $\cA^{'}$ is also an $\eps$-DP algorithm. Suppose $X_1^s$, and $X_1^{'s}$ be two sequences in $[\ab]^s$ that could differ only on the last coordinate, namely $\Xos = \Xosmo X_s$, and $\Xosp = \Xosmo X_s^{'}$. 

Consider two sequences $\Yos= \Yosmo Y_s$, and $\Yosp= \Yosmo Y_s^{'}$ in $[6\ab]^s$ that could differ on only the last coordinate. Since $\cA$ is $\eps$-DP, 
\begin{align}
\cA(\Yos=\unifabs{6\ab}) \le \cA(\Yosp=\unifabs{6\ab})\cdot e^{\eps}.\label{eqn:dp-unif}
\end{align}
Moreover, since $F_{\q}$ is applied independently to each coordinate, 
\[
\probof{F_{\q}(\Xos) = \Yos} = \probof{F_{\q}(\Xosmo) = \Yosmo}\probof{F_{\q}(X_s) = Y_s}.
\]
Then, 
\begin{align}
&\ \ \ \ \probof{\cA^{'}(\Xos) = \q} \nonumber\\
& = \probof{\cA(F_{\q}(X_1^s)) = \unifabs{6\ab}}\nonumber\\
& = \sum_{\Yos}\ \probof{\mathcal{A}(Y^s_1) = \unifabs{6\ab}}\probof{F_{\q}(\Xos) = \Yos)}\nonumber\\
& = \sum_{\Yosmo} \sum_{Y_s\in[6\ab]}\ \probof{\mathcal{A}(Y^s_1) = \unifabs{6\ab}} \probof{F_{\q}(\Xosmo) = \Yosmo}\probof{F_{\q}(X_s) = Y_s}\nonumber\\
& = \sum_{\Yosmo}\probof{F_{\q}(\Xosmo) = \Yosmo}\Brack{ \sum_{Y_s\in[6\ab]}\ \probof{\mathcal{A}(Y^s_1) = \unifabs{6\ab}} \probof{F_{\q}(X_s) = Y_s}}.\label{eqn:first-term-prob}
\end{align}
Similarly, 
\begin{align}
\probof{\cA^{'}(\Xosp) = \q}
 \!\!=\!\! \sum_{\Yosmo}\probof{F_{\q}(\Xosmo)\!\! = \!\!\Yosmo}\Brack{ \sum_{Y_s^{'}\in[6\ab]} \probof{\mathcal{A}(\Yosp) = \unifabs{6\ab}} \probof{F_{\q}(X_s^{'}) = Y_s^{'}}}.\label{eqn:second-term-prob}
\end{align}
For a fixed $\Yosmo$, the term within the bracket in~\eqref{eqn:first-term-prob}, and~\eqref{eqn:second-term-prob} are both expectations over the final coordinate. However, by~\eqref{eqn:dp-unif} these expectations differ at most by a multiplicative $e^{\eps}$ factor. This implies that
\begin{align}
	\probof{\cA'(\Xos) = \q}\le \probof{\cA'(\Xosp) = \q}e^{\eps}.\nonumber
\end{align}
The argument is similar for the case when the testing output is \textbf{not $\unifabs{6\ab}$}, and is omitted here. We only considered sequences that differ on the last coordinate, and the proof remains the same when any of the coordinates is changed. This proves the privacy guarantees of the algorithm.
\end{proof}

\subsection{Identity Testing -- Upper Bounds} \label{sec:upper-uniformity}

 In this section, we will show that by privatizing the statistic proposed in~\cite{DiakonikolasGPP17} we can achieve the sample complexity in Theorem~\ref{thm:main-identity} for all parameter ranges. The procedure is described in Algorithm~\ref{algorithm_uniformity}.

Recall that $\Mltsmb{\smb}{\Xon}$ is the number of appearances of $\smb$ in $\Xon$. Let
\begin{align}\label{eqn:dggp-stat}
S(\Xon) \ed \frac12 \cdot \sum_{x=1}^{n} \absv{ \frac {\Mltsmb{\smb}{\Xon} } {\ns} -\frac1{k} },
\end{align}
be the TV distance from the empirical distribution to the uniform distribution. Let $\mu(p) = \expectation{S(\Xon)}$ when the samples are drawn from distribution $p$. They show the following separation result on the expected value of $S(\Xon)$. 
\begin{lemma}[\cite{DiakonikolasGPP17}] \label{lem:diako}
Let $p$ be a distribution over $[\absz]$ and $\dtv{\p}{\unifab} \ge \alpha$, then there is a constant $c$ such that 
\begin{center}
	$\mu(\p) - \mu (\unifab) \ge c \dist ^2 \min \left\{ \frac{\ns^2}{k^2}, \sqrt { \frac{\ns}{k} }, \frac1\alpha\right\}$.
\end{center}
\end{lemma}
\cite{DiakonikolasGPP17} used this result to show that thresholding $S(\Xon)$ at 0 is an optimal algorithm for identity testing. We first normalize the statistic to simplify the presentation of our DP algorithm. Let
\begin{equation}\label{eqn:statistic}
    Z (\Xon) \ed
    \begin{cases}
                                   k \Paren {S(\Xon) - \mu (\unifab) - \frac1{2} c \alpha^2 \cdot  \frac{\ns^2}{k^2} ~}, & \text{when $\ns \le k$,} \\
                                   \ns \Paren {S(\Xon) - \mu (\unifab) - \frac1{2} c \alpha^2 \cdot \sqrt { \frac{\ns}{k} } ~}, & \text{when $k< \ns \le \frac{k}{\alpha^2}$,} \\
                                   \ns \Paren {S(\Xon) - \mu (\unifab)  - \frac1{2} c \alpha}, & \text{when  $\ns \ge \frac{k}{\alpha^2}$. } 
    \end{cases}
\end{equation}
where $c$ is the constant in Lemma~\ref{lem:diako}, and $\mu(\unifab)$ is the expected value of  ${S(\Xon)}$ when $\Xon$ are drawn from uniform distribution.

\begin{algorithm}
    \caption{Uniformity testing}
    \label{algorithm_uniformity}
    \hspace*{\algorithmicindent} \textbf{Input:}  $\eps$, $\dist$,  i.i.d. samples $\Xon$ from $\p$
    \begin{algorithmic}[1] 
    \State Let $Z(\Xon)$ be evaluated from~\eqref{eqn:dggp-stat}, and~\eqref{eqn:statistic}.
    ~~~~~~~~
    \State Generate $Y\sim\Bern{\sigma\Paren{\eps\cdot Z}}$, $\sigma$ is the sigmoid function.
    \State{ {\bf if} $Y=0$, {\bf return} $\p=\unifab$, {\bf else}, {\bf return} $\p \neq \unifab$}.
\end{algorithmic}
\end{algorithm}

We now prove that this algorithm is $\eps$-DP. We need the following sensitivity result. 
\begin{lemma}
	$\Delta(Z)\le 1$ for all values of $\ns$, and $\ab$.  
\end{lemma}
\begin{proof}
	Recall that $S(\Xon) = \frac12 \cdot \sum_{x=1}^{n} \absv{ \frac {\Mltsmb{\smb}{\Xon} } {\ns} -\frac1{\ab} }$. Changing any one symbol changes at most two of the $\Mltsmb{\smb}{\Xon}$'s. Therefore at most two of the terms change by at most $\frac1\ns$. Therefore, $\Delta(S(\Xon))\le\frac1\ns$, for any $\ns$. When $\ns\le \ab$, this can be strengthened with observation that $\Mltsmb{\smb}{\Xon}/\ns\ge \frac1\ab$, for all $\Mltsmb{\smb}{\Xon}\ge1$. Therefore, 
	$S(\Xon) = \frac12\cdot \Paren{\sum_{\smb: \Mltsmb{\smb}{\Xon}\ge1}\Paren{\frac {\Mltsmb{\smb}{\Xon} } {\ns} -\frac1{k}}+ \sum_{\smb: \Mltsmb{\smb}{\Xon}=0} \frac1{k}} = \frac{\Phi_0(\Xon)}{\ab},$
	where $\Phi_0(\Xon)$ is the number of symbols not appearing in $\Xon$. This changes by at most one when one symbol is changed, proving the result. 
\end{proof}

Using this lemma, $\eps\cdot Z(\Xon)$ changes by at most $\eps$ when $\Xon$ is changed at one location. Invoking Lemma~\ref{lem:sig-cont}, the probability of any output changes by a multiplicative $\exp(\eps)$, and the algorithm is $\eps$-differentially private.

To prove the sample complexity bound, we first show that the mean of the test statistic is well separated using Lemma~\ref{lem:diako}. Then we use the concentration bound of the test statistic from~\cite{DiakonikolasGPP17} to get the final complexity.


Because of the normalization in Equation~\ref{eqn:statistic} and lemma~\ref{lem:diako}, for $\Xon$ drawn from $\unifab$
\begin{equation}\label{eqn:statistic-expectation-uniform}
\expectation{Z (\Xon)} \le
\begin{cases}
- \frac1{2} c \alpha^2 \cdot  \frac{\ns^2}{k}, & \text{when $\ns \le k$,} \\
- \frac1{2} c \alpha^2 \cdot  { \frac{\ns^{3/2}}{k^{1/2}} }, & \text{when $k< \ns \le \frac{k}{\alpha^2}$,} \\
-\frac1{2} c\ns \alpha, & \text{when  $\ns \ge \frac{k}{\alpha^2}$ . } 
\end{cases}
\end{equation}

For $\Xon$ drawn from $\p$ with $\dtv{\p}{\unifab}\ge\alpha$, 
\begin{equation}\label{eqn:statistic-expectation-far}
\expectation{Z (\Xon)} \ge
\begin{cases}
\frac1{2} c \alpha^2 \cdot  \frac{\ns^2}{k}, & \text{when $\ns \le k$,} \\
\frac1{2} c \alpha^2 \cdot  { \frac{\ns^{3/2}}{k^{1/2}} }, & \text{when $k< \ns \le \frac{k}{\alpha^2}$,} \\
\frac1{2} c \ns\alpha, & \text{when  $\ns \ge \frac{k}{\alpha^2}$ . } 
\end{cases}
\end{equation}

In order to prove the utility bounds, we also need the following (weak) version of the result of~\cite{DiakonikolasGPP17}, which is sufficient to prove the sample complexity bound for constant error probability.
\begin{lemma}
	\label{lem:identityerror}
	There is a constant $C>0$, such that when $\ns>C\sqrt{\ab}/\dist^2$, then for $\Xon\sim \p$, where either $\p=\unifab$, or $\dtv{\p}{\unifab}\ge\alpha$, 
	\[
	\probof{\absv{{Z(\Xon)-\expectation{Z(\Xon)}}}> \frac{2\expectation{Z(\Xon)}}{3} }<0.01.
	\]
\end{lemma}
\noindent The proof of this result is in Section~\ref{app:iden-lemma}.

We now proceed to prove the sample complexity bounds. Assume that $\ns>C\sqrt{\ab}/\alpha^2$, so Lemma~\ref{lem:identityerror} holds. Suppose $\eps$ be any real number such that $\eps|\expectation{Z(\Xon)}|>3\log {100}$. Let $\cA(\Xon)$ be the output of Algorithm~\ref{algorithm_uniformity}. Denote the output by 1 when $\cA(\Xon)$ is ``$\p\ne\unifab$'', and 0 otherwise. Consider the case when $\Xon\sim \p$, and $\dtv{\p}{\unifab}\ge\alpha$. Then, 
\begin{align*}
\probof{\cA(\Xon) = 1} &\ge  \probof{\cA(\Xon) = 1 \text{ and } Z(\Xon)>\frac{\expectation{Z(\Xon)}}3}\\
& =\probof{Z(\Xon)>\frac{\expectation{Z(\Xon)}}3}\cdot \probof{\cA(\Xon) = 1|Z(\Xon)>\frac{\expectation{Z(\Xon)}}3}\\
& \ge 0.99\cdot \probof{B\Paren{\sigma(\eps\cdot \frac{\expectation{Z(\Xon)}}3)}=1}\\
&\ge 0.99 \cdot 0.99 \ge 0.9,
\end{align*}
where the last step uses that $\eps \expectation{Z(\Xon)}/3>\log {100}$, along with Lemma~\ref{lem:sig-cont}. The case of $\p=\unifab$ follows from the same argument. 

Therefore, the algorithm is correct with probability at least $0.9$, whenever, $\ns>C\sqrt{\ab}/\dist^2$, and $\eps|\expectation{Z(\Xon)}|>3\log {100}$. By ~\eqref{eqn:statistic-expectation-far}, note that $\eps|\expectation{Z(\Xon)}|>3\log {100}$ is satisfied when, 
\begin{align*}
c \alpha^2 \cdot  {\ns^2}/{k}\ge&(6\log {100})/\eps ,\ \text{ for $\ns \le k$}, &\\
c \alpha^2 \cdot  { {\ns^{3/2}}/{k^{1/2}} }\ge&(6\log {100})/\eps,  \ \text{ for $k< \ns \le {k}/{\alpha^2}$}, &\\
c \alpha\cdot \ns\ge& (6\log {100})/\eps,\ \text{ for  $\ns \ge {k}/{\alpha^2}$}.&
\end{align*}
This gives the upper bounds for all the three regimes of $\ns$.

\subsection{Sample Complexity Lower bounds for Uniformity Testing}\label{sec:lower-identity}

In this section, we will show the lower bound part of Theorem~\ref{thm:main-identity}. The first term is the lower bound without privacy constraints, proved in~\cite{Paninski08}. In this section, we will prove the terms associated with privacy. 

Our lower bound is based on the following theorem, which can be viewed as a direct corollary of our DP Le Cam's method (Theorem~\ref{thm:le_cam}).

\begin{theorem}
\label{thm:coupling}
Suppose there is a coupling between $\p$ and $\q$ over $\cX^\ns$, such that $\expectation{\ham{\Xon}{\Yon}} \le D$ where $\Xon \sim p, \Yon \sim q$. Then, any $(\eps,\delta)$-differentially private \newhz{hypothesis testing} algorithm $\cA:\cX^\ns\to\{\p,\q\}$ on $\p$ and $\q$
must satisfy $\eps+\delta = \Omega\Paren{\frac1{D}}$ .
\end{theorem}
\begin{proof}
Note that the error probability is smaller than $0.1$. By Theorem~\ref{thm:le_cam}, we have $ e^{-10 \eps\cdot D} - 10D\delta <0.1$.
Hence, either $e^{\eps\cdot 10 D} = \Omega(1)$ or $10 D \delta = \Omega(1)$, which implies that 
$D = \Omega\Paren{\min \left\{\frac1\eps, \frac1\delta \right\}} = \Omega\Paren{\frac1{\eps + \delta}},$
proving the theorem.
\end{proof}

The simplest argument is for $\ns \ge \frac{k}{\alpha^2}$, which hopefully will give you a sense of how coupling argument works. We consider the case of binary identity testing where the goal is to test whether the bias of a coin is $1/2$ or $\alpha$-far from $1/2$. This is a special case of identity testing for distributions over $[k]$ (when $k-2$ symbols have probability zero). This is strictly harder than the problem of distinguishing between $\Bern{1/2}$ and $\Bern{1/2+\alpha}$. The coupling given in Example~\ref{exm:coin} has expected hamming distance of $\alpha \ns$. Hence combing with Theorem~\ref{thm:coupling}, we get a lower bound of $\Omega({\frac{1}{\dist \eps}})$.

We now consider the cases $\ns \le k$ and $k< \ns \le \frac{k}{\alpha^2}$.


To this end, we invoke LeCam's two point theorem, and design a hypothesis testing problem that will imply a lower bound on uniformity testing. The testing problem will be to distinguish between the following two cases. 

{\bf Case 1:} We are given $\ns$ independent samples from the uniform distribution $\unifab$.

{\bf Case 2:} Generate a distribution $\p$ with $\dtv{\p}{\unifab}\ge\dist$ according to some prior over all such distributions. We are then given $\ns$ independent samples from this distribution $\p$. 

Le Cam's two point theorem~\cite{Yu97} states that any lower bound for distinguishing between these two cases is a lower bound on identity testing problem. 

We now describe the prior construction for {\bf Case 2}, which is the same as considered by~\cite{Paninski08} for lower bounds on identity testing without privacy considerations. For each $\textbf{z} \in\{\pm1\}^{\ab/2}$, define a distribution $\p_{\textbf{z}}$ over $[\ab]$ such that 
\begin{align}
\p_{\textbf{z}}(2i-1) = \frac{1+\textbf{z}_i\cdot 2\dist}{\ab}, \text{ and } \p_{\textbf{z}}(2i) = \frac{1-\textbf{z}_i\cdot 2\dist}{\ab}.\nonumber
\end{align}
Then for any $\textbf{z}$, $\dtv{P_{\textbf{z}}}{\unifab}= \alpha$. For {\bf Case 2}, choose  $\p$ uniformly from these $2^{\ab/2}$ distributions. Let $Q_2$ denote the distribution on $[\ab]^\ns$ by this process. In other words, $Q_2$ is a mixture of product distributions over $[\ab]$. 

In {\bf Case 1}, let $Q_1$ be the distribution of $\ns$ $i.i.d.$ samples from $\unifab$. 

To obtain a sample complexity lower bound for distinguishing the two cases, we will design a coupling between $Q_1$, and $Q_2$, and bound its expected Hamming distance. While it can be shown that the Hamming distance of the coupling between the uniform distribution with any \emph{one} of the $2^{\ab/2}$ distributions grows as $\alpha\ns$, it can be significantly smaller, when we consider the mixtures. In particular, the following lemma shows that there exist couplings with bounded Hamming distance. 

\begin{lemma}\label{lem:coupling-hamming}
	There is a coupling between $\Xon$ generated by $Q_1$, and $\Yon$ by $Q_2$ such that 
\begin{center}
   $\expectation{\ham{\Xon}{\Yon}} \le C \cdot \dist^2 \min \{ \frac{\ns^2}{\ab}, \frac{\ns^{3/2}}{{\ab}^{1/2}} \}$.
\end{center}
\end{lemma}
The lemma is proved in Appendix~\ref{app:lb-unif}. Now applying Theorem~\ref{thm:coupling}, we get the bound in Theorem~\ref{thm:main-identity}.

\section{Closeness Testing} \label{sec:close:}

Recall the closeness testing problem from Section~\ref{sec:testing_preliminaries}, and the tight non-private bounds from Table~\ref{fig:badass:table}. Our main result in this section is the following theorem characterizing the sample complexity of differentially private algorithms for closeness testing. 
\newhz{
\begin{theorem}
	\label{thm:close-main}	
		\[
		\Sct= \Theta \Paren{\frac{\ab^{1/2}}{\dist^2}+\frac{\ab^{2/3 }}{\dist^{4/3}}+ \frac1{\dist\eps}+\frac{\ab^{1/2}}{\dist\eps^{1/2}}+\frac{\ab^{1/3}}{\dist^{4/3 }\eps^{2/3} }}.
		\]
\end{theorem}
}
{
This theorem shows that  our bounds are tight up to constant factors in all parameters. 

\subsection{Closeness Testing -- Upper Bounds} \label{sec:close-upper}

To prove the upper bounds, we only consider the case when $\delta = 0$, which would suffice by lemma~\ref{lm:epsdelta}. We privatize the closeness testing algorithm of~\cite{DiakonikolasGKPP20}. 

The statistic used by~\cite{DiakonikolasGKPP20} is
\begin{align*}
&~~~~Z(\Xon,\Xonn,\Yon,\Yonn) \\
& = \sum_{i\in[\ab]} \Paren{|\Mltsmb{i}{\Xon} - \Mltsmb{i}{\Yon}| +|\Mltsmb{i}{\Xonn} - \Mltsmb{i}{\Yonn}| - |\Mltsmb{i}{\Xon} - \Mltsmb{i}{\Xonn}| - |\Mltsmb{i}{\Yon} - \Mltsmb{i}{\Yonn}|},
\end{align*}
where $\Xon$ and $\Xonn$ are generated from distribution $\p$, and  $\Yon$ and $\Yonn$ are generated from distribution $\q$.  It turns out that this statistic has a constant sensitivity, as shown in Lemma~\ref{sens_close}.
\begin{lemma} \label{sens_close}
	$\Delta(Z)\le2$. 
\end{lemma}
\begin{proof}
	Since $Z$  is symmetric, without loss of generality assume that one of the symbols is changed in $\Yon$. This would cause at most two of the $\Mltsmb{i}{\Yon}$ to change, which changes $Z$ by at most two.
\end{proof}
We use the same approach with the test statistic as with uniformity testing to obtain a differentially private closeness testing method, described in Algorithm~\ref{algorithm_closeness}. Since the sensitivity of the statistic is at most 2, the input to the sigmoid changes by at most $\eps$ when any input sample is changed.  Invoking Lemma~\ref{lem:sig-cont}, the probability of any output changes by a multiplicative $\exp(\eps)$, and the algorithm is $\eps$-differentially private. 
\begin{algorithm}
	\caption{}
	\label{algorithm_closeness}
	\hspace*{\algorithmicindent} \textbf{Input:}
	$\eps$, $\alpha$, sample access to distribution $p$ and $q$
	\begin{algorithmic}[1] 
		\State $Z' \gets (Z - C_1\sqrt{\ns} - \frac{C_2}{\eps})/2$, where $C_1$ and $ C_2$ are universal constants
		\State Generate $Y\sim\Bern{\sigma\Paren{\exp(\eps\cdot Z'}}$
		\State{ {\bf if} $Y=0$, {\bf return} $\p=\q$}
		\State{ {\bf else}, {\bf return} $\p\ne\q$}
	\end{algorithmic}
\end{algorithm}
}


In this section, we will show that Algorithm~\ref{algorithm_closeness} satisfies sample complexity upper bounds described in Theorem~\ref{thm:close-main}.

The results in~\cite{DiakonikolasGKPP20} were proved under Poisson sampling, and we also use Poisson sampling, with only a constant factor effect on the number of samples for the same error probability. They showed the following bounds: there exists a universal constant $C$, such that
\begin{align}
&~\expectation{ Z} = 0 \text{ when } \p=\q,\label{eqn:expectation-equal}\\
&~\expectation{ Z} \ge C_1\cdot \min\Paren{\ns\dist, \frac{\ns^2\dist^2}{\ab}, \frac{\ns^{\frac{3}{2}}\dist^2 }{\sqrt{k}}} \text{ when } \dtv{\p}{\q}\ge\dist,\label{eqn:expectation-unequal}\\
&\probof{|Z- \expectation{Z}| \ge C_1 \sqrt{\ns} } \le 0.05.& \label{eqn:variance-equal}
\end{align}

We consider the case when $\p = \q$, where $\expectation{Z} = 0$. By~\eqref{eqn:variance-equal},
 \begin{align*}
\probof{Z' > - \frac{C_2}{2\eps}} \le 0.05. 
\end{align*}

Now note that, if $\frac{C_2}{2}>\log(20)$, then for all $Z' < -\frac{C_2}{2\eps} $, with probability at least 0.95, the algorithm outputs the $\p=\q$. Combining the conditions, we obtain that with probability at least 0.9, the algorithm outputs the correct answer when the input distributions satisfy $\p=\q$. 

Now we move to the case when $\dtv{\p}{\q} \ge \dist$. Note that
\begin{align}
\probof{Z' < \frac{4.5C_2}{\eps}} &\le \probof{Z < \frac{10C_2}{\eps}+C_1 \sqrt{\ns}} \nonumber\\
&= \probof{ \expectation{Z} - Z >   \expectation{Z} - \frac{10C_2}{\eps}- C_1 \sqrt{\ns}}. \nonumber
\end{align}
If we require $\expectation{Z} - \frac{10C_2}{\eps}- C_1 \sqrt{\ns} > C_1 \sqrt{\ns}$, or in other words, $\ns \ge C_3 \Paren{\frac{\ab^{1/2}}{\dist^2}+\frac{\ab^{2/3 }}{\dist^{4/3}}+ \frac1{\dist\eps}+\frac{\ab^{1/2}}{\dist\eps}+\frac{\ab^{1/3}}{\dist^{4/3 }\eps^{2/3} }}$ for some constant $C_3$, we have 
\begin{align}
\probof{Z' < \frac{4.5C_2}{\eps}} &= \probof{ \expectation{Z} - Z >   \expectation{Z} - \frac{10C_2}{\eps}- C_1 \sqrt{\ns}}. \nonumber\\
&< \probof{ \expectation{Z} - Z >  C_1 \sqrt{\ns}} \nonumber \\
&\le 0.05,
\end{align}
when the last inequality comes from~\eqref{eqn:variance-equal}.

Now note that, if $4.5C_2>\log(20)$, then for all $Z' > \frac{4.5C_2}{\eps} $, with probability at least 0.95, the algorithm outputs the $\p \neq \q$. Combining the conditions, we obtain that with probability at least 0.9, the algorithm outputs the correct answer when the input distributions satisfy $\dtv{\p}{\q}\ge \dist$. 

\subsection{Closeness Testing -- Lower Bounds}  \label{sec:close:lower}

To show the lower bound part of Theorem~\ref{thm:close-main}, we need the following simple result.

\begin{lemma}  \label{lem:close-to-ident}
	 $\Sit \leq \Sct$. 
\end{lemma}

\begin{proof}
Suppose we want to test identity with respect to $\q$. Given $\Xon$ from $\p$, generate $\Yon$ independent samples from $\q$. If $\p=\q$, then the two samples are generated by the same distribution, and otherwise they are generated by distributions that are at least $\eps$ far in total variation. Therefore, we can simply return the output of an \newhz{$(\ab,\dist,\eps)$}-closeness testing algorithm on $\Xon$, and $\Yon$. 
\end{proof}

By Lemma~\ref{lem:close-to-ident} we know that a lower bound for identity testing is also a lower bound on closeness testing.

Besides, $\Omega\Paren{\frac{\ab^{2/3}}{\dist^{4/3}}+ \frac{\ab^{1/2}}{\dist^2}}$ is the lower bound of non-private closeness testing~\cite{ChanDVV14}, which is naturally a lower bound for DP closeness testing. Therefore, we have proved the lower bound part in Theorem~\ref{thm:close-main}.
\section{Proofs}
\label{sec:testing_proof}
\subsection{Proof of  Lemma~\ref{lem:identityerror}}
\label{app:iden-lemma}
In order to prove the lemma, we need the following lemma, which is proved in \cite{DiakonikolasGPP17}.

\begin{lemma}{(Bernstein version of McDiarmid's inequality)}
\label{lem:McDiarmid}
	Let $\Yon$ be independent random variables taking values in the set $\cY$. Let $f: \cY^\ns \rightarrow \RR$ be a function of $\Yon$ so that for every $j \in [\ns]$, and $y_1,...y_\ns,{y_j  ^{\prime}} \in \cY$, we have that:
$$\absv{f(y_1,...y_j,...y_\ns) - f(y_1,...,{y_j ^{\prime}},...y_\ns)} \le B, $$

Then we have

$$\probof{f-\expectation{f} \ge z} \le \exp\Paren{\frac{-2z^2}{\ns B^2}}.$$

In addition, if for each $j \in [m]$ and $y_1,...y_{j-1},y_{j+1},...y_m$ we have that 

$$ \mathrm{Var}_{Y_j} [f(y_1,...y_j,...y_\ns)] \le \sigma_j^2,$$

then we have 
$$\probof{f-\expectation{f} \ge z} \le \exp\Paren{\frac{-z^2}{\sum_{j=1}^{\ns} \sigma_j^2 +2Bz/3 }}.$$
\end{lemma}

The statistic we use $Z(X_1^m)$ has sensitivity at most 1, hence we can use $B =1$ in Lemma~\ref{lem:McDiarmid}. 

We first consider the case when $k< \ns \le \frac{k}{\alpha^2}$. When $p=\unifab$, we get $\expectation{Z(\Xon)} = - \frac1{2} c \ns \alpha^2 \cdot \sqrt { \frac{\ns}{k} }$, then by the first part of Lemma~\ref{lem:McDiarmid},
 \begin{align}
\probof{Z(\Xon)>\frac{\expectation{Z(\Xon)}}3} =& \probof{Z(\Xon)> -   \frac1{6} c \ns \dist^2 \cdot \sqrt{\frac{\ns}{k}}}  \nonumber\\
\le& \probof{Z(\Xon) - \expectation{Z(\Xon) } > \frac{2}{3}c \ns \dist^2 \cdot \sqrt{\frac{\ns}{k}}} \nonumber\\
\le & \exp\Paren{-\frac{8c^2\ns^2 \dist^4}{9k}}.
\end{align}

Therefore, there is a $C_1$ such that if $\ns\ge C_1\sqrt{\ab}/\alpha^2$, then under the uniform distribution $\probof{Z(\Xon)>\frac{\expectation{Z(\Xon)}}3}$  is at most 1/100. The non-uniform distribution part is similar and we omit the case. 

Then we consider the case when $\frac{k}{\alpha^2} < \ns$. When $p=\unifab$, we get $\expectation{Z(\Xon)} = - \frac1{2} c \ns \alpha$, then also by the first part of Lemma~\ref{lem:McDiarmid},
 \begin{align}
\probof{Z(\Xon)>\frac{\expectation{Z(\Xon)}}3} =& \probof{Z(\Xon)> -   \frac1{6} c \ns \dist} \nonumber\\   
\le& \probof{Z(\Xon) - \expectation{Z(\Xon) } > \frac{2}{3}c \ns \dist } \nonumber\\
\le & \exp\Paren{ -\frac{8c^2\ns \dist^2}{9} }.\nonumber
\end{align}

Using the same argument we can show that there is a constant $C_2$ such that for $\ns\ge C_2/\alpha^2$, then under the uniform distribution $\probof{Z(\Xon)>\frac{\expectation{Z(\Xon)}}3}$  is at most 1/100. The case of non-uniform distribution is omitted because of the same reason.

At last we consider the case when $\ns \le k$. In this case we need another result proved in~\cite{DiakonikolasGPP17}:
\[
	\mathrm{Var}_{X_j} [Z(x_1, x_2, ...,X_j, .. ,x_\ns) ] \le \frac \ns k, \forall j, x_1,x_2,...,x_{j-1}, x_{j+1}, ... x_n.
\]

 When $p=\unifab$, we get $\expectation{Z(\Xon)} = - \frac1{2} c k \alpha^2 \cdot \frac{\ns^2}{k^2}$, then by the second part of Lemma~\ref{lem:McDiarmid},
 \begin{align}
\probof{Z(\Xon)>\frac{\expectation{Z(\Xon)}}3} =& \probof{Z(\Xon)> -   \frac1{6} c \dist^2 \cdot \frac{\ns^2}{k}}   \nonumber\\
\le& \probof{Z(\Xon) - \expectation{Z(\Xon) } > \frac{2}{3}c  \dist^2 \cdot \frac{\ns^2}{k}} \nonumber\\
\le & \exp\Paren{ \frac{-\frac{4}{9}c^2 \dist^4 \frac{\ns^4}{k^2} }{ \frac{\ns^2}{k} + \frac{4}{9}c \dist^2 \frac{\ns^2}{k}} } \nonumber\\
\le & \exp\Paren{-\frac{2}{9}c \alpha^4 \frac{\ns^2}{k}}.\nonumber
\end{align}

Therefore, there is a $C_3$ such that if $\ns\ge C_3\sqrt{\ab}/\alpha^2$, then under the uniform distribution $\probof{Z(\Xon)>\frac{\expectation{Z(\Xon)}}3} $ is at most 1/100. The case of non-uniform distribution is similar and is omitted. 

Therefore, if we take $C = \max\{C_1, C_2, C_3\}$, we prove the result in the lemma.

\subsection{Proof of Lemma~\ref{lem:coupling-hamming}}
\label{app:lb-unif}

We first consider the case when $\ns\le \ab$, where $\min \{ \frac{\ns^2}{\ab}, \frac{\ns^{3/2}}{{\ab}^{1/2}} \} = \frac{\ns^2}{\ab}$.

Before proving the lemma, we consider an example that will provide insights and tools to analyze the distributions $Q_1$, and $Q_2$.
Let $t\in\NN$. Let $P_2$ be the following distribution over $\{0,1\}^t$:
\begin{itemize}
\item
Select $b\in\{\half-\dist, \half+\dist\}$ with equal probability. 
\item
Output $t$ independent samples from $\Bern{b}$. 
\end{itemize}
Let $P_1$ be the distribution over $\{0,1\}^t$ that outputs $t$ independent samples from $\Bern{0.5}$.

When $t=1$, $P_1$ and $P_2$ both become $\Bern{0.5}$. For t=2, $P_1(00) = P_1(11) = \frac14+\dist^2$, and $P_1(10)= P_1(01)= \frac14-\dist^2$, and $\dtv{P_1}{P_2}$ is $2\dist^2$. A slightly general result is the following:
\begin{lemma}
\label{lem:p1-p2}
For $t=1$, $\dtv{P_1}{P_2}=0$ and for $t \geq 2$, $\dtv{P_1}{P_2} \le 2t\alpha^2$. 
\end{lemma}

\begin{proof}
Consider any sequence $X_1^{t}$ that has $t_0$ zeros, and $t_1= t-t_0$ ones. Then, 
\begin{align}
P_1(X_1^t) = {t\choose t_0} \frac1{2^t},\nonumber
\end{align}
and 
\begin{align}
P_2(X_1^t) = {t\choose t_0} \frac1{2^t}\Paren{\frac{(1-2\dist)^{t_0}(1+2\dist)^{t_1}+ (1+2\dist)^{t_0}(1-2\dist)^{t_1}}2}.\nonumber
\end{align}
The term in the parentheses above is minimized when $t_0 = t_1 = t/2$. In this case, 
\begin{align}
P_2(X_1^t) \ge & P_1(X_1^t)\cdot (1+2\dist)^{t/2} (1-2\dist)^{t/2}= P_1(X_1^t)\cdot (1-4\dist^2)^{t/2}. \nonumber
\end{align}
Therefore, 
\begin{align}
\dtv{P_1}{P_2} = \sum_{P_1>P_2} P_1(X_1^t)-P_2(X_1^t) \le \sum_{P_1>P_2} P_1(X_1^t)\Paren{1- (1-4\dist^2)^{t/2}} \le 2t\dist^2,\nonumber
\end{align}
where we used the Weierstrass Product Inequality, which states that $1-tx\le (1-x)^{t}$ proving the total variation distance bound.
\end{proof}

As a corollary this implies:
\begin{lemma}
\label{lem:coupling-bin}
There is a coupling between $\Xot$ generated from $P_1$ and $\Yot$ from $P_2$ such that $\expectation{\ham {\Xot}{\Yot}}\le t\cdot \dtv{P_1}{P_2} \le 4(t^2-t) \dist^2$.
\end{lemma}

\begin{proof}
Observe that $\sum_{X_1^t} \min\{P_1(X_1^t),P_2(X_1^t)\} = 1-\dtv{P_1}{P_2}$. Consider the following coupling between $P_1$, and $P_2$. Suppose $X_1^t$ is generated by $P_1$, and let $R$ be a $U[0,1]$ random variable. 
\begin{enumerate}
\item{ $R<1-\dtv{P_1}{P_2}$} Generate $X_1^t$ from the distribution that assigns probability  $\frac{\min\{P_1(X_1^t),P_2(X_1^t)\}}{1-\dtv{P_1}{P_2}}$ to $X_1^t$. Output $(X_1^t, X_1^t)$.
\item{$R\ge1-\dtv{P_1}{P_2}$}
Generate $X_1^t$ from the distribution that assigns probability ${\frac{P_1(X_1^t) - \min\{P_1(X_1^t),P_2(X_1^t)\}}{\dtv{P_1}{P_2}}}$ to $X_1^t$,  and $Y_1^t$ from the distribution that assigns  probability ${\frac{P_2(Y_1^t) - \min\{P_1(Y_1^t),P_2(Y_1^t)\}}{\dtv{P_1}{P_2}}}$  to $Y_1^t$ independently. Then output $(X_1^t, Y_1^t)$.
\end{enumerate}
To prove the coupling, note that the probability of observing $X_1^t$ is 
\[
\Paren{1-\dtv{P_1}{P_2}}\cdot \frac{\min\{P_1(X_1^t),P_2(X_1^t)\}}{1-\dtv{P_1}{P_2}}+\dtv{P_1}{P_2}\cdot {\frac{P_1(X_1^t) - \min\{P_1(X_1^t),P_2(X_1^t)\}}{\dtv{P_1}{P_2}}} = P_1(X_1^t).
\]
A similar argument gives the probability of $Y_1^t$ to be $P_2(Y_1^t)$.

\noindent Then  $\expectation{\ham {\Xot}{\Yot}}\le t \cdot \dtv{P_1}{P_2} = 2t^2 \dist^2\le 4(t^2-t)\alpha^2$ when $t\ge2$, and when $t=1$, the distributions are identical and the Hamming distance of the coupling is equal to zero.
\end{proof}

We now have the tools to prove Lemma~\ref{lem:coupling-hamming} for $\ns\le\ab$. 
\begin{proof}[Proof of Lemma~\ref{lem:coupling-hamming} for $\ns\le\ab$.]
The following is a coupling between $Q_1$ and $Q_2$:
\begin{enumerate}
\item 
Generate $\ns$ samples $Z_1^\ns$ from a uniform distribution over $[\ab/2]$. 
\item
For $j \in[\ab/2]$, let $T_j\subseteq[\ns]$ be the set of locations where $j$ appears. Note that $|T_j| = \Mltsmb{j}{Z_1^\ns}$. 
\item
To generate samples from $Q_1$:
\begin{itemize}
	\item Generate $|T_j|$ samples from a uniform distribution over $\{2j-1, 2j\}$, and replace the symbols in $T_j$ with these symbols. 
\end{itemize}
\item
To generate samples from $Q_2$:
\begin{itemize}
\item
Similar to the construction of $P_1$ earlier in this section, consider two distributions over $\{2j-1, 2j\}$ with bias $\half-\dist$, and $\half+\dist$. 
\item
Pick one of these distributions at random.
	\item Generate $|T_j|$ samples from it over $\{2j-1, 2j\}$, and replace the symbols in $T_j$ with these symbols. 
\end{itemize}
\end{enumerate}

From this process the coupling between $Q_1$, and $Q_2$ is also clear:
\begin{itemize}
\item 
Given $\Xon$ from $Q_2$, for each $j\in[\ab/2]$ find all locations $\ell$ such that $X_\ell= 2j-1$, or $X_{\ell}=2j$. Call this set $T_j$. 
\item
Perform the coupling between $P_2$ and $P_1$ from Lemma~\ref{lem:coupling-bin}, after replacing $\{0,1\}$ with $\{2j-1, 2j\}$.
\end{itemize}

	Using the coupling defined above, by the linearity of expectations, we get:
	\begin{align}
		\expectation{\ham{\Xon}{\Yon}}  & =  \sum_{j=1}^{\ab/2} \expectation{\ham{X_1^{|T_j|}}{Y_1^{|T_j|}}} \nonumber\\
		& = \frac{\ab }{2} \expectation{\ham{X_1^{R}}{Y_1^{R}}} \nonumber \\
		& \le  \frac{\ab }{2}\cdot \expectation{4\dist^2(R^2-R)}, \nonumber
	\end{align}
	where $R$ is a binomial random variable with parameters $m$ and $2/\ab$. 
Now, a simple exercise computing Binomial moments shows that for $X\sim Bin(n, s)$, $\expectation{X^2-X} = s^2(n^2-n)\le n^2s^2.$ This implies that 
\[
\expectation{R^2-R} \le \frac{4\ns^2}{\ab^2}.
\]
Plugging this, we obtain
	\begin{align}
		\expectation{\ham{\Xon}{\Yon}} \le  \frac{\ab }{2}\cdot 	\frac{16\dist^2\ns^2}{\ab^2}  = \frac{8 \ns^2\dist^2}{\ab},\nonumber
	\end{align}
	proving the claim.
\end{proof}

Next we consider the case when $\ab\le\ns\le \ab/\dist^2$, where $\min \{ \frac{\ns^2}{\ab}, \frac{\ns^{3/2}}{{\ab}^{1/2}} \} =\frac{\ns^{3/2}}{{\ab}^{1/2}} $.

Lemma~\ref{lem:p1-p2} holds for all values of $t$, and $\alpha$. The lemma can be strengthened for cases where $\alpha$ is small. The following lemma is proved in Section~\ref{sec:coup_p}.

\begin{lemma}
	\label{lem:coup_p}
	Let $P_1$, and $P_2$ be the distributions over $\{0,1\}^t$ defined in the last section. 
	There is a coupling between $X_1^t$ generated by $P_1$, and $Y_1^t$ by $P_2$ such that 
	\[
	\expectation{\ham{X_1^t}{Y_1^t}} \le C\cdot (\alpha^2t^{3/2} +  \alpha^4 t^{5/2} + \alpha^5 t^3).
	\] 
\end{lemma}

Given the coupling we defined inSection~\ref{sec:coup_p} for proving Lemma~\ref{lem:coup_p}, the coupling between $Q_1$, and $Q_2$ uses the same technique in the last section for $\ns\le \absz$. 
\begin{itemize}
	\item 
	Given $\Xon$ from $Q_2$, for each $j\in[\ab/2]$ find all locations $\ell$ such that $X_\ell= 2j-1$, or $X_{\ell}=2j$. Call this set $T_j$. 
	\item
	Perform the coupling in Appendix~\ref{sec:coup_p} between $P_2$ and $P_1$ on $T_j$, after replacing $\{0,1\}$ with $\{2j-1, 2j\}$.
\end{itemize}

Using the coupling defined above, by the linearity of expectations, we get:
\begin{align}
	\expectation{\ham{\Xon}{\Yon}}  & =  \sum_{j=1}^{\ab/2} \expectation{\ham{X_1^{|T_j|}}{Y_1^{|T_j|}}} \nonumber\\
	& = \frac{\ab }{2} \expectation{\ham{X_1^{R}}{Y_1^{R}}} \nonumber \\
	& \le  \frac{\ab }{2}\cdot \expectation{64\cdot\Paren{\alpha^4R^{5/2} + {\dist^2R^{3/2} + \alpha^5 R^3}} }, \nonumber
\end{align}
where $R \sim \bino{m}{2/\ab}$.

We now bound the moments of Binomial random variables. The bound is similar in flavor to~\cite[Lemma 3]{AcharyaOST17} for Poisson random variables.
\begin{lemma}
Suppose $\frac{\ns}{k}>1$, and $Y\sim \bino{\ns}{\frac1k}$, then for $\gamma \ge 1$, there is a constant $C_{\gamma}$ such that 
\[
\expectation{Y^{\gamma}}\le C_\gamma {\Paren{\frac{\ns}{k}}}^\gamma.
\]
\end{lemma}

\begin{proof}
	For integer values of $\gamma$, this directly follows from the moment formula for Binomial distribution~\cite{Knoblauch08}, and for other $\gamma \ge 1$, by Jensen's Inequality
	\[
		\expectation{Y^{\gamma}} \le \expectation{\Paren{Y^{\ceil{\gamma} }}^{\frac{\gamma}{\ceil{\gamma}}}}\le  \expectation{\Paren{Y^{\ceil{\gamma} }}} ^{\frac{\gamma}{\ceil{\gamma}}} \le \Paren{ C_{\ceil{\gamma}} \expectation{Y}^{\ceil{\gamma} }}^{\frac{\gamma}{\ceil{\gamma}}} = C' (\expectation{Y})^{\gamma} , 
	\]
	proving the lemma. 
	\end{proof}
Therefore, letting $C = \max \{C_{5/2}, C_{3}, C_{3/2}\}$, we obtain
\begin{align*}
\expectation{\ham{\Xon}{\Yon}}  \le 32k C\cdot\Paren{\alpha^4\Paren{\frac{\ns}{k}}^{5/2} + {\dist^2 \Paren{\frac{\ns}{k}}^{3/2} + {{\alpha^5 \Paren{\frac{\ns}{k}}}}^3}}.
\end{align*}
Now, notice $\alpha\sqrt\frac{\ns}{k}<1$. Plugging this, 
\begin{align*}
\expectation{\ham{\Xon}{\Yon}}  \le\ & 32 C\cdot k\cdot\Paren{\alpha^4\Paren{\frac{\ns}{k}}^{5/2} + {\dist^2 \Paren{\frac{\ns}{k}}^{3/2} + {{\alpha^5 \Paren{\frac{\ns}{k}}}}^3}} \\
=\ & 32 C\cdot k \alpha^2\cdot\Paren{\alpha^2\frac{\ns}{k}\cdot\Paren{\frac{\ns}{k}}^{3/2} + \Paren{\frac{\ns}{k}}^{3/2} + {{\alpha^3 \Paren{\frac{\ns}{k}}^{3/2}\Paren{\frac{\ns}{k}}}}^{3/2}} \\
\le\ & 96 C\cdot k \Paren{\frac{\ns}{k}}^{3/2},
\end{align*}
completing the argument. 

\subsection{Proof of Lemma~\ref{lem:coup_p}} \label{sec:coup_p}

To prove Lemma~\ref{lem:coup_p}, we need a few lemmas first:

\begin{definition}
	A random variable $Y_1$ is said to stochastically dominate $Y_2$ if for all $t$,  $\probof{Y_1\ge t}\ge \probof{Y_2\ge t}$. 
\end{definition}

\begin{lemma}\label{stod}
	Suppose $N_1 \sim \bino{t}{\frac{1}{2}}, N_2 \sim \frac{1}{2}\bino{t}{\frac{1+\alpha}{2}} + \frac{1}{2}\bino{t}{\frac{1-\alpha}{2}} $. Then $Z_2 = \max \{N_2, t-N_2\}$ stochastically dominates $Z_1 = \max \{N_1, t-N_1\}$.
\end{lemma}

\begin{proof}
	\begin{equation}
		\probof{Z_2 \ge l} = \sum_{i = 0}^{t - l} {t \choose i} \left[\Paren{\frac{1+\dist}{2}}^i \Paren{\frac{1-\dist}{2}}^{t-i} + \Paren{\frac{1-\dist}{2}}^i \Paren{\frac{1+\dist}{2}}^{t-i} \right],\nonumber
	\end{equation}
	\begin{equation}
		\probof{Z_1 \ge l} = 2\cdot \sum_{i = 0}^{t - l} {t \choose i} \Paren{\frac12}^t.\nonumber
	\end{equation}
	Define $F(l) =\probof{Z_2 \ge l} - \probof{Z_1 \ge l} $. What we need to show is $F(l) \ge 0, \forall l \ge \frac{t}{2}$. First we observe that $\probof{Z_2 \ge \frac{t}{2}} =  \probof{Z_1 \ge \frac{t}{2}} = 1$ and  $\probof{Z_2 \ge t} = (\frac{1 +\dist}{2})^t + (\frac{1 - \dist}{2})^t  \ge 2 (\frac12)^t  = \probof{Z_1 \ge t}$. Hence $F(\frac{t}{2}) = 0, F(t) > 0$.
	Let
	\[
		f(l) = F(l+1) - F(l)= - {{t \choose l} } \left[\Paren{\frac{1+\dist}{2}}^l \Paren{\frac{1-\dist}{2}}^{t-l} + \Paren{\frac{1-\dist}{2}}^l \Paren{\frac{1+\dist}{2}}^{t-l} - 2 \Paren{\frac12}^t\right].
	\]
	Let $g(x) = \Paren{\frac{1+\dist}{2}}^x \Paren{\frac{1-\dist}{2}}^{t-x} + \Paren{\frac{1-\dist}{2}}^x \Paren{\frac{1+\dist}{2}}^{t-x} - 2 \Paren{\frac12}^t, x \in [t/2, t]$, then
	\[
		\frac{dg(x)}{dx} =  \ln\Paren{\frac{1+\dist}{1-\dist}} \cdot\left[\Paren{\frac{1+\dist}{2}}^x \Paren{\frac{1-\dist}{2}}^{t-x} - \Paren{\frac{1-\dist}{2}}^x \Paren{\frac{1+\dist}{2}}^{t-x}\right] \ge 0.
	\]
	
	We know $g(t/2) < 0, g(t) > 0$, hence $\exists x^*, s.t. g(x) \le 0, \forall x< x^*$ and $g(x) \ge 0, \forall x > x^*$. Because $f(l) =  - {t \choose l}  g(l)$, hence $\exists l^*, s.t. f(l) \le 0, \forall l \ge l^*$ and $f(l) \ge 0, \forall l < l^*.$
	Therefore, $F(l)$ first increases and then decreases, which means $F(l)$ achieves its minimum at $\frac{t}{2}$ or $t$. Hence $F(l) \ge 0$, completing the proof.
\end{proof}

For stochastic dominance, the following definition~\cite{Hollander12} will be useful.  
\begin{definition}
	A coupling $(X',Y')$ is a monotone coupling if $\probof{X' \ge Y'}=1$. 
\end{definition}

The following lemma states a nice relationship between stochastic dominance and monotone coupling, which is provided as Theorem~7.9 in \cite{Hollander12}

\begin{lemma} \label{stod:coup}
	Random variable $X$ stochastically dominates $Y$ if and only if there is a monotone coupling between $(X',Y')$ with $\probof{X' \geq Y'} = 1$. 
\end{lemma}

By Lemma~\ref{stod:coup}, there is a monotone coupling between $Z_1 = \max \{N_1, t-N_1\}$ and $Z_2 = \max \{N_2, t-N_2\}$. Suppose the coupling is $P^c_{Z_1,Z_2}$, we define the coupling between $X_1^t$ and $Y_1^t$ as following:

\begin{enumerate}
	\item Generate $X_1^t$ according to $P_1$ and count the number of one's in $X_1^t$ as $n_1$.
	\item Generate $n_2$ according to $P^c[Z_2| Z_1 = \max\{n_1, t- n_1\}]$.
	\item If $n_1 > t - n_1$, choose $n_2 - n_1$ of the zero's in $X_1^t$ uniformly at random and change them to one's to get $Y_1^t$. 
	\item If $n_1 < t - n_1$, choose $n_2 - (t - n_1)$ of the one's in $X_1^t$ uniformly at random and change them to zero's to get $Y_1^t$.
	\item If $n_1 = t - n_1$, break ties uniformly at random and do the corresponding action.
	\item Output $(X_1^t,Y_1^t)$.
\end{enumerate}

Since the coupling is monotone, and $\ham{X_1^t}{Y_1^t} = Z_2 - Z_1$ for every pair of $(X_1^t,Y_1^t)$, we get:
\[
	\expectation{\ham{X_1^t}{Y_1^t}} = \expectation{\max \{N_2, t-N_2\}} - \expectation{\max \{N_1, t-N_1\}}.
\]

Hence, to show lemma~\ref{lem:coup_p}, it suffices to show the following lemma:

\begin{lemma}
Suppose $N_1 \sim \bino{t}{\frac{1}{2}}, N_2 \sim \frac{1}{2}\bino{t}{\frac{1+\alpha}{2}} + \frac{1}{2}\bino{t}{\frac{1-\alpha}{2}} $.
\[
	\expectation{\max \{N_2, t-N_2\}} - \expectation{\max \{N_1, t-N_1\}} < C\cdot (\alpha^2t^{3/2} +  \alpha^4 t^{5/2} + \alpha^5 t^3)
\]
\end{lemma}

\begin{proof}
	
\begin{align}
	\ &\expectation{\max \{N_2, t-N_2\}} \nonumber\\
	 =\ & \sum_{0\le\ell\le t/2}(t/2 + \ell){t\choose \frac t2-\ell}\Paren{\Paren{\frac{1-\alpha}{2}}^{\frac t2-\ell}\Paren{\frac{1+\alpha}{2}}^{\frac t2+\ell}+\Paren{\frac{1+\alpha}{2}}^{\frac t2-\ell}\Paren{\frac{1-\alpha}{2}}^{\frac t2+\ell}} \nonumber\\
	=\ & \frac t2 +  \sum_{0\le\ell\le t/2} \ell{t\choose \frac t2-\ell}\Paren{\Paren{\frac{1-\alpha}{2}}^{\frac t2-\ell}\Paren{\frac{1+\alpha}{2}}^{\frac t2+\ell}+\Paren{\frac{1+\alpha}{2}}^{\frac t2-\ell}\Paren{\frac{1-\alpha}{2}}^{\frac t2+\ell}}. \nonumber
\end{align}

Consider a fixed value of $t$. Let 
\[
f(\alpha) = \sum_{0\le\ell\le t/2}\ell{t\choose \frac t2-\ell}\Paren{\Paren{\frac{1-\alpha}{2}}^{\frac t2-\ell}\Paren{\frac{1+\alpha}{2}}^{\frac t2+\ell}+\Paren{\frac{1+\alpha}{2}}^{\frac t2-\ell}\Paren{\frac{1-\alpha}{2}}^{\frac t2+\ell}}.
\]

The first claim is that this expression is minimized at $\alpha=0$. This is because of the monotone coupling between $Z_1$ and $Z_2$, which makes $\expectation{Z_2} \ge \expectation{Z_1}$. This implies that $f'(0)=0$, and by intermediate value theorem, there is $\beta\in[0,\alpha]$, such that 
\begin{align}
f(\alpha) = f(0)+\frac12\alpha^2\cdot f''(\beta).\label{eq:double-derivative}
\end{align}
We will now bound this second derivative. 
To further simplify, let
\[
g(\alpha) = \Paren{\frac{1-\alpha}{2}}^{\frac t2-\ell}\Paren{\frac{1+\alpha}{2}}^{\frac t2+\ell}+\Paren{\frac{1+\alpha}{2}}^{\frac t2-\ell}\Paren{\frac{1-\alpha}{2}}^{\frac t2+\ell}.
\]
Differentiating $g(\alpha)$, twice with respect to $\alpha$, we obtain, 
\begin{align}
g''(\alpha) = & ~\frac1{16} \cdot \Paren{\alpha^2(t^2-t)-4\alpha\ell(t-1)+4\ell^2-t}\Paren{\frac{1-\alpha}{2}}^{\frac t2-\ell-2}\Paren{\frac{1+\alpha}{2}}^{\frac t2+\ell-2}\nonumber\\
& + \frac1{16} \cdot \Paren{\alpha^2(t^2-t)+4\alpha\ell(t-1)+4\ell^2-t}\Paren{\frac{1+\alpha}{2}}^{\frac t2-\ell-2}\Paren{\frac{1-\alpha}{2}}^{\frac t2+\ell-2}.\nonumber
\end{align}

Then $g''(\alpha)$ can be bound by,
\begin{align}
g''(\alpha) \le \frac1{16}\cdot\Paren{\alpha^2t^2+4\ell^2}\Paren{\Paren{\frac{1-\alpha}{2}}^{\frac t2-\ell-2}\Paren{\frac{1+\alpha}{2}}^{\frac t2+\ell-2}+\Paren{\frac{1+\alpha}{2}}^{\frac t2-\ell-2}\Paren{\frac{1-\alpha}{2}}^{\frac t2+\ell-2}}.\nonumber
\end{align}

When $\alpha<\frac14$, $(1-\alpha^2)^2>\frac12$, and we can further bound the above expression by
\[
g''(\alpha) 
\le 2\cdot\Paren{\alpha^2t^2+4\ell^2}\Paren{\Paren{\frac{1-\alpha}{2}}^{\frac t2-\ell}\Paren{\frac{1+\alpha}{2}}^{\frac t2+\ell}+\Paren{\frac{1+\alpha}{2}}^{\frac t2-\ell}\Paren{\frac{1-\alpha}{2}}^{\frac t2+\ell}}.
\]

Suppose $X$ is a $\bino{t}{ \frac{1+\beta}2}$ distribution. Then, for any $\ell>0$,
\[
\probof{\absv{X-\frac t2}=\ell} = {t\choose \frac t2-\ell}\Paren{\Paren{\frac{1-\beta}{2}}^{\frac t2-\ell}\Paren{\frac{1+\beta}{2}}^{\frac t2+\ell}+\Paren{\frac{1+\beta}{2}}^{\frac t2-\ell}\Paren{\frac{1-\beta}{2}}^{\frac t2+\ell}}.
\]
Therefore, we can bound~\eqref{eq:double-derivative}, by
\[
f''(\beta) \le 2\cdot\Paren{\beta^2t^2 \expectation{\absv{X-\frac t2}}+4\expectation{\absv{X-\frac t2}^3}}.
\] 
For $X\sim \bino{\ns}{r}$, 
\begin{align*}
\expectation{\Paren{X-\ns r}^2}&= \ns r(1-r)\le \frac{\ns}4, \text{ and }\\
\expectation{\Paren{X-\ns r}^4}&= \ns r(1-r)\Paren{3r(1-r) (\ns-2)+1}\le 3\frac{\ns^2}{4}.
\end{align*}
We bound each term using these moments, 
\begin{align}
\expectation{\absv{X-\frac t2}} \le \expectation{\Paren{X-\frac t2}^2}^{1/2}
=& \Paren{t\frac{(1-\beta^2)}4+ \Paren{\frac{t\beta}{2}}^2}^{1/2}\nonumber \le   \sqrt t + {t\beta}.
\end{align}
We similarly bound the next term, 
\begin{align}
\expectation{\absv{X-\frac t2}^3} &\le  \expectation{\Paren{X-\frac t2}^4}^{3/4}\nonumber\\
&\le  \expectation{\Paren{X- \frac{t(1+\beta)}{2}+\frac{t\beta}{2}}^4}^{3/4}\nonumber\\
&\le   8\Paren{\expectation{\Paren{X- \frac{t(1+\beta)}2}^4}^{3/4}+\Paren{\frac{t\beta}2}^3}\nonumber\\
&\le  8\Paren{t^{3/2} + \Paren{\frac{t\beta}2}^3},\nonumber
\end{align}
where we use $(a+b)^4\le 8(a^4+b^4)$. 

Therefore, 
\[
f''(\beta) \le 64\cdot\Paren{\beta^2t^{5/2} + {t^{3/2} + {(t\beta)}^3}} \le 64\cdot\Paren{\alpha^2t^{5/2} + {t^{3/2} + {(t\alpha)}^3}}.
\]

As a consequence,
\[
	\expectation{\max \{N_2, t-N_2\}} - \expectation{\max \{N_1, t-N_1\}} = \dist^2 f''(\beta) \le 64\cdot (\alpha^2t^{3/2} +  \alpha^4 t^{5/2} + \alpha^5 t^3).
\]
completing the proof.
\end{proof}

\chapter{Privately Estimating Distribution Properties}
\label{cha:ins}
\section{Introduction}

How can we infer distribution properties given samples?
If data is in abundance, the solution may be simple -- the empirical distribution will approximate the true distribution.
However, challenges arise when data is scarce in comparison to the size of the domain. 
For example, it has recently been observed that there are several very rare genetic mutations which occur in humans, and we wish to know how many  such mutations exist~\cite{KeinanC12,TennessenBOFKGMDLJKJLGRAANBSBBABSN12,NelsonWEKSVSTBFWAZLZZLLLWTHNWACZWCNM12}. This is also a good example of performing statistical inference on a sensitive dataset.

Our focus in this chapter is to develop tools for privately estimating distribution properties, which is another important problem in statistical inference.
In particular, we study the tradeoff between statistical accuracy, privacy, and error rate in the sample size. 
Our model is that we are given sample access to some unknown discrete distribution $p$, over a domain of size $\ab$, which is possibly unknown in some tasks.
We wish to estimate the following properties:
\begin{itemize}
\item {\bf Support Coverage}: If we take $m$ samples from the distribution, what is the expected number of unique elements we expect to see?
\item {\bf Support Size}: How many elements of the support have non-zero probability?
\item {\bf Entropy}: What is the Shannon entropy of the distribution?
\end{itemize}
For more formal statements of these problems, see Section~\ref{sec:probs}.
We require that our output is $\dist$-accurate, satisfies $(\eps, 0)$-differential privacy, and is correct with probability $1 - \fp$.
The goal is to give an algorithm with minimal sample complexity $n$, while simultaneously being computationally efficient. 

\subsection{Results and Techniques}
{\bf \noindent Theoretical Results.} Our main results show that privacy can be achieved for all these problems at a very low cost.
For example, if one wishes to privately estimate entropy, this incurs an additional additive cost in the sample complexity which is very close to linear in $1/\dist\eps$.
We draw attention to two features of this bound.
First, this is independent of $\ab$.
All the problems we consider have complexity $\Theta(\ab/\log \ab)$, so in the primary regime of study where $\ab \gg 1/\dist\eps$, this small additive cost is dwarfed by the inherent sample complexity of the non-private problem.
Second, the bound is almost linear in $1/\dist\eps$.
We note that performing even the most basic statistical task privately, estimating the bias of a coin, incurs this linear dependence.
Surprisingly, we show that much more sophisticated inference tasks can be privatized at almost no cost.
In particular, these properties imply that the additive cost of privacy is $o(1)$ in the most studied regime where the support size is large.
In general, this is not true -- for many other problems, including distribution estimation and hypothesis testing, the additional cost of privacy depends significantly on the support size or dimension~\cite{DiakonikolasHS15,CaiDK17,AcharyaSZ18,AliakbarpourDR18}.
We also provide lower bounds, showing that our upper bounds are almost tight.
A more formal statement of our results appears in Section~\ref{sec:ins_results}.

{\bf \noindent  Experimental Results.} We demonstrate the efficacy of our method with experimental evaluations.
As a baseline, we compare with the non-private algorithms of~\cite{OrlitskySW16} and~\cite{WuY18}.
Overall, we find that our algorithms' performance is nearly identical, showing that, in many cases, privacy comes (essentially) for free. 
We begin with an evaluation on synthetic data.
Then, inspired by~\cite{ValiantV13a,OrlitskySW16}, we analyze a text corpus consisting of words from Hamlet, in order to estimate the number of unique words which occur.
Finally, we investigate name frequencies in the US census data.
This setting has been previously considered by~\cite{OrlitskySW16}, but we emphasize that this is an application where private statistical analysis is critical.
This is proven by efforts of the US Census Bureau to incorporate differential privacy into the 2020 US census~\cite{DajaniLSKRMGDGKKLSSVA17}.

{\bf \noindent  Techniques.} Our approach works by choosing statistics for these tasks which possess bounded sensitivity, which is well-known to imply privacy under the Laplace or Gaussian mechanism.
We note that bounded sensitivity of statistics is not always something that can be taken for granted.
Indeed, for many fundamental tasks, optimal algorithms for the non-private setting may be highly sensitive, thus necessitating crucial modifications to obtain differential privacy~\cite{AcharyaDK15, CaiDK17}.
Thus, careful choice and design of statistics must be a priority when performing inference with privacy considerations.

To this end, we leverage recent results of~\cite{AcharyaDOS17}, which studies estimators for non-private versions of the problems we consider.
The main technical work in their paper exploits bounded sensitivity to show sharp cutoff-style concentration bounds for certain estimators, which operate using the principle of best-polynomial approximation.
They use these results to show that a single algorithm, the Profile Maximum Likelihood (PML), can estimate all these properties simultaneously.
On the other hand, we consider the sensitivity of these estimators for purposes of privacy -- the same property is utilized by both works for very different purposes, a connection which may be of independent interest.

We note that bounded sensitivity of a statistic may be exploited for purposes other than privacy.
For instance, by McDiarmid's inequality, any such statistic also enjoys very sharp concentration of measure, implying that one can boost the success probability of the test at an additive cost which is logarithmic in the inverse of the failure probability.
One may naturally conjecture that, if a statistical task is based on a primitive which concentrates in this sense, then it may also be privatized at a low cost.
However, this is not true -- estimating a discrete distribution in $\ell_1$ distance is such a task, but the cost of privatization depends significantly on the support size~\cite{DiakonikolasHS15}.

One can observe that, algorithmically, our method is quite simple: compute the non-private statistic, and add a relatively small amount of Laplace noise.
The non-private statistics have recently been demonstrated to be practical~\cite{OrlitskySW16, WuY18}, and the additional cost of the Laplace mechanism is minimal.
This is in contrast to several differentially private algorithms which invoke significant overhead in the quest for privacy.
Our algorithms attain almost-optimal rates (which are optimal up to constant factors for most parameter regimes of interest), while simultaneously operating effectively in practice, as demonstrated in our experimental results.

\subsection{Related Work}
Over the last decade, there have been a flurry of works on the problems we study by the computer science and information theory communities, including Shannon and R\'enyi entropy estimation~\cite{Paninski03, ValiantV17b, JiaoVHW17, AcharyaOST17, ObremskiS17, WuY18}, support coverage and support size estimation~\cite{ValiantV17b, ValiantV16, OrlitskySW16, RaghunathanVZ17, WuY18}. 
A recent paper studies the general problem of estimating functionals of discrete distribution from samples in terms of the smoothness of the functional~\cite{FukuchiS17}.
These have culminated in a nearly-complete understanding of the sample complexity of these properties, with optimal sample complexities (up to constant factors) for most parameter regimes.


Recently, there has been significant interest in performing statistical tasks under differential privacy constraints.
Perhaps most relevant to this work are~\cite{CaiDK17, AcharyaSZ18, AliakbarpourDR18,Sheffet18,AcharyaCFT18}, which study the sample complexity of differentialy privately performing classical distribution testing problems, including identity and closeness testing.
Some recent work focuses on the testing of simple hypotheses: \cite{CanonneKMSU19} studies the sample complexity of this problem, while~\cite{AwanS18} provides a uniformly most powerful (UMP) test for binomial data (though~\cite{BrennerN14} shows that UMP tests can not exist in general).
Other works investigating private hypothesis testing include~\cite{WangLK15,GaboardiLRV16,KiferR17,KakizakiSF17,Rogers17,GaboardiR17}, which focus less on characterizing the finite-sample guarantees of such tests, and more on understanding their asymptotic properties and applications to computing p-values.
There has also been study on private distribution learning~\cite{DiakonikolasHS15,DuchiJW17,KarwaV18,AcharyaSZ18a,KamathLSU18}, in which we wish to estimate parameters of the distribution, rather than just a particular property of interest.
Similar to our work, \cite{Smith11} shows that the cost of privacy in statistical estimation can be a lower order term -- roughly, he shows that this is the case for any statistic which is asymptotically normal.
A number of other problems have been studied with privacy requirements, including clustering~\cite{WangWS15,BalcanDLMZ17}, principal component analysis~\cite{ChaudhuriSS13,KapralovT13,HardtP14b}, ordinary least squares~\cite{Sheffet17}, and much more.

\subsection{Organization}
We begin with notation and preliminaries in Section~\ref{sec:ins_prelim}.
Our theoretical results and analysis are described in Section~\ref{sec:ins_results} and~\ref{sec:ins_theory}.
Finally, our experimental investigations are in~\ref{sec:ins_exp} and~\ref{sec:supp-experiments}.

\section{Preliminaries}
\label{sec:ins_prelim}


We now describe the classical distribution property estimation problem, and then state the problem under differential privacy. 

\paragraph{Property Estimation.} Given $\dist,\beta$, $f$, and independent samples $\Xon$ from an unknown distribution $\p$, design an estimator $\hat f:\Xon\to\RR$ such that with probability at least $1-\beta$, $\absv{\hat{f}(\Xon)-f(\p)}<\dist$. The \emph{sample complexity} of $\hat f$,
$
S_{\tt PE}(f, \hat{f}, \dist, \beta) \ed \min\{n: \probof{\absv{\hat{f}(\Xon)-f(\p)}>\dist}<\beta\}
$ is the smallest number of samples to estimate $f$ to accuracy $\dist$, and error $\beta$.  
We study the problem for $\beta = 1/3$, and by the median trick, we can boost the success probability to $1 - \beta$ with an additional multiplicative $\log (1/\beta)$ more samples. Therefore, focusing on $\beta = 1/3$, we define $S_{\tt PE}(f,\hat{f}, \dist) \ed S_{\tt PE} (f,\hat{f}, \dist, 1/3)$.
The sample complexity of estimating a property $f(\p)$ is the minimum sample complexity over all estimators: $S_{\tt PE} (f,\dist) = \min_{\hat f} S_{\tt PE} (f, \hat f, \dist)$.


\paragraph{Private Property Estimation.} Given $\dist, \eps, \beta$, $f$, and independent samples $\Xon$ from an unknown distribution $\p$, design an $\eps$-differentially private estimator $\hat f:\Xon\to\RR$ such that with probability at least $1-\beta$, $\absv{\hat{f}(\Xon)-f(\p)}<\dist$. Similar to the non-private setting, the \emph{sample complexity} of $\eps$-differentially private estimation problem is $S_{\tt PE}(f,\dist, \eps) = \min_{\hat f: \hat f \text{ is $\eps$-DP}}S_{\tt PE} (f, \hat f, \dist, 1/3)$, the smallest number of samples $\ns$ for which there exists such an $\eps$-DP $\pm\dist$ estimator with error probability at most 1/3. 

The following lemma can be viewed as a direct corollary of the Laplace mechanism, which is introduced in Lemma~\ref{lem:laplace_mechanism}.
%

\begin{lemma}
\label{lem:main-sensitivity}
Let the \emph{sensitivity} of an estimator $\hat{f}:[\ab]^{\ns}\to\RR$ be $\Delta_{n,\hat{f}}\ed \max_{\ham{\Xon}{\Yon}\le1} \absv{\hat{f}(\Xon)-\hat{f}(\Yon)}$, and $D_{\hat f}(\dist,\eps) = \min\{{\ns}: \Delta_{n,\hat{f}}\le \dist\eps\}$, then
\[
S_{\tt PE} (f,\dist, \eps)  = O\Paren{\min_{\hat f} \left\{S_{\tt PE}(f,\hat f, \dist/2)+ D_{\hat f}\left(\frac{\dist}{4},\eps\right)\right\}}.
\]
\end{lemma}

\begin{proof}
Lemma~\ref{lem:laplace_mechanism} showed that for a function with sensitivity $\Delta_{n,\hat{f}}$, adding Laplace noise $X \sim Lap(\Delta_{n,\hat{f}}/\eps)$ makes the output $\eps$-differentially private.
By the definition of $D_{\hat f}(\frac{\dist}{4},\eps)$, the Laplace noise we add has parameter at most $\frac{\dist}{4}$. Recall that the probability density function of $Lap(b)$ is $\frac1{2b} e^{-\frac{|x|}{b}}$, hence we have $\probof{|X| > \alpha/2} < \frac{1}{e^2}$. By the union bound, we get an additive error larger than $\alpha = \frac{\alpha}{2} + \frac{\alpha}{2}$ with probability at most $1/3 + \frac{1}{e^2} < 0.5$. Hence, with the median trick, we can boost the error probability to $1/3$, at the cost of a constant factor in the number of samples.
\end{proof}

%

\subsection{Problems of Interest}
\label{sec:probs}

\paragraph{Support Size.} The support size of a distribution $\p$ is $S(\p) =\absv{\{x:\p(x)>0\}}$, the number of symbols with non-zero probability values. However, notice that estimating $S(\p)$ from samples can be hard due to the presence of symbols with negligible, yet non-zero probabilities. To circumvent this issue,~\cite{RaskhodnikovaRSS09} proposed to study the problem when the smallest probability is bounded. Let $\Dgk
  \ed\left\{\p\in\Delta:
  \p(x)\in\{0\}\cup\left[1/\ab,1\right]\right\}$
be the set of all distributions where all non-zero probabilities have value at least $1/\ab$. For $\p\in\Dgk$, our goal is to estimate $S(\p)$ up to $\pm\dist\ab$ with the least number of samples from $\p$.

\paragraph{Support Coverage.} For a distribution $\p$, and an integer $m$, let $S_m(\p) = \sum_{\smb} (1 - (1-{\p(\smb)})^m)$, be the expected number of symbols that appear when we obtain $m$ independent samples from the distribution $\p$. The objective is to find the least number of samples $\ns$ in order to estimate $S_m\Paren{\p}$ to an additive $\pm \dst m$.

Support coverage arises in many ecological and biological studies~\cite{ColwellCGLMCL12} to quantify the number of \emph{new} elements (gene mutations, species, words, etc) that can be expected to be seen in the future.  
Good and Toulmin~\cite{GoodT56} proposed an estimator that  for any constant $\dst$, requires $m/2$ samples to estimate  $S_m(p)$.

\paragraph{Entropy.} 
The Shannon entropy of a distribution $\p$ is $H(p) =\sum_x p(x)\log\frac1{p(x)}$,
$H(p)$ is a central object in information theory~\cite{CoverT06}, and also arises in many fields such as machine learning~\cite{Nowozin12}, neuroscience~\cite{BerryWM97, NemenmanBRS04}, and others. Estimating $H(p)$ is hard with any finite number of samples due to the possibility of infinite support. To circumvent this,
    a natural approach is to consider distributions in $\Dk$.
 The goal is to estimate the entropy of a distribution in $\Dk$ to an additive $\pm \dist$, where $\Dk$ is all discrete distributions over at most $k$ symbols.

\section{Statement of Results}
\label{sec:ins_results}
Our theoretical results for estimating support coverage, support size, and entropy are given below.
Algorithms for these problems and proofs of these statements are provided in Section~\ref{sec:ins_theory}.
Our experimental results are described and discussed in Section~\ref{sec:ins_exp}.


\begin{restatable}{theorem}{supportcoverage}
\label{thm:supportcoverage}
The sample complexity of support coverage estimation is 

\begin{equation*}
S_{\tt PE}(S_m, \dist, \eps)=
\begin{cases}
O\Paren{ \frac{m\log (1/\dist)}{\log m}+ \frac{m\log (1/\dist)}{\log (2 + \eps m)} }, & \text{when $m \ge \frac1{\dist \eps}$} \\
O\Paren{\frac1{\dist^2} + \frac1{\dist\eps}}, & \text{when $\frac1{\dist} \le m \le \frac1{\dist \eps}$} \\
O\Paren{m^2 + \frac{m}{\eps}}. & \text{when $ m \le \frac1{\dist}$} 
\end{cases}
\end{equation*}

Furthermore, 
$$
S_{\tt PE}(S_m, \dist, \eps)= \Omega\Paren{\frac{m\log (1/\dist)}{\log m} + \frac{1}{\dist \eps}  }.
$$
\end{restatable}


\begin{restatable}{theorem}{ssize}
\label{thm:ssize}
The sample complexity of support size estimation is 
\begin{equation*}
S_{\tt PE}(S, \dist, \eps)= 
\begin{cases}
O\Paren{ \frac{\ab\log^2 (1/\dist)}{\log \ab}+ \frac{\ab\log^2 (1/\dist)}{\log (2 + \eps \ab)} }, & \text{when $\ab \ge \frac1{\dist \eps}$} \\
O\Paren{\ab \log(1/\dist) + \frac1{\dist\eps}}, & \text{when $\frac1{\dist} \le \ab \le \frac1{\dist \eps}$} \\
O\Paren{\ab \log \ab + \frac{\ab}{\eps}}. & \text{when $ \ab \le \frac1{\dist}$} 
\end{cases}
\end{equation*}
Furthermore,
\begin{equation*}
S_{\tt PE}(S, \dist, \eps)= 
\begin{cases}
 \Omega\Paren{ \frac{\ab\log^2 (1/\dist)}{\log \ab} + \frac{1}{\dist \eps} }, & \text{when $\ab \ge \frac1{\dist}$} \\
 \Omega\Paren{\ab \log \ab + \frac{\ab}{\eps}}. & \text{when $ \ab \le \frac1{\dist}$} 
\end{cases}
\end{equation*}
\end{restatable}

\begin{restatable}{theorem}{entropy}
\label{thm:entropy}
Let $\lambda>0$ be \emph{any} small fixed constant. 
For instance, $\lambda$ can be chosen to be any constant between $0.01$ and $1$.
We have the following upper bounds on the sample complexity of entropy estimation:
$$
S_{\tt PE}(H, \dist, \eps)= O\Paren{\frac{\ab}{\dist}+\frac{\log^2(\min\{\ab,\ns\})}{\dist^2}+\frac{1}{\dist\eps}\log\Paren{\frac1{\dist\eps}}}
$$
and
$$
S_{\tt PE}(H, \dist, \eps)= O\Paren{\frac{\ab}{\lambda^2\dist\log\ab}+\frac{\log^2(\min\{\ab,\ns\})}{\dist^2} + \Paren{\frac1{\dist\eps}}^{1+\lambda}}.
$$
Furthermore,
$$
S_{\tt PE}(H, \dist, \eps) =\Omega\Paren{\frac{\ab}{\dist\log\ab}+\frac{\log^2(\min\{\ab,\ns\})}{\dist^2}+\frac{\log\ab}{\dist\eps}}.
$$
\end{restatable}

\noindent We provide some discussion of our results. 
At a high level, we wish to emphasize the following two points:
\begin{enumerate}
\item Our upper bounds show that the cost of privacy in these settings is often negligible compared to the sample complexity of the non-private statistical task, especially when we are dealing with distributions over a large support.
Furthermore, our upper bounds are almost tight in all parameters.
\item The algorithmic complexity introduced by the requirement of privacy is minimal, consisting only of a single step which noises the output of an estimator.
In other words, our methods are realizable in practice, and we demonstrate the effectiveness on several synthetic and real-data examples.
\end{enumerate}

Before we continue, we emphasize that, in Theorems~\ref{thm:supportcoverage} and~\ref{thm:ssize}, we consider the ``sublinear'' regime to be of primary interest (when $m \geq \frac1{\dist\eps}$ or $k \geq \frac1{\dist\eps}$, respectively), both technically, and in terms of parameter regimes which may be of greatest interest in practice.
We include results for other regimes mostly for completeness.

First, we examine our results on support coverage and support size estimation in the sublinear regime, when $m \geq \frac{1}{\dist\eps}$ (focusing on support coverage for simplicity, but support size is similar).
In this regime, if $\eps = \Omega(m^{\gamma}/m)$ for any constant $\gamma > 0$, then up to constant factors, our upper bound is within a constant factor of the optimal sample complexity without privacy constratints.
In other words, for most meaningful values of $\eps$, privacy comes for free. 
In the non-sublinear regime for these problems, we provide upper and lower bounds which match in a number of cases.
We note that in this regime, the cost of privacy may not be a lower order term -- however, this regime only occurs when one requires very high accuracy, or unreasonably large privacy, which we consider to be of somewhat lesser interest.

Next, we turn our attention to entropy estimation.
We note that the second upper bound in Theorem~\ref{thm:entropy} has a parameter $\lambda$ that indicates a tradeoff between the sample complexity incurred in the first and third term.
This parameter determines the degree of a polynomial to be used for entropy estimation.
As the degree becomes smaller (corresponding to a large $\lambda$), accuracy of the polynomial estimator decreases, however, at the same time, low-degree polynomials have a small sensitivity, allowing us to privatize the outcome. 

In terms of our theoretical results, one can think of $\lambda = 0.01$.
With this parameter setting, it can be observed that our upper bounds are almost tight.
For example, one can see that the upper and lower bounds match to either logarithmic factors (when looking at the first upper bound), or a very small polynomial factor in $1/\dist\eps$ (when looking at the second upper bound).
For our experimental results, we empirically determined an effective value for the parameter $\lambda$ on a single synthetic instance.
We then show that this choice of parameter generalizes, giving highly-accurate private estimation in other instances, on both synthetic and real-world data.

\section{Algorithms and Analysis}
\label{sec:ins_theory}
In this section, we prove our results for support coverage in Section~\ref{sec:coverage}, support size in Section~\ref{sec:ssize}, and entropy in Section~\ref{sec:entropy}. 
In each section, we first describe and analyze our algorithms for the relevant problem.
We then go on to describe and analyze a lower bound construction, showing that our upper bounds are almost tight.

All our algorithms fall into the following simple framework:
\begin{enumerate}
\item Compute a non-private estimate of the property;
\item Privatize this estimate by adding Laplace noise, where the parameter is determined through analysis of the estimator and potentially computation of the estimator's sensitivity.
\end{enumerate}

\subsection{Support Coverage Estimation}
\label{sec:coverage}
In this section, we prove Theorem~\ref{thm:supportcoverage}, about support coverage estimation: 
\supportcoverage*
Our upper bound is analyzed in Section~\ref{sec:coverage-ub}, while our lower bound is proved in Section~\ref{sec:coverage-lb}.
\subsubsection{Upper Bound for Support Coverage Estimation}
\label{sec:coverage-ub}
We split the analysis into two regimes.
First, we focus on the case where $m \leq \frac1{\dist\eps}$, and we prove the upper bound $O\left(\frac{1}{\dist^2} + \frac{1}{\dist \eps}\right)$.
Note that the problem is identical for any $\dist < \frac1m$, since this corresponds to estimating the support coverage exactly, and the above bound simplifies to $O\left(m^2 + \frac{m}{\eps}\right)$.
The algorithm in this case is simple: since $\ns = \Omega(m)$, we group the dataset into $\ns/m$ batches of size $m$.
Let $Y_j$ be the number of unique symbols observed in batch $j$.
Our estimator is 
\[
\hat{S}_m(\Xon) = \frac{m}{\ns} \sum_{j =1}^{\ns/m} Y_j.
\]

Observe that $\expectation{Y_j} = S_m(\p)$, and that $\Var[Y_j] \leq m$.
The latter can be seen by observing that $Y_j$ is the sum of $m$ negatively correlated indicator random variables, each one being the indicator of whether that sample in the batch is the first time the symbol is observed.
This gives that $\hat{S}_m(\Xon)$ is an unbiased estimator of $S_m(\p)$, with variance $O(m^2/\ns)$.
By Chebyshev's inequality, since we want an estimate which is accurate up to $\pm \dist m$, this gives us that $C_{\hat{S}_m}(S_m(\p),\dist/2) = O\left(\frac{1}{\dist^2}\right)$.
Furthermore, we can see that the sensitivity of $\hat{S}_m(\Xon)$ is at most $2m/\ns$.
By Lemma~\ref{lem:main-sensitivity}, there is a private algorithm for support coverage estimation as long as 
\[
\Delta\Paren{\frac{\hat{S}_m(\Xon)}m} \le \dist\eps.
\]
With the above bound on sensitivity, this is true with $n = O(1/\dist\eps)$, giving the desired upper bound.

Now, we turn our attention to the case where $m \geq \frac1{\dist\eps}$, and we prove the upper bound $O\left(\frac{m\log (1/\dist)}{\log m}+ \frac{m\log (1/\dist)}{\log (2 + \eps m)}\right)$.
Let $\prfi{i}$ be the number of symbols that appear $i$ times in $\Xon$. We will use the following non-private support coverage estimator from~\cite{OrlitskySW16}:
\[
\hat{S}_m(\Xon) = \sum_{i=1}^{\ns}\prfi{i}\Paren{1-(-t)^i\cdot \probof{Z\ge i}},
\] 
where $Z$ is a Poisson random variable with mean $r$ (which is a parameter to be instantiated later), and $t=(m-n)/n$. 

Our private estimator of support coverage is derived by adding Laplace noise to this non-private estimator with the appropriate noise parameter, and thus the performance of our private estimator, is analyzed by bounding the sensitivity and the bias of this non-private estimator according to Lemma~\ref{lem:main-sensitivity}.

The sensitivity and bias of this estimator is bounded in the following lemmas.  
\begin{lemma} \label{lem:sensitivity-coverage}
Suppose $m>2\ns$, then the maximum coefficient of $\prfi{i}$ in $\hat{S}_m(p)$ is at most $1+e^{r(t-1)}$. 
\end{lemma}

\begin{proof}
By the definition of $Z$, we know $\probof{Z\ge i} = \sum_{k = i}^{\infty} e^{-r}\frac{r^k}{k!}$, hence we have:
\begin{align}
|1+(-t)^i\cdot \probof{Z\ge i}| &\le  1+ t^i \sum_{k = i}^{\infty} e^{-r}\frac{r^k}{k!}   \nonumber \\
						&\le  1 + e^{-r} \sum_{k = i}^{\infty}\frac{(rt)^k}{k!}  \nonumber \\
					   	&\le 1 + e^{-r} \sum_{k = 0}^{\infty}\frac{(rt)^k}{k!}  \nonumber \\
						&= 1 + e^{r(t-1)} \nonumber
\end{align}
\end{proof}

\noindent The bias of the estimator is bounded in Lemma 4 of~\cite{AcharyaDOS17}:
\begin{lemma}\label{lem:bias-coverage}
Suppose $m>2\ns$, then
\[
\absv{\expectation{\hat{S}_m(\Xon)} - S_m(\p)} \le 2+2e^{r(t-1)}+\min(m, S(p))\cdot e^{-r}. 
\]
\end{lemma}

Using these results, letting $r = \log (1/\alpha)$,~\cite{OrlitskySW16} showed that there is a constant $C$, such that with $n = C\frac{m}{\log m}\log(1/\dist)$ samples, with probability at least 0.9, 
\[
\absv{\frac{\hat{S}_m(\Xon)}m- \frac{S_m(\p)}m} \le \dist.
\]



Our upper bound in Theorem~\ref{thm:supportcoverage} is derived by the following analysis of the sensitivity of $\frac{\hat{S}_m(\Xon)}m$.
	
If we change one sample in $\Xon$, at most two of the $\prfi{j}$'s change. Hence by Lemma~\ref{lem:sensitivity-coverage}, the sensitivity of the estimator satisfies
\begin{align}
	\Delta\Paren{\frac{\hat{S}_m(\Xon)}m} \le & \frac2m\cdot\Paren{1+e^{r(t-1)}}.\label{eqn:bound-sen-cov}
\end{align}

By Lemma~\ref{lem:main-sensitivity}, there is a private algorithm for support coverage estimation as long as 
\[
\Delta\Paren{\frac{\hat{S}_m(\Xon)}m} \le \dist\eps,
\]
which by~\eqref{eqn:bound-sen-cov} holds if
$$
2(1 + \exp(r(t-1))) \le \dist\eps m.
$$
Let $r = \log (3/\alpha)$, note that $t-1 = \frac m{\ns} -2$. Suppose $\dist\eps m>2$, then, the condition above reduces to 
\[
\log \Paren{\frac{3}{\dist}} \cdot \Paren{\frac mn-2} \le \log \Paren{\frac12\dist\eps m -1}.
\]

This is equivalent to
\begin{align}
 n & \ge \frac{m \log (3/\dist)}{\log (\frac12\dist \eps m - 1) + 2 \log (3/\dist)}  \nonumber \\
    &= \frac{m \log (3/\dist)}{\log ( \frac32 \eps m - 3 / \dist) +  \log  (3/\dist)}  \nonumber
\end{align}

Suppose $\dist\eps m >2$, then the condition above reduces to the requirement that
\[
n = \Omega\Paren{\frac{m\log (1/\dist)}{\log (2 + \eps m)}}.
\]

\subsubsection{Lower Bound for Support Coverage Estimation}
\label{sec:coverage-lb}

We now prove the lower bound described in Theorem~\ref{thm:supportcoverage}.
Note that the first term in the lower bound is the sample complexity of non-private support coverage estimation, shown in~\cite{OrlitskySW16}.
Therefore, we turn our attention to prove the last term in the sample complexity.

%
%

Consider the following two distributions. $u_1$ is uniform over $[m(1+\dist)]$. $u_2$ is distributed over $m+1$ elements $[m] \cup \{\triangle\}$ where $u_2[i] = \frac{1}{m(1+\dist)} \forall i \in [m]$ and $u_2[\triangle] = \frac{\dist}{1+\dist}$. Moreover, $\triangle \notin [m(1+\dist)]$. Then,

\[
S_m(u_1) = m(1+\dist)\cdot \Paren {1-  \Paren{1-\frac1{m(1+\dist)}}^m},
\]

and 
\[
S_m(u_2) = m\cdot \Paren {1-  \Paren{1-\frac1{m(1+\dist)}}^m} + \Paren {1-  \Paren{1-\frac\dist{1+\dist}}^m} \nonumber 
\]

hence,
\begin{align}
&S_m(u_2) - S_m(u_1) \nonumber \\
& = m\dist \cdot \Paren {1-  \Paren{1-\frac1{m(1+\dist)}}^m} -  \Paren {1-  \Paren{1-\frac\dist{1+\dist}}^m} \nonumber \\
& = \Omega(\dist m) \nonumber 
\end{align}

Hence we know there support coverage differs by $\Omega(\dist m)$. Moreover, their total variation distance is $\frac{\dist}{1+\dist}$.
The following lemma is folklore, based on the coupling interpretation of total variation distance, and the fact that total variation distance is subadditive for product measures.

\begin{lemma} \label{lem:min_coupling}
For any two distributions $\p$, and $\q$, there is a coupling between $\ns$ i.i.d.\ samples from the two distributions with an expected Hamming distance of $\dtv{\p}{\q}\cdot\ns$.  
\end{lemma}


Using Lemma~\ref{lem:min_coupling} and $\dtv{u_1}{u_2} = \frac{\dist}{1+\dist}$, we have
\begin{lemma} \label{lem:coupling_coverage}
Suppose $u_1$ and $u_2$ are as defined before, there is a coupling between $u_1^\ns$ and $u_2^\ns$ with expected Hamming distance equal to $\frac{\dist}{1+\dist} \ns$.
\end{lemma}

Moreover, given $\ns$ samples, we must be able to privately distinguish between $u_1$ and $u_2$ given an $\dist$ accurate estimator of support coverage with privacy considerations. Thus, according to Theorem~\ref{thm:coupling} and Lemma~\ref{lem:coupling_coverage}, we have:
\[
\frac\dist{1+\dist} \ns\ge \frac1{\eps}\Rightarrow n = \Omega\Paren{\frac1{\eps\dist}}.
\]

%
%
%

\subsection{Support Size Estimation}
In this section, we prove our main theorem about support size estimation, Theorem~\ref{thm:ssize}:

\ssize*

\noindent Our upper bound is described and analyzed in Section~\ref{sec:ssize-ub}, while our lower bound appears in Section~\ref{sec:ssize-lb}.
\label{sec:ssize}

\subsubsection{Upper Bound for Support Size Estimation}
\label{sec:ssize-ub}


We split the analysis into two regimes. First we consider the ``sparse'' case, where the amount of data is relatively small.
In particular, $n < \frac{k \log \frac{3}{\alpha}}{2}$. In this case we show a bound of $O\Paren{ \frac{\ab\log^2 (1/\dist)}{\log \ab}+ \frac{\ab\log^2 (1/\dist)}{\log (2 + \eps \ab)} }$. This upper bound is less than $\frac{k \log \frac{3}{\alpha}}{2}$ only when $\ab = \Omega \Paren{\frac1{\dist \eps} }$, which is the condition for the sparse case.

\paragraph{Sparse case} In \cite{OrlitskySW16}, it is shown that the support coverage estimator can be used to obtain optimal results for estimating the support size of a distribution.
In this fashion, taking $m=\ab\log(3/\dist)$, we we may use an estimator of the support coverage $S_m(\p)$ as an estimator of $S(\p)$.
In particular, their result is based on the following observation. 

\begin{lemma}
	\label{lem:sssc}
	Suppose $m\ge \ab\log (3/\dist)$, then for any $\p\in\Dgk$, 
	\[
	\absv{S_m(\p) - S(\p)}\le \frac{\dist\ab}3.
	\]
\end{lemma}

\begin{proof}
	From the definition of $S_m(\p)$, we have $S_m(\p)\le S(\p)$. For the other side, 
	\begin{align}
	S(p) -S_m(p) &= \sum_{\smb} \Paren{1-\p(x)}^m \le \sum_{\smb} e^{-m\p(x)}  \nonumber \\
	&\le \ab\cdot e^{-\log (3/\dist)}  ~~~~~ = \frac{\ab\dist}{3}. \qedhere
	\end{align}
\end{proof}

Therefore, estimating $S_m(\p)$ for $m = k\log (3/\dist)$, up to $\pm \dist\ab/3$.. Therefore, the goal is to determine the smallest value of $\ns$ to solve the support coverage problem for $m = k\log (3/\dist)$.


Suppose $r = \log (3/\dist)$, and $m = \ab\log (3/\dist) = \ab \cdot r$ in the support coverage problem. Then, we have 
\begin{align}
t= \frac mn-1 = \frac{\ab\log (3/\dist)}{n}-1.\label{eq:exp-t}
\end{align}
Then, by Lemma~\ref{lem:bias-coverage} in the previous section, we have 
\begin{align}
&~~~\absv{\expectation{\hat{S}_m(\Xon)} - S(\p)} \nonumber\\ 
&\le \absv{\expectation{\hat{S}_m(\Xon)} - S_m(\p)}  + \absv{S_m(\p) - S(\p)}\nonumber\\ 
&\le 2+2e^{r(t-1)}+\min\{m, \ab\}\cdot e^{-r} + \frac{k\dist}{3}\nonumber\\
&\le 2+2e^{r(t-1)}+\ab\cdot e^{-\log (3/\dist)} + \frac{k\dist}{3}\nonumber\\
&\le  2 + 2e^{r(t-1)}+ 2\frac{k\dist}{3}.\nonumber
\end{align}

We will find conditions on $\ns$ such that the middle term above is at most $\ab\dist$. Toward this end, note that $2e^{r(t-1)} \le \dist\ab$ holds if and only if ${r(t-1)} \le \log\Paren{\frac{\dist\ab}{2}}$. Plugging in~\eqref{eq:exp-t}, this holds when
\begin{align}
\log (3/\dist)\cdot \Paren{\frac{\ab\log (3/\dist)}{n}-2} \le \log\Paren{\frac{\dist\ab}{2}},\nonumber
\end{align}
which is equivalent to 

\begin{align}
\ns\ge \frac{\ab\log^2(3/\dist)}{\log{\frac{\dist\ab}{2}}+2\log \frac 3\dist } = O\Paren{ \frac{\ab\log^2(1/\dist)}{\log{\ab }} }\nonumber
\end{align}
where we have assumed without loss of generality that $\dist>\frac1\ab$. 

The computations for sensitivity are very similar. From Lemma~\ref{lem:main-sensitivity}, we need to find the value of $\ns$ such that 
\[
2+2e^{r(t-1)}\le \dist\eps\ab, 
\]
where we assume that $\ns \le \frac12\ab\log (3/\dist)$, else we just add noise to the true number of observed distinct elements.
By computations similar to the previous case, this reduces to 
\begin{align}
\ns \ge \frac{\ab\log^2(3/\dist)}{\log{\frac{\dist\eps\ab}{2}}+\log \frac 3\dist}.\nonumber
\end{align}
Therefore, this gives us a sample complexity of
\begin{align*}
\ns  = O\Paren{\frac{\ab\log^2(1/\dist)}{\log\Paren{2 + \eps\ab}}}
\end{align*}
for the sensitivity result to hold.

\paragraph{Dense case}
Then let us consider the dense case when $k \le \frac{1}{\alpha \eps}$. The algorithm under this case will be the following. Let $W(\Xon)$ denote the set of symbols which appear in $\Xon$  and let $N_x$ denote the number of times $x$ appears, then our non-private estimator is 
\[
\hat{S}(\Xon) = \sum_{x \in W(\Xon)} \min \left \{ 1,\frac{N_x}{\frac{\ns}{3\ab}} \right\}.
\]

To analyze the performance of the algorithm, we consider two cases, the case when $ \ab \le \frac1{\dist}$ and the case when $\frac1{\dist} \le \ab \le \frac1{\dist \eps}$.

 When $ \ab \le \frac1{\dist}$, we have $\ab \dist <1$, which means we need to know the exact support size. Our algorithm gives correct answer when all the symbols appearing at least $\frac{\ns}{3\ab}$ times. For any symbol $x$ with $\p(x) \ge \frac1{\ab}$, according to the Chernoff bound, $\probof{N_x < \frac{\ns}{3\ab}} \le \exp(- \frac{\frac{2\ns^2}{9\ab^2}}{ \ns \cdot \frac1{\ab}}) = \exp(-\frac{2\ns}{9\ab})$. Let $\ns \ge 18\ab \log \ab$, we have $\probof{N_x < \frac{\ns}{3\ab}} \le \frac1{\ab^4}$. Then according to the union bound, the probability of all the symbols appearing at least $\frac{\ns}{3\ab}$ is greater than $1-\frac1{\ab^3}$. When $\ab \ge 2$, this is larger than $2/3$, which means our algorithm gives correct answer with probability more than $\frac{2}{3}$.

Furthermore, we can see that the sensitivity of $\hat{S}(\Xon)$ is at most $3\ab/\ns$.
By Lemma~\ref{lem:main-sensitivity}, there is a private algorithm for support size estimation as long as 
\[
\Delta\Paren{\hat{S}(\Xon) } \le \eps.
\]
With the above bound on sensitivity, this is true with $n = O(\ab/\eps)$, giving the desired upper bound.

Next we consider the case when $\frac1{\dist} \le \ab \le \frac1{\dist \eps}$. For any symbol $x$ with $\p(x) \ge \frac1{\ab}$, according to the  same argument, $\probof{N_x < \frac{\ns}{3\ab}} \le \exp(- \frac{\frac{2\ns^2}{9\ab^2}}{ \ns \cdot \frac1{\ab}}) = \exp(-\frac{2\ns}{9\ab})$. When $\ns \ge 9\ab \log(1/\dist)$, we have $\probof{N_x < \frac{\ns}{3\ab}} \le \dist^2 \le 0.5\dist$ if we suppose $\dist<0.5$. Let $Y(\Xon) \ed \sum_{x \in S(\p)} \mathbf{1}\{ N_x \ge \frac{\ns}{3\ab}\}$, which is the number of symbols appearing more than $\frac{\ns}{3\ab}$ times. We know that $\expectation{Y(\Xon)} > S(\p)(1-0.5\dist)$ by linearity of expectations. Moreover, $\variance{Y(\Xon)} < 0.5 \dist \cdot S(\p)$ since it is the sum of $S(\p)$ negatively related Bernoulli random variables with bias less than $0.5 \alpha$.  According to Chebyshev's inequality, 

\[ \probof{ (1-0.5\dist)S(\p)<Y(\Xon) < S(\p) + \dist S(\p)} \ge 1-\frac{1}{4.5 \dist S(\p)} \ge 1-\frac{1}{ 4.5\ab \dist} \ge \frac{2}{3},\] 

where the last inequality comes from the fact $\ab\dist \ge 1$. Therefore,  

\noindent\[\probof{ (S(\p) - \dist \ab<Y(\Xon) < S(\p) + \dist \ab} \ge \probof{ (1-0.5\dist)S(\p)<Y(\Xon) < S(\p) + \dist S(\p)} \ge \frac{2}{3}.\]

Furthermore, we can see that the sensitivity of $\hat{S}(\Xon)$ is the same, which is at most $3\ab/\ns$.
By Lemma~\ref{lem:main-sensitivity}, there is a private algorithm for support coverage estimation as long as 
\[
\Delta\Paren{\hat{S}(\Xon) } \le \ab \dist \eps.
\]
With the above bound on sensitivity, this is true with $n = O(\frac1{\dist\eps})$, giving the desired upper bound.

\subsubsection{Lower Bound for Support Size Estimation}
\label{sec:ssize-lb}
In this section, we prove a lower bound for support size estimation, as described in Theorem~\ref{thm:ssize}.
The techniques are similar to those for support coverage in Section~\ref{sec:coverage-lb}.

First let us focus on the case when $\ab \ge \frac1{\dist}$, 
The first term of the complexity is the lower bounds for the non-private setting,
which follows by combining the lower bound of~\cite{OrlitskySW16} for support coverage, with the equivalence between estimation of support size and coverage as implied by Lemma~\ref{lem:sssc}.
We focus on the final term in the sequel.

Consider the following two distributions: $u_1$ is a uniform distribution over $[k]$ and $u_2$ is a uniform distribution over $[(1-\dist)k]$. Then the support size of these two distribution differs by $\dist k$, and $\dtv{u_1}{u_2} = \dist$. 

Hence by Lemma~\ref{lem:min_coupling}, we know the following:

\begin{lemma} \label{lem:coupling_ssize}
	Suppose $u_1 \sim U[k]$ and $u_2 \sim U[(1-\dist)k]$, there is a coupling between $u_1^\ns$ and $u_2^\ns$ with expected Hamming distance equal to $ \dist \ns $.
\end{lemma}

Moreover, given $\ns$ samples, we must be able to privately distinguish between $u_1$ and $u_2$ given an $\dist$ accurate estimator of entropy with privacy considerations. 
Thus, according to Theorem~\ref{thm:coupling}
and Lemma~\ref{lem:coupling_ssize}, we have:
\[
\dist \ns\ge \frac1{\eps}\Rightarrow n = \Omega\Paren{\frac1{\eps\dist}}.
\]

Then we move to the second case when $\ab \le \frac1{\dist}$. Because $\ab \dist<1$, we need to recover the support size exactly. 
The first term of the complexity is the lower bound for the non-private setting which can be proved using a coupon collector style argument, so here we focus on the second term.

We consider the following two distributions: $u_1$ is a uniform distribution over $[k]$ and $u_2$ is a uniform distribution over $[k-1]$. We must distinguish between these two distributions, for which $\dtv{u_1}{u_2} = \frac1{k}$. Hence, by Lemma~\ref{lem:min_coupling}, we have 
\[
\frac{\ns}{\ab}\ge \frac1{\eps}\Rightarrow n = \Omega\Paren{\frac{\ab}{\eps}}.
\]

\subsection{Entropy Estimation}
\label{sec:entropy}
In this section, we prove our main theorem about entropy estimation, Theorem~\ref{thm:entropy}:

\entropy*

\noindent We describe and analyze two upper bounds.
The first is based on the empirical entropy estimator, and is described and analyzed in Section~\ref{sec:entropy-ub-empirical}.
The second is based on the method of best-polynomial approximation, and appears in Section~\ref{sec:entropy-ub-poly}.
Finally, our lower bound is in Section~\ref{sec:entropy-lb}.
\subsubsection{Upper Bound for Entropy Estimation: The Empirical Estimator}
\label{sec:entropy-ub-empirical}
Our first private entropy estimator is derived by adding Laplace noise into the empirical estimator. The parameter of the Laplace distribution is $\frac{ \Delta(\entemp)}{\eps}$, where $\Delta(\entemp)$ denotes the sensitivity of the empirical estimator. By analyzing its sensitivity and bias, we prove an upper bound on the sample complexity for private entropy estimation and get the first upper bound in Theorem~\ref{thm:entropy}.

Let $\emp$ be the empirical distribution, and let $\entemp$ be the entropy of the empirical distribution.
The theorem is based on the following three facts:
\begin{align}
&\Delta(\entemp) = O\Paren{\frac{\log\ns}{\ns}}.\label{eqn:sensitivity-emp}\\
&\absv{\entp{\p}-\expectation{\entemp}} =O\Paren{\frac{\ab}{\ns}},\label{eq:bias-emp}\\
&\variance{\entemp} =O\Paren{\frac{\log^2(\min\{\ab,\ns\})}{\ns}},\label{eq:variance-emp}
\end{align}
With these three facts in hand, the sample complexity of the empirical estimator can be bounded as follows.
 By Lemma~\ref{lem:main-sensitivity}, 
we need $\Delta(\entemp) \le \dist \eps$, which gives $\ns=O\Paren{\frac1{\dist \eps} \log(\frac1{\dist \eps})}$. We also need $\absv{\entp{\p}-\expectation{\entemp}} = O\Paren{\dist}$ and $\variance{\entemp} = O\Paren{\dist^2}$, which gives $\ns = O\Paren{\frac{\ab}{\dist}+   \frac{ \log^2(\min\{\ab,\ns\})}{\dist^2}}$.

\paragraph{Proof of~\eqref{eqn:sensitivity-emp}.}
The largest change in any $\nsmb$ when we change one symbol is one. Moreover, at most two $\nsmb$ change. Therefore, 
\begin{align}
\Delta(\entemp) &\leq  2\cdot \max_{j=1\ldots \ns-1} \absv{ \frac{j+1}{n}\log\frac{n}{j+1}-\frac{j}{n}\log\frac{n}{j}}\nonumber\\
&=  2\cdot \max_{j=1\ldots \ns-1}\absv{\frac{j}{\ns}\log \frac{j}{j+1}+\frac1{\ns}\log \frac{\ns}{j+1}}\label{eqn:sum-one}\\
&\leq  2 \cdot \max_{j=1\ldots \ns-1}\max\left\{\absv{\frac{j}{\ns}\log \frac{j}{j+1}}, \absv{ \frac1{\ns}\log \frac{\ns}{j+1}}\right\}\label{eqn:sum-two}\\
&\leq  2\cdot \max\left\{\frac 1\ns, \frac{\log \ns}{\ns}\right\},\nonumber\\
& = 2\cdot \frac{\log \ns}{\ns}.
\end{align}

\paragraph{Proof of~\eqref{eq:bias-emp}.} By the concavity of entropy function, we know that 
\[
\expectation{\entp{\emp}} \le \entp{\p}.
\]
Therefore, 
\begin{align}
&\expectation{\absv{\entp{\p}-\entp{\emp}}} =  \entp{\p}-\expectation{{\entp{\emp}}}\nonumber\\
&=  \expectation{\sum_{\smb}\Paren{\pemps\log \pemps-\p(\smb)\log\p(\smb)}} \nonumber\\
&=  \expectation{\sum_{\smb}\pemps\log \frac{\pemps}{\p(\smb)}}+\expectation{\sum_{\smb}(\pemps-\p(\smb))\log\p(\smb)}\nonumber\\
&=  \expectation{\kldist{\emp}{\p}}\label{eqn:kl-exp}\\
&\leq \expectation{\chisquare{\emp}{\p}}\label{eqn:kl-chisq}\\
&=  \expectation{\sum_{\smb}\frac{(\pemps-\p(\smb))^2}{\p(\smb)}}\nonumber\\
&\le  {\sum_{\smb}\frac{(\p(\smb)/\ns)}{\p(\smb)}}\label{eqn:var-bin}\\
&= \frac\ab\ns.
\end{align}

\paragraph{Proof of~\eqref{eq:variance-emp}.}The variance bound of $\frac{\log^2 k}{n}$ is given precisely in Lemma 15 of~\cite{JiaoVHW17}. 
To obtain the other half of the bound of, we apply the bounded differences inequality in the form stated in Corollary 3.2 of~\cite{BoucheronLM13}.
\begin{lemma}
	\label{lem:bdd-diff-var}
	Let $f:\Omega^{\ns}\to\RR$ be a function. Suppose further that 
	\[
	\max_{z_1,\ldots, z_n, z_{i}^{'}} \absv{f(z_1,\ldots,z_n) - f(z_1,\ldots, z_{i-1}, z_i^{'},\ldots,z_n)}\le c_i.
	\]
	Then for independent variables $Z_1,\ldots,Z_{\ns}$, 
	\[
	\Var\Paren{f(Z_1,\ldots,Z_{\ns})} \le \frac14\sum_{i=1}^{\ns} c_i^2.
	\]
\end{lemma}

Therefore, using Lemma~\ref{lem:bdd-diff-var} and Equation~\eqref{eqn:sensitivity-emp}
\[
\variance{\entemp}\le \ns\cdot\Paren{\frac{4\log^2\ns}{\ns^2}} = \frac{4\log^2\ns}{\ns}.
\]

%

\subsubsection{Upper Bound for Entropy Estimation: Best-Polynomial Approximation}
\label{sec:entropy-ub-poly}
We prove an upper bound on the sample complexity for private entropy estimation if one  adds Laplace noise into best-polynomial estimator.
This will give us the second upper bound in Theorem~\ref{thm:entropy}.

In the non-private setting the optimal sample complexity of estimating $H(\p)$ over $\Delta_k$ is given by Theorem 1 of~\cite{WuY16}
\[
\Theta\Paren{\frac{\ab}{\dist\log\ab}+\frac{\log^2(\min\{\ab,\ns\})}{\dist^2}}.
\]
However, this estimator can have a large sensitivity.~\cite{AcharyaDOS17} designed an estimator that has the same sample complexity but a smaller sensitivity. We restate Lemma 6 of~\cite{AcharyaDOS17} here:
\begin{lemma}
Let $\lambda>0$ be a fixed small constant, which may be taken to be any value between $0.01$ and $1$. Then there is an entropy estimator with sample complexity 
\[
\Theta\Paren{\frac1{\lambda^2}\cdot\frac{\ab}{\dist\log\ab}+\frac{\log^2(\min\{\ab,\ns\})}{\dist^2}},
\]
and has sensitivity $\ns^\lambda/\ns$. 
\end{lemma}
We can now invoke Lemma~\ref{lem:main-sensitivity}
on the estimator in this lemma to obtain the upper bound on private entropy estimation. 

%

\subsubsection{Lower Bound for Entropy Estimation}
\label{sec:entropy-lb}
We now prove the lower bound for entropy estimation. 
Note that any lower bound on privately testing two distributions $\p$, and $\q$ such that $H(\p)-H(\q)= \Theta(\alpha)$ is a lower bound on estimating entropy. 

We analyze the following construction for Proposition 2 of~\cite{WuY16}. The two distributions $\p$, and $\q$ over $[k]$ are defined as:
\begin{align}
p(1) = \frac23,& p(i) = \frac{1-p(1)}{k-1}, \text{for $i=2,\ldots,k$},\\
q(1) = \frac{2-\eta}3,& q(i) = \frac{1-q(1)}{k-1}, \text{for $i=2,\ldots,k$}.
\end{align}
Then, by the grouping property of entropy, 
\[
H(\p) = h(2/3) + \frac13\cdot \log (k-1), \text{ and } H(\q) = h((2-\eta)/3) + \frac{1+\eta}3\cdot \log (k-1),
\]
which gives
\[
H(\p)-H(\q) = \Omega(\eta\log k). 
\]
For $\eta = \dist/\log\ab$, the entropy difference becomes $\Theta(\dist)$. 

The total variation distance between $\p$ and $\q$ is $\eta/3$. By Lemma 5, 
there is a coupling over $\Xon$, and $\Yon$ generated from $\p$ and $\q$ with expected Hamming distance at most $\dtv{\p}{\q}\cdot \ns$. 
This along with Lemma 2 
gives a lower bound of $\Omega\Paren{\log k/\dist\eps}$ on the sample complexity.  

\section{Experiments}
\label{sec:ins_exp}
We evaluated our methods for entropy estimation and support coverage on both synthetic and real data.
Overall, we found that privacy is quite cheap: private estimators achieve accuracy which is comparable or near-indistinguishable to non-private estimators in many settings.
Our results on entropy estimation and support coverage appear in Sections~\ref{sec:exp-entropy} and~\ref{sec:exp-coverage}, respectively.
Code of our implementation is available at \url{https://github.com/HuanyuZhang/INSPECTRE}.

\subsection{Entropy}
\label{sec:exp-entropy}
We compare the performance of our entropy estimator with a number of alternatives, both private and non-private.
Non-private algorithms considered include the plug-in estimator (\plugin),  the Miller-Madow Estimator (\MM)~\cite{Miller55}, the sample optimal polynomial approximation estimator (\polyn) of~\cite{WuY16}.
We analyze the privatized versions of plug-in, and \polyn\ in Sections~\ref{sec:entropy-ub-empirical} and~\ref{sec:entropy-ub-poly}, respectively.
The implementation of the latter is based on code from the authors of~\cite{WuY16}\footnote{See \url{https://github.com/Albuso0/entropy} for their code for entropy estimation.}.
We compare performance on different distributions including uniform, a distribution with two steps, Zipf(1/2), a distribution with Dirichlet-1 prior, and a distribution with Dirichlet-$1/2$ prior, and over varying support sizes.

While \plugin, and \MM\ are parameter free,  \polyn\ (and its private counterpart) have to choose the degree $L$ of the polynomial  to use, which manifests in the parameter $\lambda$ in the statement of Theorem~\ref{thm:entropy}.~\cite{WuY16} suggests the value of $L=1.6\log k$ in their experiments. However, since we add further noise, we choose a single $L$ as follows: (i) Run  privatized \polyn\ for different $L$ values and distributions for $k=2000$, $\eps=1$, (b) Choose the value of $L$ that performs well across different distributions (See Figure~\ref{fig:Lcomparison}). We choose $L=1.2\cdot \log k$ from this, and  use it for all other experiments. 
To evaluate the sensitivity of \polyn, we computed the estimator's value at all possible input values, computed the sensitivity, (namely, $\Delta = \max_{\ham{\Xon}{\Yon}\le 1}|\polyn(\Xon)-\polyn(\Yon)|$), and added noise distributed as $\text{Lap}\left(0,\frac{\Delta}{\eps}\right)$.

%
%
\begin{figure*}[h]
\centering
\captionsetup[subfigure]{labelformat=empty}
\subfigure{
\includegraphics[width=0.15\textwidth]{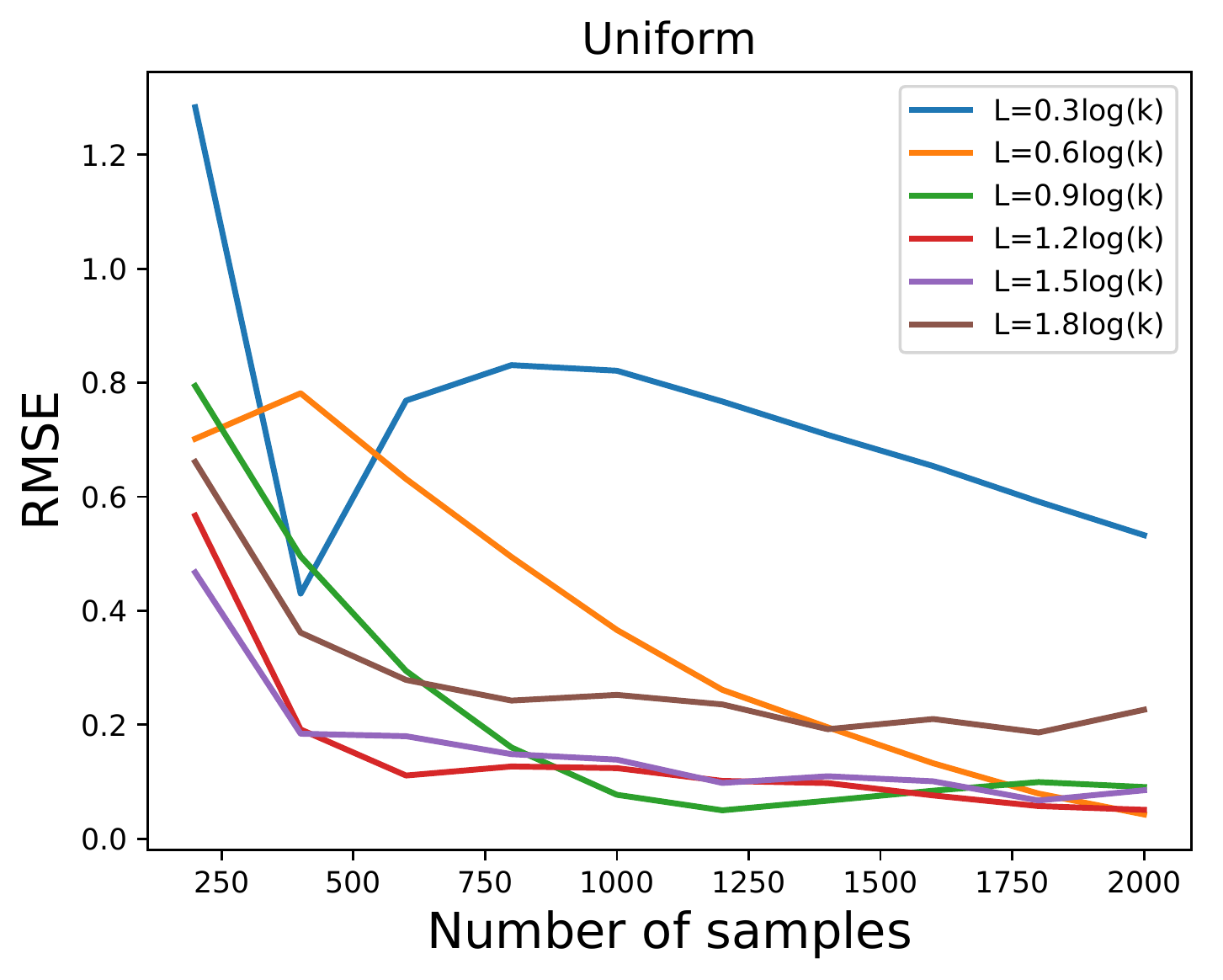}
}
\subfigure{
\includegraphics[width=0.15\textwidth]{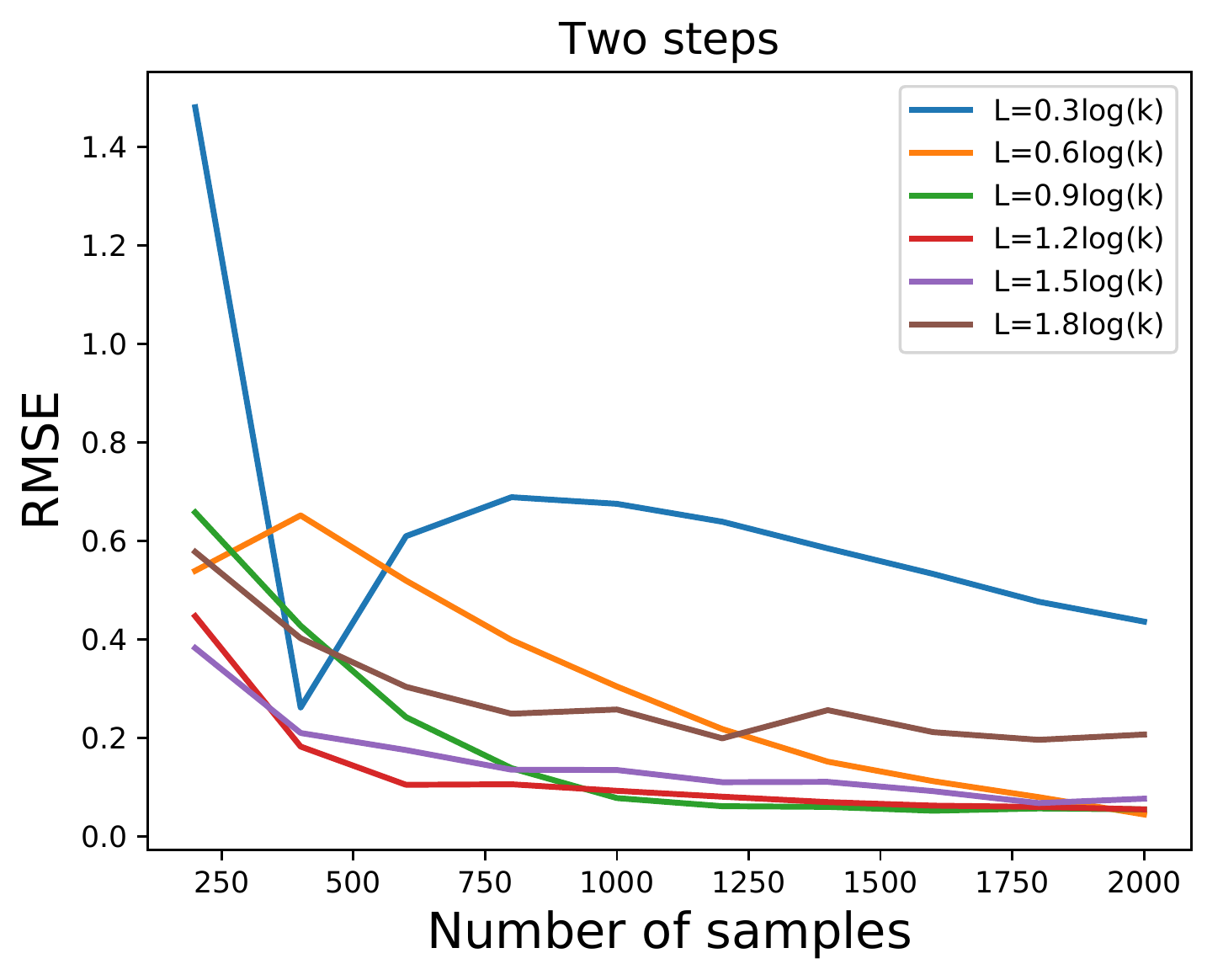}
}
\subfigure{
\includegraphics[width=0.15\textwidth]{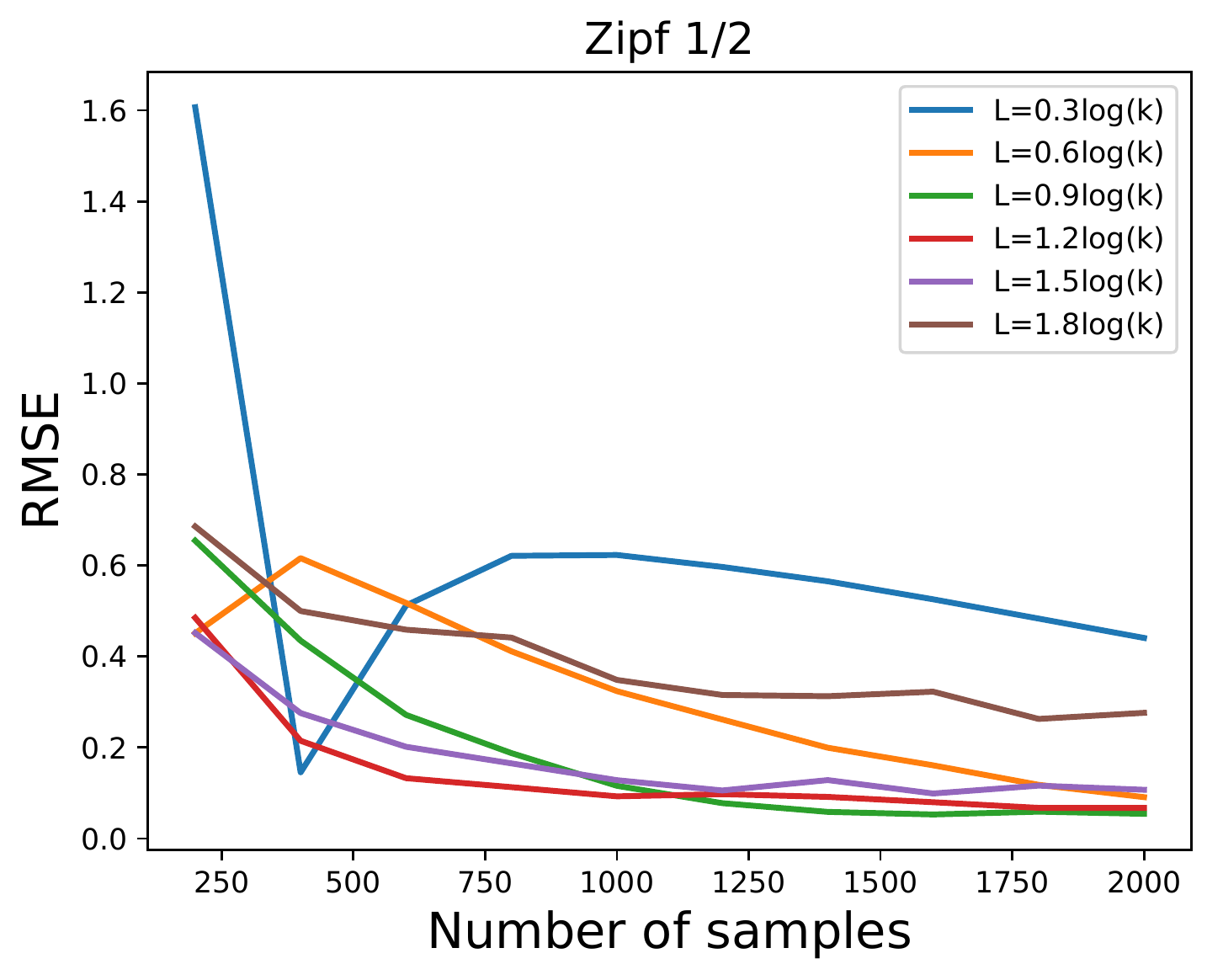}
}
\subfigure{
\includegraphics[width=0.15\textwidth]{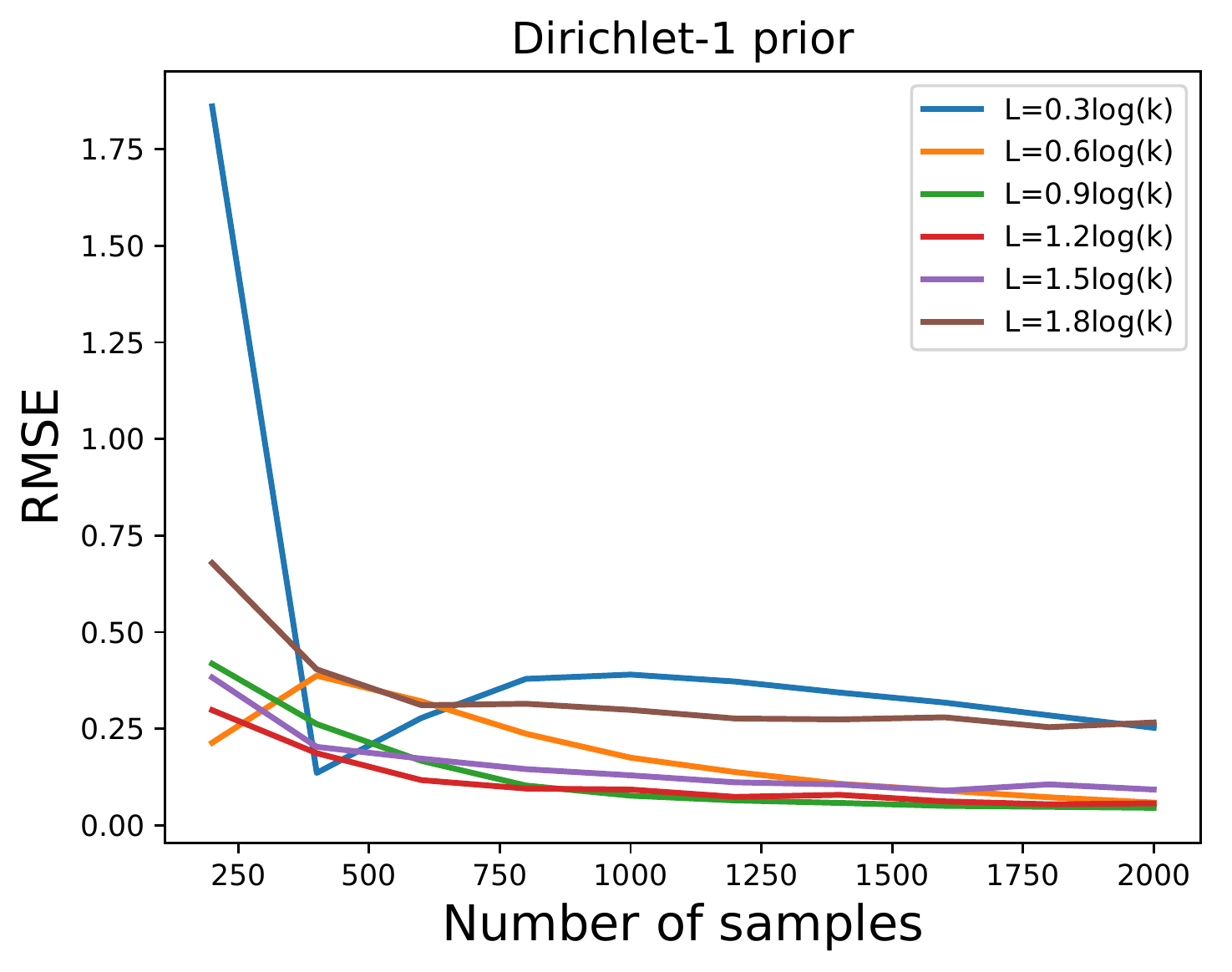}
}
\subfigure{
\includegraphics[width=0.15\textwidth]{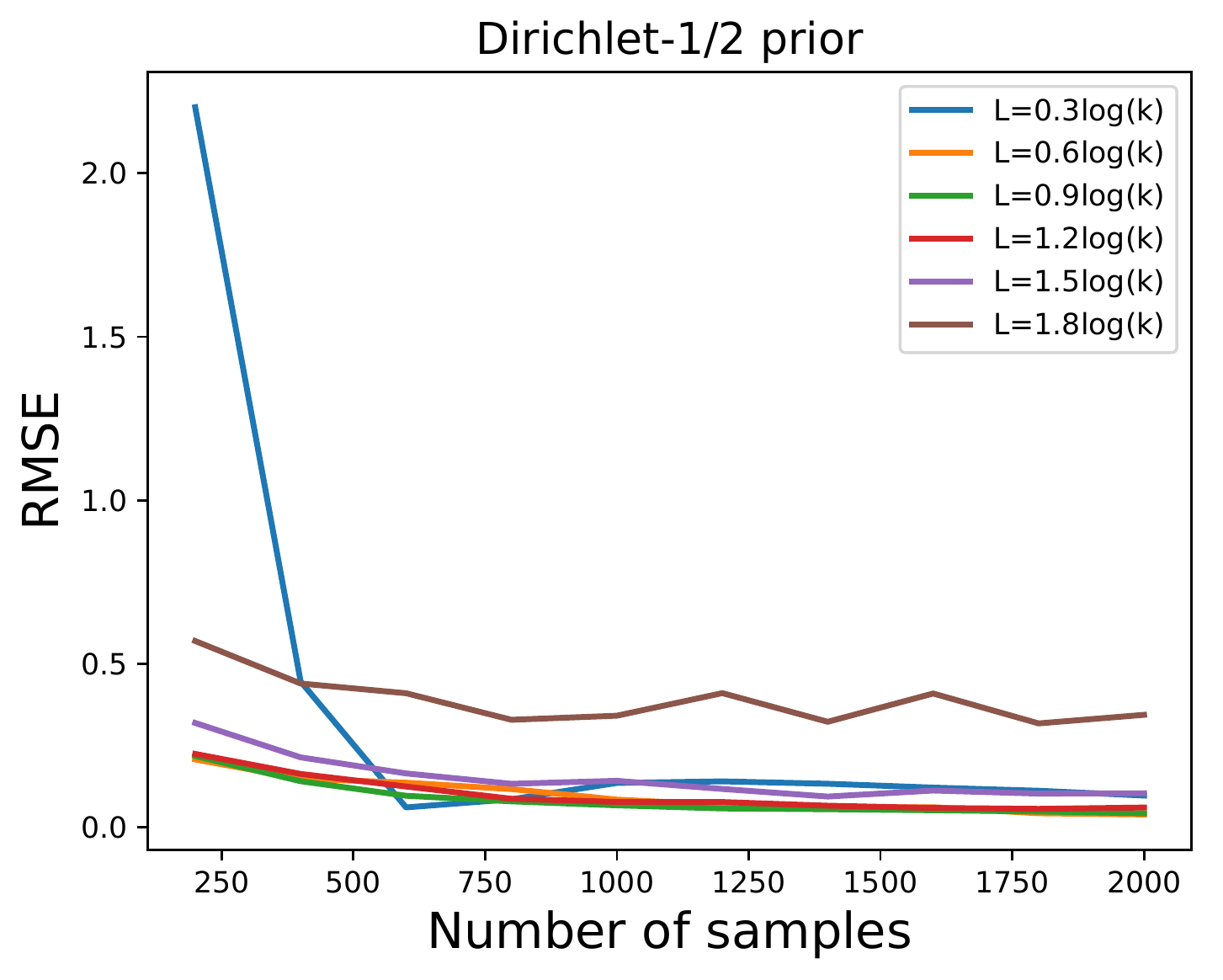}
}
\caption{RMSE comparison between private Polynomial Approximation Estimators for entropy with various values for degree $L$, $k=2000$, $\eps=1$. The degree $L$ represents a bias-variance tradeoff: a larger degree decreases the bias but increases the sensitivity, necessitating the addition of Laplace noise with a larger variance.}
\label{fig:Lcomparison} 

\end{figure*}


The RMSE of various estimators for $k=1000$, and $\eps=1$ for various distributions  are illustrated in Figure~\ref{fig:entropy_k1000}.
The RMSE is averaged over 100 iterations in the plots.

%
%
\begin{figure*}[h]
\centering
\captionsetup[subfigure]{labelformat=empty}
\subfigure{
\includegraphics[width=0.15\textwidth]{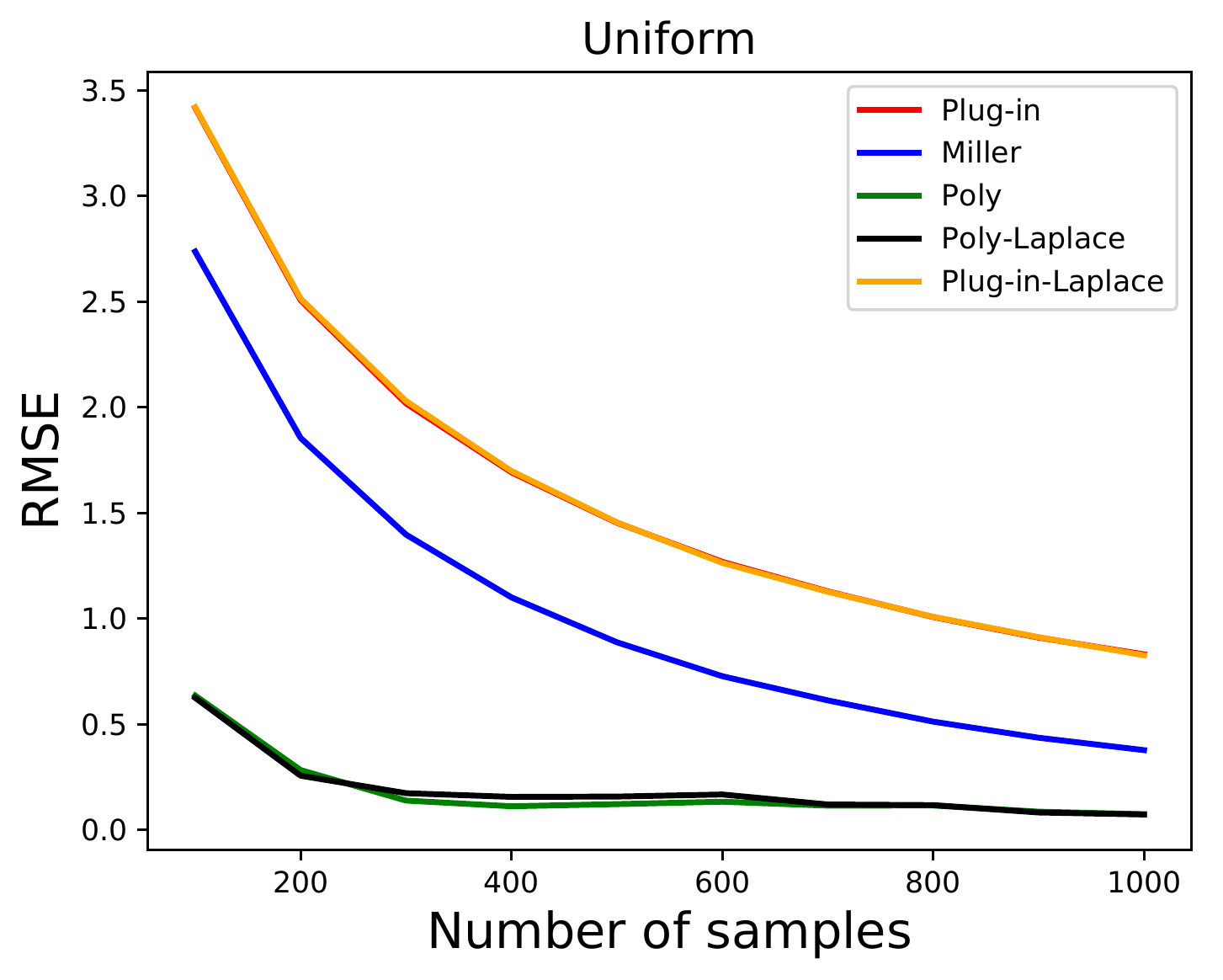}
}
\subfigure{
\includegraphics[width=0.15\textwidth]{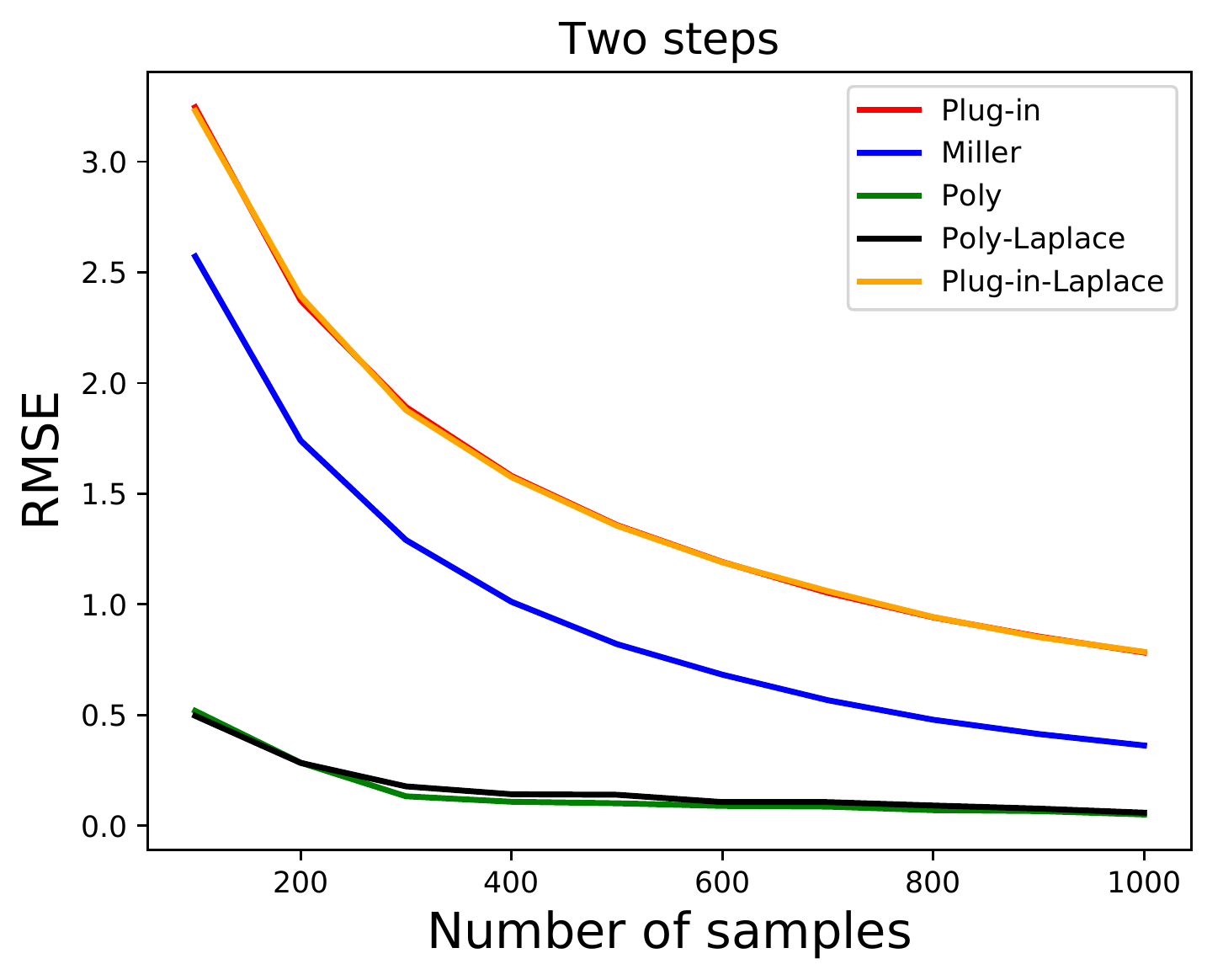}
}
\subfigure{
\includegraphics[width=0.15\textwidth]{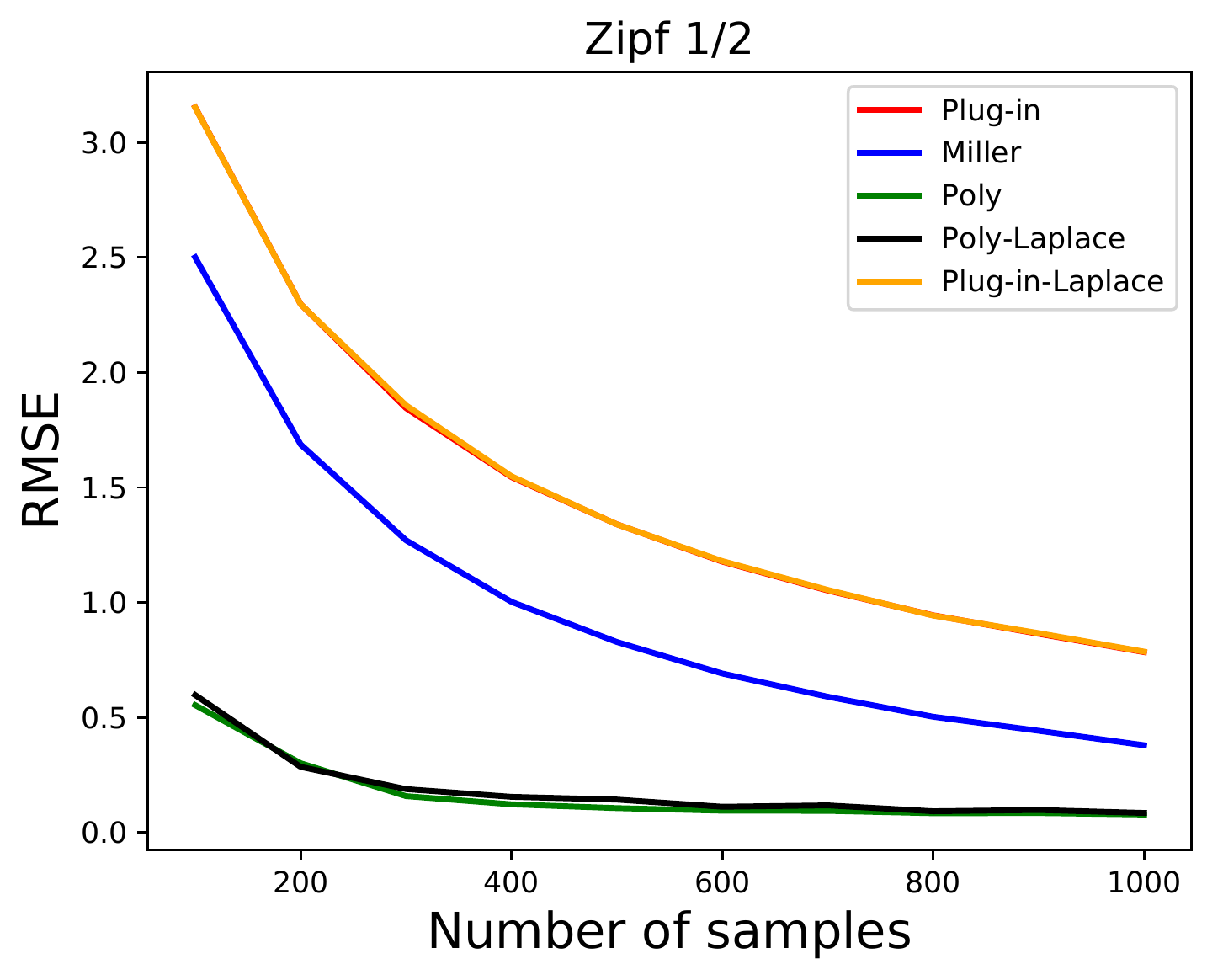}
}
\subfigure{
\includegraphics[width=0.15\textwidth]{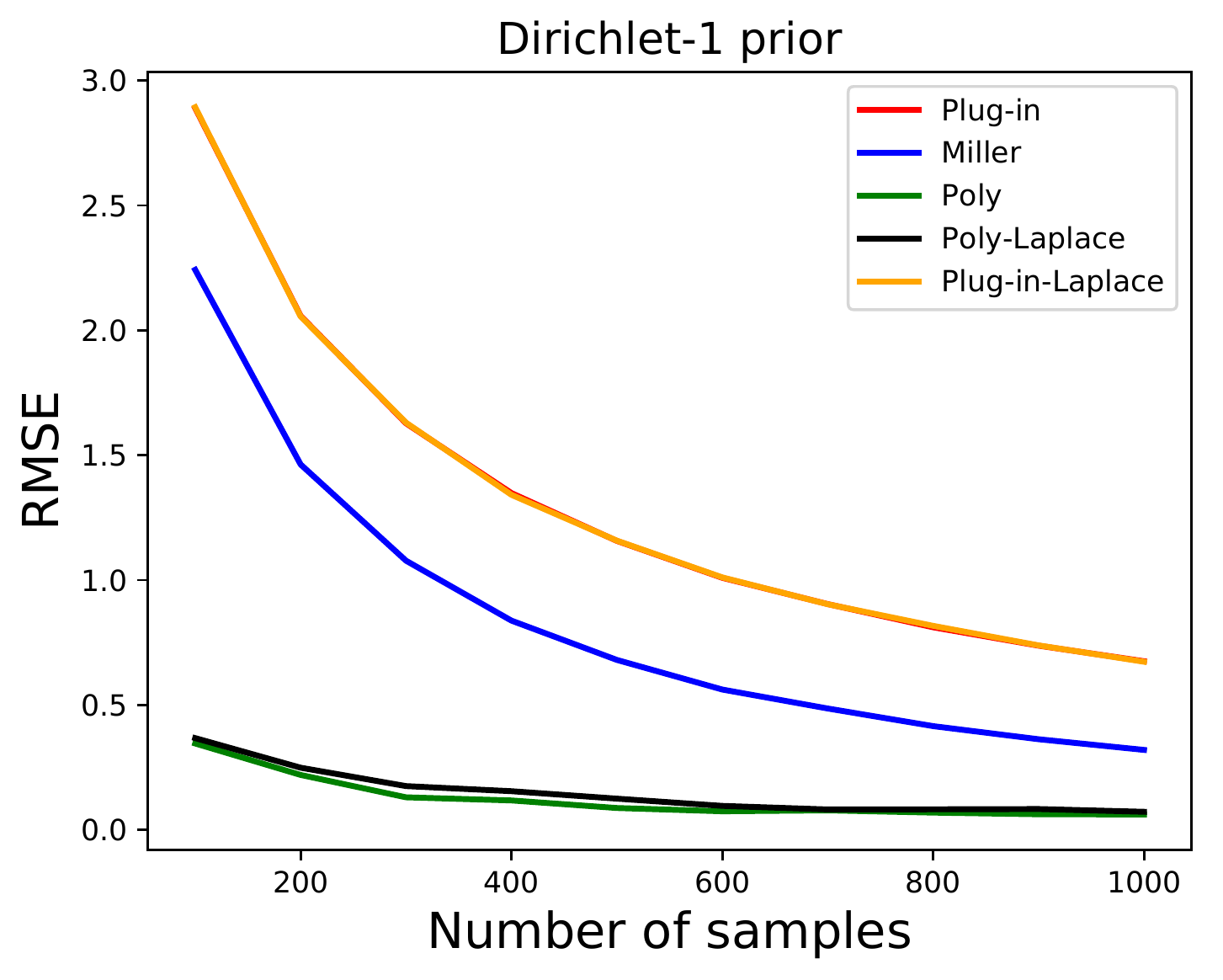}
}
\subfigure{
\includegraphics[width=0.15\textwidth]{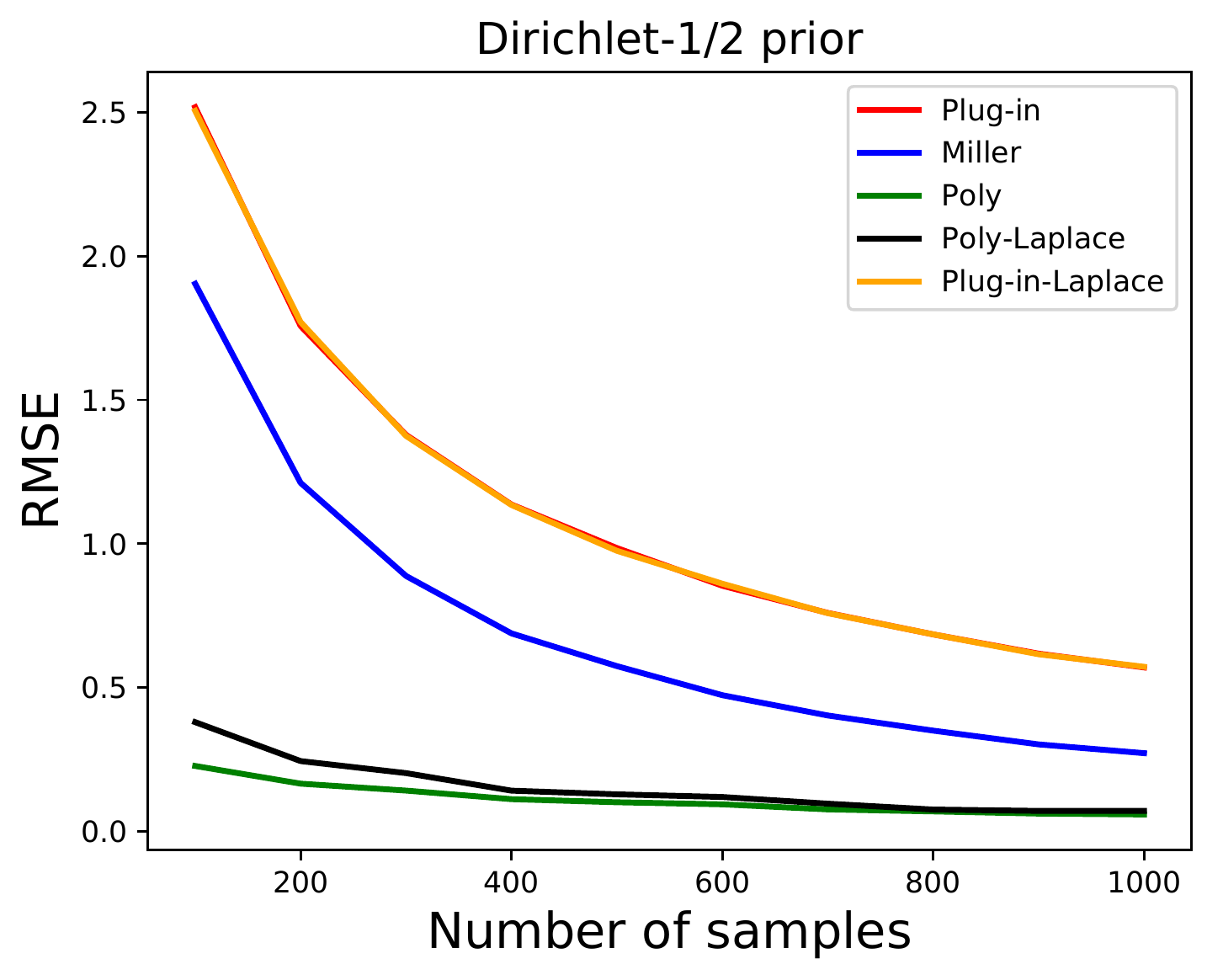}
}
\caption{Comparison of various estimators for entropy, $k=1000$, $\eps =1$.} 
\label{fig:entropy_k1000} 

\end{figure*}

We observe that the performance of our private-\polyn\ is near-indistinguishable from the non-private \polyn, particularly as the number of samples increases.
It also performs significantly better than all other alternatives, including the non-private Miller-Madow and the plug-in estimator.
The cost of privacy is minimal for several other settings of $k$ and $\varepsilon$, for which results appear in Section~\ref{sec:supp-experiments}.

\subsection{Support Coverage}
\label{sec:exp-coverage}

We investigate the cost of privacy for the problem of support coverage.
We provide a comparison between the Smoothed Good-Toulmin  estimator (SGT) of~\cite{OrlitskySW16} and our algorithm, which is a privatized version of their statistic (see Section~\ref{sec:coverage-ub}).
Our implementation is based on code provided by the authors of~\cite{OrlitskySW16}.
As shown in our theoretical results, the sensitivity of SGT  is at most $2 (1+e^r(t-1))$, necessitating the addition of Laplace noise with parameter  $2 (1+e^{r(t-1)})/\eps$.
Note that while the theory suggests we select the parameter $r = \log(1/\dist)$, $\dist$ is unknown.
We instead set $r = \frac{1}{2t} \log_{e} \frac{n(t+1)^2}{t-1}$, as previously done in~\cite{OrlitskySW16}.

\subsubsection{Evaluation on Synthetic Data}

In our synthetic experiments, we consider different distributions over different support sizes $\ab$. We generate $\ns = \ab/2$ samples, and then estimate the support coverage at $m=\ns \cdot t$. For large $t$, estimation is harder. 
Some results of our evaluation on synthetic are displayed in Figure~\ref{fig:synthetic_k20000}.
We compare the performance of SGT, and privatized versions of SGT with parameters $\eps = 1, 2,$ and $10$.
For this instance, we fixed the domain size $k = 20000$.
We ran the methods described above with $n = k/2$ samples, and estimated the support coverage at $m = nt$, for $t$ ranging from $1$ to $10$.
The performance of the estimators is measured in terms of RMSE over 1000 iterations.
\begin{figure*}[h]
\centering

\subfigure{
\includegraphics[width=0.15\textwidth]{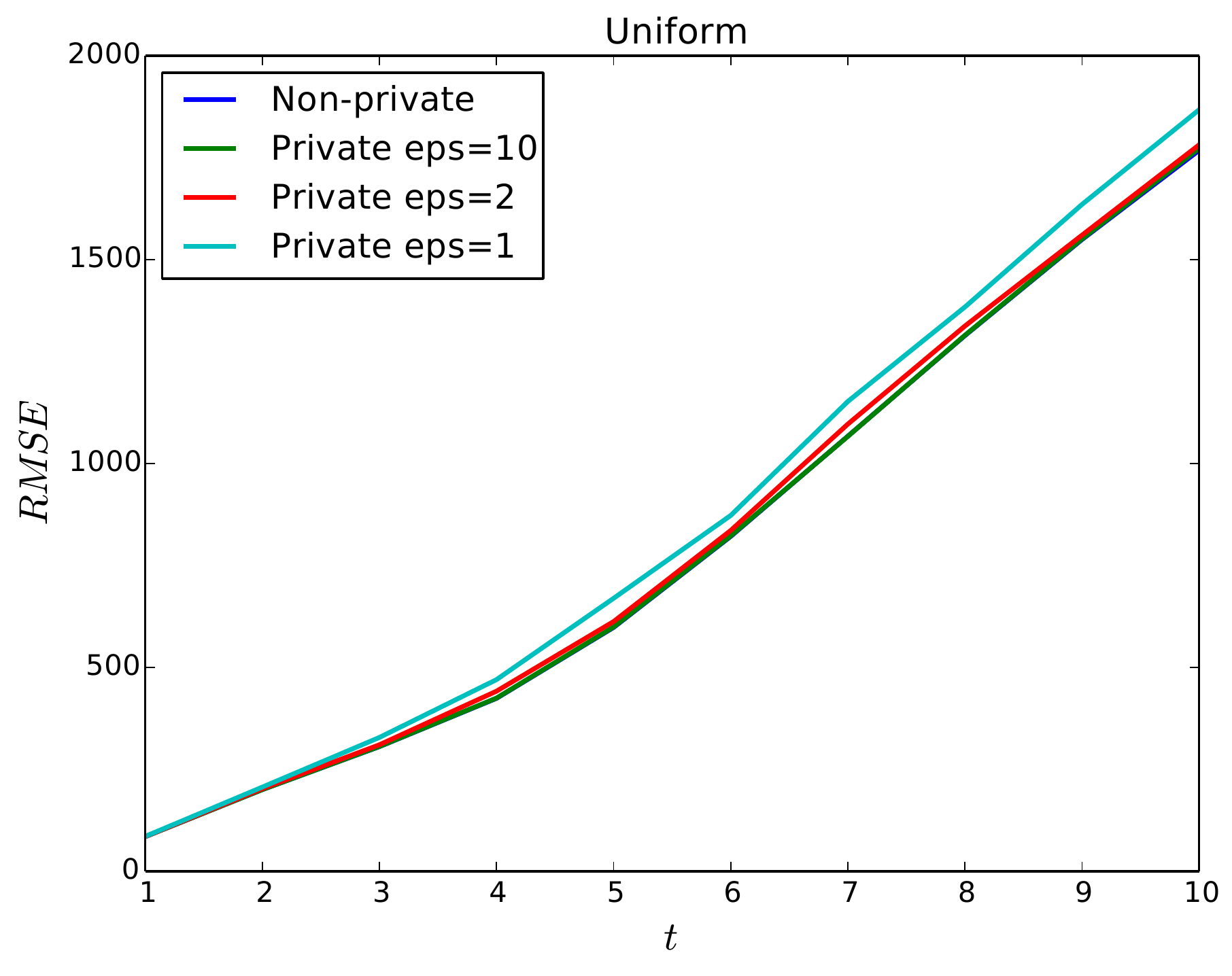}
}
\subfigure{
\includegraphics[width=0.15\textwidth]{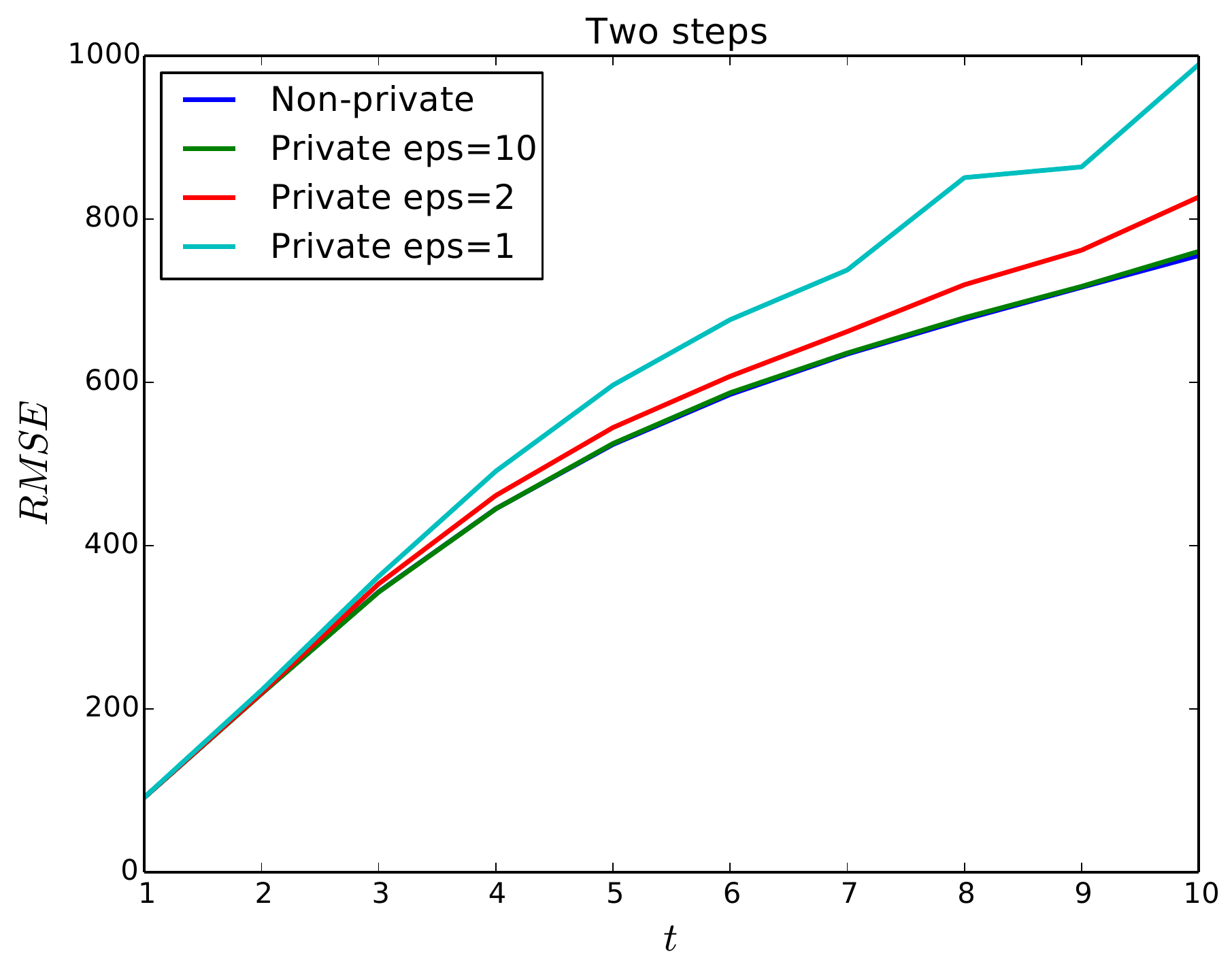}
}
\subfigure{
\includegraphics[width=0.15\textwidth]{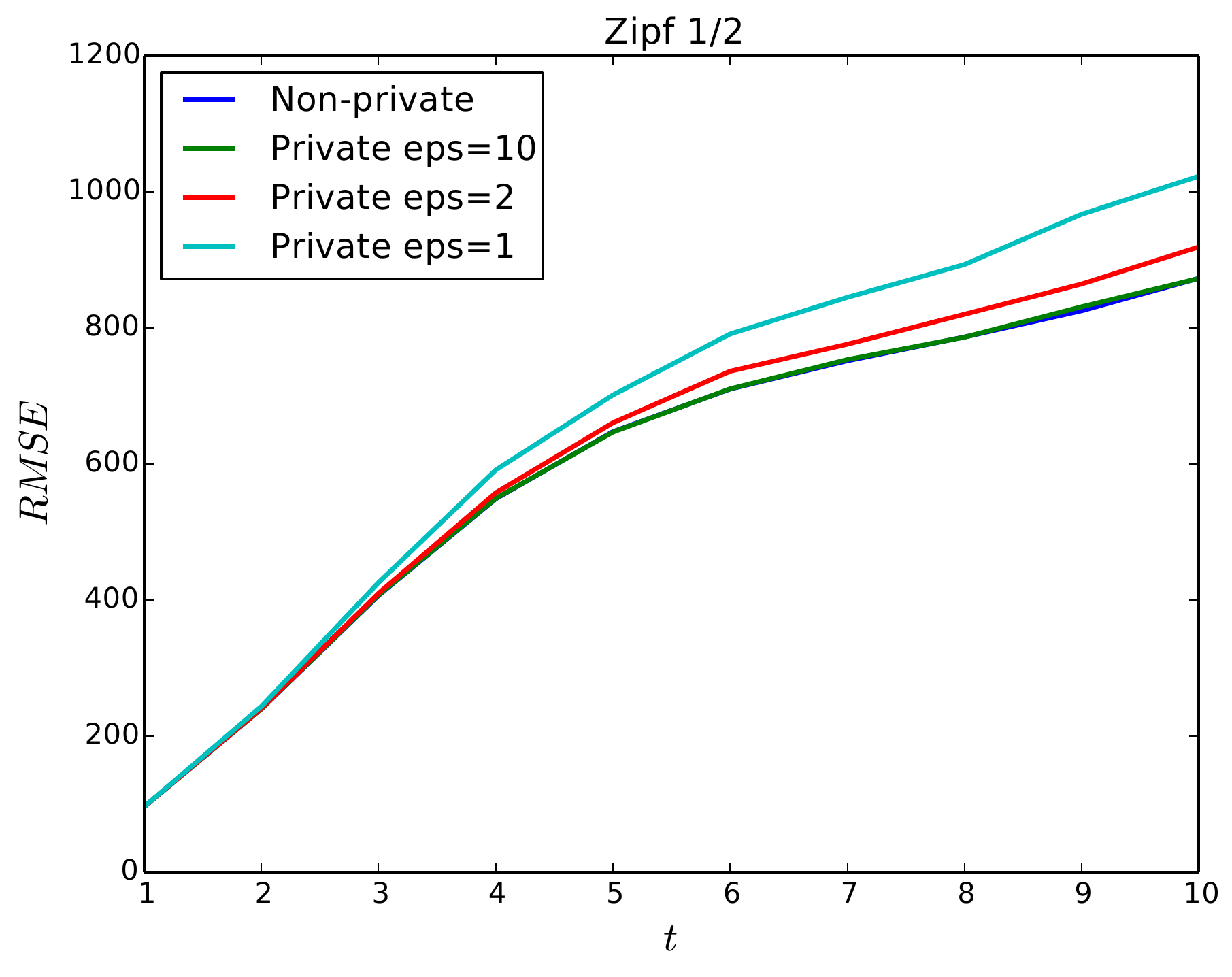}
}
\subfigure{
\includegraphics[width=0.15\textwidth]{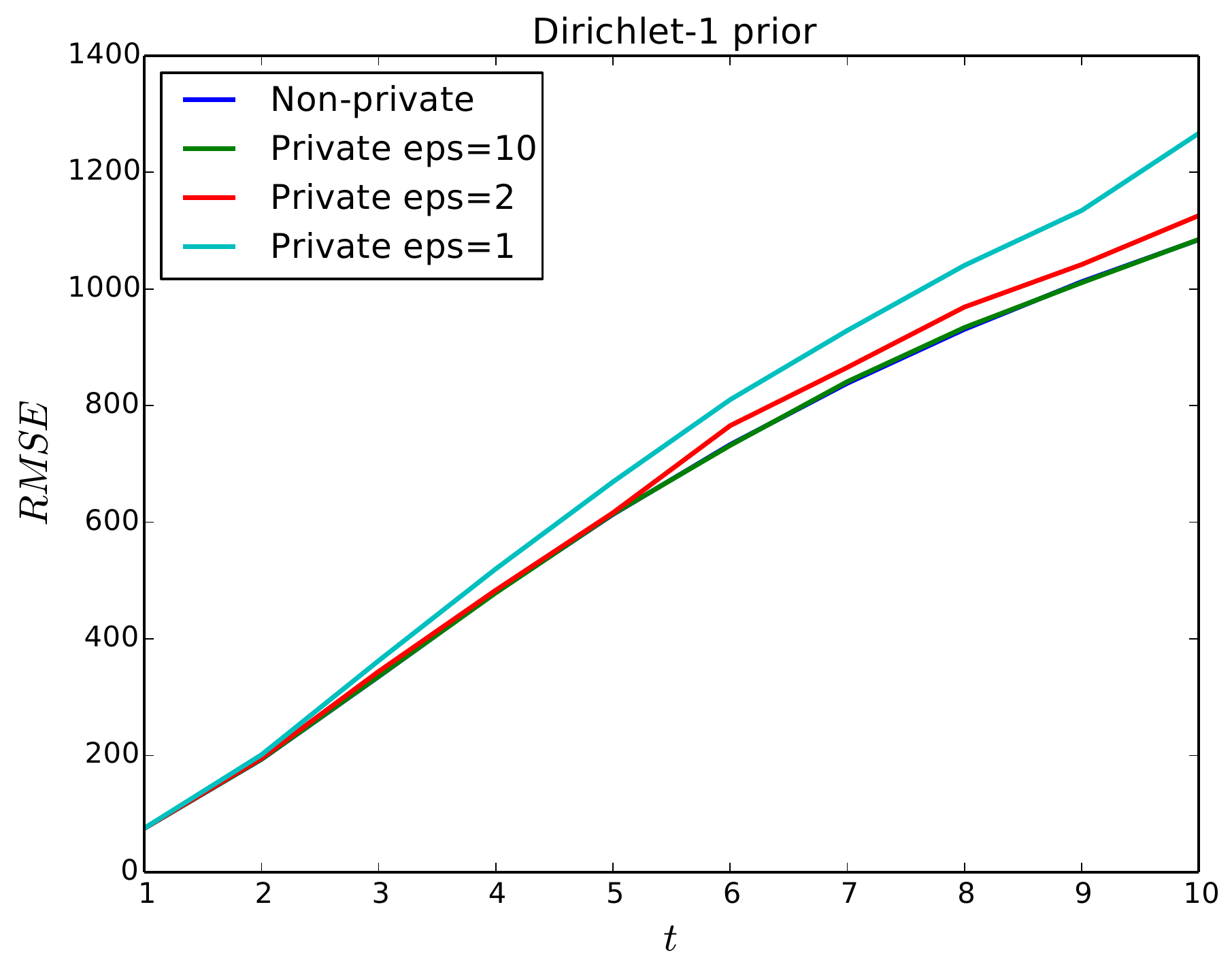}
}
\subfigure{
\includegraphics[width=0.15\textwidth]{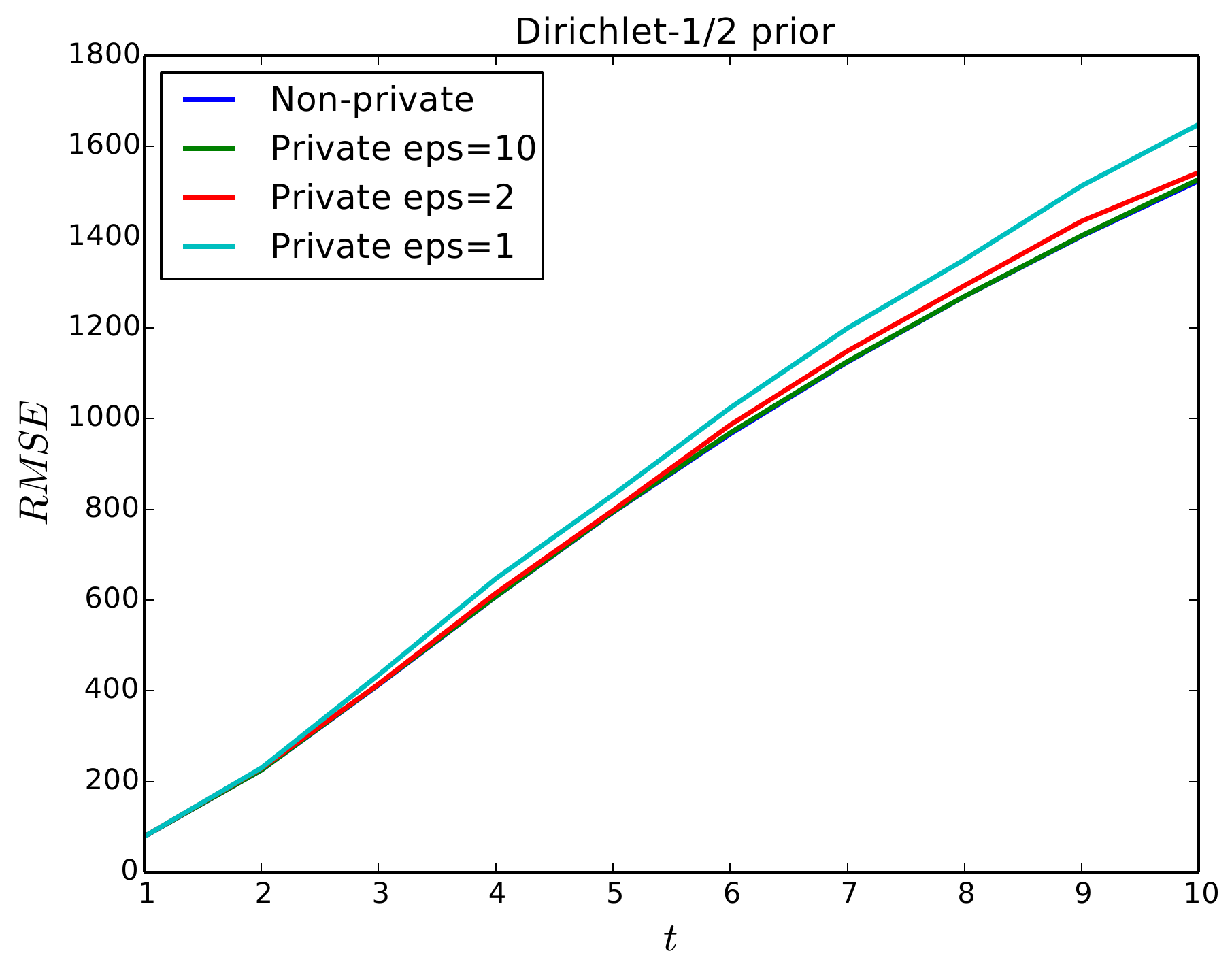}
}

\caption{Comparison between our private support coverage estimator with non-private SGT when $k=20000$} 
\label{fig:synthetic_k20000} 

\end{figure*}

We observe that, in this setting, the cost of privacy is relatively small for reasonable values of $\eps$.
This is as predicted by our theoretical results, where unless $\eps$ is extremely small (less than $1/k$) the non-private sample complexity dominates the privacy requirement.
However, we found that for smaller support sizes (as shown in Section~\ref{sec:supp-exp-coverage}), the cost of privacy can be  significant.
We provide an intuitive explanation for why no private estimator can perform well on such instances.
To minimize the number of parameters, we instead argue about the related problem of support-size estimation.
Suppose we are trying to distinguish between distributions which are uniform over supports of size $100$ and $200$.
We note that, if we draw $n = 50$ samples, the ``profile'' of the samples (i.e., the histogram of the histogram) will be very similar for the two distributions.
In particular, if one modifies only a few samples (say, five or six), one could convert one profile into the other.
In other words, these two profiles are almost-neighboring datasets, but simultaneously correspond to very different support sizes.
This pits the two goals of privacy and accuracy at odds with each other, thus resulting in a degradation in accuracy.

\subsubsection{Evaluation on Census Data and Hamlet}
We conclude with experiments for support coverage on two real-world datasets, 
the 2000 US Census data and the text of Shakespeare's play Hamlet, inspired by investigations in~\cite{OrlitskySW16} and~\cite{ValiantV17b}.
Our investigation on US Census data is also inspired by the fact that this is a setting where privacy is of practical importance, evidenced by the proposed adoption of differential privacy in the 2020 US Census~\cite{DajaniLSKRMGDGKKLSSVA17}.

The Census dataset contains a list of last names that appear at least 100 times.
Since the dataset is so oversampled, even a small fraction of the data is likely to contain almost all the names.
As such, we make the task non-trivial by subsampling $m_{total} = 86080$ individuals from the data, obtaining $20412$ distinct last names.
We then sample $n$ of the $m_{total}$ individuals without replacement and attempt to estimate the total number of last names.
Figure~\ref{fig:name} displays the RMSE over 100 iterations of this process. We observe that even with an exceptionally stringent privacy budget of $\eps = 0.5$, the performance is almost indistinguishable from the non-private SGT estimator.


\begin{figure}[h]
\centering 
\includegraphics [scale=0.3]{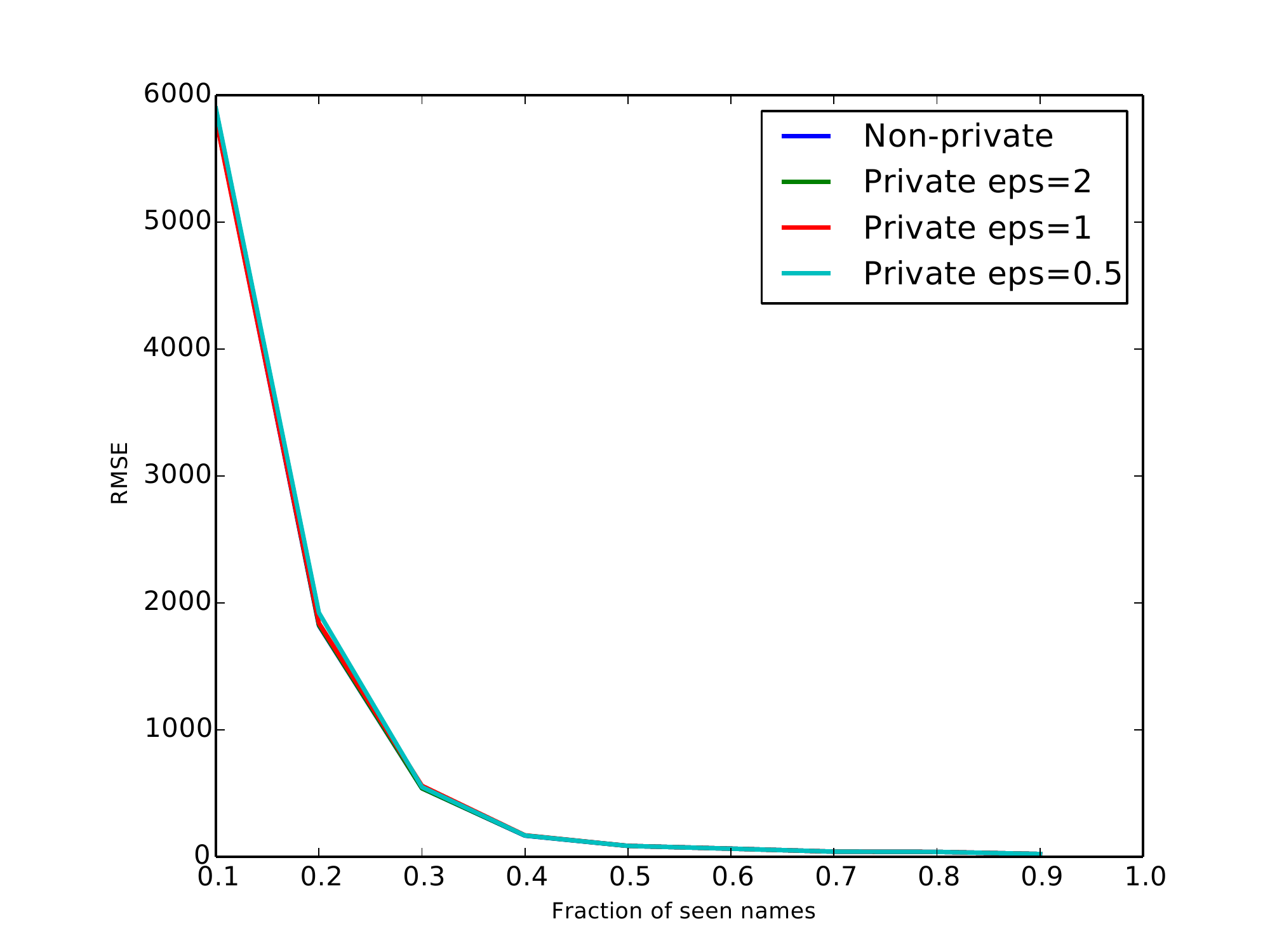}
\caption{Comparison between our private support coverage  estimator with the SGT on Census Data.}
\label{fig:name} 
\end{figure}

The Hamlet dataset has $m_{total} = 31,999$ words, of which 4804 are distinct.
Since the distribution is not as oversampled as the Census data, we do not need to subsample the data.
Besides this difference, the experimental setup is identical to that of the Census dataset.
Once again, as we can see in Figure~\ref{fig:hamlet}, we get near-indistinguishable performance between the non-private and private estimators, even for very small values of $\eps$.
Our experimental results demonstrate that privacy is realizable in practice, with particularly accurate performance on real-world datasets. 

\begin{figure}[h]
\centering 
\includegraphics [scale=0.3]{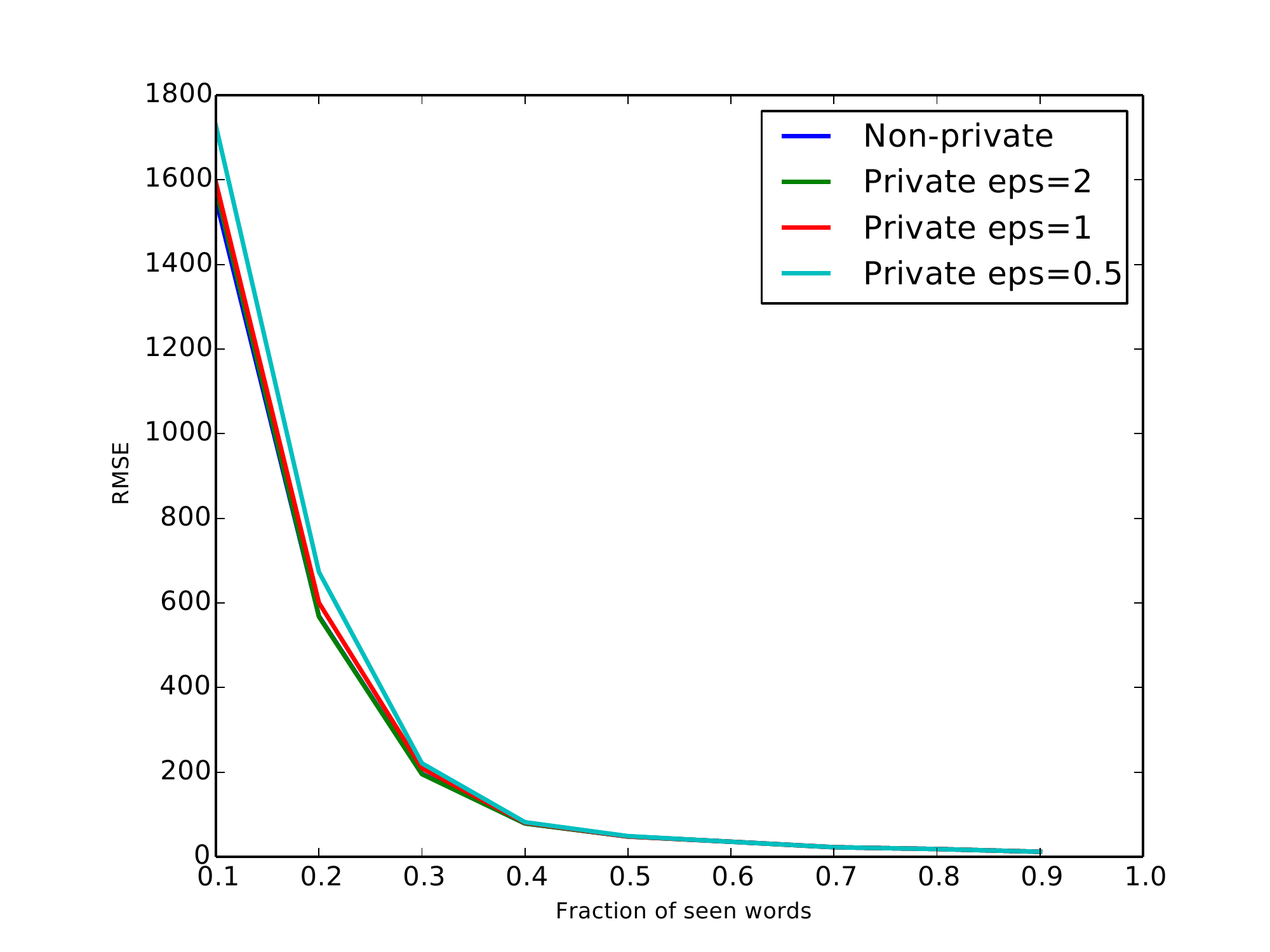}
\caption{Comparison between our private support coverage estimator with the SGT on Hamlet.}
\label{fig:hamlet} 
\end{figure}

\section{Additional Experimental Results}
\label{sec:supp-experiments}
This section contains additional plots of our synthetic experimental results.
Section~\ref{sec:supp-exp-entropy} contains experiments on entropy estimation, while Section~\ref{sec:supp-exp-coverage} contains experiments on estimation of support coverage.
\subsection{Entropy Estimation}
\label{sec:supp-exp-entropy}
We present four more plots of our synthetic experimental results for entropy estimation.
Figures~\ref{fig:entropy-k100-eps1} and~\ref{fig:entropy-k100-eps2} are on a smaller support of $k =100$, with $\eps = 1$ and $2$, respectively.
Figures~\ref{fig:entropy-k1000-eps05} and~\ref{fig:entropy-k1000-eps2} are on a support of $k=1000$, with $\eps = 0.5$ and $2$.
\begin{figure*}
\centering
\subfigure[]{
\begin{minipage}[b]{0.3\textwidth}
\includegraphics[width=1\textwidth]{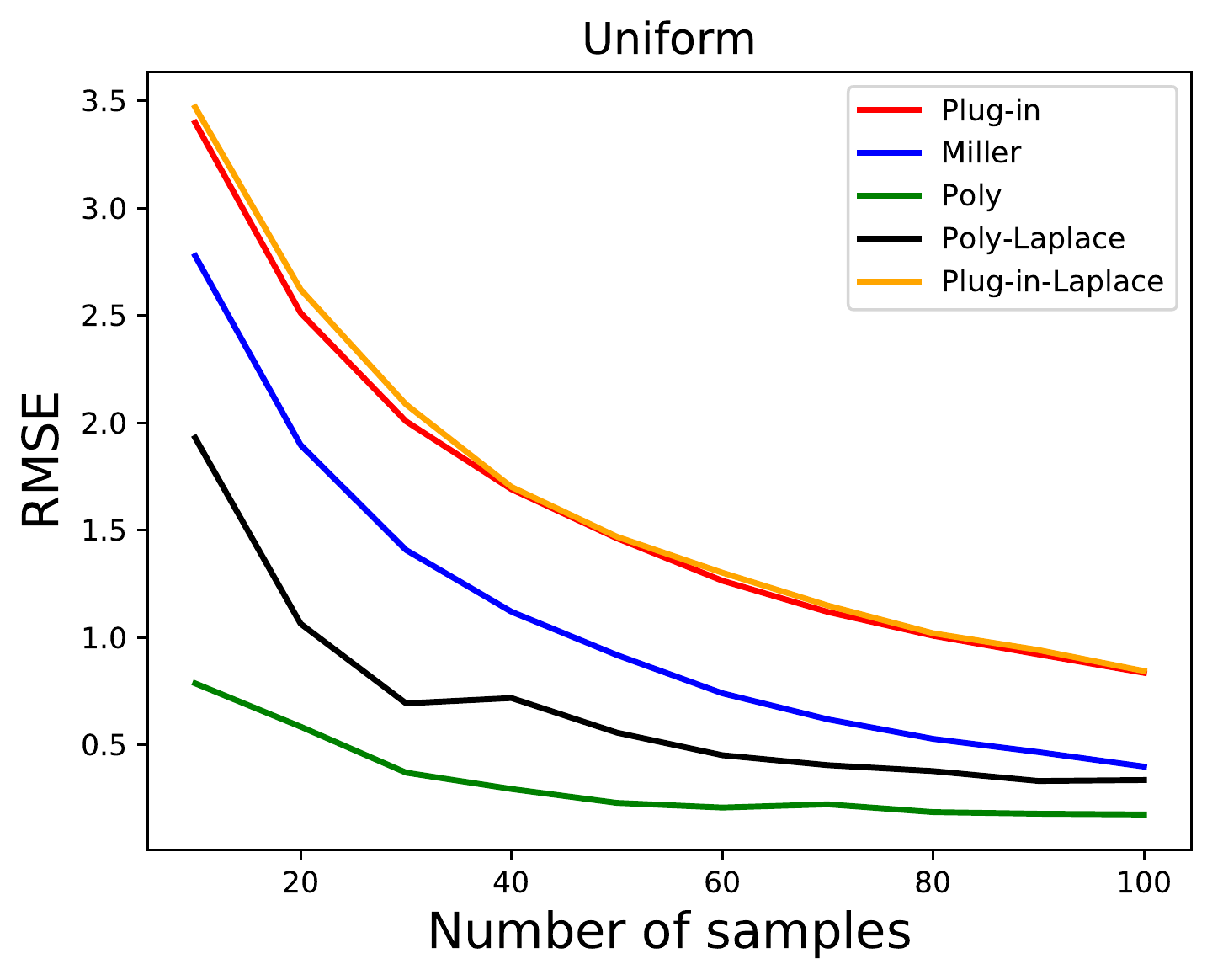}
\includegraphics[width=1\textwidth]{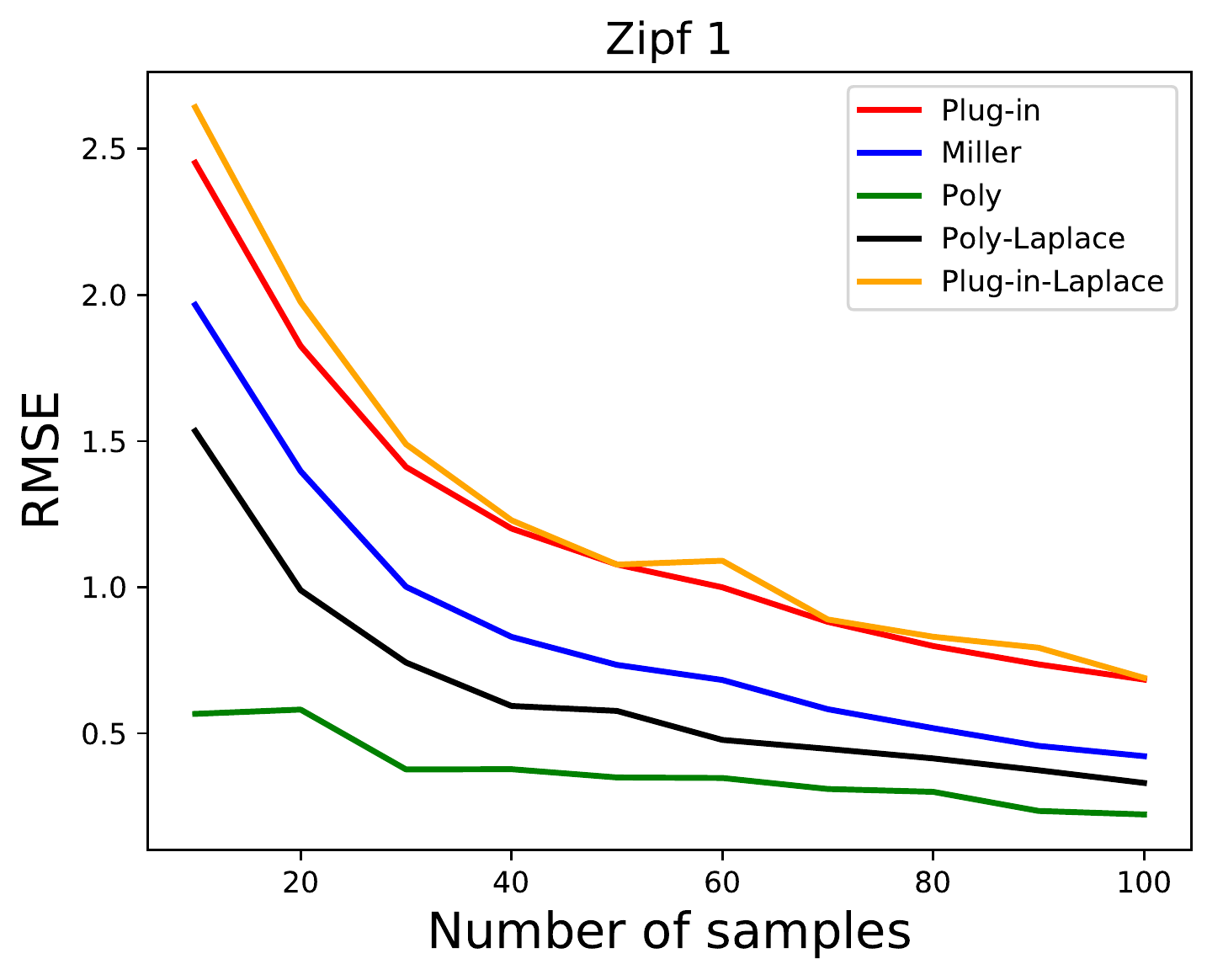}
\end{minipage}
}
\subfigure[]{
\begin{minipage}[b]{0.3\textwidth}
\includegraphics[width=1\textwidth]{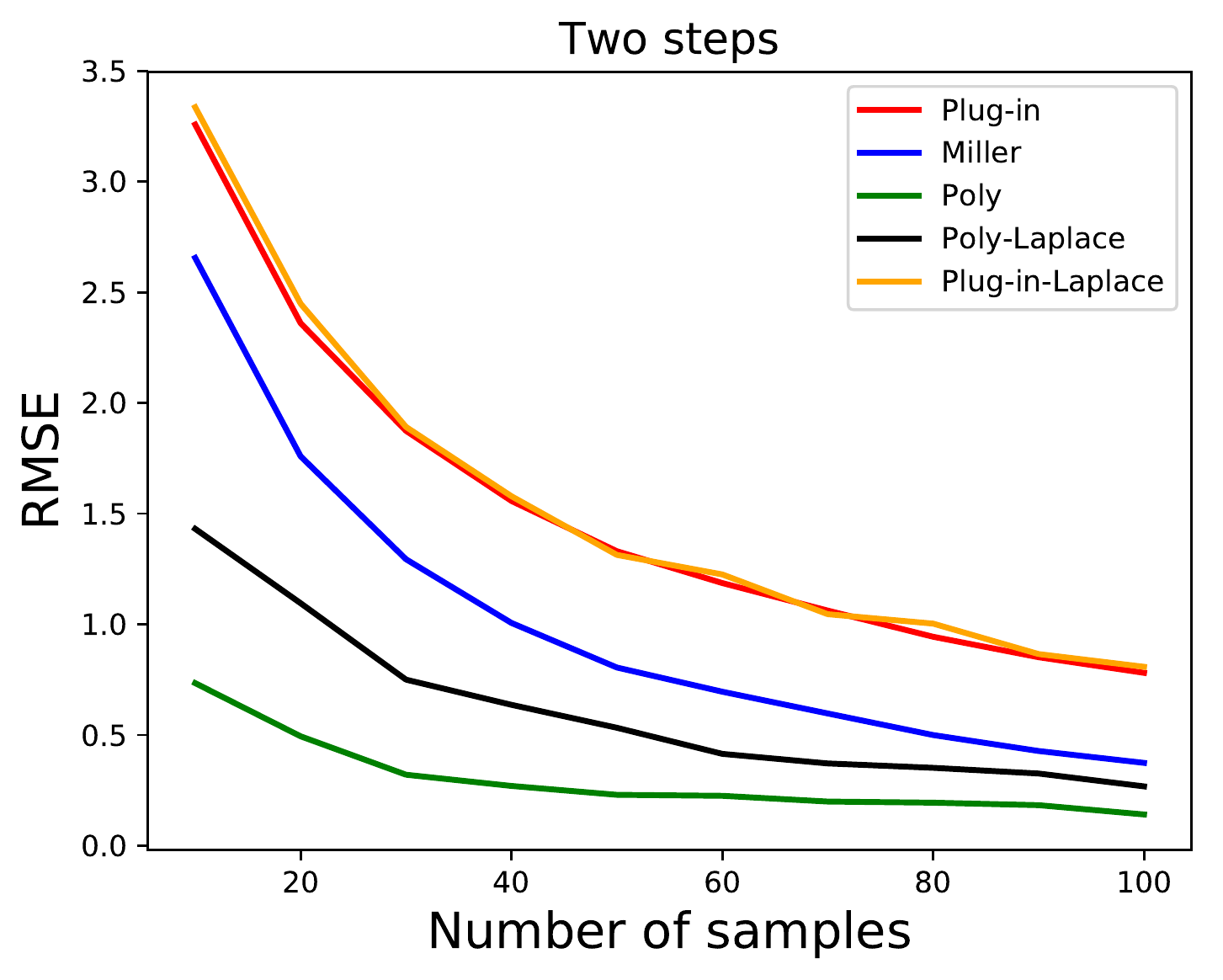}
\includegraphics[width=1\textwidth]{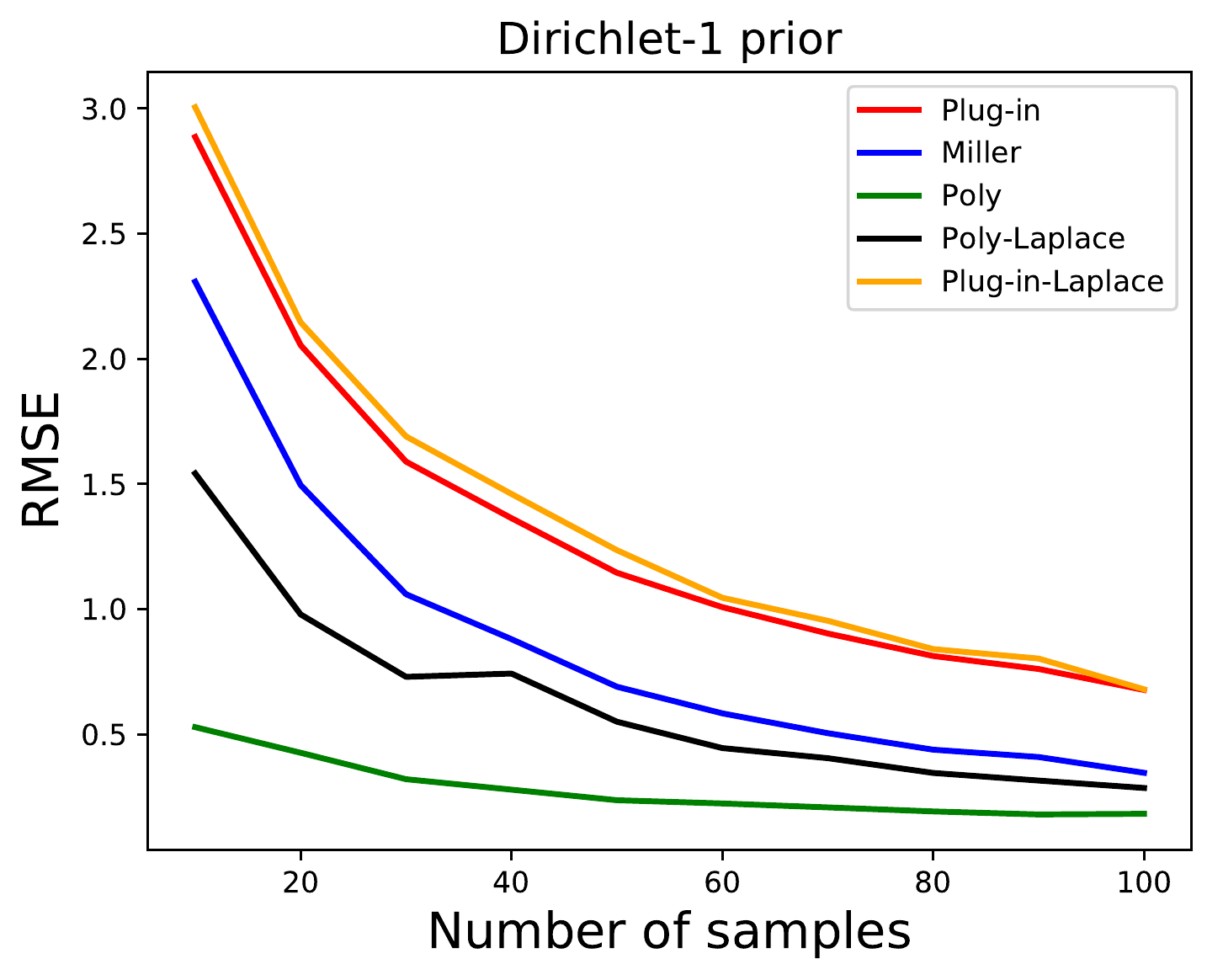}
\end{minipage}
}
\subfigure[]{
\begin{minipage}[b]{0.3\textwidth}
\includegraphics[width=1\textwidth]{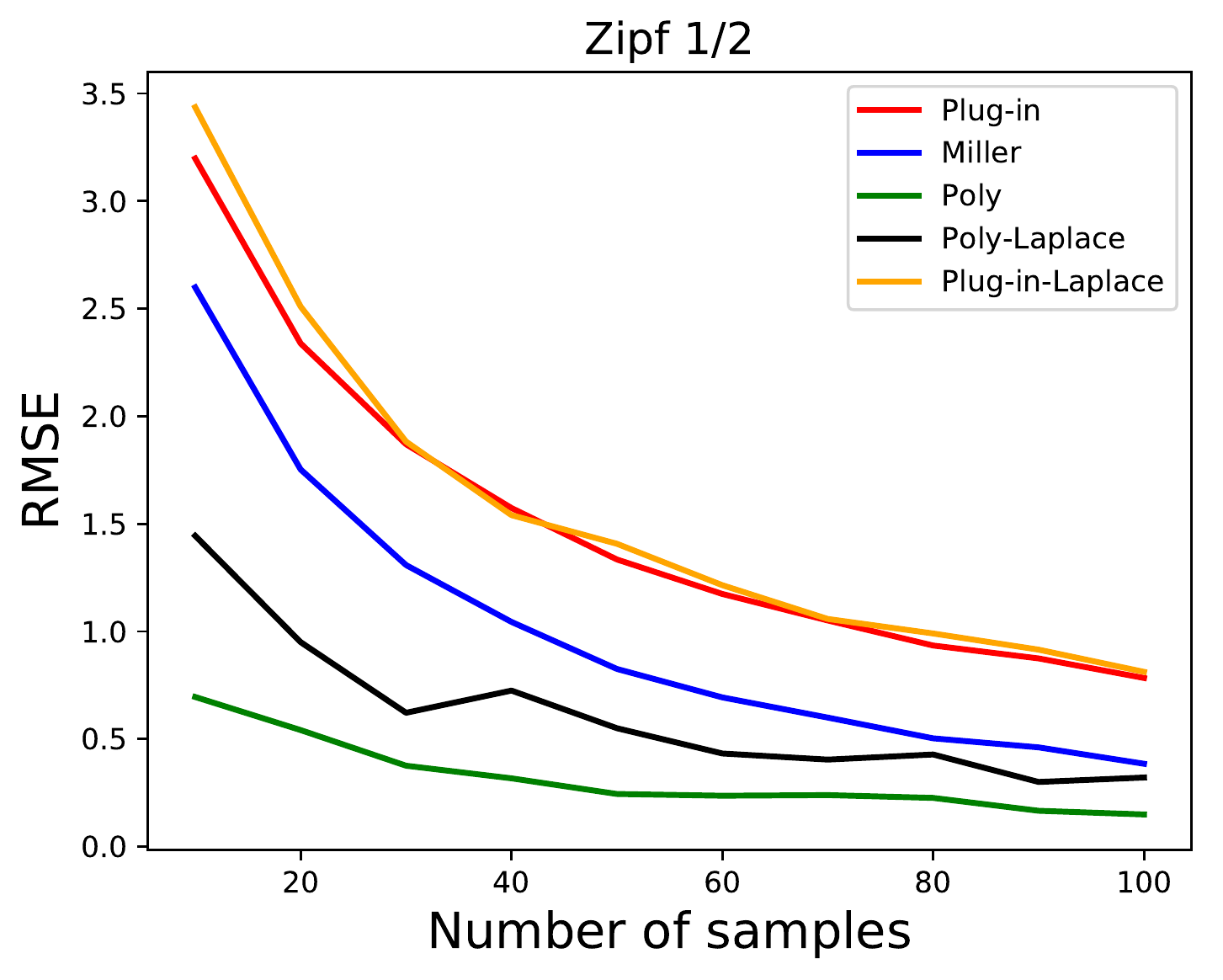}
\includegraphics[width=1\textwidth]{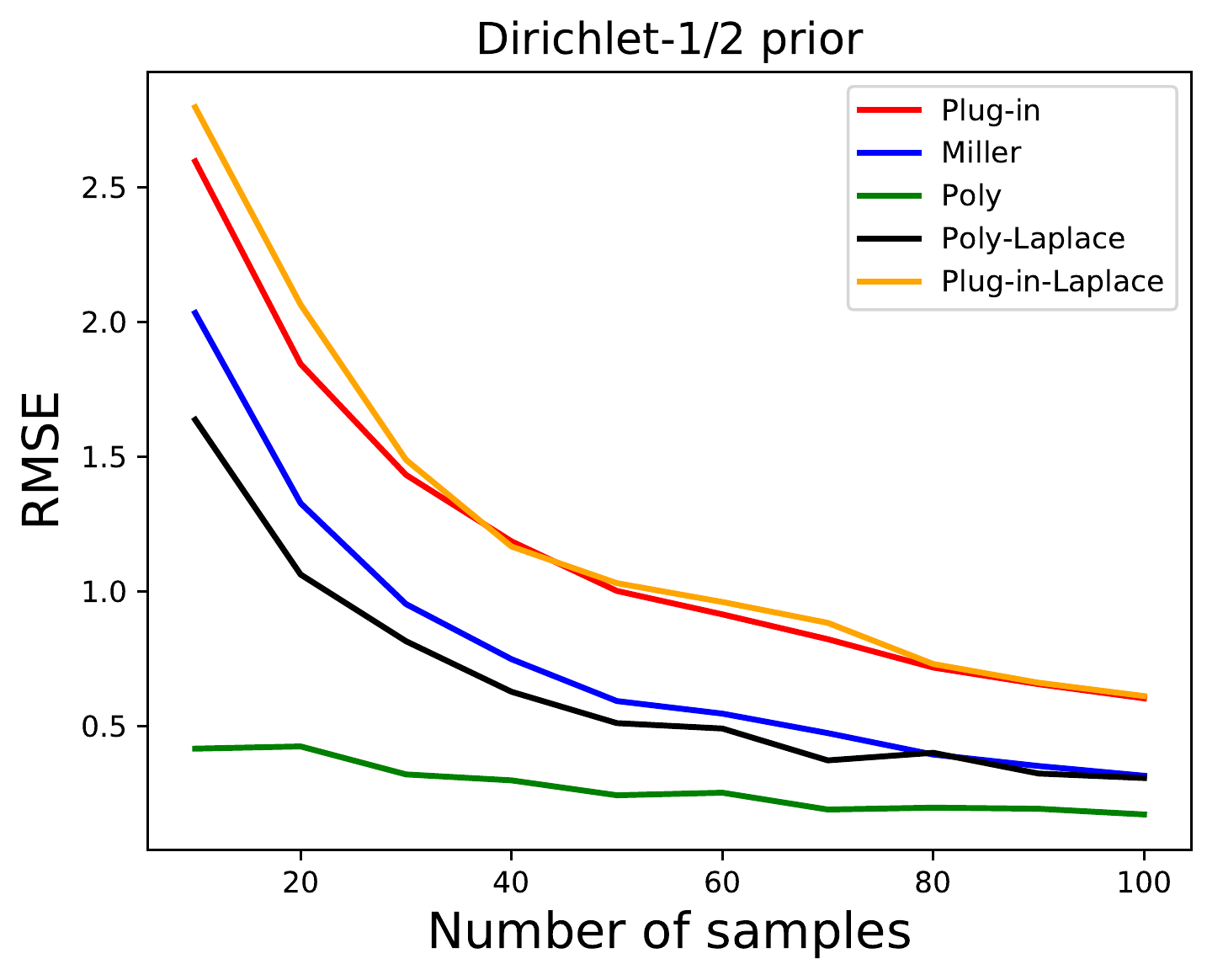}
\end{minipage}
}
\caption{Comparison of various estimators for the entropy, $k=100$, $\eps =1$.} 
\label{fig:entropy-k100-eps1}
\end{figure*}
\begin{figure*}
\centering
\subfigure[]{
\begin{minipage}[b]{0.3\textwidth}
\includegraphics[width=1\textwidth]{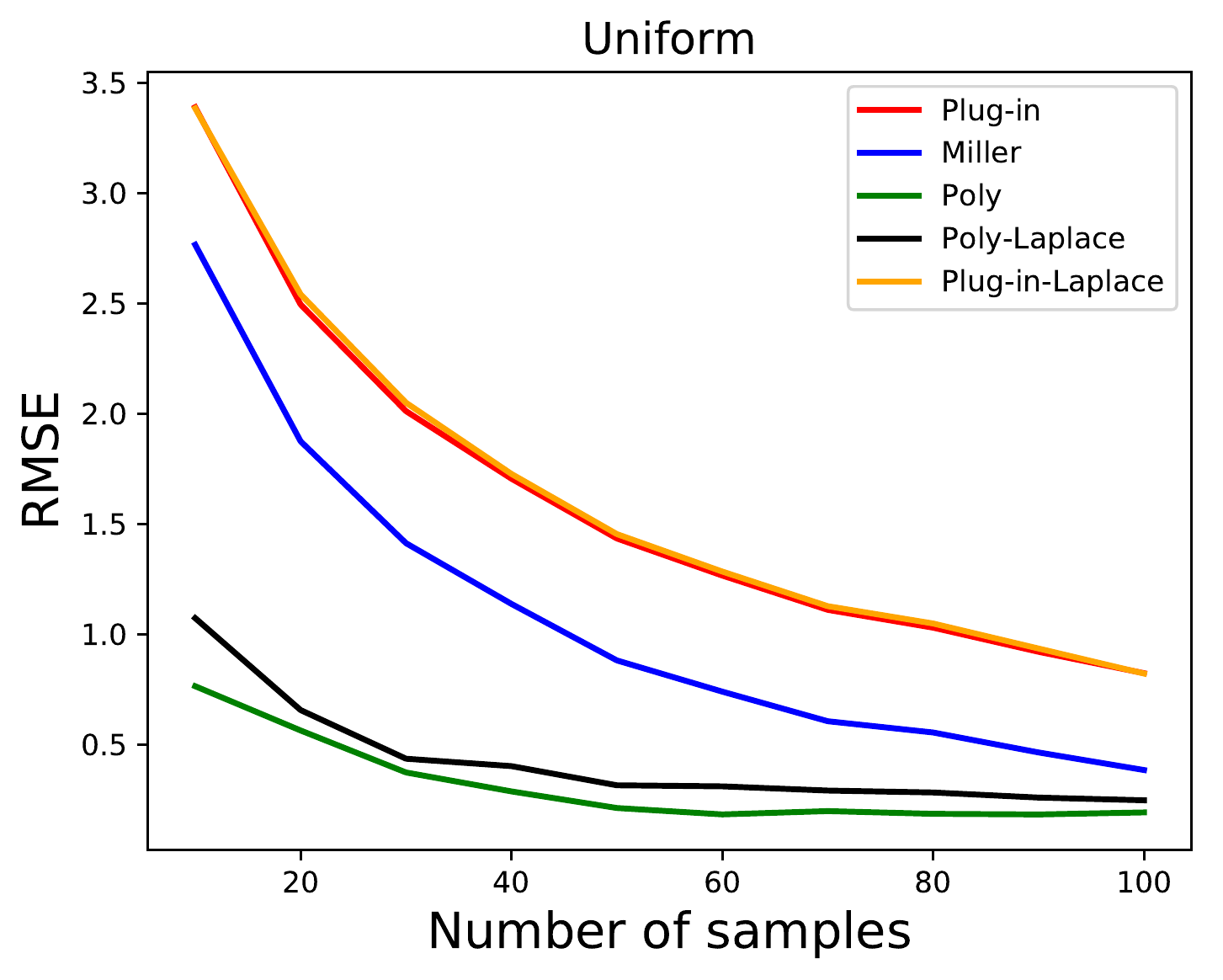}
\includegraphics[width=1\textwidth]{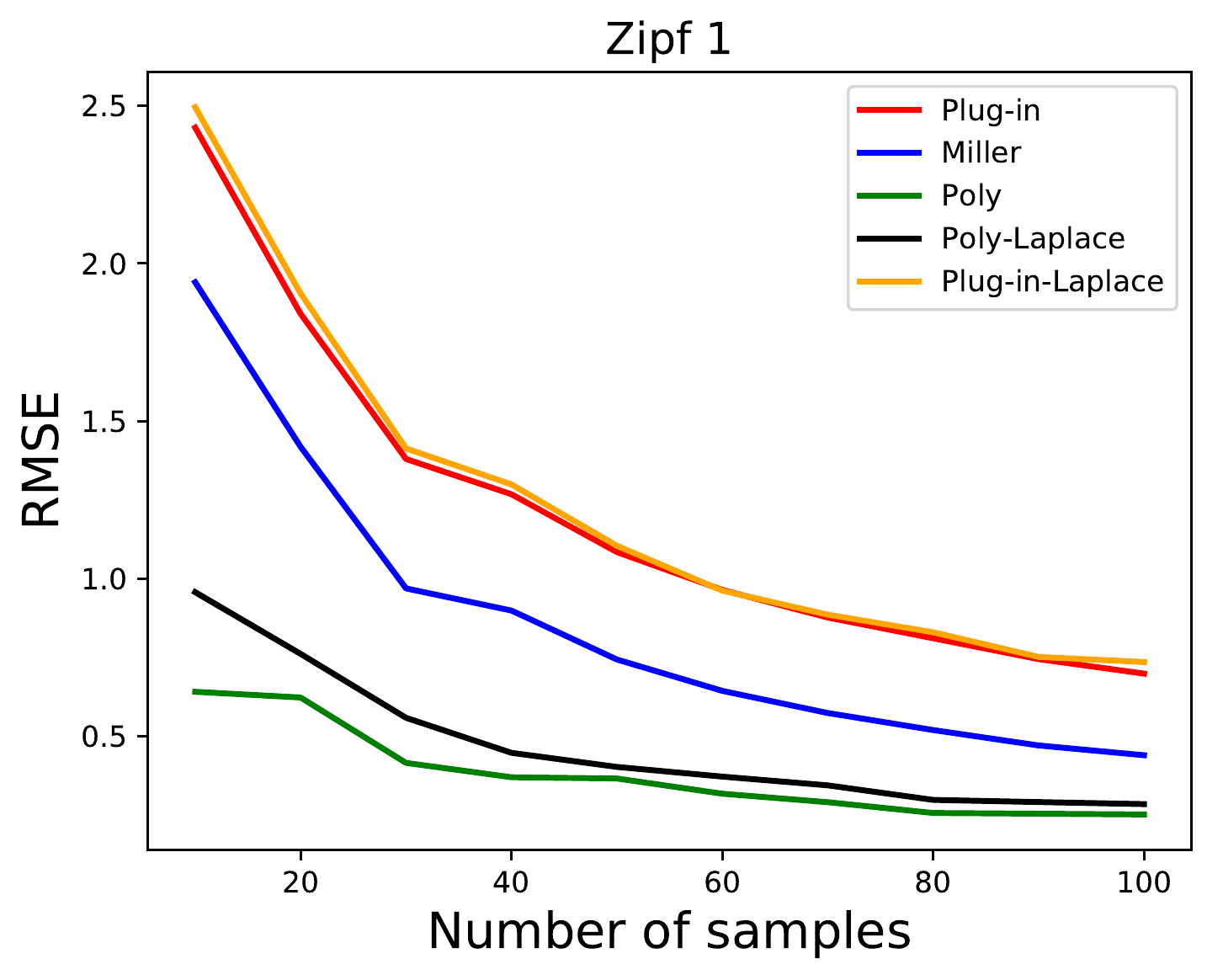}
\end{minipage}
}
\subfigure[]{
\begin{minipage}[b]{0.3\textwidth}
\includegraphics[width=1\textwidth]{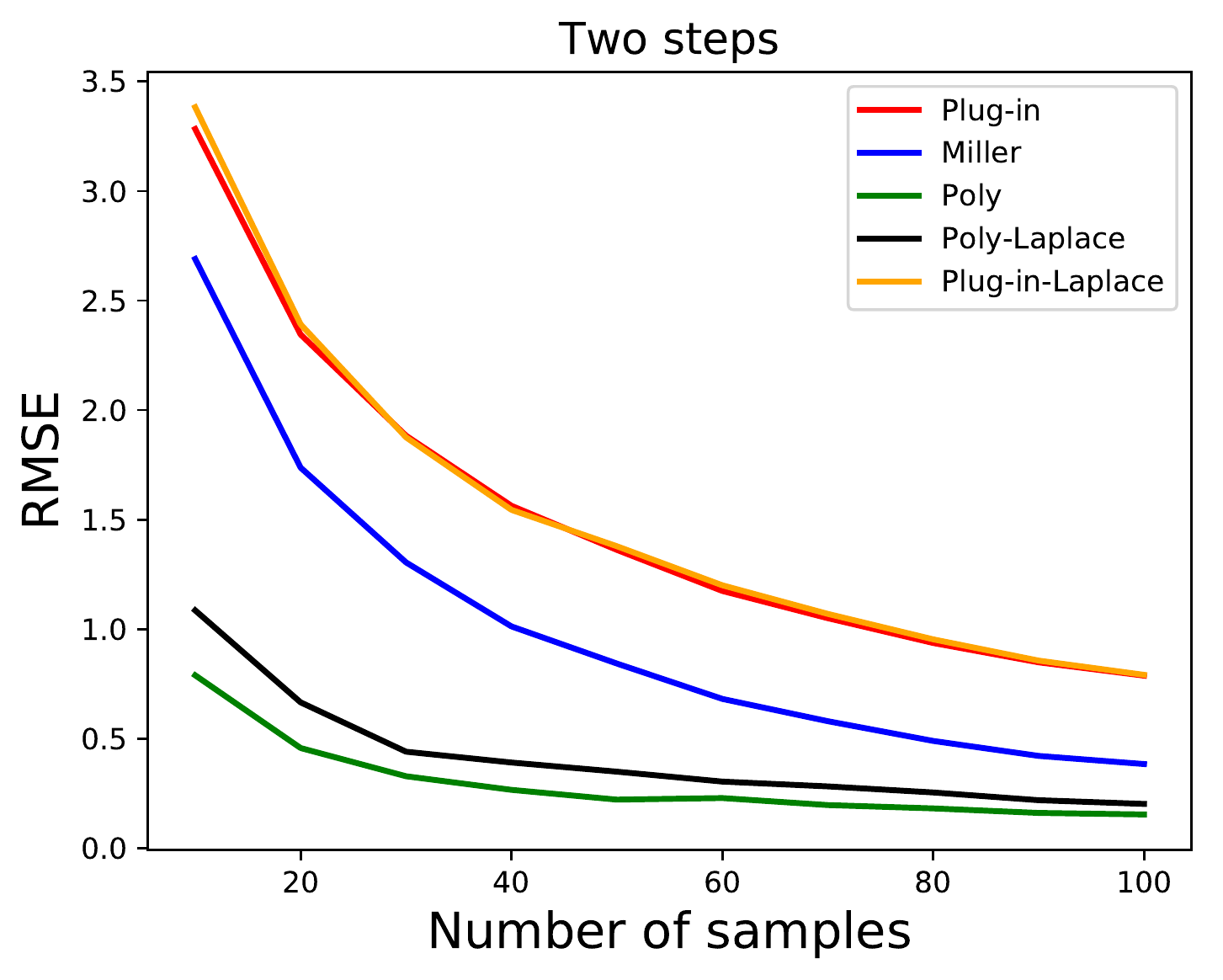}
\includegraphics[width=1\textwidth]{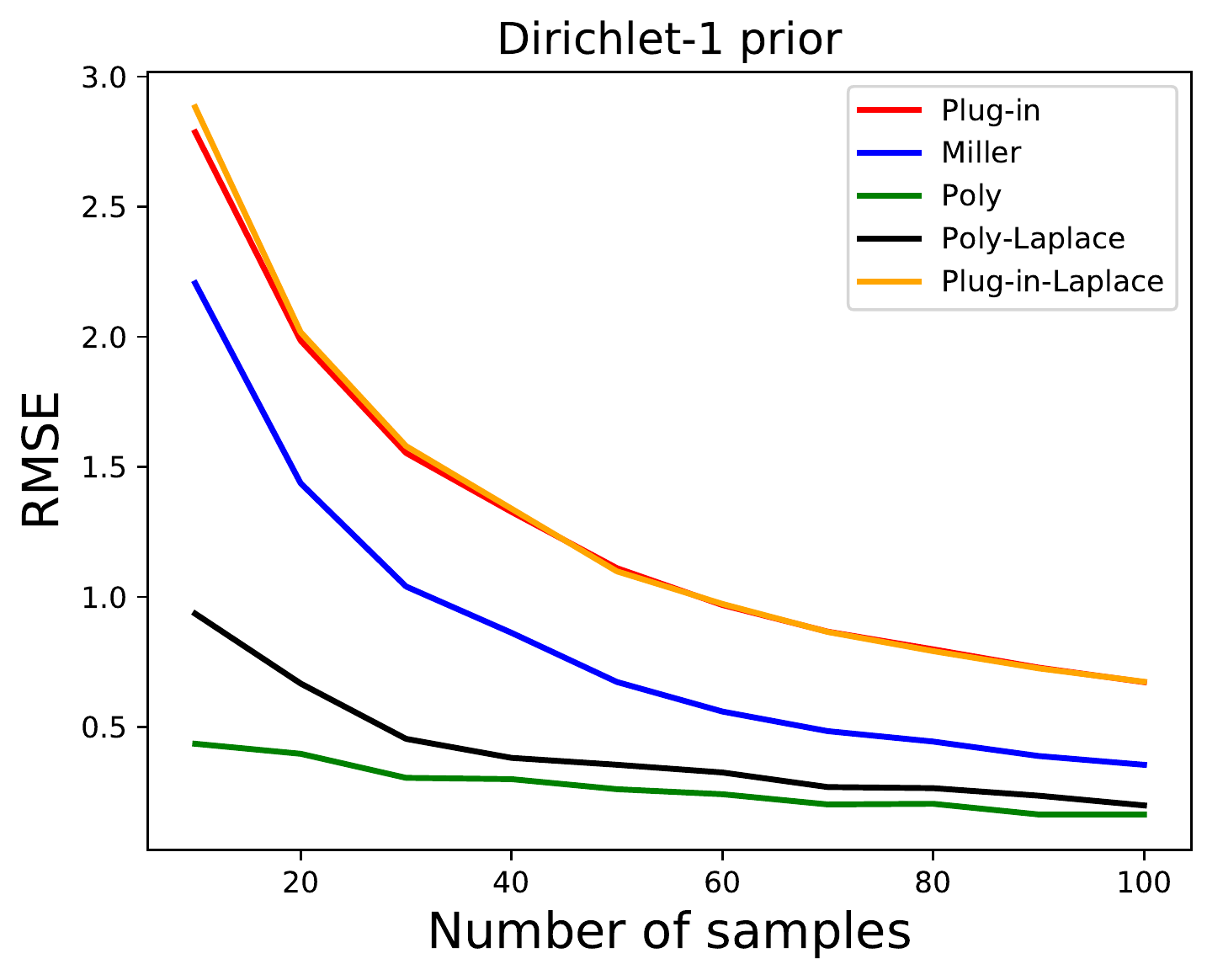}
\end{minipage}
}
\subfigure[]{
\begin{minipage}[b]{0.3\textwidth}
\includegraphics[width=1\textwidth]{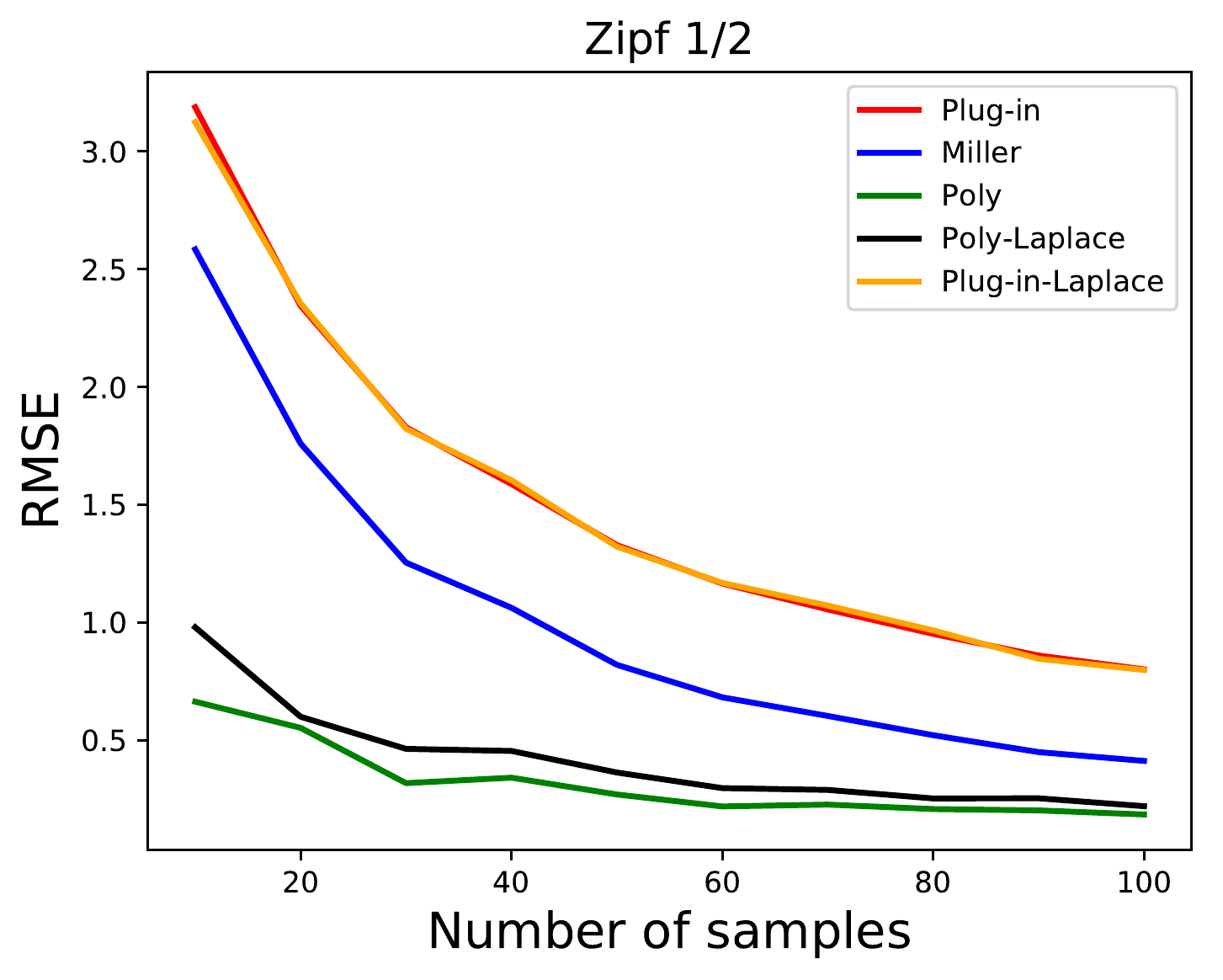}
\includegraphics[width=1\textwidth]{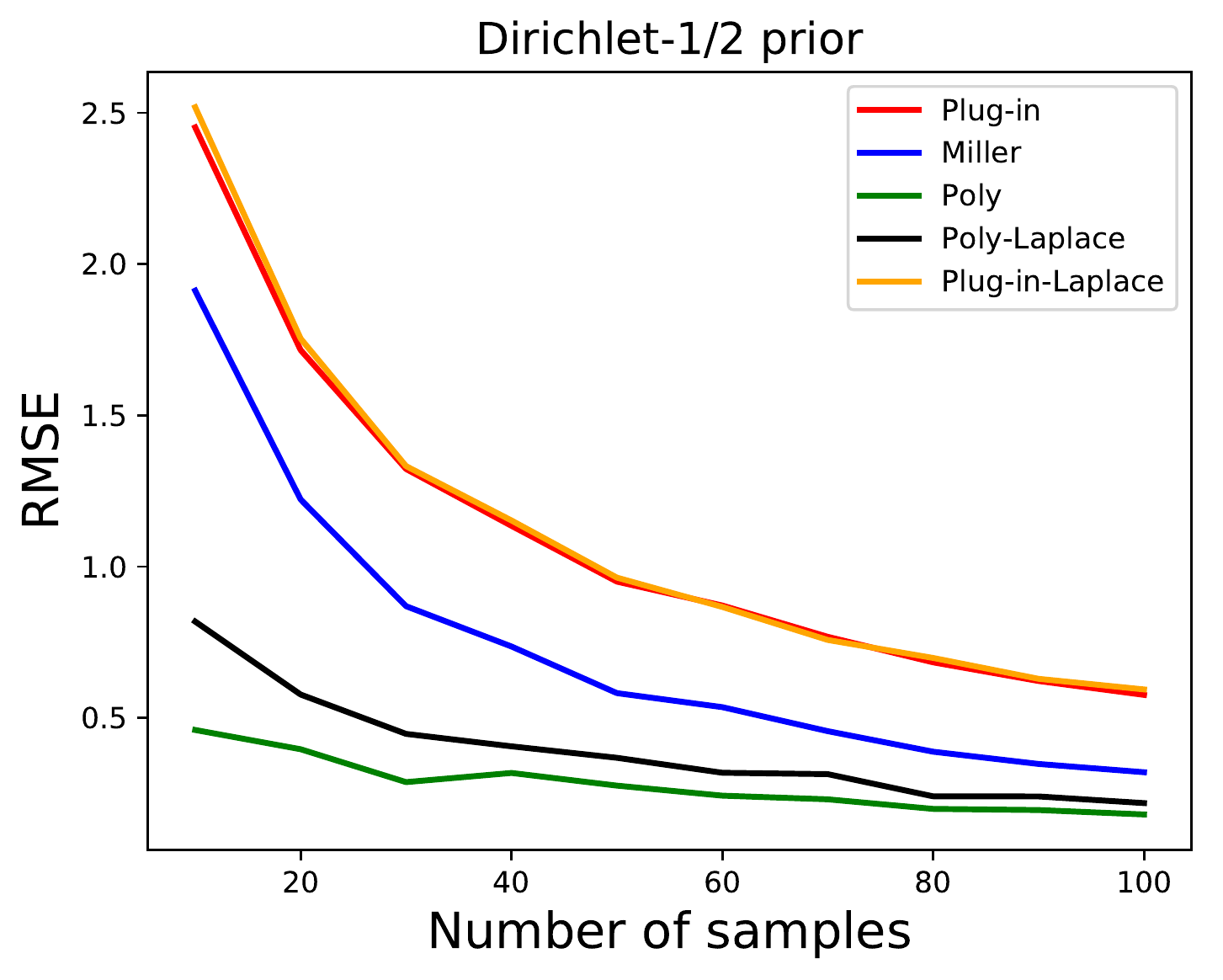}
\end{minipage}
}
\caption{Comparison of various estimators for the entropy, $k=100$, $\eps =2$.} 
\label{fig:entropy-k100-eps2}
\end{figure*}
\begin{figure*}
\centering
\subfigure[]{
\begin{minipage}[b]{0.3\textwidth}
\includegraphics[width=1\textwidth]{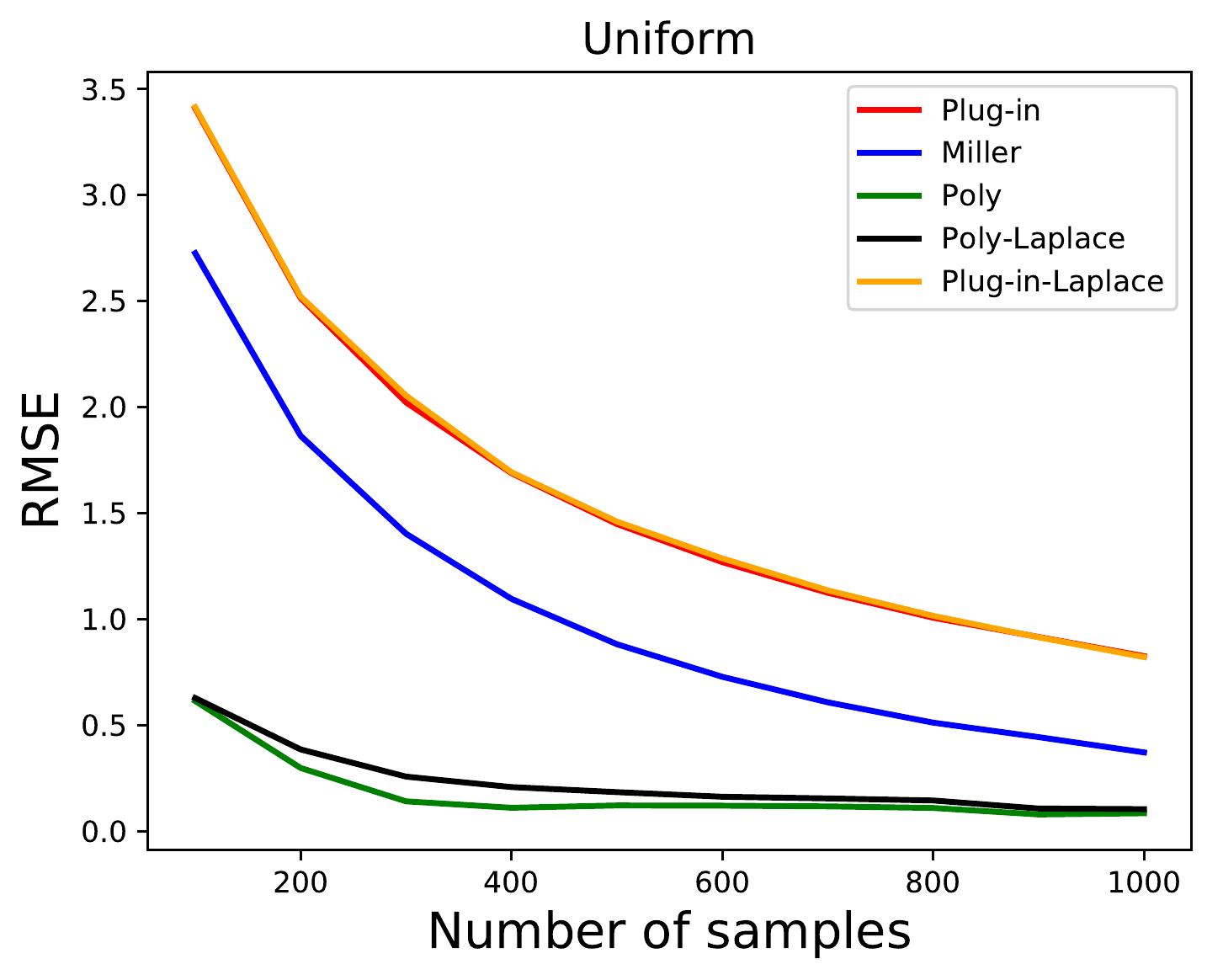}
\includegraphics[width=1\textwidth]{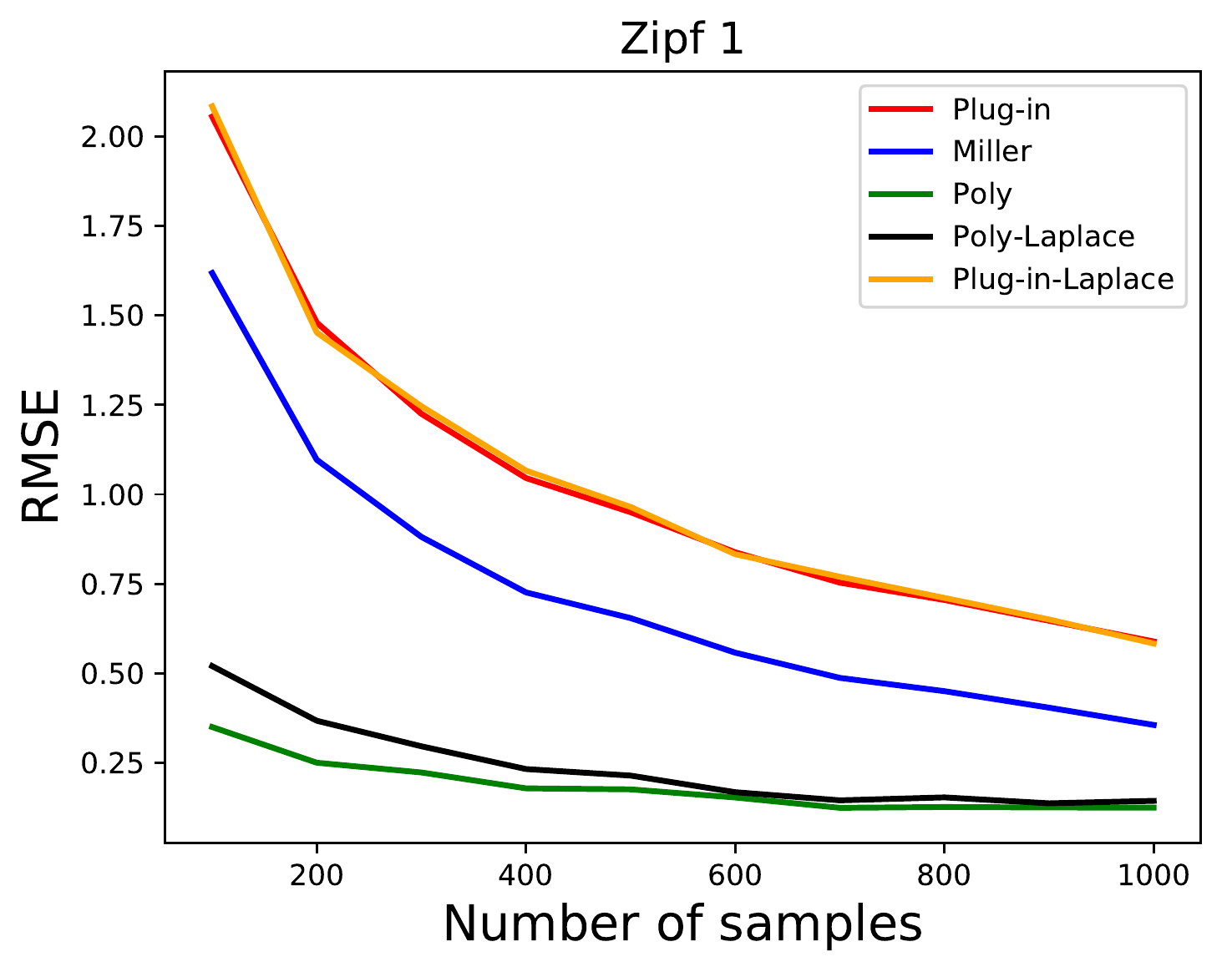}
\end{minipage}
}
\subfigure[]{
\begin{minipage}[b]{0.3\textwidth}
\includegraphics[width=1\textwidth]{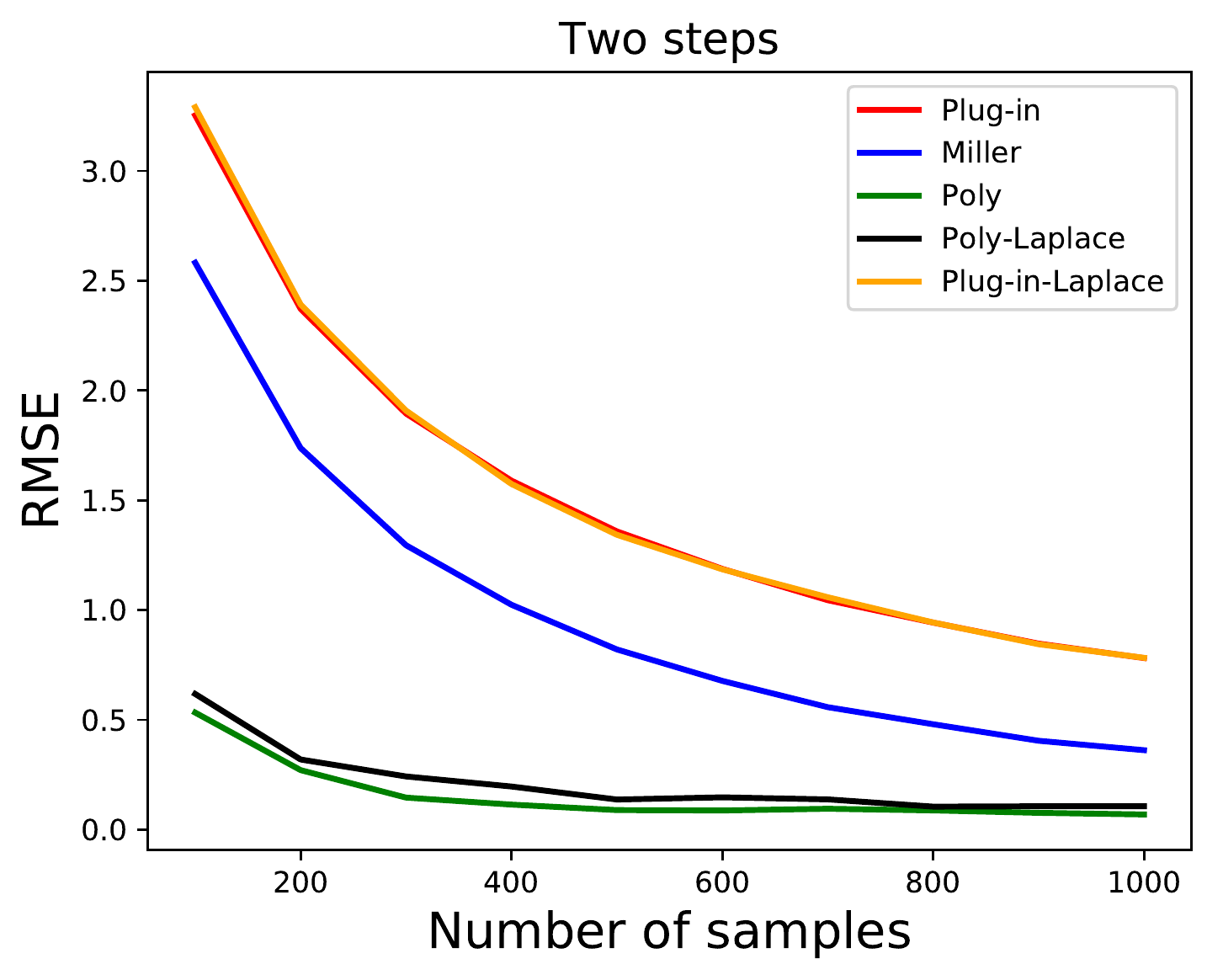}
\includegraphics[width=1\textwidth]{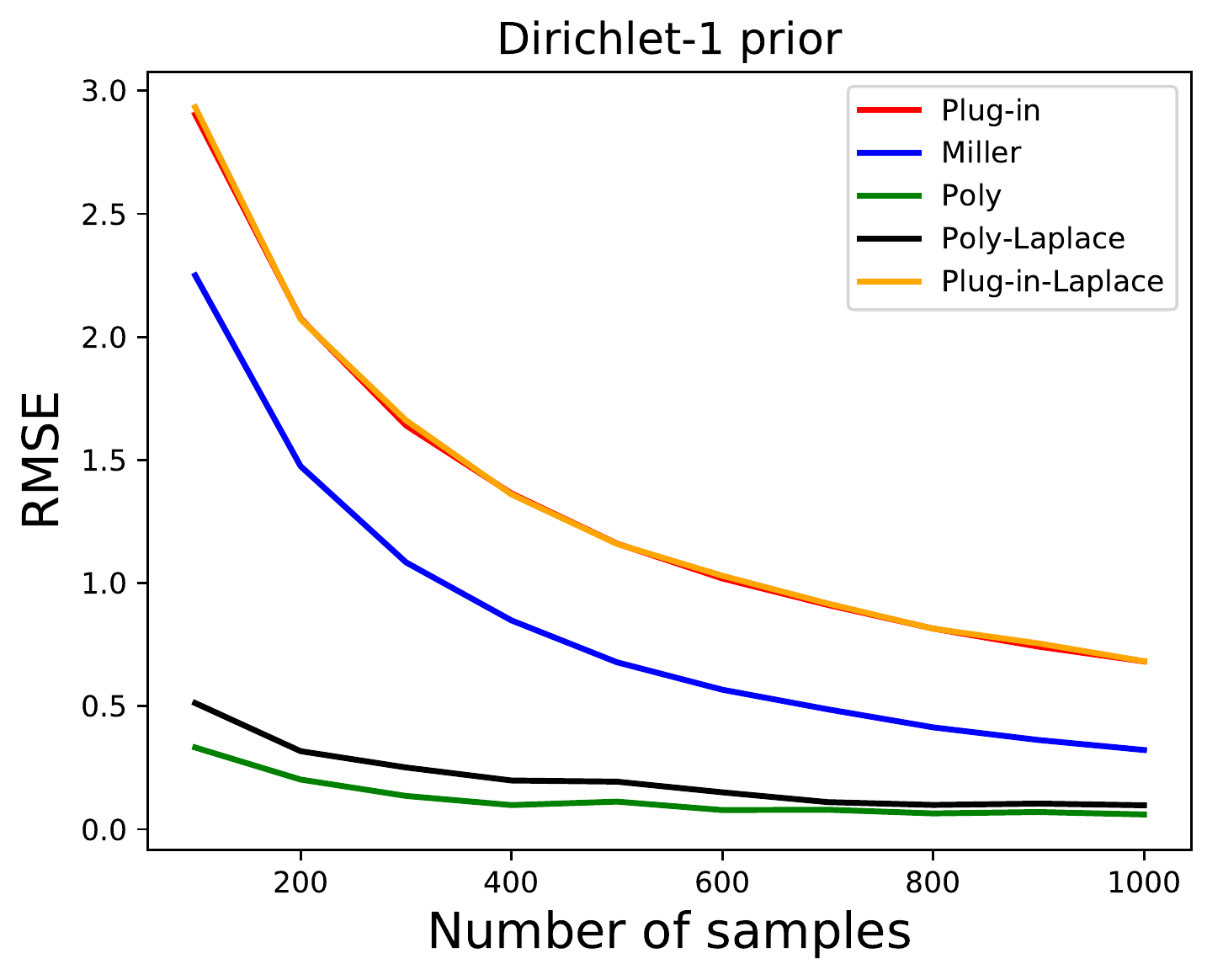}
\end{minipage}
}
\subfigure[]{
\begin{minipage}[b]{0.3\textwidth}
\includegraphics[width=1\textwidth]{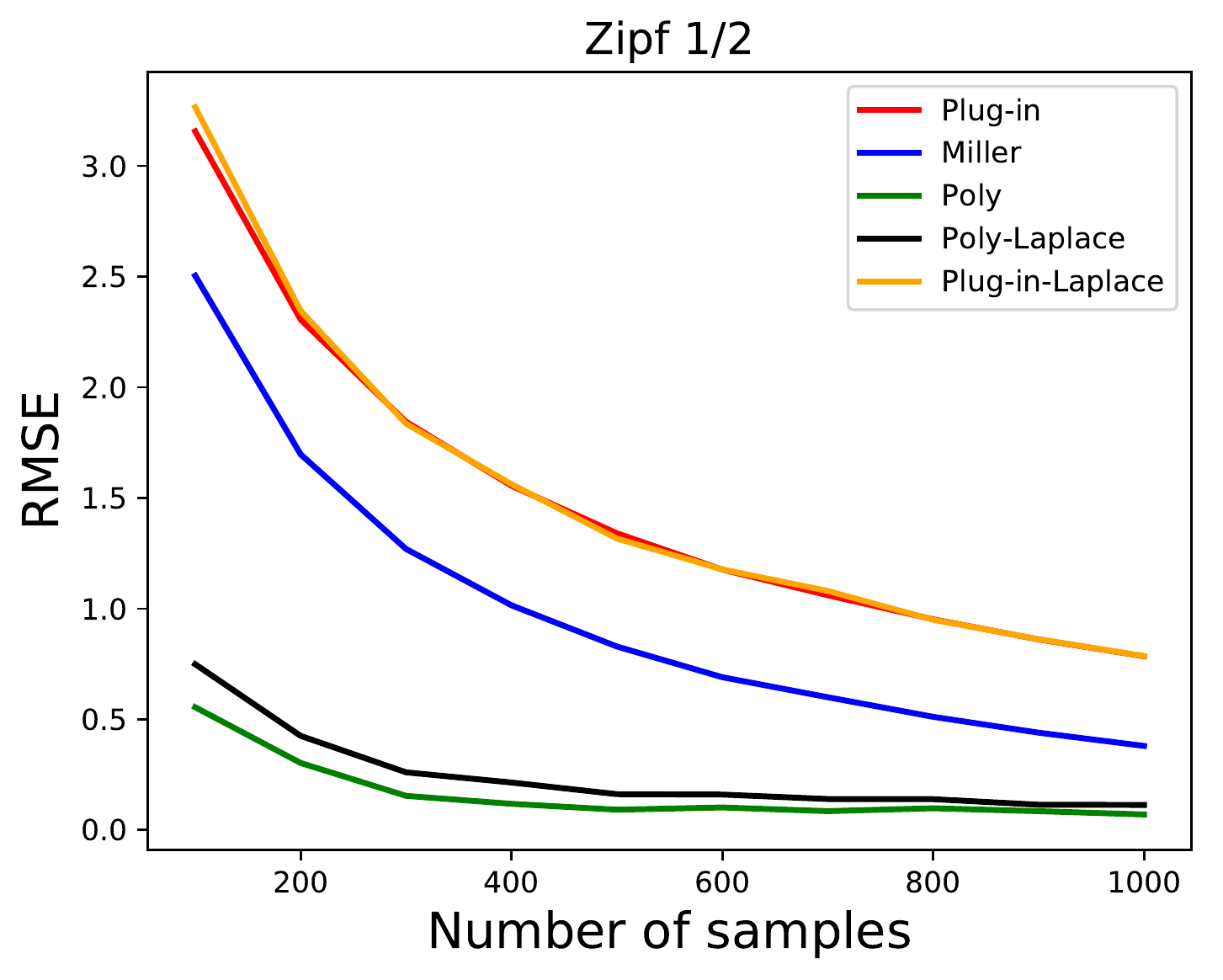}
\includegraphics[width=1\textwidth]{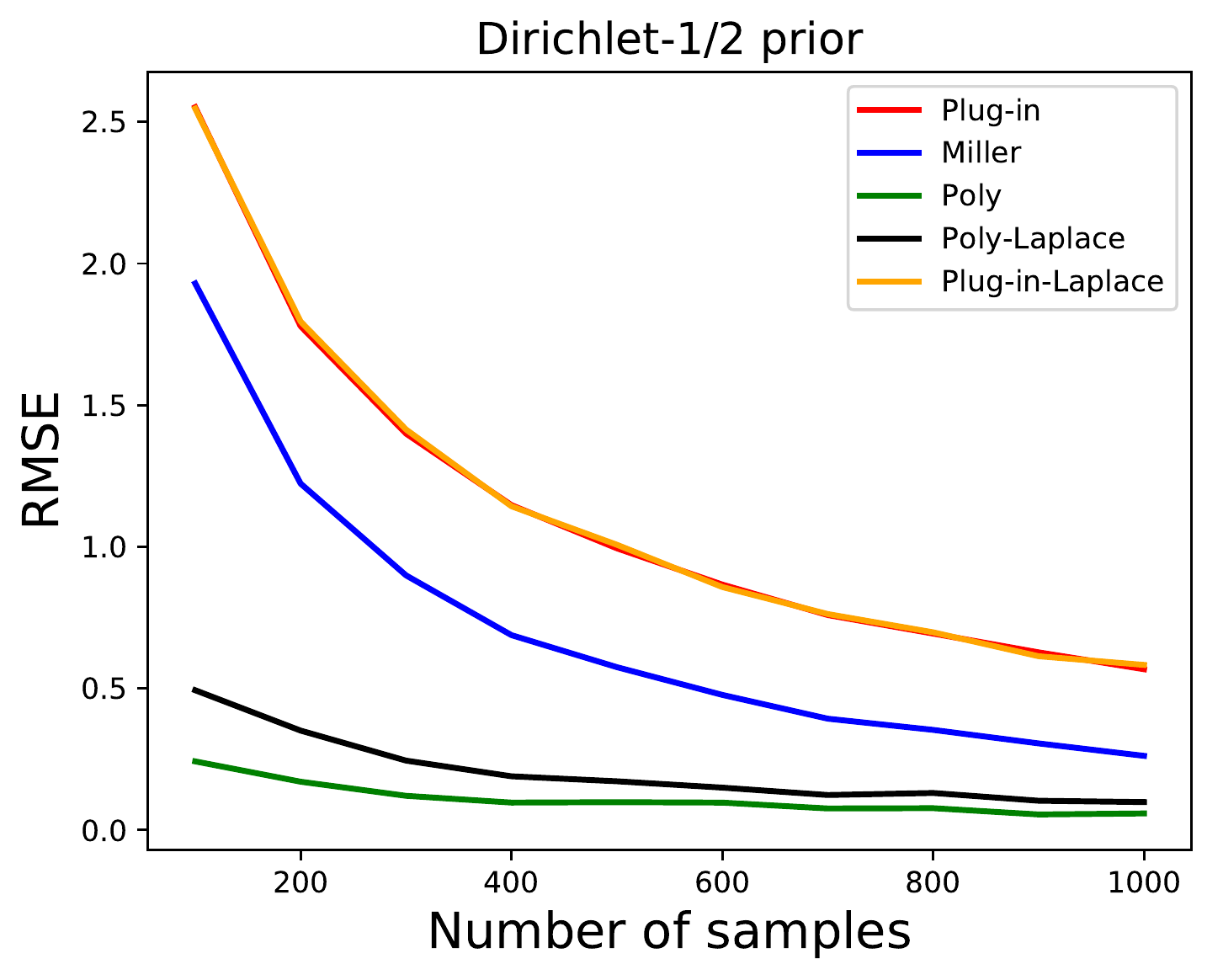}
\end{minipage}
}
\caption{Comparison of various estimators for the entropy, $k=1000$, $\eps =0.5$.} 
\label{fig:entropy-k1000-eps05}
\end{figure*}
\begin{figure*}
\centering
\subfigure[]{
\begin{minipage}[b]{0.3\textwidth}
\includegraphics[width=1\textwidth]{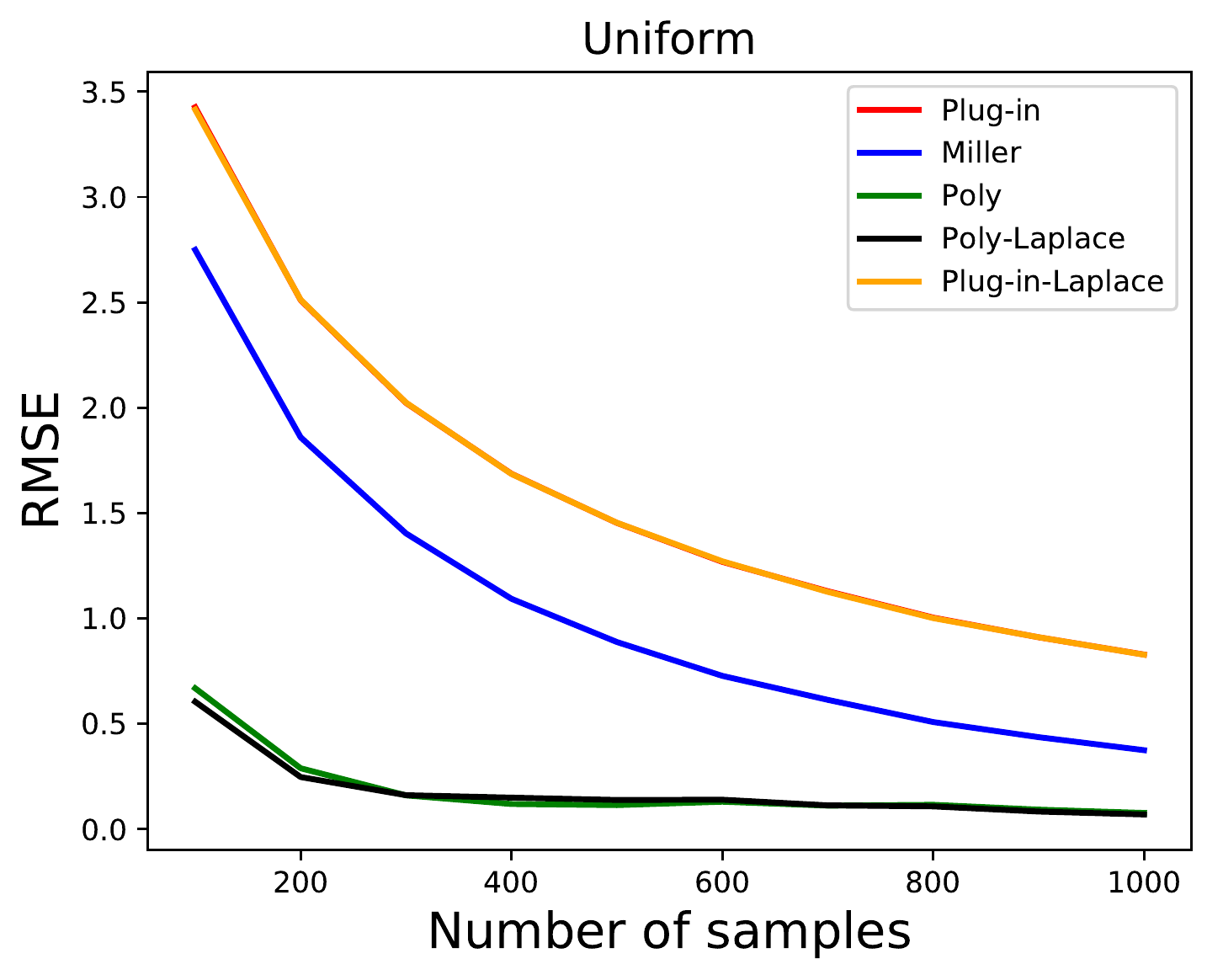}
\includegraphics[width=1\textwidth]{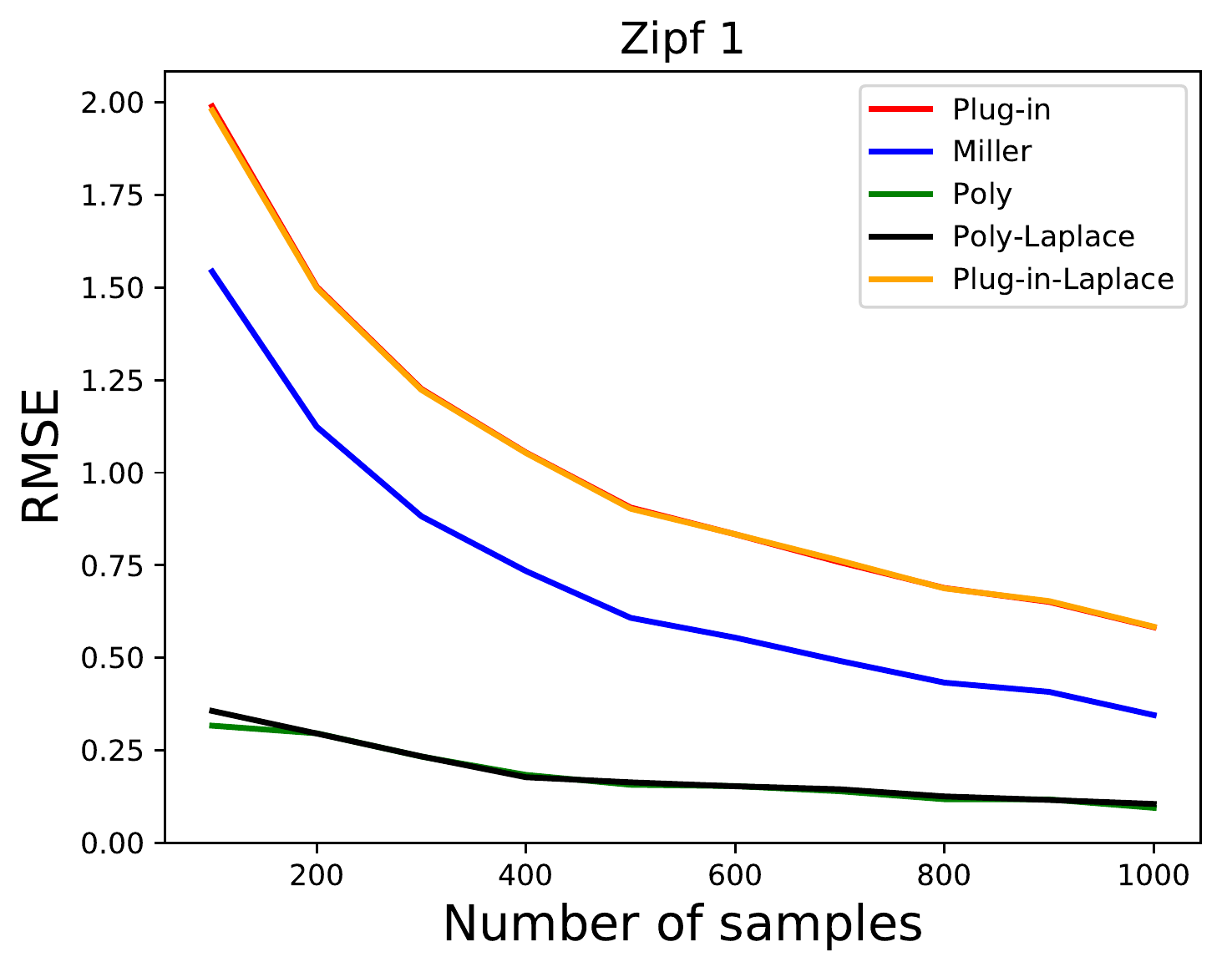}
\end{minipage}
}
\subfigure[]{
\begin{minipage}[b]{0.3\textwidth}
\includegraphics[width=1\textwidth]{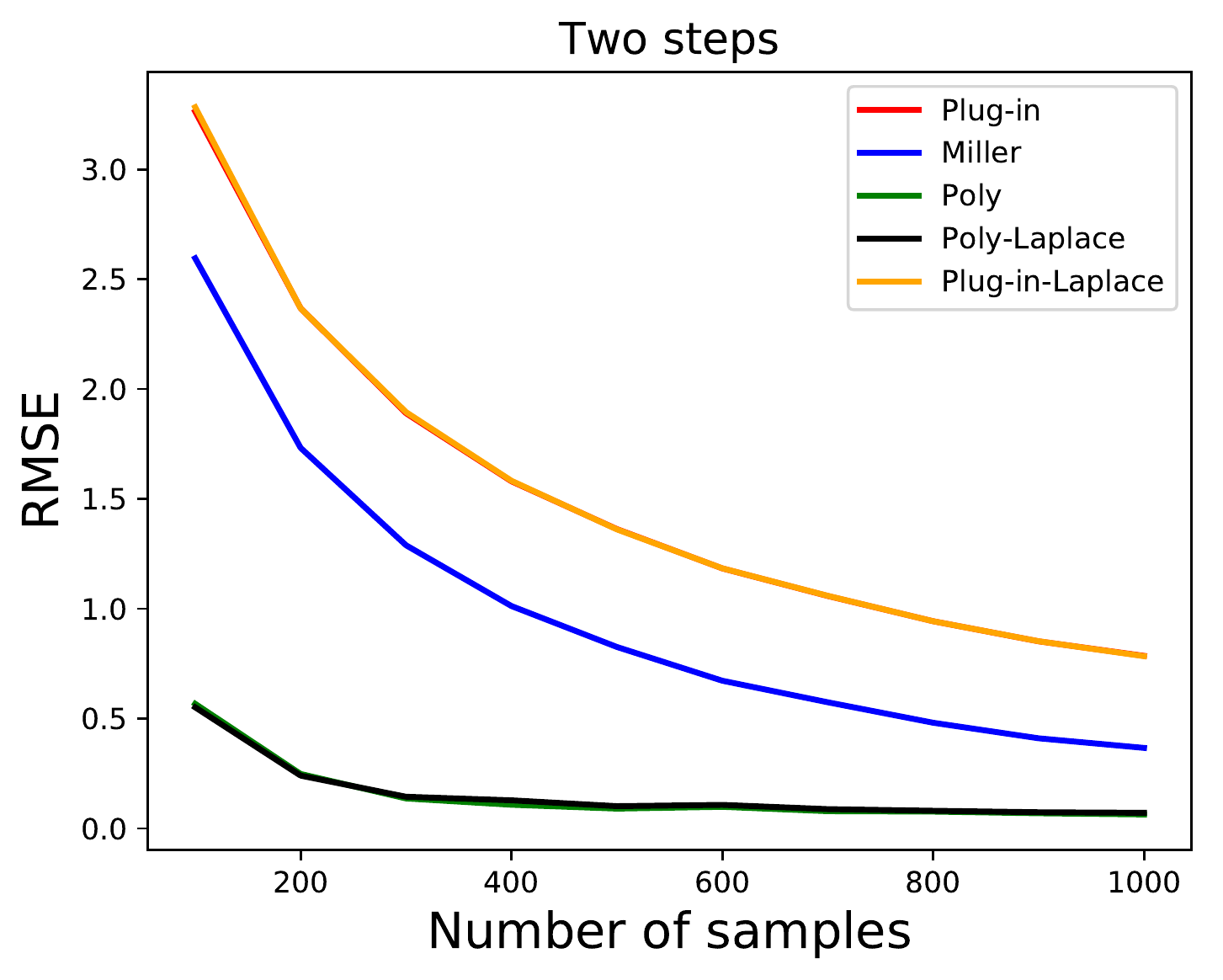}
\includegraphics[width=1\textwidth]{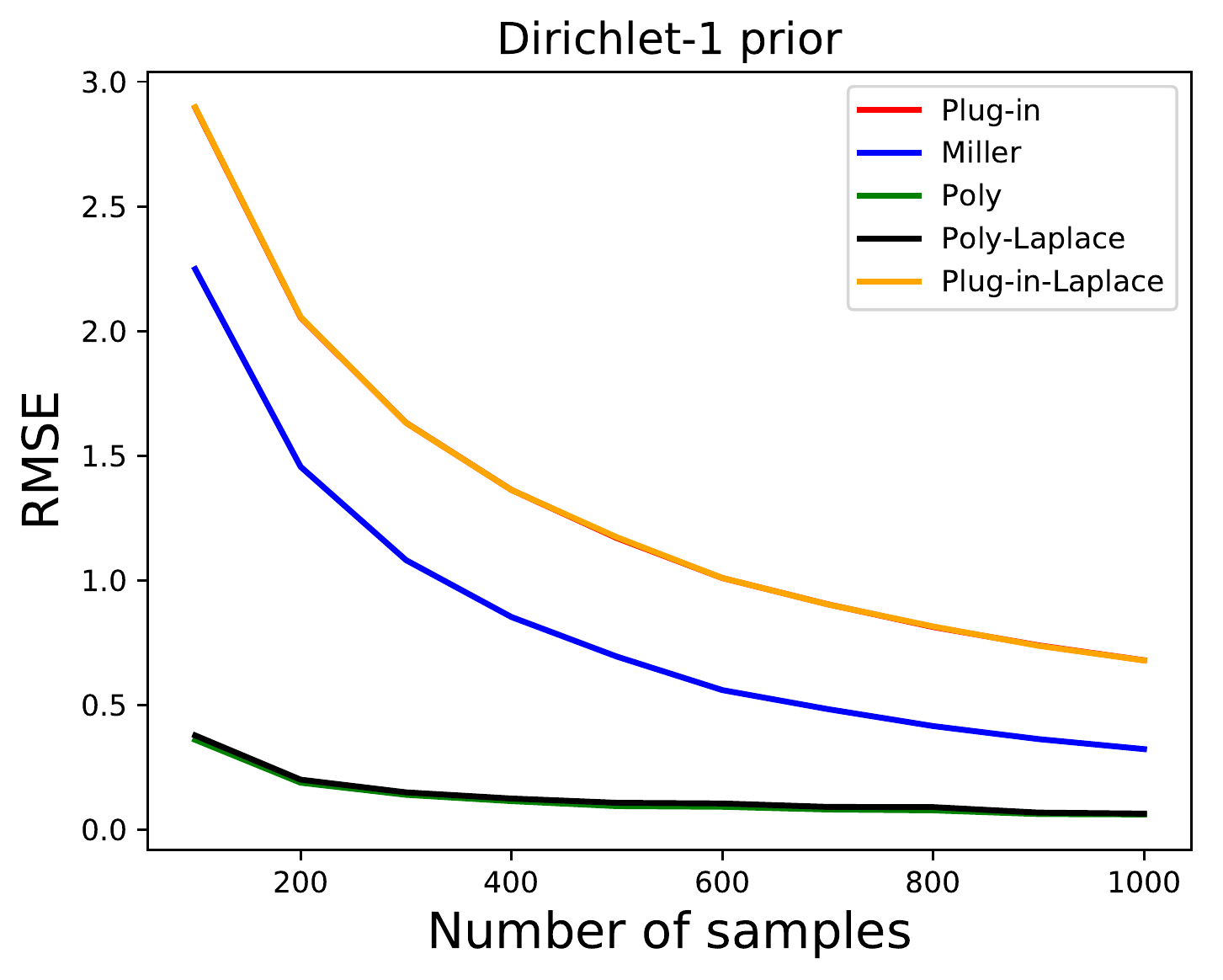}
\end{minipage}
}
\subfigure[]{
\begin{minipage}[b]{0.3\textwidth}
\includegraphics[width=1\textwidth]{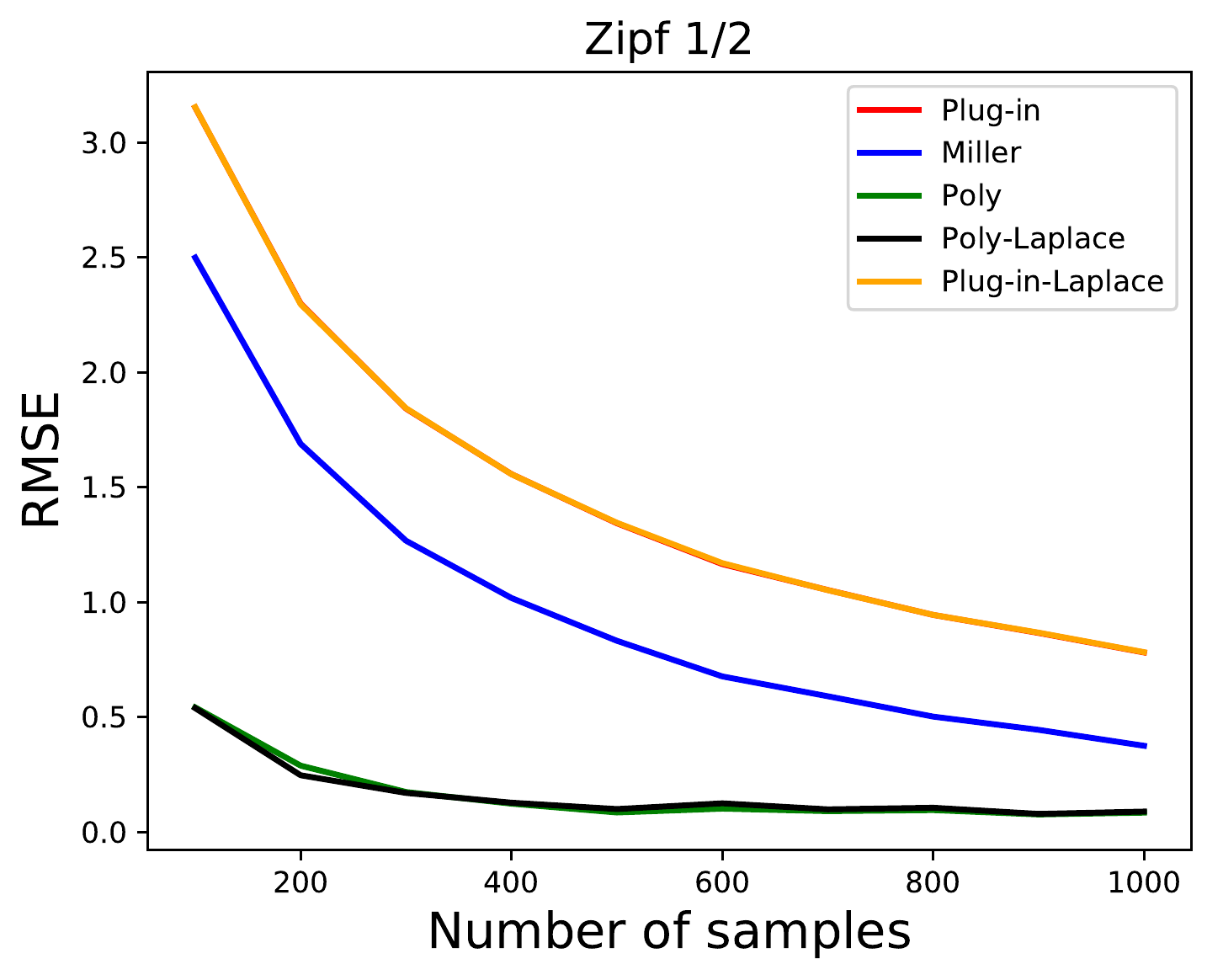}
\includegraphics[width=1\textwidth]{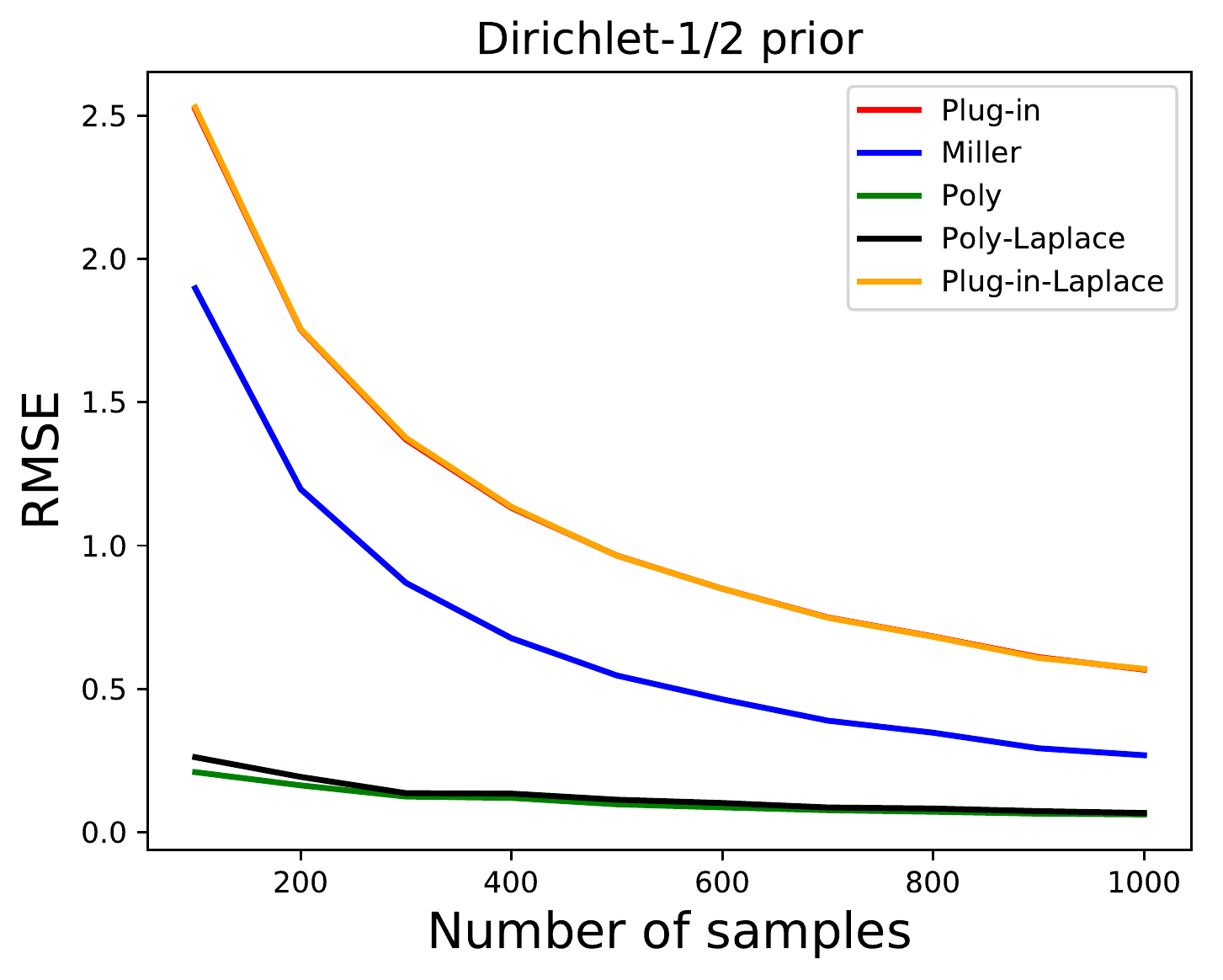}
\end{minipage}
}
\caption{Comparison of various estimators for the entropy, $k=1000$, $\eps =2$.} 
\label{fig:entropy-k1000-eps2}
\end{figure*}

\subsection{Support Coverage}
\label{sec:supp-exp-coverage}
We present three additional plots of our synthetic experimental results for support coverage estimation.
In particular, Figures~\ref{fig:coverage-k1000}, \ref{fig:coverage-k5000}, and~\ref{fig:coverage-k100000} show support coverage for $k$ = 1000, 5000, 100000.
\begin{figure*}
\centering
\subfigure[]{
\begin{minipage}[b]{0.3\textwidth}
\includegraphics[width=1\textwidth]{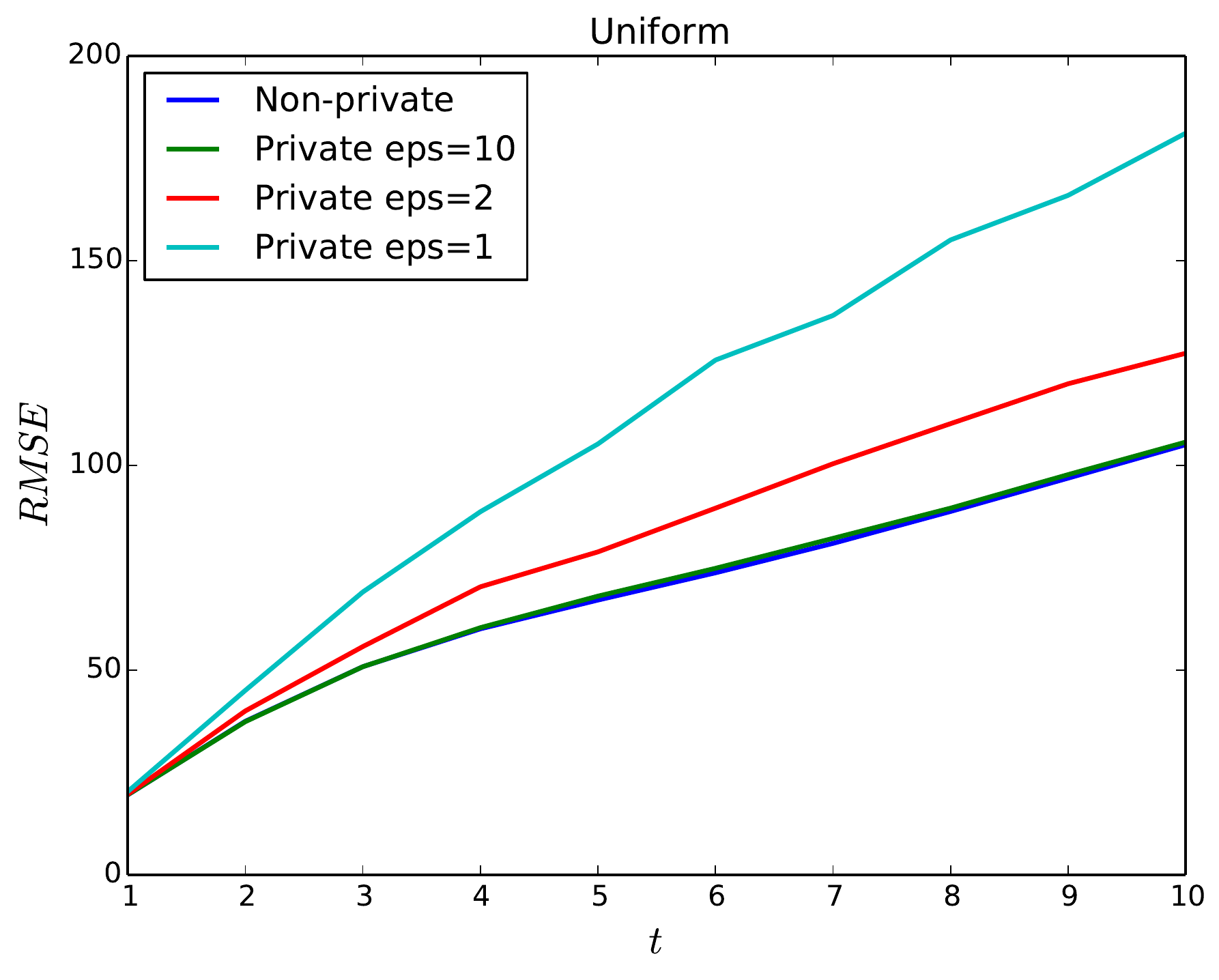}
\includegraphics[width=1\textwidth]{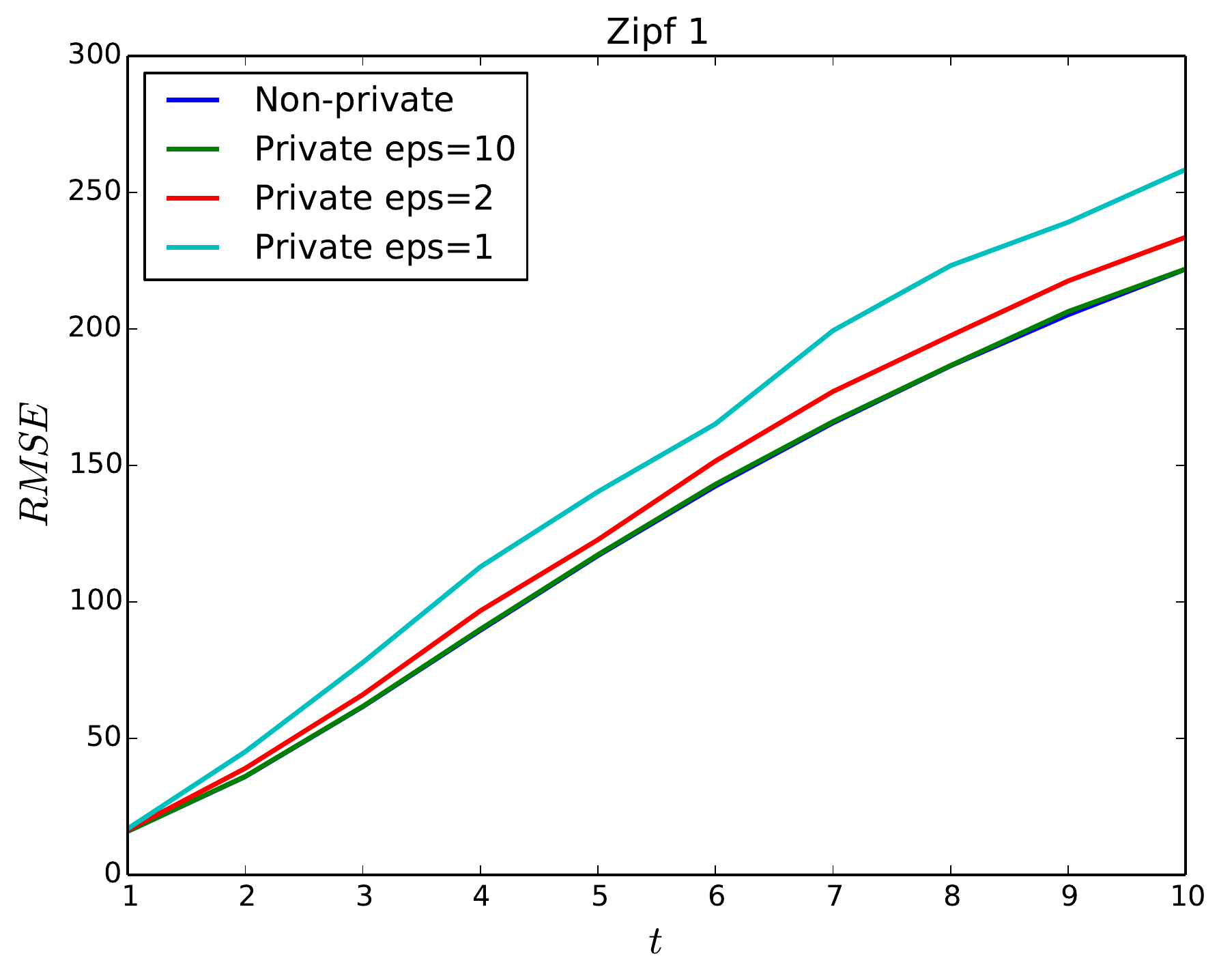}
\end{minipage}
}
\subfigure[]{
\begin{minipage}[b]{0.3\textwidth}
\includegraphics[width=1\textwidth]{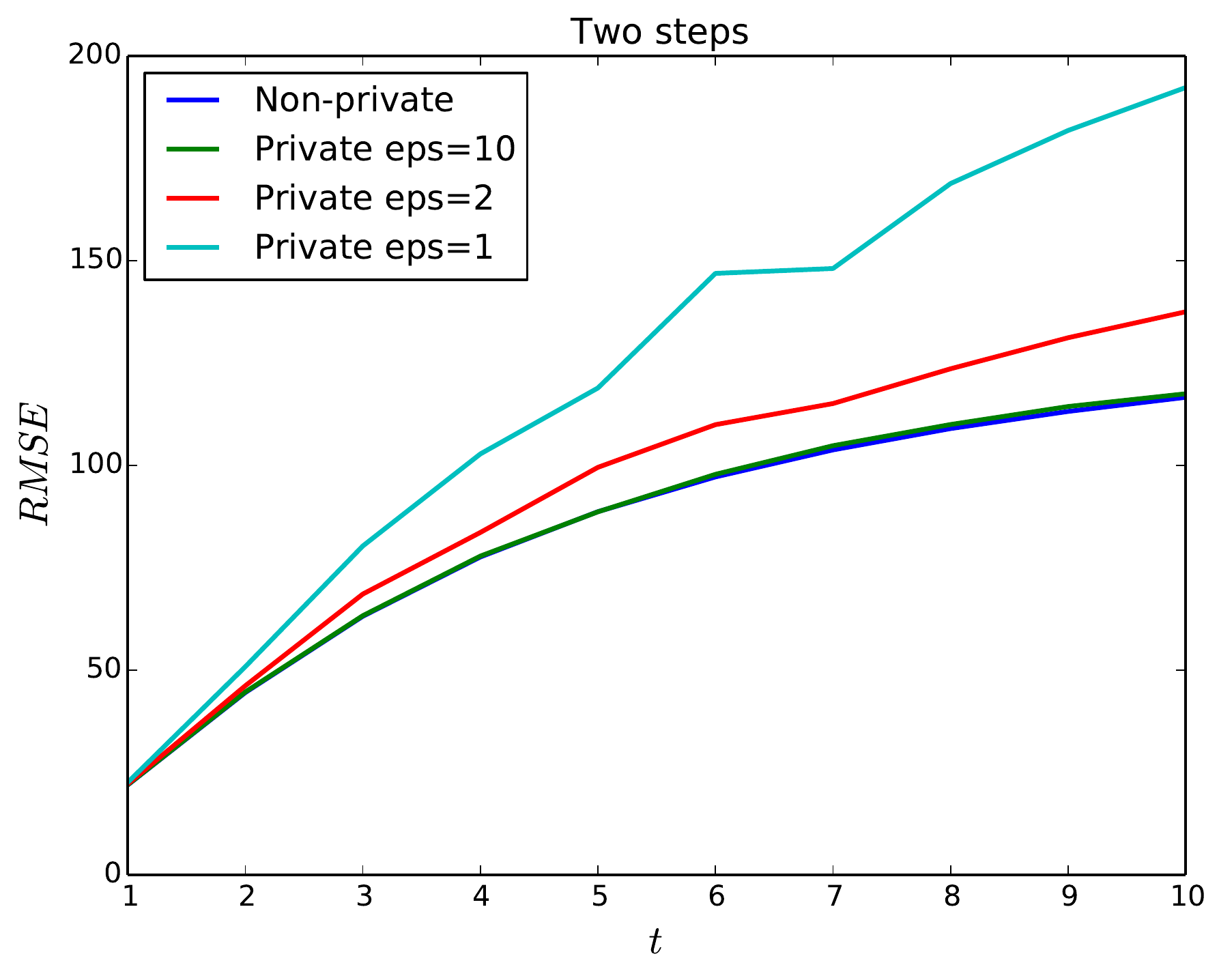}
\includegraphics[width=1\textwidth]{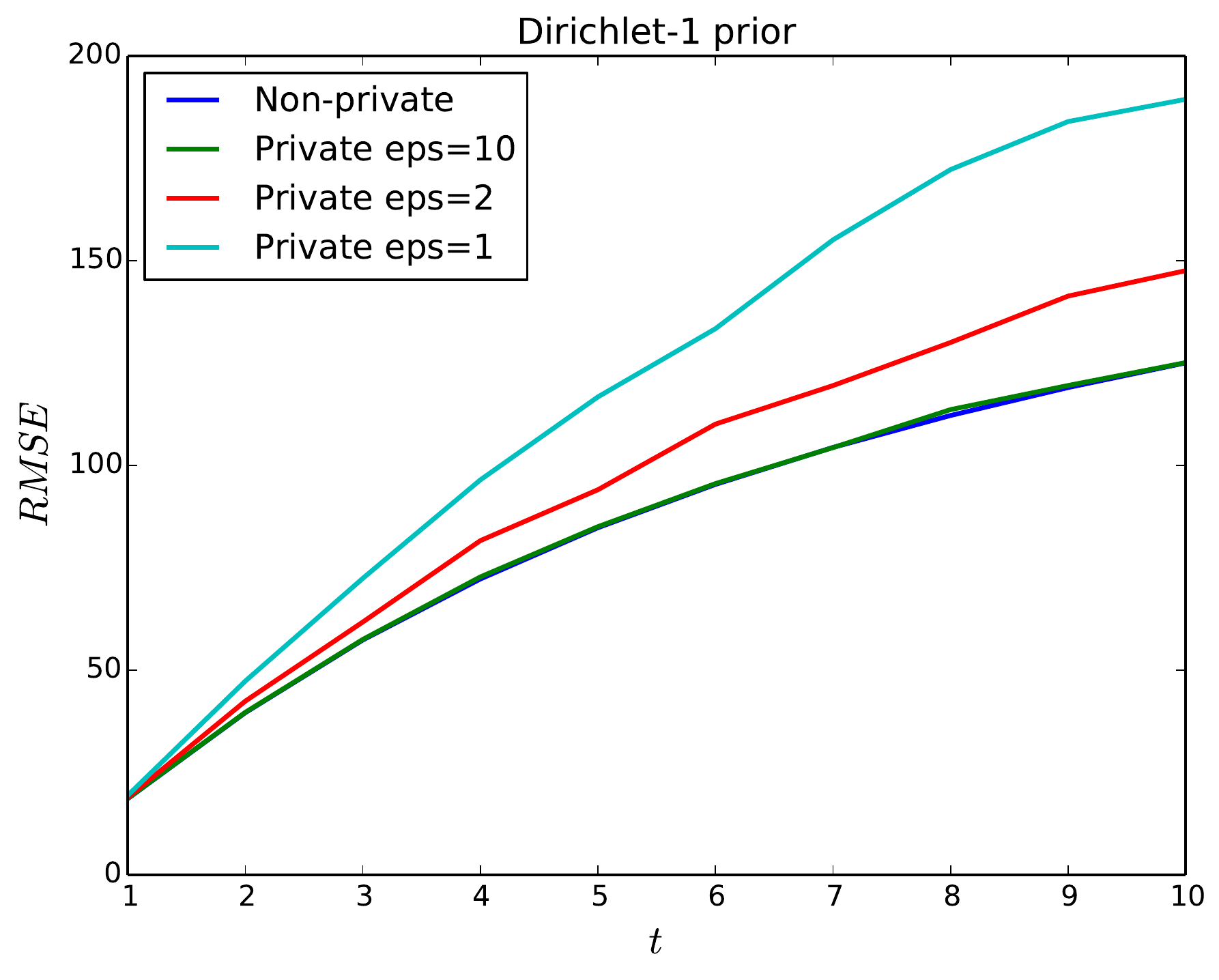}
\end{minipage}
}
\subfigure[]{
\begin{minipage}[b]{0.3\textwidth}
\includegraphics[width=1\textwidth]{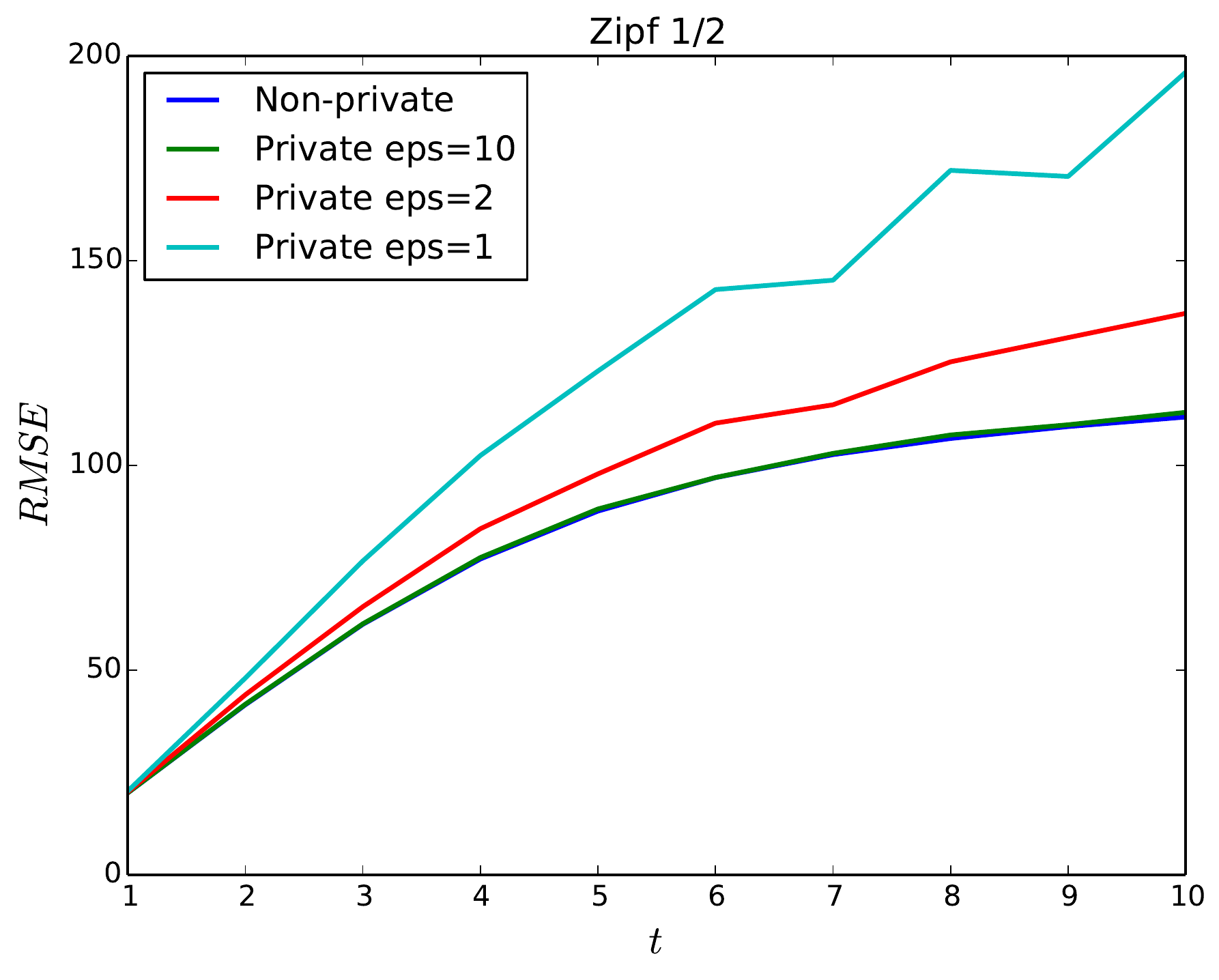}
\includegraphics[width=1\textwidth]{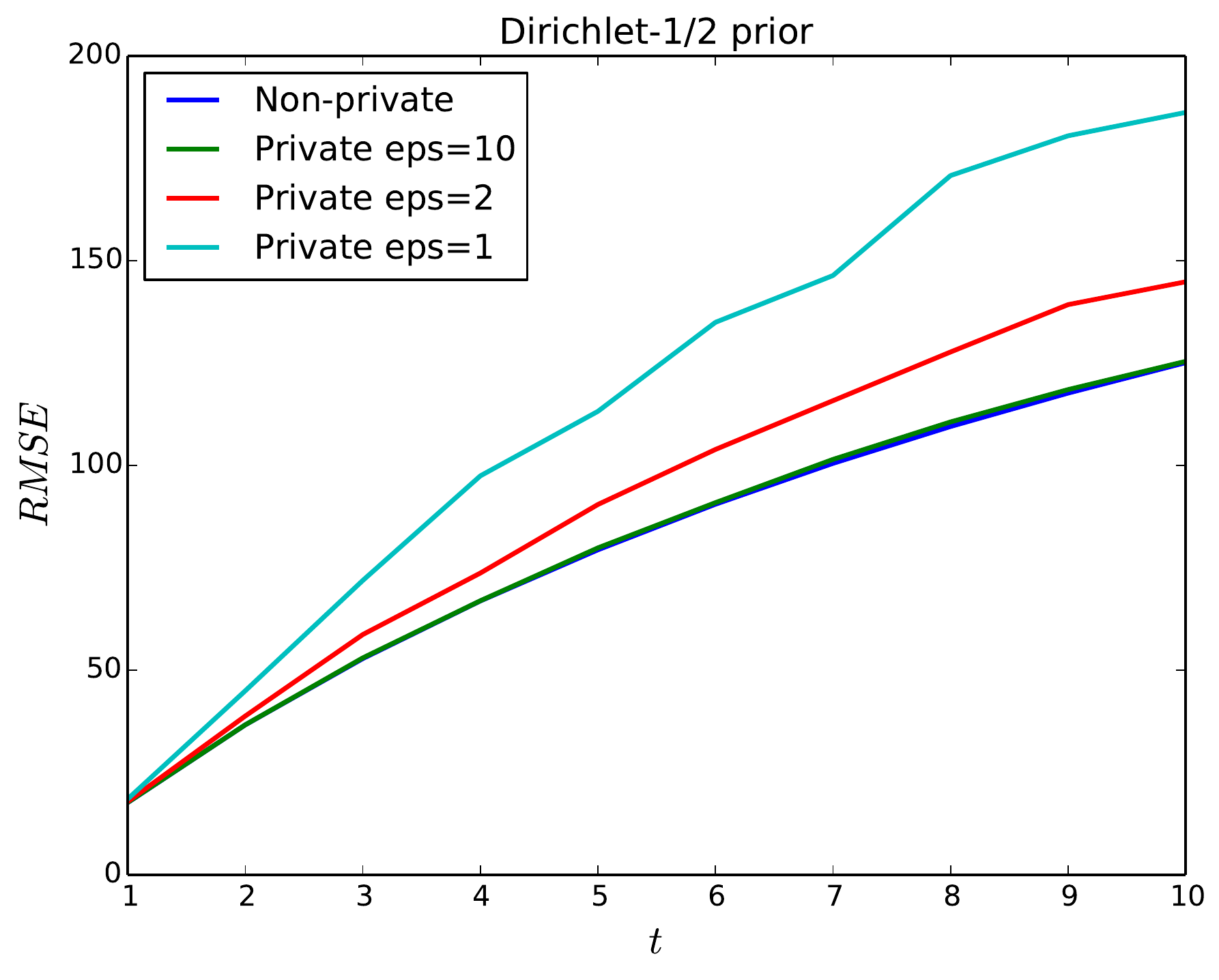}
\end{minipage}
}
\caption{Comparison between the private estimator with the non-private SGT when $k=1000$.} 
\label{fig:coverage-k1000}
\end{figure*}
\begin{figure*}
\centering
\subfigure[]{
\begin{minipage}[b]{0.3\textwidth}
\includegraphics[width=1\textwidth]{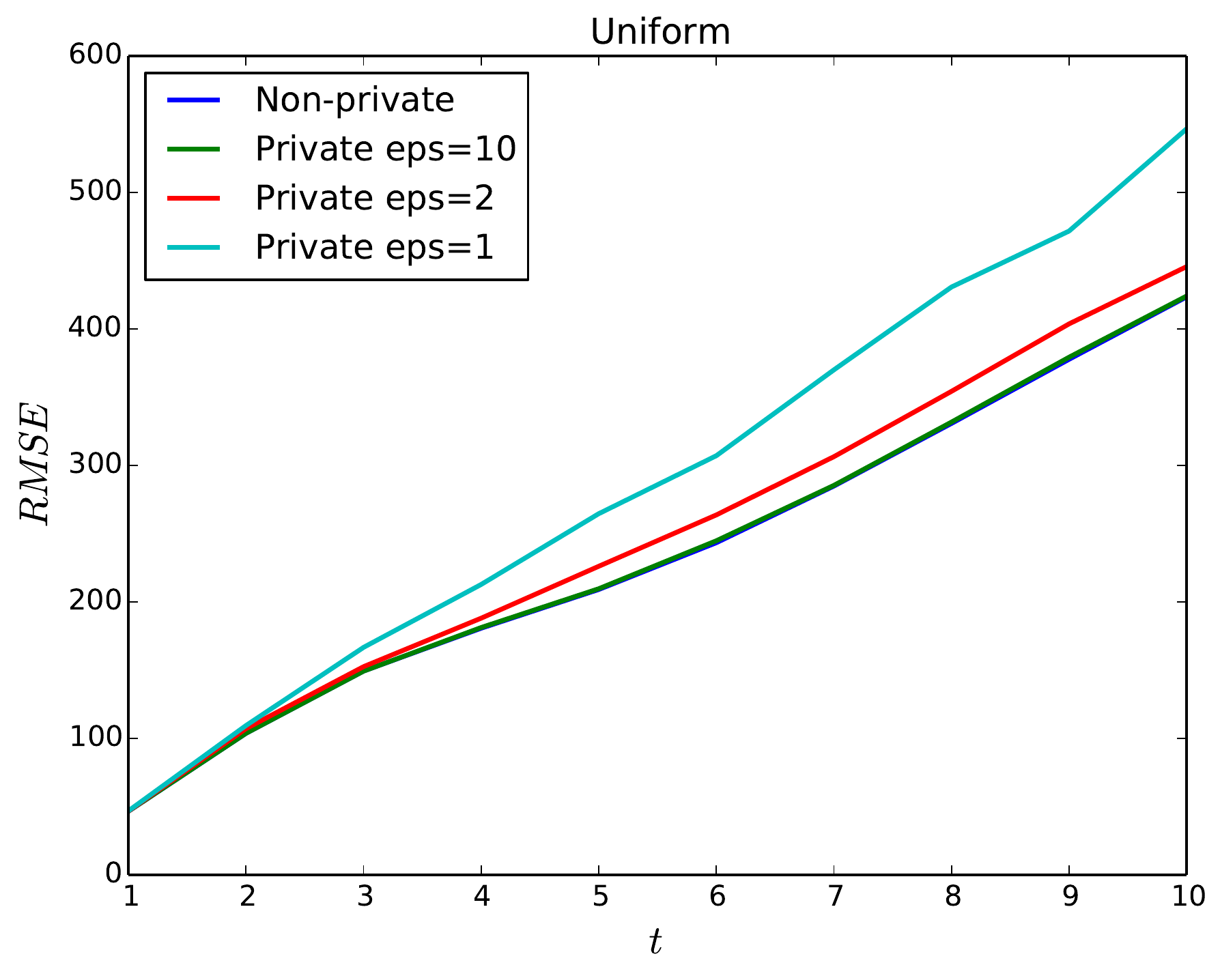}
\includegraphics[width=1\textwidth]{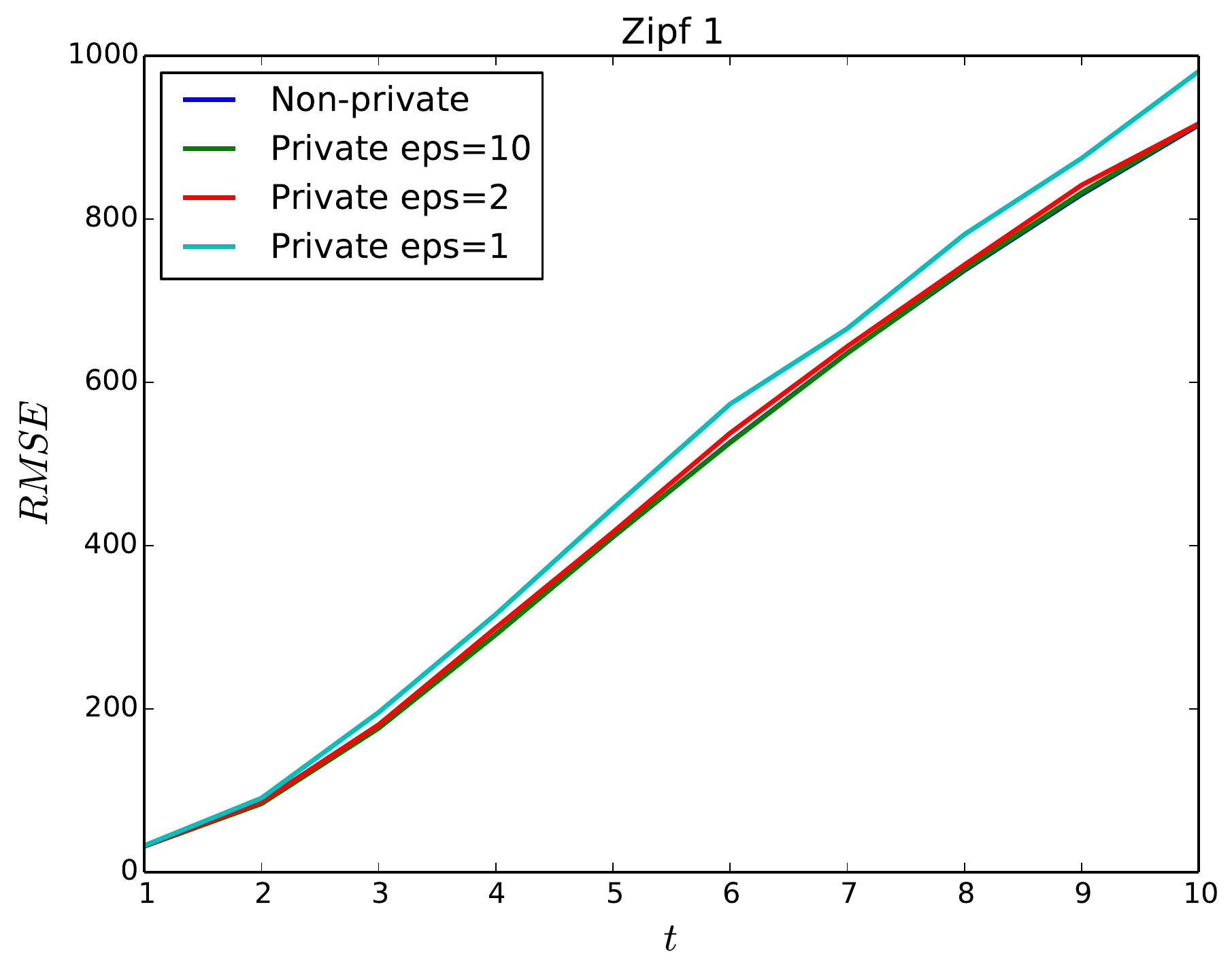}
\end{minipage}
}
\subfigure[]{
\begin{minipage}[b]{0.3\textwidth}
\includegraphics[width=1\textwidth]{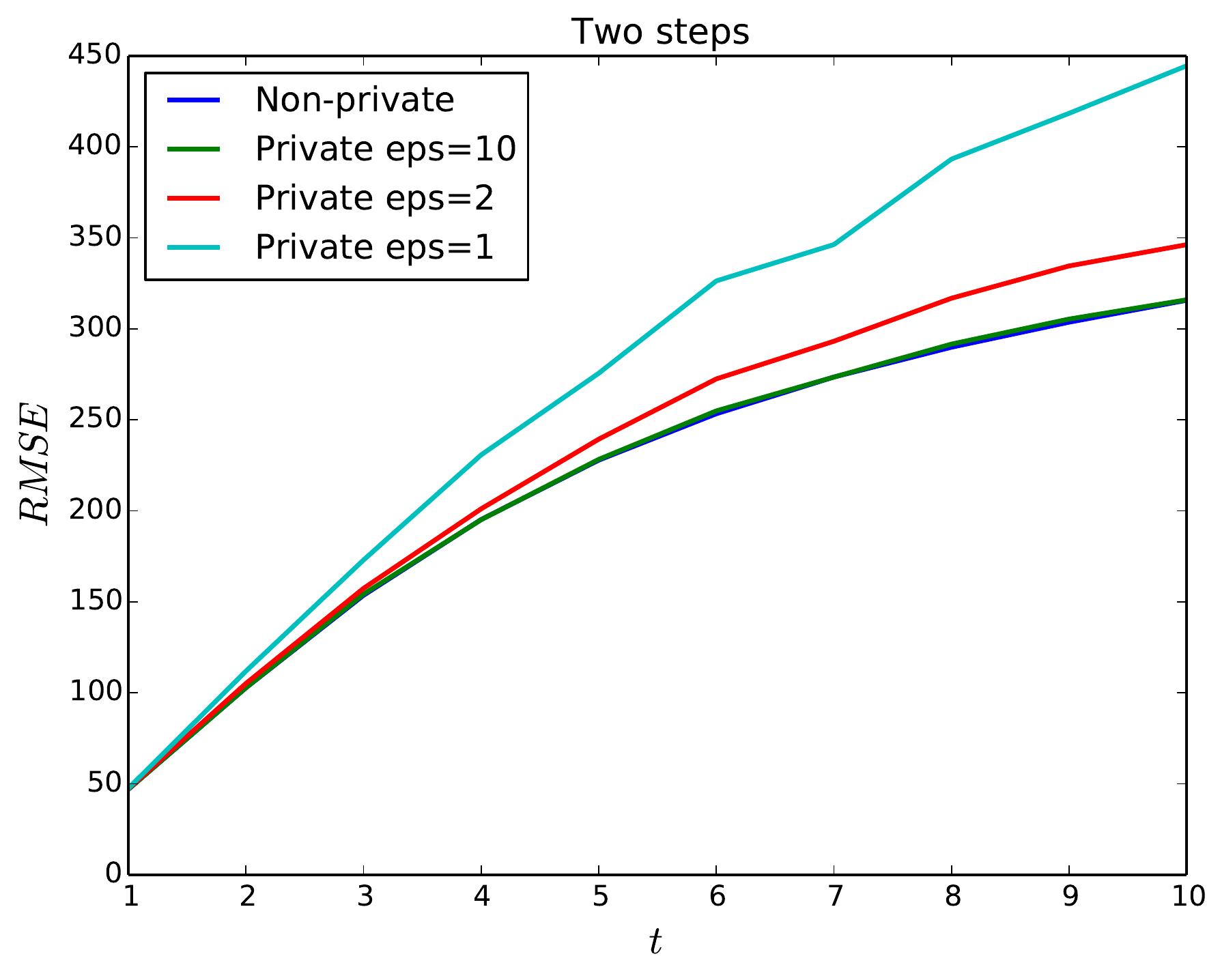}
\includegraphics[width=1\textwidth]{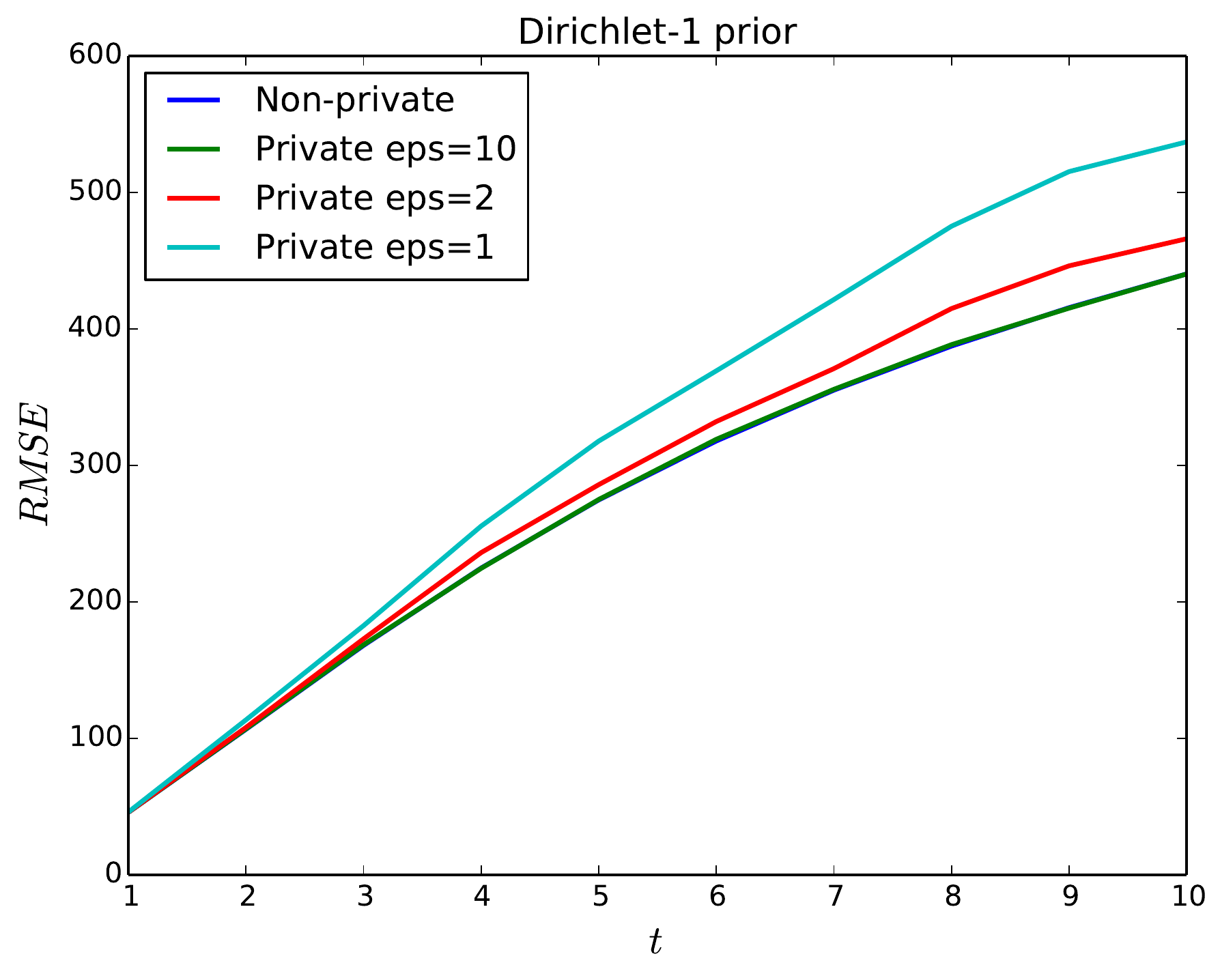}
\end{minipage}
}
\subfigure[]{
\begin{minipage}[b]{0.3\textwidth}
\includegraphics[width=1\textwidth]{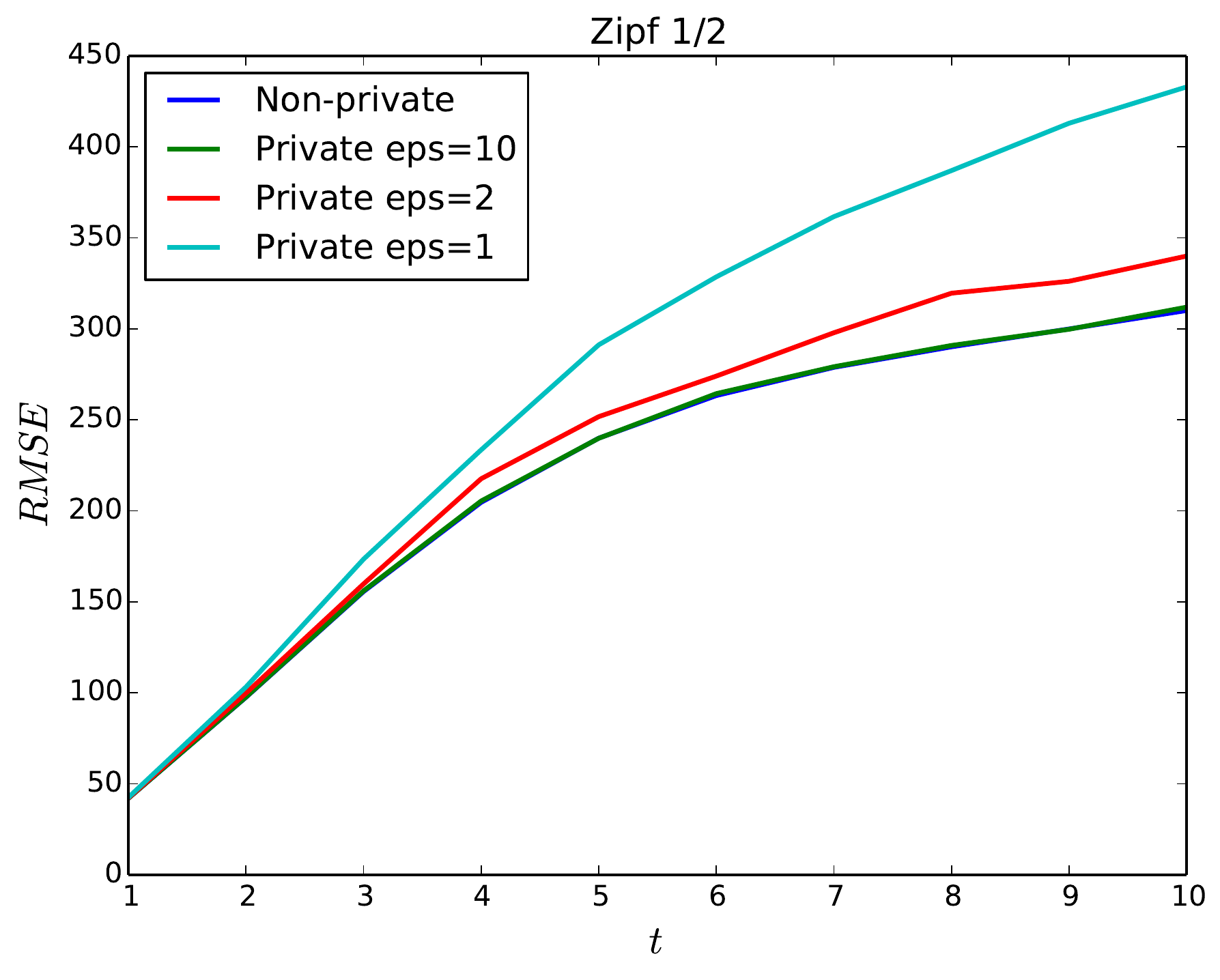}
\includegraphics[width=1\textwidth]{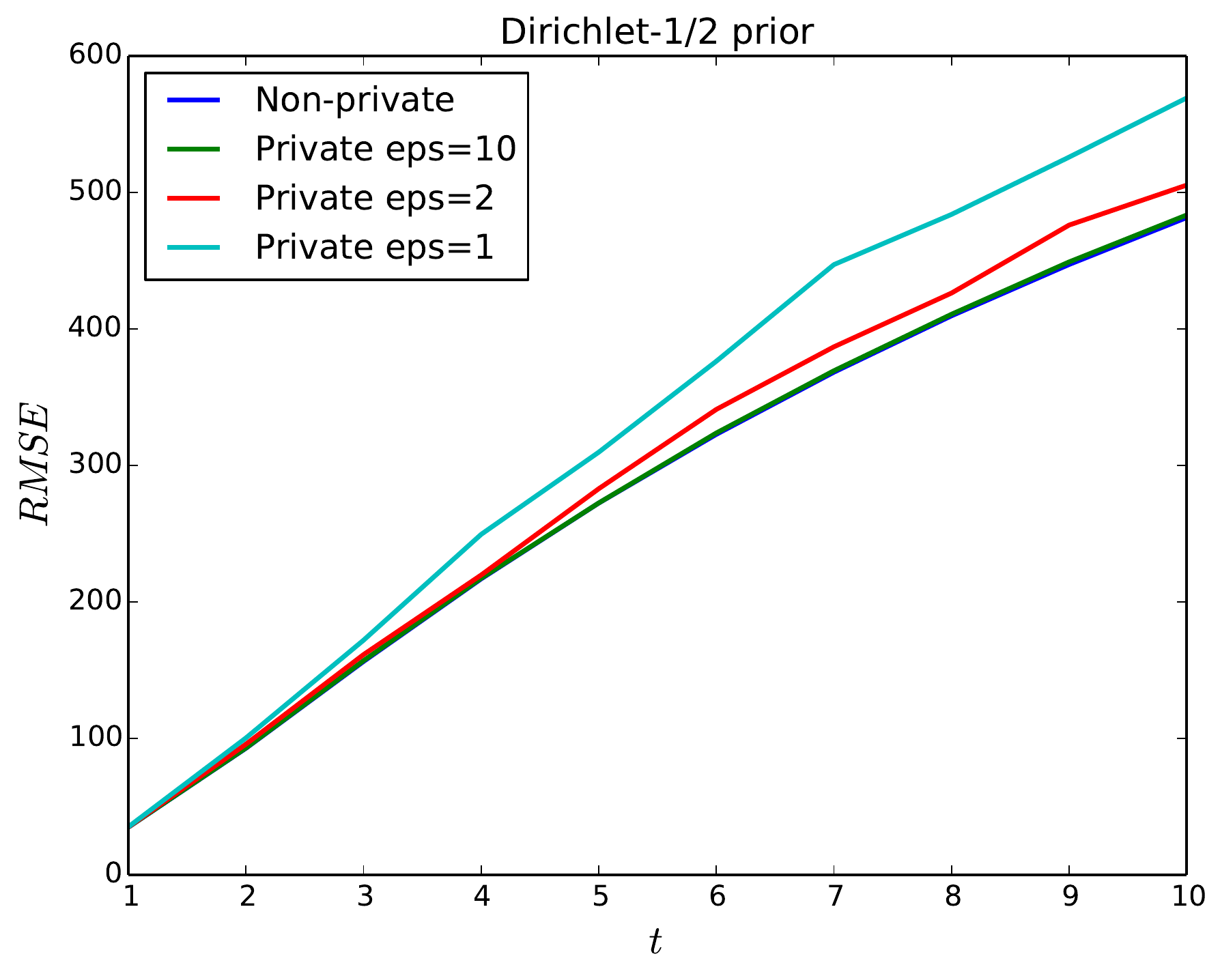}
\end{minipage}
}
\caption{Comparison between the private estimator with the non-private SGT when $k=5000$.} 
\label{fig:coverage-k5000}
\end{figure*}
\begin{figure*}
\centering
\subfigure[]{
\begin{minipage}[b]{0.3\textwidth}
\includegraphics[width=1\textwidth]{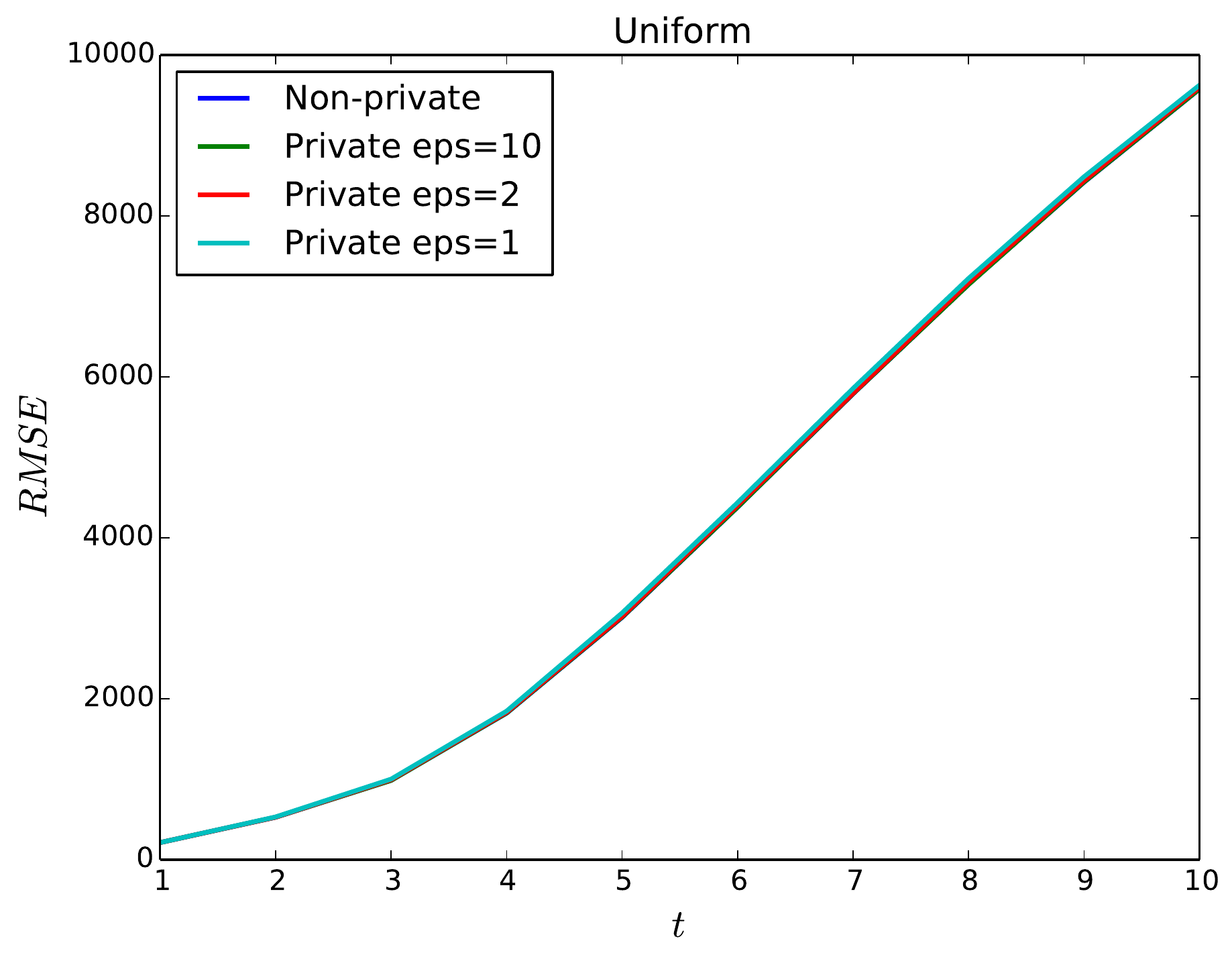}
\includegraphics[width=1\textwidth]{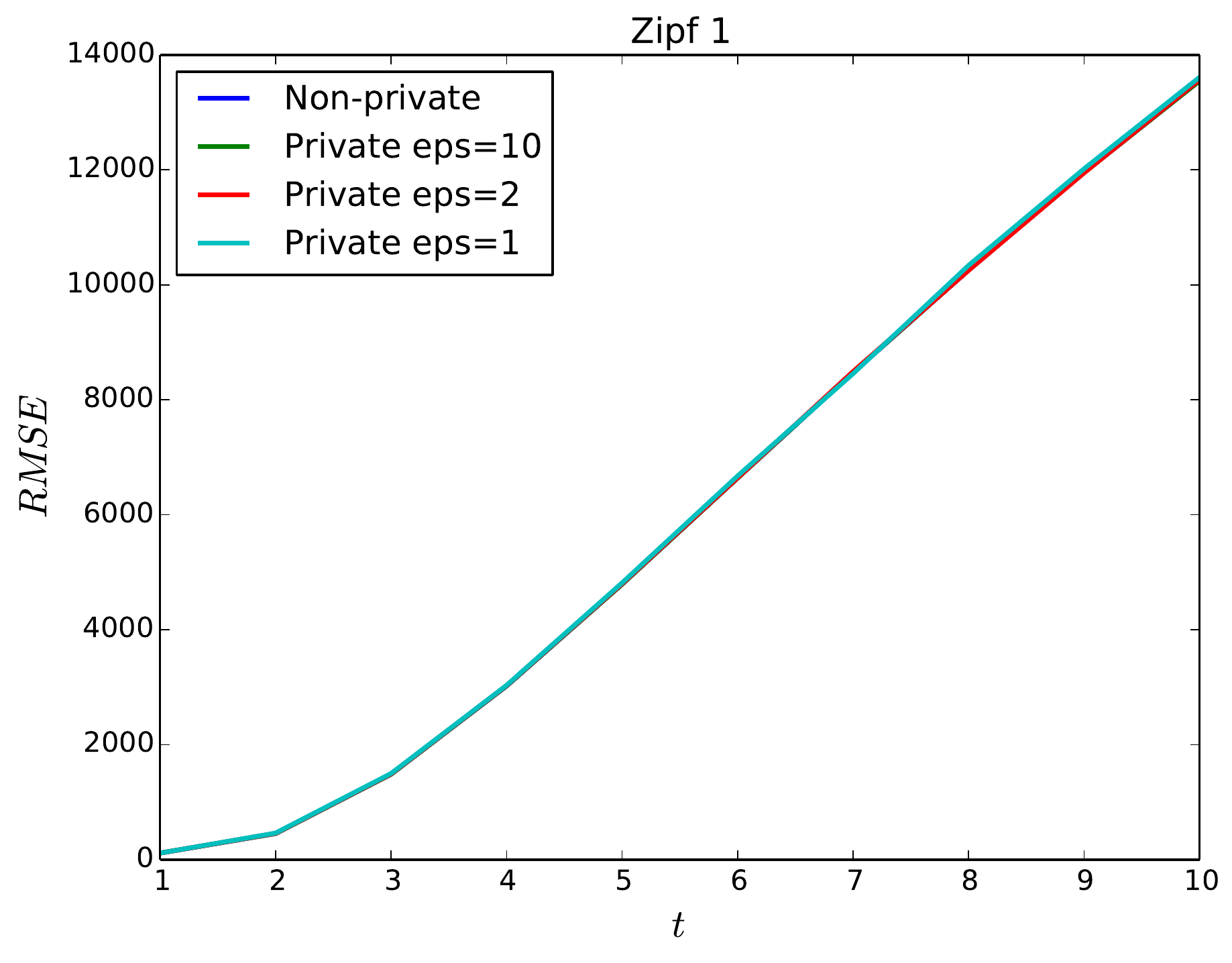}
\end{minipage}
}
\subfigure[]{
\begin{minipage}[b]{0.3\textwidth}
\includegraphics[width=1\textwidth]{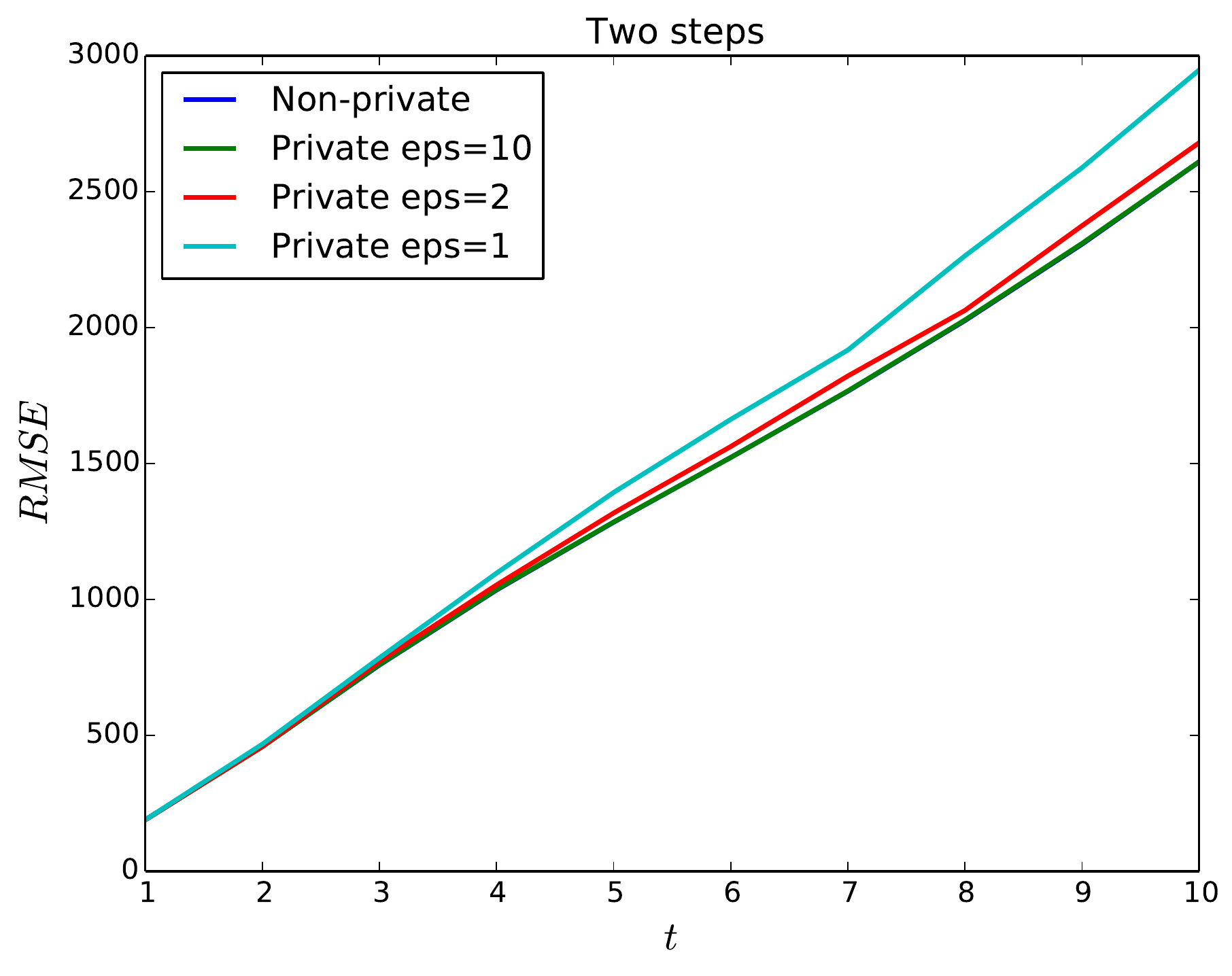}
\includegraphics[width=1\textwidth]{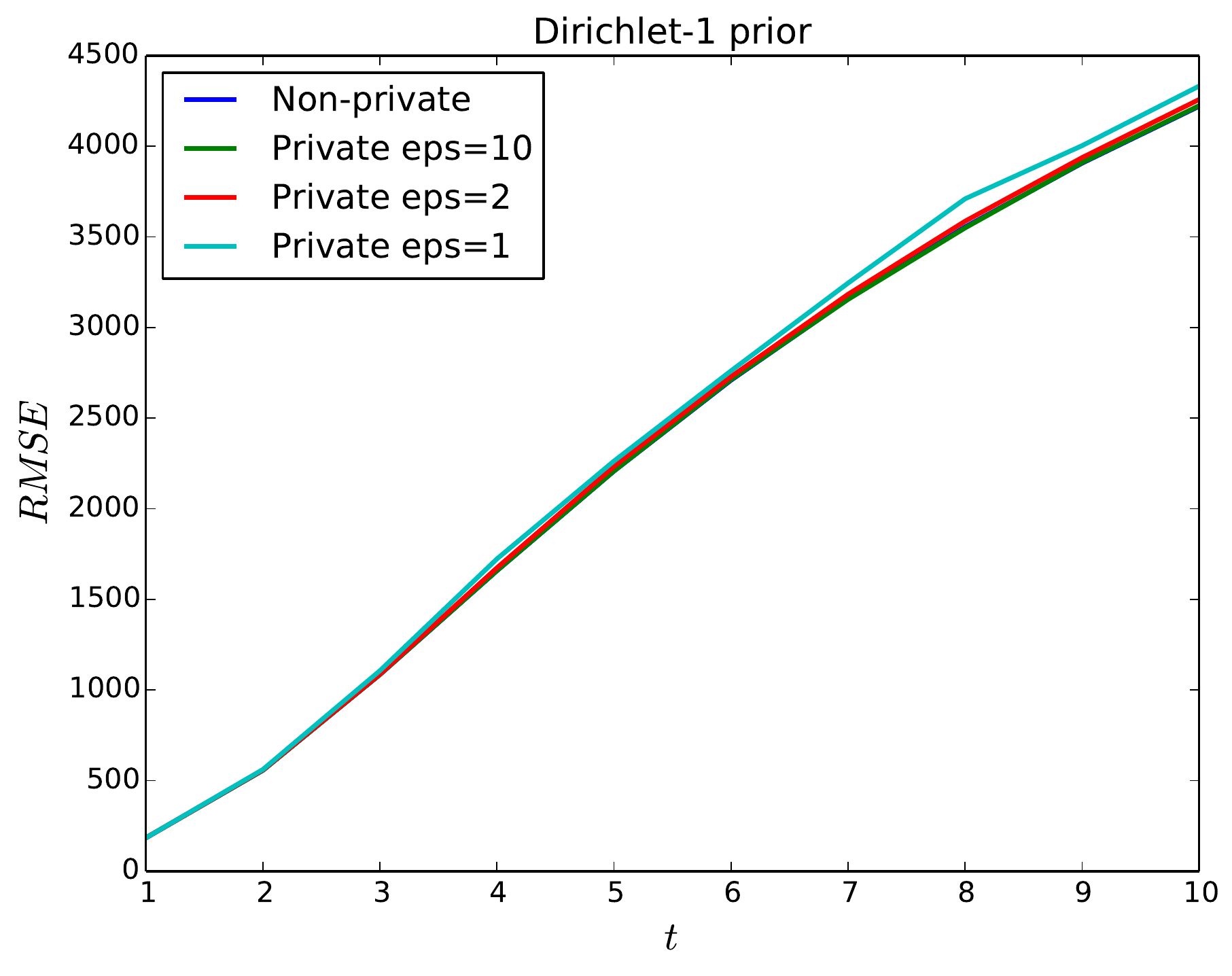}
\end{minipage}
}
\subfigure[]{
\begin{minipage}[b]{0.3\textwidth}
\includegraphics[width=1\textwidth]{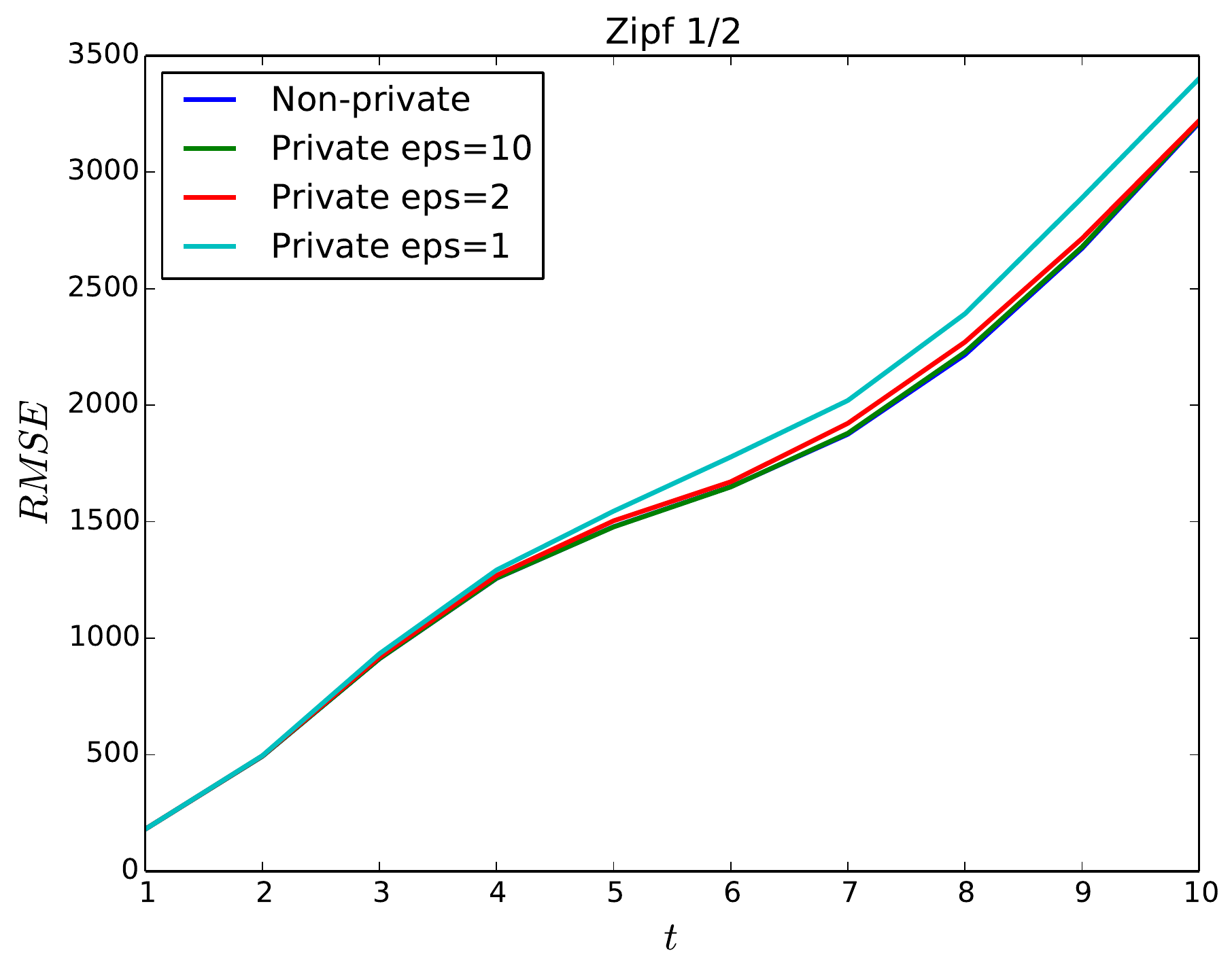}
\includegraphics[width=1\textwidth]{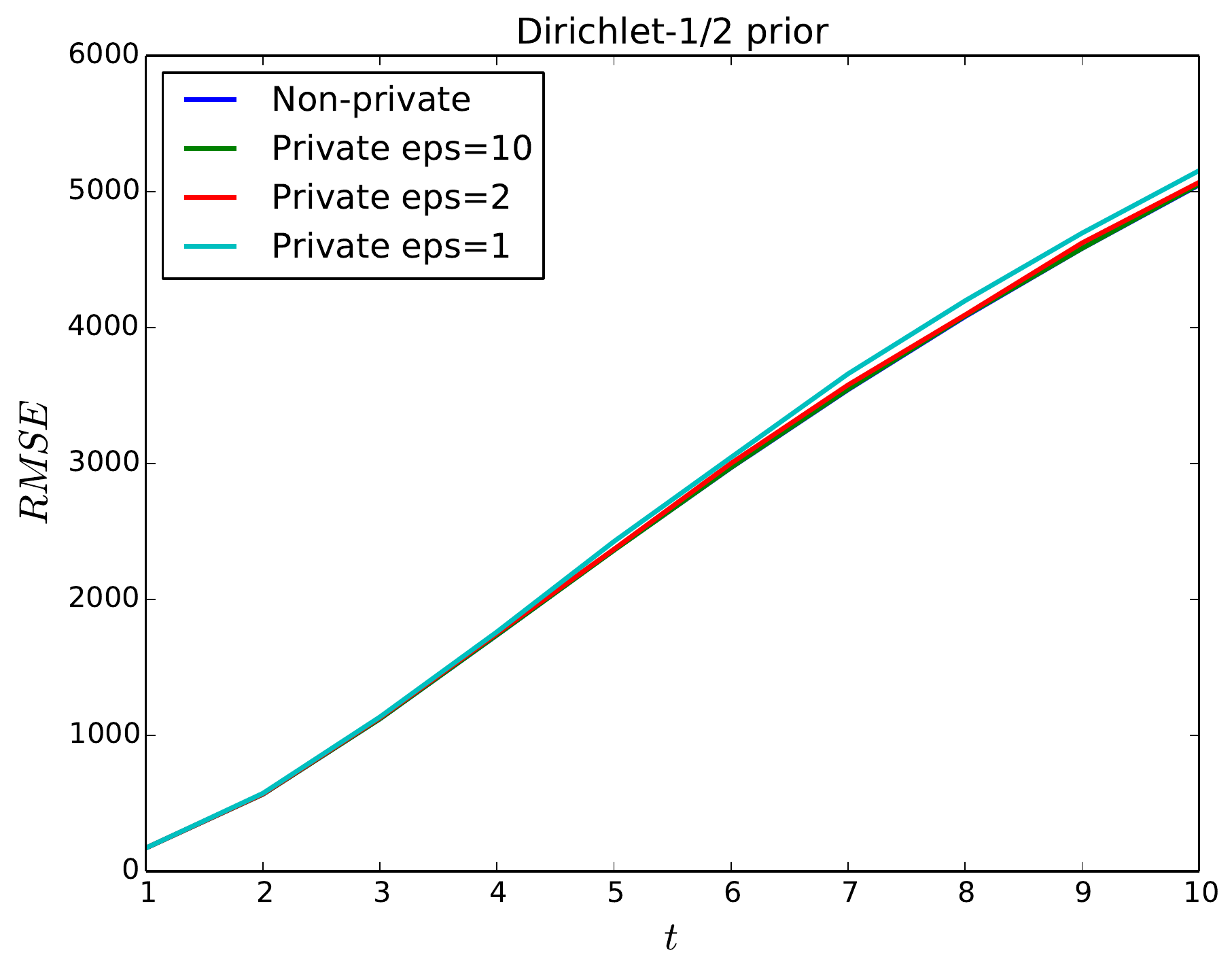}
\end{minipage}
}
\caption{Comparison between the private estimator with the non-private SGT when $k=100000$.} 
\label{fig:coverage-k100000}
\end{figure*}

\chapter{Private Distribution Estimation}
\label{cha:estimation}
 
\section{Introduction}
\label{sec:de_applications}
In this chapter, we explore the problem of private distribution estimation (defined in Section~\ref{sec:preliminary_problem_formulation}), which is one of the most fundamental problem in statistical inference. Given i.i.d. samples, what we want to estimate is the underlying distribution itself. 

In this chapter,
we will apply Theorem~\ref{thm:dp_fano} (private Fano's inequality) and Theorem~\ref{thm:assouad} (private Assouad's lemma) to some classic distribution estimation problems. Our results are summarized in Table~\ref{tab:pure} and Table~\ref{tab:approx}. Before presenting the results, we  firstly introduce the following theorem, which can be used to prove lower bounds on the sample complexity in this distribution estimation framework.  We remark that it can be viewed as a corollary of Theorem~\ref{thm:dp_fano}.


\begin{restatable}[$\eps$-DP distribution estimation]{theorem}{distest}
	\label{coro:fano}
	Given $\eps>0$, let $\cV=\{\q_1, \q_2,...,\q_M\} \subseteq \cQ$ be a set distributions over $\cX$ with size $M$, such that for all $i\ne j$, 	\begin{enumerate}[label=(\alph*)]
		\item $\ell \Paren{\q_i, \q_j} \ge 3 \tau$,
		\item $D_{KL} \Paren{\q_i,\q_j} \le \beta$,
		\item $d_{TV} \Paren{\q_i,\q_j} \le \gamma$,
	\end{enumerate}
	then
\[
		S(\cQ, \ell, \tau, \eps) = \Omega\Paren{\frac{\log M}{\beta}+ \frac{\log M}{\gamma \eps}}.
\]
\end{restatable}
\noindent\textbf{Remark.} With only conditions $(a)$ and $(b)$, we obtain 
the first term of the sample complexity lower bound which is the original 
Fano's bound for sample complexity. \newnewzs{By Pinsker's inequality, 
a bound on the KL divergence (Condition $(c)$) would imply a bound on TV 
distance (Condition $(b)$), i.e., $\gamma \le \sqrt{\beta/2}$. Hence 
Conditions $(a)$ and $(b)$ can also imply a lower bound on the sample 
complexity. We include all three conditions here since it is possible that in 
certain applications $\gamma \ll \sqrt{\beta/2}$, and hence a better bound 
can be obtained.}

\begin{proof}
Recall that $\cQ^{\ns} := \{ \q^\ns | \q \in \cQ \}$ is the set of induced distributions over $\cX^\ns$ and $\q^\ns \in \cQ^\ns, \theta(\q^\ns) = \q$. Then,
	$\forall i \neq j \in [M]$,
$\ell \Paren{\theta(\q_i^n),\theta(\q_j^n)} \ge 3 \tau$, and $D_{KL} \Paren{\q_i^n,\q_j^n} = n D_{KL} \Paren{\q_i,\q_j}  \le n \beta.$

{The following lemma is a corollary of maximal coupling~\cite{Hollander12}, which states that for two distributions there is a coupling of their $n$ fold product distributions with an expected Hamming distance $n$ times their total variation distance.}
	\begin{lemma}
		\label{lem:couplingdistance}
		Given distributions $\q_1,\q_2$ over $\cX$, there exists a coupling $(X,Y)$ between $\q_1^n$ and $\q_2^n$ such that  
		\[\expectation{\ham{X}{Y}} = \ns \cdot d_{TV} \Paren{\q_1,\q_2},
		\]
		where $X \sim \q_1^n$ and $Y \sim \q_2^n.$
	\end{lemma}
\noindent By Lemma~\ref{lem:couplingdistance}, $\forall i,j \in [M]$, there exists a coupling $(X, Y)$ between $\q_i^n$ and $\q_j^n$ such that $\expectation{\ham{X}{Y}} \le n\gamma. $
	Now by Theorem~\ref{thm:dp_fano},
	\begin{align}
	R(\cQ^*, \ell, \eps) \ge \max \Bigg\{\frac{3\tau}{2} \left(1 - \frac{n \beta + \log 2}{\log M}\right), 1.2\tau \min\left\{1, \frac{M}{e^{10\eps n \gamma}}\right\} \Bigg\}.
	\end{align}
Therefore, for $R(\cQ^\ns, \ell, \eps) \le \tau$, 
	\[
	S(\cQ, \ell, \tau, \eps) = \Omega\Paren{\frac{\log M}{\beta}+ \frac{\log M}{\gamma \eps}}.
	\]
\end{proof}

\noindent We now present examples of distribution classes we consider.

\medskip
\noindent\textit{$k$-ary discrete distribution estimation.} {Suppose $\cX=[k]:=\{1, \ldots, \ab\}$, and $\cQ:=\Delta_k$ is the simplex of $k$-ary distributions over $[k]$.  We consider estimation in both total variation and $\ell_2$ distance.} 

\medskip
\noindent\textit{$(k,d)$-product distributions.} Consider $\cX=[k]^d$, and let $\cQ:=\Delta_k^d$ be the set of product distributions over $[k]^d$, where the marginal distribution on each coordinate is over $[k]$ and independent of the other coordinates. We study estimation under total variation distance. A special case of this is Bernoulli product distributions ($k=2$), where each of the $d$ coordinates is an independent Bernoulli random variable.  

\medskip
\noindent\textit{$d$-dimensional Gaussian mixtures.} Suppose $\cX=\RR^{\dims}$, and $\cG_{\dims} := \{ \cN (\mu, I_d): \norm{\mu}\le R\}$ is the set of all Gaussian distributions in $\RR^d$ with bounded mean and identity covariance matrix. The bounded mean assumption is unavoidable, since by~\cite{BunKSW2019}, it is not possible to learn a single Gaussian distribution under pure DP without this assumption. {We consider 
\[
\cQ = \cG_{\ab, \dims} :=\left\{\sum_{j=1}^{\ab} w_j \p_j: \forall j \in [k], w_j \ge 0, \p_j\in\cG_\dims, w_1+\ldots+w_k =1\right\},
\] the collection of mixtures of $k$ distributions from $\cG_\dims$.} 

\medskip

\subsection{Results}
\label{sec:de_results}

\begin{table*}[htb]
\centering
      \begin{tabular}{| c | c | c |}
      \hline
      {\bf Problem} & {\bf Upper Bounds} & {\bf Lower Bounds} \\ \hline
      {\bf $\ab$-ary} & \multicolumn{2}{c|} {\parbox[c][1.2cm]{10cm}{\centering$ \Theta \Paren{\frac{\ab}{\alpha^2}+\frac{\ab}{\alpha\eps}}$~(\cite{DiakonikolasHS15}, Theorem~\ref{thm:pure-dv})} }\\\hline
 
       {\bf $\ab$-ary, $\ell_2$ distance} &\parbox[c][1.5cm]{5cm}{\centering $O\Paren{\frac{1}{\dist^2}+\min\Paren{\frac{\sqrt{\ab}}{\dist\eps},\frac{\log\ab}{\dist^2\eps}  } }$ \\\small{(Theorem~\ref{thm:pure-l2}) }}
       & \parbox[c]{5cm}{\centering $\Omega \Paren{\frac{1}{\dist^2}+\min\Paren{\frac{\sqrt{\ab}}{\dist\eps},\frac{\log (\ab \dist^2)}{\dist^2\eps}  } }$ \\\small{(Theorem~\ref{thm:pure-l2}) } } \\ \hline
       
       {\bf product distribution} & \parbox[c][1.5cm]{5cm}{\centering $O\Paren{\ab\dims \log \Paren{ \frac{\ab\dims}{\dist}} \Paren{\frac1{\dist^2}+\frac1{\dist\eps}}}$ \\\small{\cite{BunKSW2019} }}
       &\parbox[c]{5cm}{\centering $\Omega \Paren{{\ab\dims}\Paren{\frac1{\alpha^2}+\frac1{\alpha \eps}}}$ \\\small{(Theorem~\ref{thm:main_product}) }}\\\hline
       
              {\bf Gaussian mixtures} & \parbox[c][1.5cm]{5cm}{\centering ${O}\Paren{kd\log(\frac{dR}{\alpha})(\frac1{\alpha^2}+ \frac1{\alpha \eps})}$ \\\small{\cite{BunKSW2019} }}
       &\parbox[c]{5cm}{\centering $\Omega \Paren{{\ab\dims}\Paren{\frac1{\alpha^2}+\frac1{\alpha \eps}}}$ \\\small{(Theorem~\ref{thm:main_Gaussian}) }}\\\hline
       
      \end{tabular}
    \caption{\label{tab:pure} Summary of the sample complexity bounds for $\eps$-DP discrete distribution estimation. Unless mentioned, the bounds are all for estimation under total variation distance.}
\end{table*}

%

\begin{table*}[htb]
\centering
      \begin{tabular}{| c | c | c |}
      \hline
      {\bf Problem} & {\bf Upper Bounds} & {\bf Lower Bounds} \\ \hline
      {\bf $\ab$-ary} &  \parbox[c][1.5cm]{5cm}{\centering $O\Paren{\frac{\ab}{\dist^2}+\frac{\ab} {\dist\eps}}$ \\\small{(\cite{DiakonikolasHS15}, Theorem~\ref{thm:ApproximateDiscreteTotalVariation})} }
       &\parbox[c]{5cm}{\centering $\Omega \Paren{\frac{\ab}{\alpha^2}+\frac{\ab}{\alpha (\eps+\delta)}}$ \\\small{(Theorem~\ref{thm:ApproximateDiscreteTotalVariation})} }\\\hline
 
       {\bf $\ab$-ary, $\ell_2$ distance} &\parbox[c][1.5cm]{5cm}{\centering 
       $O 
       \Paren{\frac{1}{\dist^2}+\min\Paren{\frac{\sqrt{\ab}}{\dist\eps},\frac{\log\ab}{\dist^2\eps}
         } }$ \\\small{(Theorem~\ref{thm:ApproximateDiscreteL2}) }}
       & \parbox[c]{5cm}{\centering $\Omega \Paren{\frac{1}{\dist^2}+\min\Paren{\frac{\sqrt{\ab}}{\dist(\eps+\delta)},\frac{1}{\dist^2(\eps+\delta)}  } }$ \\\small{(Theorem~\ref{thm:ApproximateDiscreteL2}) } } \\ \hline
       
              {\makecell { \bf product distribution \\ ($k = 2$)}} & \parbox[c][1.5cm]{5cm}{\centering $O\Paren{\dims \log \Paren{ \frac{\dims}{\dist}} \Paren{\frac1{\dist^2}+\frac1{\dist\eps}}}$ \\\small{\cite{KamathLSU18, BunKSW2019}}}
       &\parbox[c]{5cm}{\centering $\Omega \Paren{\frac{\dims}{\alpha^2}+\frac{\dims}{\alpha (\eps + \delta)}}$ \\\small{(Theorem~\ref{thm:ApproximateProductTotalVariation},~\cite{KamathLSU18}) }}\\ \hline
      \end{tabular}
    \caption{\label{tab:approx} Summary of the sample complexity bounds for $(\eps,\delta)$-DP discrete distribution estimation. Unless mentioned, the bounds are all for estimation under total variation distance.}
\end{table*}

\medskip

\noindent\textbf{Applications of Theorem~\ref{thm:dp_fano}.}
We apply Corollary~\ref{coro:fano} and obtain sample complexity lower bounds for the tasks mentioned above under pure  differential privacy.

\medskip
{\noindent\textit{$\ab$-ary distribution estimation.} Without privacy constraints, the sample complexity of $\ab$-ary discrete distributions under total variation, and $\ell_2$ distance is $\Theta(k/\alpha^2)$ and $\Theta(1/\alpha^2)$ respectively, achieved by the empirical estimator. Under $\eps$-DP constraint, an upper bound of $O \Paren{\ab/\alpha^2+\ab/\alpha\eps}$
samples for total variation distance is known using Laplace 
mechanism~\cite{DworkMNS06} (e.g.~\cite{DiakonikolasHS15}). In 
Theorem~\ref{thm:pure-dv}, we establish the sample complexity of this 
problem by providing a lower bound that matches this upper 
bound. \newnewzs{The bound shows that when $\eps \ll \alpha$, the cost 
due 
to privacy dominates the statistical error and when $\eps \ge \alpha$, the 
privacy cost is almost negligible. The same break point (up to logarithmic 
factors) has also been observed for product distributions and mixtures of 
Gaussian distributions, as listed below.}} 

Under $\ell_2$ distance, in Theorem~\ref{thm:pure-l2} we design estimators and establish their optimality whenever $\dist< \ab^{-1/2}$ or $\dist \ge \ab^{-0.499}$, which contains almost all the parameter range. Note that under $\ell_2$ distance, estimation without privacy has sample complexity independent of $k$, whereas {an unavoidable} logarithmic dependence on $\ab$ is introduced due to privacy requirements. The results are presented in Section~\ref{sec:pdp_kary}. 
 
\medskip
\noindent\textit{$(\ab,\dims)$-product distribution estimation.} For $(\ab,\dims)$-product distribution estimation under $\eps$-DP,~\cite{BunKSW2019} proposed an algorithm that uses
$O\Paren{\ab\dims \log \Paren{ \ab\dims/\dist} \Paren{1/\dist^2+1/\dist\eps}}$ samples. 
In this chapter, we present a lower bound of $\Omega\Paren{{\ab\dims}/{\dist^2}+{\ab\dims}/{\dist\eps}}$, which matches their upper bound up to logarithmic factors. For Bernoulli product distributions, \cite{KamathLSU18} proved a lower bound of $\Omega\Paren{{d}/{\alpha \eps}}$ under $(\eps,{3}/{64\ns})$-DP, which is naturally a lower bound for pure DP. 
The details are presented in Section~\ref{sec:product}.

\medskip
\noindent\textit{Estimating Gaussian mixtures.}{~\cite{BunKSW2019} provided an upper bound of $\widetilde{O}\Paren{{\ab\dims}/{\dist^2}+{\ab\dims}/{\dist\eps}}$ samples.} Without privacy, a tight bounds of $\Omega({\ab\dims}/{\alpha^2})$ was shown in~\cite{SureshOAJ14, DaskalakisK14, AshtianiBHLMP18}. {In this chapter, we prove a lower bound of $\Omega\Paren{{\ab\dims}/{\dist^2}+{\ab\dims}/{\dist\eps}}$, which matches the upper bound up to logarithmic factors.} \newzs{For the special case of estimating a single Gaussian ($k=1$), a lower bound of $\ns = {\Omega}\Paren{{\dims}/{(\alpha\eps\log \dims)}}$ was given in~\cite{KamathLSU18} for $(\eps,{3}/{64\ns})$-DP, which implies a lower bound that is $\log d$ factor weaker than our result under pure DP.}


\medskip
\noindent\textbf{Applications of Theorem~\ref{thm:assouad}.} As remarked earlier, Theorem~\ref{thm:dp_fano} only works for pure DP (or approximate DP with very small $\delta$). 
Assouad's lemma can be used to obtain lower bounds for distribution estimation under $(\eps,\delta)$-DP. For $k$-ary distribution estimation under $TV$ distance, we get a lower bound of $\Omega\Paren{{k}/{\alpha^2} + {k}/{\alpha(\eps + \delta)}}$. This shows that even up to $\delta=O(\eps)$, the sample complexity for $(\eps,\delta)$-DP is the same as that under $\eps$-DP.

\newzs{For Bernoulli ($k = 2$) product distributions, \cite{KamathLSU18} provides an efficient $(\eps, \delta)$-DP algorithm that achieves an upper bound of $O\Paren{\dims \log \Paren{ \dims/\dist} \Paren{1/\dist^2+1/\dist\eps}}$.\footnote{\newzs{The algorithm in~\cite{BunKSW2019} works for $\eps$-DP and general $k$ but it is not computationally efficient.}}} The lower bound $\Omega ({d}/{\alpha^2} + {d}/{\alpha \eps})$ obtained in~\cite{KamathLSU18} by fingerprinting holds for small values of $\delta = O(1/\ns)$. Note by the definition of DP, if $\delta>1/\ns$, a DP algorithm can blatantly disregard the privacy of $\delta\ns$ users. Therefore in most of the literature, $\delta$ is assumed to be $O(1/\ns)$. We want to make a complimentary remark that we
can obtain the same lower bound all the way up to $\delta = O(\eps)$. This shows that there is no gain even if we compromise the privacy of a $\delta$ fraction of users. Therefore, there is no incentive to do it. We describe the details about these applications in Section~\ref{sec:adp_applications}. 
\subsection{Related and Prior Work}
\label{sec:de_related}

Protecting privacy generally comes at the cost of performance degradation. Previous literature has studied various problems and established utility privacy trade-off bounds, including distribution estimation, hypothesis testing, property estimation, empirical risk minimization, etc~\cite{ChaudhuriMS11, Lei11, BassilyST14, DiakonikolasHS15, CaiDK17, AcharyaSZ18, KamathLSU18, AliakbarpourDR18, AcharyaKSZ18}. 

There has been significant recent interest in differentially private distribution estimation.~\cite{DiakonikolasHS15} gives upper bounds for privately learning $\ab$-ary distributions under total variation distance.~\cite{KamathLSU18, BunKSW2019, KarwaV18} focus on high-dimensional distributions, including product distributions and Gaussian distributions. As discussed in the previous section, our proposed lower bounds improve upon their lower bounds in various settings. \cite{BunNSV15} studies the problem of privately estimating a distribution in Kolmogorov distance, which is weaker than total variation distance. Upper and lower bounds for private estimation of the mean of product distributions in $\ell_\infty$ distance, heavy tailed distributions, and Markov Random fields are studied in~\cite{BlumDMN05, DworkMNS06, SteinkeU17a, BunUV18, KamathSU20, ZhangKKW20}.

Several estimation tasks including distribution estimation and hypothesis testing have also been considered under the distributed notion of local differential privacy, e.g.,~\cite{Warner65, KasiviswanathanLNRS11, ErlingssonPK14, DuchiJW13, KairouzBR16, WangHWNXYLQ16, Sheffet17, YeB18, GaboardiR17, AcharyaSZ18a, AcharyaS19, AcharyaCFT18}.

\section{$\eps$-DP Distribution Estimation} \label{sec:pdp_estimation}
In this section, we use Corollary~\ref{coro:fano} to prove sample complexity lower bounds for various $\eps$-DP distribution estimation problems. The general idea is to construct a subset of distributions in $\cQ$ such that they are close in both $TV$ distance and $KL$ divergence while being separated {in the loss function $\ell$}. The larger the subsets we construct, {the better} the lower bounds we can get. In Section~\ref{sec:pdp_kary}, we {derive} sample complexity lower bounds for $k$-ary distribution estimation under both \emph{TV} and $\ell_2$ distance that are tight up to constant factors. Tight sample complexity lower bounds up to logarithmic factors for $(k,d)$-product distributions and $d$-dimensional Gaussian mixtures are derived in Section~\ref{sec:product} and~\ref{sec:gaussian} respectively. 

{Corollary~\ref{coro:fano} requires a \emph{packing} of distributions with pairwise distance at least $3\tau$ apart in $\ell$. A standard method to construct such distributions is using results from coding theory.}

\newzs{We start with some definitions. An $h$-ary code of length $k$ is a set $\cC\subseteq \{0,1, \ldots, h-1\}^k$, and each $c\in\cC$ is a \emph{codeword}. The \emph{minimum distance} of a code $\cC$ is the smallest Hamming distance between two codewords in $\cC$. The code is called binary when $h = 2$. The weight of a binary codeword $c\in\cC$ is $wt(c)= |\{i:c_i=1\}|$, the number of 1's in $c$. A binary code $\cC$ is a \emph{constant weight code} if each $c\in\cC$ has the same weight.  We now present some useful variants of the classic Giblert Varshamov bounds on the existence of codes with certain properties. We prove these in Section~\ref{sec:codes}.
}
{
\begin{lemma}
\label{lem:GV}
Let $l$ be an integer at most $k/2$ and at least $20$. There exists a constant weight binary code $\cC$ which has code length $k$, weight $l$, minimum distance $l/4$ with $|\cC|\ge \Paren{\frac{k}{2^{7/8}l}}^{7l/8}$.
\end{lemma}} 

\begin{lemma}
\label{lem:constantGV2}
There exists an $h$-ary code $\cH$ with code length $\dims$ and minimum Hamming distance $\frac{\dims}{2}$, which satisfies that $\absv{\cH} \ge (\frac{h}{16})^{\frac{\dims}{2}}$.
\end{lemma}}

%
%

\subsection{$\ab$-ary Distribution Estimation} \label{sec:pdp_kary}
We establish the sample complexity of $\eps$-DP $\ab$-ary distribution estimation under $TV$ and $\ell_2$ distance. 
\begin{theorem}
\label{thm:pure-dv}
The sample complexity of $\eps$-DP $\ab$-ary distribution estimation under $TV$ distance  is
\begin{align}
S_{\tt DE}(\Delta_{\ab}, d_{TV}, \alpha, \eps)=\Theta \Paren{\frac{\ab}{\alpha^2}+\frac{\ab}{\alpha\eps}}.
\end{align}
\end{theorem}

\begin{theorem}
\label{thm:pure-l2}
The sample complexity of $\eps$-DP $\ab$-ary distribution estimation under $\ell_2$ distance is
\begin{align}
S_{\tt DE}(\Delta_{\ab}, \ell_2, \alpha, \eps) =\Theta \Paren{\frac{1}{\dist^2}+\frac{\sqrt{\ab}}{\dist\eps}},\ \ \ \text{for $\dist < \frac1{\sqrt{\ab}}$, and}
\end{align}
\begin{align}
\Omega\Paren{\frac{1}{\dist^2}+\frac{\log(\ab\dist^2)}{\dist^2\eps}}\le S(\Delta_{\ab}, \ell_2, \alpha, \eps) \le O\Paren{\frac{1}{\dist^2}+\frac{\log\ab}{\dist^2\eps}} \ \ \ \text{for $\dist > \frac1{\sqrt{\ab}}$.} 
\end{align}
\end{theorem}

For $\ell_2$ loss, our bounds are tight within constant factors when $\dist < \frac1{\sqrt{\ab}}$ or $\dist > \ab^{-(\frac12- 0.001)}$.

\subsubsection {Total variation distance}
\label{sec:puredp-tv}

In this section, we derive the sample complexity of $\eps$-DP $k$-ary distribution estimation under $TV$ distance, which is stated in Theorem~\ref{thm:pure-dv}.

\medskip

\noindent \textbf{Upper bound:}
\cite{DiakonikolasHS15} provides an upper bound based on Laplace mechanism~\cite{DworkMNS06}. {We state the algorithm and a proof for completeness and we will use it for estimation under $\ell_2$ distance.}

Given a $\Xon$ from an unknown distribution $\p$ over $[\ab]$. Let $M_x(\Xon)$ be the number of appearances of $x$ in $\Xon$. {Let $p^{\text{erm}}$ be the empirical estimator where $ p^{\text{erm}}(x) :=\frac{M_x(\Xon)}{\ns}$. We note that changing one $X_i$ in $X^n$ can change at most two coordinates of $p^{\text{erm}}$, each by at most $\frac1n$, and thus changing one $X_i$ changes the $p^{\text{erm}}$ by at most $2/n$ in $\ell_1$ distance. Therefore, by~\cite{DworkMNS06}, adding a Laplace noise of parameter $2/\ns\eps$ to each coordinate of $p^{\text{erm}}$ makes it $\eps$-DP. For $x\in[k]$, let
\[ 
h(x) = p^{\text{erm}}(x)+\text{Lap}\Paren{\frac{2}{n \eps}}, \]
where $Lap(\beta)$ is a Laplace random variable with parameter $\beta$.}
The final output $\hat{p}$ is the projection of $h$ on the simplex $\Delta_k$ {in $\ell_2$ distance}.
The expected $\ell_2$ loss between $h$ and $p$ can be upper bounded by
\begin{align}
\Paren{\expectation{\norm{h-p}}}^2 &\le \expectation{\norms{h-p}} \le \expectation{\norms{p^{\text{erm}}-p}}+\expectation{\norms{h - p^{\text{erm}}}}, \nonumber
\end{align}
where the first inequality comes from the Jensen's inequality and the second inequality comes from the triangle inequality.

The first term $\expectation{\norms{p^{\text{erm}} -p}}$ is upper bounded by $\frac{1}{n}$ by an elementary analysis of the empirical estimator. For the second term, note that $\expectation{\norms{h - p^{\text{erm}}}} = \sum_{i=1}^\ab \expectation{Z_i^2}$, where $\forall i, Z_i \sim \text{Lap}\Paren{\frac{2}{\ns\eps}}$. 
{By the variance of Laplace distribution}, we have $\expectation{\norms{p^{\text{erm}} -h}} = O\Paren{\frac{k}{n^2 \eps^2}}$.  Therefore $\expectation{\norm{h-p}}\le O\Paren{\frac1{\sqrt{\ns}}+\frac{\sqrt{k}}{\ns\eps}}$.

Note that since $\Delta_k$ is convex, $\norm{\hat{p}-p} \le\norm{h-p}$. Finally, by Cauchy-Schwarz Inequality, $\expectation{\absvn{\hat{p}- p}} \le \sqrt{\ab}\cdot \expectation{\norm{\hat{p}-\p}}  \le \sqrt{\ab}\cdot\expectation{\norm{h-\p}} = O\Paren{\sqrt{\frac{k}{n}}+ {\frac{k}{n \eps}}} $. Therefore $\expectation{\absvn{\hat{p}- p}} \le \alpha$ when $n= O \Paren{\frac{k}{\alpha^2}+\frac{k}{\alpha\eps}}$.

\medskip

\noindent \textbf{Lower bound.}
{We will construct a large set of distributions such that the conditions of Corollary~\ref{coro:fano} hold.} 
\ignore{Our lower bound is through the following two steps. The first step is to construct a set of distributions $\cV_\ab \subset \Delta_k$, such that each pair of distributions in $\cV_\ab$ is $\Omega(\dist)$-separated in total variation distance and with $|\cV_\ab|$ as large as possible. 
In the second step, we bound the $KL$ divergence and $TV$ distance between each pair of distributions, and apply Corollary~\ref{coro:fano} to get the sample complexity lower bound.}
{Suppose $\alpha < 1/48$. Applying Lemma~\ref{lem:GV} with $l={\ab}/{2}$, there exists a constant weight binary code $\cC$ of weight $k/2$, and minimum distance $k/8$, and $|\cC|>2^{7k/128}$. For each codeword $c\in\cC$, a distribution $\p_c$ over $[k]$ is defined as follows:
\[
\p_c (i)  = \left\{
\begin{array}{rcl}
\frac{1+24\dist}{\ab} ,       &      & \text{if ~ }c_i=1,\\
\frac{1-24\dist}{\ab},       &      & \text{if ~ }c_i=0.
\end{array} \right.
\]}

{We choose $\cV = \{ \p_c : c \in \cC \}$ to apply Corollary~\ref{coro:fano}. By the minimum distance property, any two distributions in $\cV$ have a total variation distance of at least $24\alpha/k\cdot k/8=3\alpha$, and at most $24\alpha$. Furthermore, by using $\log (1+x) \le x$, we can bound the KL divergence between distributions by their $\chi^2$ distance, 
\[d_{KL} (\p, \q) \le  \chi^2 (\p,\q)  = \sum_{x=1}^k \frac{\Paren{\p(x) -\q(x)}^2 }{\q(x)} < {10000\alpha^2}.
\]
Setting $\tau=\alpha$, $\gamma=24\alpha$, and $\beta=10000\alpha^2$, and using $\log M > 7\ab/64$ in Corollary~\ref{coro:fano}, we obtain $S(\Delta_{\ab}, d_{TV}, \alpha, \eps)= \Omega \Paren{\frac{k}{\alpha^2}+\frac{k}{\alpha\eps}}.$}
\subsubsection{$\ell_2$ distance}
In this section, we derive the sample complexity of $\eps$-DP $\ab$-ary distribution estimation under $\ell_2$ distance, which is stated in Theorem~\ref{thm:pure-l2}.

\medskip

\noindent \textbf{Upper bound:}
We use the same algorithm as in Section~\ref{sec:puredp-tv}. Following the same argument as in Section~\ref{sec:puredp-tv}, the square of expected $\ell_2$ loss of $\hat{p}$ can be upper bounded by
\begin{align}
\Paren{\expectation{\norm{\hat{p}-p}}}^2 &\le \expectation{\norm{h-p}^2} \le \expectation{\norms{p^{\text{erm}}-p}}+\expectation{\norms{h - p^{\text{erm}}}} = O\Paren{\frac1{\ns}+\frac{\ab}{\ns^2\eps^2}}. \nonumber
\end{align}
Since $\Delta_k$ is convex, we have $\norm{\hat{p}-p} \le\norm{h-p}$. Moreover, the following lemma gives another bound for $\norm{\hat{p}-p}$ (See Corollary 2.3 in~\cite{Bassily18}).
\begin{lemma}
Let $L \subset \RR^d$ be a symmetric convex body of $\ab$ vertices $\{a_j\}_{j=1}^{\ab}$, and let $y \in L$ and $\bar{y} = y+z$ for some $z \in \RR^d$. Let $\hat{y} = \arg \min_{w\in L} \norm{w-\bar{y}}^2$. Then we must have 
$$ \norm{y-\hat{y}}^2 \le 4\max_{j\in[k]} \{ \langle z, a_j \rangle\}. $$
\end{lemma} 

From the lemma, we have $\expectation{\norm{\hat{p}-h}^2}\le 4\cdot \expectation{ \max_{j\in[\ab]} \absv{Z_j}}$, where $\forall j \in [\ab]$, $Z_j \sim \text{Lap}(\frac{2}{\ns\eps})$. Note that $\expectation{ \max \absv{Z_j}} = O\Paren{\frac{\log\ab}{\ns\eps}}$ due to the tail bound of Laplace distribution. We have $\Paren{\expectation{\norm{\hat{p}-p}}}^2 = O\Paren{\frac1{\ns}+ \frac{\log\ab}{\ns\eps}}$. Combined with the previous analysis,  $\Paren{\expectation{\norm{\hat{p}-p}}}^2 = O\Paren{\frac1{\ns}+ \min\Paren{\frac{\ab}{\ns^2\eps^2} ,\frac{\log\ab}{\ns\eps}} }$. Therefore $\expectation{\norm{\hat{p}-p}} \le \frac1{10}\alpha$ when $n= O \Paren{\frac{1}{\alpha^2}+ \min \Paren {\frac{\sqrt{k}}{\alpha\eps}, \frac{\log\ab}{\dist^2\eps} }}$.

\medskip

\noindent \textbf{Lower bound:}
We first consider the case when $\dist < \frac1{\sqrt{\ab}}$, where we can 
derive the lower bound simply by a reduction. By Cauchy-Schwarz 
inequality, for any estimator $\hat{p}$, $ \expectation{\absvn{\hat{p}-p}} 
\le \sqrt{\ab} \cdot \expectation{\norm{\hat{p}-p}}$. Therefore 
$S(\Delta_k,  \ell_2, \alpha, \eps) \ge S(\Delta_k,  d_{TV}, \sqrt{k}\alpha, 
\eps)$, which gives us $S(\Delta_k,  \ell_2, \alpha, \eps)= 
\Omega\Paren{\frac{1}{\dist^2}+\frac{\sqrt{\ab}}{\dist\eps}}$.

Now we consider $\dist \ge \frac1{\sqrt{k}}$. Note that it is enough if we prove the lower bound of $\Omega\Paren{\frac{\log \Paren{\alpha^2\ab}}{\dist^2\eps}}$, since $\Omega\Paren{\frac1{\dist^2}}$ 
is the sample complexity of non-private estimation problem for all range of $\dist$. Similarly, we follow Corollary~\ref{coro:fano}, except that we need to construct a different set of distributions. 

Without loss of generality, we assume $\dist<0.1$. Now we use the codebook in Lemma~\ref{lem:GV} to construct our distribution set. {We fix weight $l = \lfloor \frac1{50\dist^2} \rfloor$. Note that for any $x>2$, $\lfloor x \rfloor > \frac{x}{2}$. Then we have  $\frac1{100\dist^2} < \lfloor l \rfloor  \le \frac1{50\dist^2} $ since $\frac1{50\dist^2}>2$. Therefore we get a codebook $\cC$ with $|\cC| \ge (k\alpha^2)^{\frac{1}{200 \alpha^2}}$}. Given $c \in \cC$, we construct the following distribution  $\p_c$ in $\Delta_\ab$:
$$ \p_c (i)  = \frac1{l} c_i. $$

We use $\cV_\ab = \{ \p_c : c \in \cC \}$ to denote the set of all these distributions. It is easy to check that $\forall \p \in \cV_\ab$ is a valid distribution. Moreover, for any pair of distributions $\p, \q \in \cV_\ab$, we have $\norm{\p - \q}> \frac1{2\sqrt{l}} = \Omega\Paren{\dist}$.

For any pair $\p,\q \in \cV_k$, $d_{TV} (\p, \q) \le 1$, which is a naive upper bound for $TV$ distance. Finally by setting $\ell$ in Corollary~\ref{coro:fano} to be $\ell_2$ distance, we have $S(\Delta_{\ab}, \ell_2, \alpha, \eps)  =  \Omega \Paren{\frac{\log \absv{\cC} }{\eps}} = \Omega \Paren{\frac{\log ( k\alpha^2)}{\alpha^2\eps}}$.

\subsection{Product Distribution Estimation} \label{sec:product}


Recall that $\Delta_{\ab, \dims}$ is the set of all $(\ab,\dims)$-product distributions.~\cite{BunKSW2019} proves an upper bound of $O\Paren{\ab\dims \log \Paren{ \frac{\ab\dims}{\dist}} \Paren{\frac1{\dist^2}+\frac1{\dist\eps}}}$. We prove a sample complexity lower bound for $\eps$-DP $(\ab,\dims)$-product distribution estimation in Theorem~\ref{thm:main_product}, which is optimal up to logarithmic factors.

\begin{theorem}
\label{thm:main_product}
The sample complexity of $\eps$-DP $(\ab,\dims)$-product distribution estimation satisfies
\begin{align}
&S_{\tt DE}(\Delta_{\ab, \dims}, d_{TV}, \alpha, \eps) = \Omega \Paren{\frac{\ab\dims}{\dist^2}+\frac{\ab\dims}{\dist \eps}}. \nonumber
\end{align}
\end{theorem}

\begin{proof}
{
We start with the construction of the distribution set. First we use the same binary code as in Lemma~\ref{lem:GV} with weight $l = \frac{\ab}{2}$.
Let $h := \absv{\cC}$ denote the size of the codebook. Given $j \in [h]$, we construct the following $\ab$-ary distribution $\p_j$ based on $c_j \in \cC$:}

{$$ \p_{j} (i)  = \frac{1}{\ab} +\frac{\dist}{\ab\sqrt{\dims}} \cdot \mathbb{I}\Paren{ c_{j,i}=1},$$}
{ where $ c_{j,i}$ denotes the $i$-th coordinate of $c_j$.}
%

Now we have designed a set of $\ab$-ary distributions of size $h = \Omega \Paren{2^{\frac{7\ab}{128}}}$. To construct a set of product distributions, {we use the codebook construction in Lemma~\ref{lem:constantGV2} to get an $h$-ary codebook $\cH$ with length $d$ and minimum hamming distance $d/2$. Moreover, $|\cH| \ge (\frac{h}{16})^{\frac{\dims}{2}}$. }

Now we can construct the distribution set of $(\ab,\dims)$-product distributions.
Given $b \in \cH$, define
$$\dP_b = \p_{b_1} \times  \p_{b_2}\times \cdots \times \p_{b_d}.$$

Let $\cV_{\ab,\dims}$ denote the set of distributions induced by $\cH$. We want to prove that $\forall \dP \neq \dQ \in \cV_{\ab,\dims},$
\begin{align}
	\dtv{\dP}{\dQ} \ge C\dist, \label{eqn:dtv_prod}\\
	D_{KL} (\dP,\dQ)  \le 4\dist^2, \label{eqn:dkl_prod}
\end{align}
for some constant $C$. Suppose these two inequalities hold, using~\eqref{eqn:dkl_prod}, by Pinsker's Inequality, we get $d_{TV}(\dP, \dQ) \le \sqrt{2D_{KL} (\dP, \dQ) } \le 2\sqrt2\alpha$. Then using Corollary~\ref{coro:fano}, we can get
\begin{align}
&S_{\tt DE}(\Delta_{\ab, \dims}, d_{TV}, \alpha, \eps) = \Omega \Paren{\frac{\ab\dims}{\dist^2}+\frac{\ab\dims}{\dist \eps}}. \nonumber
\end{align}
Now it remains to prove~\eqref{eqn:dtv_prod} and~\eqref{eqn:dkl_prod}. For~\eqref{eqn:dkl_prod}, note that for any distribution pair $\dP,\dQ \in \cV_{\ab,\dims}$,
\begin{align}
	D_{KL} (\dP,\dQ) \le \dims \cdot \max_{i, j \in [h]} d_{KL}\Paren{\p_i, \p_j} \le 4\dist^2, \nonumber
\end{align}
\newzs{
where the first inequality comes from the additivity of $KL$ divergence for independent distributions and $\forall i, j \in [h]$, 
\[
	d_{KL}\Paren{\p_i, \p_j} = \sum_{x \in [k]} \p_i(x) \log \frac{\p_i(x)}{\p_j(x)} \le \sum_{x \in [\dims]} \frac{(\p_i(x) - \p_j(x))^2}{\p_j(x)} \le \ab \Paren{\frac{\dist}{\ab\sqrt{\dims}}}^2/ \frac{1}{\ab} = \frac{\alpha^2}{d}.
\]
}
Next we prove~\eqref{eqn:dtv_prod}. For any $b \in \cH$ and $\forall i \in[k]$, define set
\[
	S_i  = \{j \in [\ab] : c_{ b_{i},j} = 1\},
\]
which contains the locations of $+1$'s in the code at the $i$th coordinate of $b$. Based on this, we define a product distribution
\[
	\dP'_{b} = \prod_{i = 1}^{d} \cB(\mu_{i} ),
\]
where $\mu_{i} = \sum_{j \in S_i} \p_{b_{i}} (j)$ and $\cB(t)$ is a Bernoulli distribution with mean $t$. For any $b' \neq b \in \cH$, we define
\[
	\dP'_{b'} = \prod_{i = 1}^{d} \cB(\mu'_{i} ),
\]
 where $\mu'_{i} = \sum_{j \in S_i} \p_{b'_{i}}(j)$. Then we have:
 \[
 	\dtv{\dP'_{b}}{\dP'_{b'}} \le \dtv{\dP_{b}}{\dP_{b'}},
 \]
since $\dP'_{b}$ and $\dP'_{b'}$ can be viewed as a post processing of $\dP_{b}$ and $\dP_{b'}$ by mapping elements in $S_i$ to 1 and others to 0 at the $i$-th coordinate. Moreover, we have $\ham{b}{b'} \ge \frac{d}{2}$, and $\forall i$, if $b_{i} \neq b'_{i}$, we have $d_{H}(c_{b_{i}}, c_{b'_{i}}) > \frac{k}{8}$. By the definition of $\p_i$'s, we have
\[
	\| \mu_1 - \mu_2 \|_2^2 \ge \frac{d}{2} \times \left( \frac{\ab}{8} \times \frac{\dist}{\ab \sqrt{d}} \right)^2 = \frac{\dist^2}{128}.
\] 
By Lemma 6.4 in~\cite{KamathLSU18}, there exists a constant $C$ such that $\dtv{\dP'_{b}}{\dP'_{b'}} \ge C\dist$, proving~\eqref{eqn:dtv_prod}.


\end{proof}

\subsection{Gaussian Mixtures 	Estimation} \label{sec:gaussian}



Recall $\cG_{\dims} = \{ \cN (\mu, I_d): \norm{\mu}\le R\}$ is the set of $\dims$-dimensional spherical Gaussians with unit variance and bounded mean and $\cG_{\ab,\dims} = \{ \p: \p \text{ is a $\ab$-mixture of } \cG_{\dims} \}$ consists of mixtures of $k$ distributions in $\cG_{\dims}$.~\cite{BunKSW2019} proves an upper bound of $\widetilde{O}\Paren{\frac{\ab\dims}{\dist^2}+\frac{\dims}{\dist\eps}}$ for estimating $k$-mixtures of Gaussians. We provide a sample complexity lower bound for estimating mixtures of Gaussians in Theorem~\ref{thm:main_Gaussian}, which matches the upper bound up to logarithmic factors. 

\begin{theorem}
\label{thm:main_Gaussian}
Given $\ab \le \dims$ and $R\ge \sqrt{64\log\Paren{\frac{8\ab}{\dist}}}$, {or $\ab \ge \dims$ and $R\ge (\ab)^{\frac1\dims} \cdot \sqrt{64\dims \log\Paren{\frac{8\ab}{\dist}}}$,}
\begin{align}
&S_{\tt DE}(\cG_{\ab, \dims}, d_{TV}, \alpha, \eps) = \Omega \Paren{\frac{\ab\dims}{\dist^2}+\frac{\ab\dims}{\dist \eps}}. \nonumber
\end{align}
\end{theorem}

\begin{proof}
	
{We first consider the case when $\ab \le \dims$ and $R\ge \sqrt{64\log\Paren{\frac{8\ab}{\dist}}}$.} Let $\cC$ denote the codebook in Lemma~\ref{lem:GV} with weight $l=\frac{\dims}{2}$. Then we have $|\cC| \ge 2^{\frac{7\ab}{128}}$. Given $c_i$ in codebook $\cC$, we construct the following $\dims$-dimensional Gaussian distribution $\p_i$, with identity covariance matrix and mean $\mu_i$ satisfying 

\[\mu_{i,j}  = \frac{\dist}{\sqrt{ \dims}} c_{i,j},\]
{ where   $\mu_{i,j}$ denotes the $j$-th coordinate of $\mu_i$.}

{Let $h = \absv{\cC}$. Similar to the product distribution case, using Lemma~\ref{lem:constantGV2}, we can get an $h$-ary codebook $\cH$ with length $d$ and minimum hamming distance $d/2$. Moreover, $|\cH| \ge (\frac{h}{16})^{\frac{\dims}{2}}$. }

$\forall i \in [h]$ and $j \in k$, define $\p_{i}^{(j)} = \cN(\mu_i + \frac{R}{2}e_j , I_d)$, where $e_j$ is the $j$th standard basis vector. It is easy to verify their means satisfy the norm bound. 
For a codeword $b \in \cH$, let
\[
	\mP_b = \frac1k \Paren{ \p_{b_1}^{(1)} +  \p_{b_2}^{(2)}  +\ldots + \p_{b_k}^{(k)} }.
\]

Let $\cV_{\cG} = \{ \mP_b: b \in \cH\}$ be the set of the distributions defined above. Next we prove that $\forall \mP_b \neq \mP_{b'}\in \cV_{\cG}$, 
\begin{align}
	\dtv{\mP_b}{\mP_{b'}} \ge C \alpha, \label{eqn:dtv_gaussian} \\
	\dkl{\mP_b}{\mP_{b'}} \le 4 \alpha^2.\label{eqn:dkl_gaussian}
\end{align}
where $C$ is a constant. If these two inequalities hold, using~\eqref{eqn:dkl_gaussian}, by Pinsker's Inequality, we get $\dtv{\mP_b}{\mP_{b'}} \le \sqrt{2\dkl{\mP_b}{\mP_{b'}}  } \le 2\sqrt2\alpha$. Using Corollary~\ref{coro:fano}, we get
\[
S_{\tt DE}(\cG_{\ab, \dims}, d_{TV}, \alpha, \eps) =\Omega \Paren{\frac{\ab\dims}{\alpha^2}+\frac{\ab\dims}{\alpha \eps}}.
\]
It remains to prove~\eqref{eqn:dtv_gaussian} and~\eqref{eqn:dkl_gaussian}.

For~\eqref{eqn:dkl_gaussian}, note that for any distribution pair $\mP_b \neq \mP_{b'} \in \cV_{\cG}$,
\begin{align}
\dkl{\mP_b}{\mP_{b'}} &\le \frac1{k} \sum_{t=1}^{\ab} \dkl{\p^{(t)}_{b_t}}{\p^{(t)}_{b'_{t}}}\le \max_{i,j \in [h]} d_{KL}\Paren{\p_i, \p_j} \le 4\dist^2, \nonumber
\end{align}
{where the first inequality comes from the convexity of $KL$ divergence and the last inequality uses the fact that the KL divergence between two Gaussians with identity covariance is at most the $\ell_2^2$ distance between their means.} 

Next we prove~\eqref{eqn:dtv_gaussian}. Let $B_j = B_{j,1} \times \cdots \times B_{j,\dims}$, where

$$B_{j,i}=
\begin{cases}
[ \frac{R}{4}, \frac{3R}{4}], & \text{when $i=j$,}\\
[ -\frac{R}{4}, \frac{R}{4} ], & \text{when $i \neq j$ and $i \le \ab$,}\\
[-\infty, \infty], & \text{when $\ab<i \le \dims$.}
\end{cases}$$

Then by Gaussian tail bound and union bound, for any $\mP \in \cV_{\cG}$, the mass of the $j$-th Gaussian component outside $B_j$ is at most $2ke^{-\frac12 \cdot \Paren{\frac1{4} R}^2}$. And the mass of other Gaussian components inside $B_j$ is at most $e^{-\frac12 \cdot \Paren{\frac1{4} R}^2}$. Hence we have:
\begin{align}
	\dtv{\mP_b}{\mP_{b'}} &= \frac1{2\ab} \int_{z \in \RR^{\dims}} \absv{ \p_{b_1}^{(1)}(z) +\cdots+ \p_{b_{\ab}}^{(\ab)}(z)  - \p_{b'_{1}}^{(1)}(z)  - \cdots -  \p_{b'_{\ab}}^{(\ab)}(z) }  dz \nonumber\\
	&\ge \frac1{2\ab}  \sum_{j=1}^{\ab} \int_{z \in B_j} \absv{ \p_{b_1}^{(1)}(z) +\cdots+ \p_{b_{\ab}}^{(\ab)}(z)  - \p_{b'_{1}}^{(1)}(z)  - \cdots -  \p_{b'_{\ab}}^{(\ab)}(z) }  dz \nonumber\\
	& \ge \frac1{2\ab} \cdot \sum_{j=1}^{\ab} (\int_{z \in B_j}  \absv{  \p_{b_{j}}^{(j)}(z)  - \p_{b'_{j}}^{(j)}(z) } dz -(\ab - 1) \cdot e^{-\frac12 \cdot \Paren{\frac1{4} R}^2}) \nonumber \\
	& \ge \frac1{2\ab} \cdot \sum_{j=1}^{\ab} (\int_{z \in \RR^{\dims}}  \absv{  \p_{b_{j}}^{(j)}(z)  - \p_{b'_{j}}^{(j)}(z) } dz - 3\ab \cdot e^{-\frac12 \cdot \Paren{\frac1{4} R}^2}) \nonumber \\
	& = \frac1{2\ab} \cdot \sum_{j=1}^{\ab} \dtv{\p_{b_{j}}}{\p_{b'_{j}}} - \frac{3\dist^2}{64k}. \nonumber
\end{align}
By Fact 6.6 in~\cite{KamathLSU18}, there exists a constant $C_1$ such that for any pair $i \neq j \in [h]$, 
$$\dtv{\p_i}{\p_j} \ge C_1\dist.$$
Hence we have
\[
	\frac1{2\ab} \cdot \sum_{j=1}^{\ab} \dtv{\p_{b_{j}}}{\p_{b'_{j}}} \ge \frac{C_1\dist}{2\ab} \ham{b}{b'} \ge \frac{C_1\dist}{4},
\]
where the last inequality comes from the property of the codebook. WLOG, we can assume $
\frac{3\dist}{64k} < C_1/8$. Taking $C = \frac{C_1}{8}$ completes the proof of~\eqref{eqn:dtv_gaussian}.

{Now we considers the case when $\ab \ge \dims$ and $R\ge (\ab)^{\frac1\dims} \cdot \sqrt{64\dims \log\Paren{\frac{8\ab}{\dist}}}$. Let $r = \sqrt{16\dims \log\Paren{\frac{8\ab}{\dist}}}$, we note that there exists a packing set $S = \{v_1, v_2, ..., v_k\} \subset \mathbb{R}^{\dims}$ which satisfies $\forall u, v \in S$,
	\[
		\norm{u - v} > r,~~~~ \norm{u}\le R,~~~~\norm{v}\le R/3, 
	\]
	and $|S| = \ab$ since $R\ge 2 (\ab)^{\frac1\dims}  r$. Consider the set of mixture distributions as following: 
	For a codeword $b \in \cH$, let
	\[
		\mP'_b = \frac1k \Paren{ \p_{b_1}^{(1)'} +  \p_{b_{2}}^{(2)'}  +\ldots + \p_{b_{\ab}}^{(k)'} },
	\]
	where $\forall i \in [k], \p_{b_j}^{(j)'} = \cN(\mu_{b_j} + v_j, I_d)$.
	Let $B'_j$ denote the $\ell_2$ ball centering at the $v_j$ with radius $\frac{r}{2}$. We note that by similar analysis using the tail bound of the Gaussian distribution, the mass of the $j$-th Gaussian component outside $B_j'$ is at most $\frac{\dist^2}{64\ab^2}$. Meanwhile, the mass of other Gaussian components inside $B_j'$ is also at most $\frac{\dist^2}{64\ab^2}$. Hence the remaining analysis follows from the previous case.}
\end{proof}
\section{$(\eps,\delta)$-DP Distribution Estimation} \label{sec:adp_applications}
{In the previous section we used Theorem~\ref{thm:dp_fano} to obtain sample complexity lower bounds for pure differential privacy. We will now use Theorem~\ref{thm:assouad} to prove sample complexity lower bounds under $(\eps, \delta)$-DP.}

\subsection{$\ab$-ary Distribution Estimation} \label{sec:adp_kary}


\begin{theorem}
\label{thm:ApproximateDiscreteTotalVariation}
The sample complexity of $(\eps, \delta)$-DP $\ab$-ary distribution estimation under total variation distance is 
\[
S_{\tt DE}(\Delta_k,  d_{TV}, \alpha, \eps, \delta) =\Omega \Paren{\frac{k}{\alpha^2}+\frac{k}{\alpha(\eps+\delta)}}.
\]
\end{theorem}
{In practice, $\delta$ is chosen to be $\delta = O\Paren{\frac1\ns}$, and the privacy parameter is chosen as a small constant, $\eps = \Theta(1)$. In particular, when $\delta\le \eps$, the theorem above shows
\[
S_{\tt DE}(\Delta_k,  d_{TV}, \alpha, \eps, \delta) =\Omega \Paren{\frac{k}{\alpha^2}+\frac{k}{\alpha\eps}}.
\]
Since the sample complexity of $\eps$-DP is at most the sample complexity of $(\eps,\delta)$-DP, this shows that the bound above is tight for {$ \delta\le \eps$}. The lower bound part is proved using Theorem~\ref{thm:assouad} in Section~\ref{sec:approximate-tv}.}

\begin{theorem}
\label{thm:ApproximateDiscreteL2}
The sample complexity of $(\eps, \delta)$-DP discrete distribution estimation under $\ell_2$ distance,
\[
 \Omega \Paren {\frac{1}{\alpha^2}+\frac{\sqrt{\ab}}{\alpha (\eps+\delta)}} \le S_{\tt DE}(\Delta_k, \ell_2, \alpha, \eps, \delta) \le O \Paren{\frac{1}{\alpha^2}+\frac{\sqrt{\ab}}{\alpha\eps}},\ \ \ \ \text{for $\dist < \frac1{\sqrt{\ab}}$},
\] 
\begin{align}
\Omega\Paren{\frac{1}{\dist^2}+\frac{1}{\dist^2(\eps+\delta)}}\le S_{\tt DE}(\Delta_k, \ell_2, \alpha, \eps, \delta)  \le O\Paren{\frac{1}{\dist^2}+\frac{\log\ab}{\dist^2\eps}}, \ \ \ \text{for $\dist > \frac1{\sqrt{\ab}}$}. \nonumber
\end{align}
\end{theorem}
{When $\delta = O(\eps)$, the bounds are tight when $\dist < 1/\sqrt{\ab}$ and differ by a factor of $\log k$ when $\dist \ge 1/\sqrt{\ab}$. We prove this result in Section~\ref{sec:approximate-l2}.}

\subsubsection{Proof of Theorem~\ref{thm:ApproximateDiscreteTotalVariation}.}
\label{sec:approximate-tv}

{The first term $k/\alpha^2$ is the tight sample complexity without privacy. We prove that $S_{\tt DE}(\Delta_{\ab}, d_{TV}, \alpha, \eps, \delta) = \Omega \Paren{\frac{\ab}{\alpha(\eps+\delta)}}$.} 

{Suppose $k$ is even and $\dist < 1/10$. Let $ \cE_{k/2} =\{-1,+1\}^{k/2}$, for  $e \in \cE_{k/2}$, we define $\p_e\in\Delta_k$ as follows. 
\begin{align}
\text{For $i=1, \ldots, k/2$}~~~~ \p_e(2i-1) = \frac{1+ 10 e_i \cdot \alpha}{\ab},  ~~ \p_e(2i) = \frac{1- 10e_i  \cdot \dist}{\ab}. \label{eqn:constr-approx-k-ary}
\end{align}}
{To apply Theorem~\ref{thm:assouad}, let $\cV_{\ab/2} = \{ \p_e^n, e \in \cE_{\ab/2}\}$. $\p_e^n$ is the distribution of $n$ i.i.d. samples from distribution $\p_e$, and $\theta(\p_e^n) = \p_e$. For $u, v \in \cE_{\ab/2}$,
\[
\ell( \theta(\p_u^n),  \theta(\p_v^n)) = d_{TV} (p_{u},p_{v}) = \frac{20\dist}{\ab} \cdot \sum_{i=1}^{\frac{k}{2}}\mathbb{I}\Paren{u_i \neq v_i}, 
\]
thus obeying~\eqref{eqn:assouad-loss} with $\tau=10\dist/k$.}
{Recall the mixture distributions $\p_{+i}$ and $\p_{-i}$,} 
\[
\dP_{+i}  =  \frac{2}{|\cE_{\ab/2}|} \sum_{e \in \cE_{\ab/2}: e_i= + 1} \p^{\ns}_e, ~~~ \dP_{-i}  =  \frac{2}{|\cE_{\ab/2}|} \sum_{e \in \cE_{\ab/2}: e_i= -1}\p^{\ns}_e.
\]
{To apply Theorem~\ref{thm:assouad}, we prove the following bound on the Hamming distance between a coupling between $\p_{+i}$ and $\p_{-i}$.
\begin{lemma} \label{lem:coupling_distance}
For any $i$, there is a coupling $(X,Y)$ between $\dP_{+i}$ and $\dP_{-i}$, such that
\[
\expectation{\ham{X}{Y}} \le \frac{20\dist \ns}{\ab}.
\]
\end{lemma}}
\begin{proof}
{By the construction in~\eqref{eqn:constr-approx-k-ary}, note that the distributions $\dP_{+i}$ and $\dP_{-i}$ only have a difference in the number of times $2i-1$ and $2i$ appear. To generate $Y\sim \p_{-i}$ from from $X\sim\p_{+i}$, we scan through $X$ and independently change the coordinates that have the symbol $2i-1$ to the symbol $2i$ with probability $\frac{20\alpha}{1+10\alpha}$. The expected Hamming distance is bounded by $\frac{20\alpha}{1+10\alpha} \cdot \frac{1+10\alpha}{\ab} \cdot \ns = \frac{20\alpha \ns}{\ab}$.}
\end{proof}

Note that $\cV\subset\cP := \{\p^n | \p \in \Delta_k \}$. By Theorem~\ref{thm:assouad}, using the bound on $D$ from Lemma~\ref{lem:coupling_distance}, and $\tau=10\dist/k$,
\[
R(\cP, d_{TV}, \eps, \delta) \ge \frac{5\alpha}{\ab} \cdot \ab \cdot \Paren{0.9 e^{-10 \eps D} - 10 D \delta}\ge  {5\alpha} \cdot \Paren{0.9 e^{-200 n\eps\alpha/k} - 200 \frac{n\eps\alpha\delta}k}.
\]
{To achieve $R(\cP, d_{TV}, \eps, \delta) \le \alpha$, either $n\eps\alpha/k=\Omega(1)$ or ${n\eps\alpha\delta}/k=\Omega(1)$, which implies that $n=\Omega( \frac{k}{\alpha(\eps+\delta)})$}.

\subsubsection{Proof of Theorem~\ref{thm:ApproximateDiscreteL2}} \label{sec:approximate-l2}

We first consider the case where $\dist < \frac1{\sqrt{\ab}}$.  By Cauchy-Schwarz inequality, $S_{\tt DE}(\Delta_k, \ell_2, \alpha, \eps, \delta) \ge S_{\tt DE}(\Delta_k,  d_{TV}, \sqrt{k}\alpha, \eps, \delta) $, and therefore $S_{\tt DE}(\Delta_k, \ell_2, \alpha, \eps, \delta) = \Omega\Paren{\frac{1}{\dist^2}+\frac{\sqrt{\ab}}{\dist (\eps+\delta)}}$ by Theorem~\ref{thm:ApproximateDiscreteTotalVariation}.

For $\alpha \ge \frac{1}{\sqrt{k}}$, we have $l = \lfloor \frac1{16\dist^2}\rfloor \le k$. Therefore, $\Delta_l \subset \Delta_k$ and $\dist < \frac1{\sqrt{l}}$. Hence,
\[
	S_{\tt DE}(\Delta_k, \ell_2, \alpha, \eps, \delta) \ge S_{\tt DE}(\Delta_l, \ell_2, \alpha, \eps, \delta) = \Omega\Paren{\frac{1}{\alpha^2} + \frac1{\dist^2(\eps+\delta)}}.
\]

\subsection{Binary Product Distribution Estimation} 
\label{sec:adp_product}

{We now consider estimation of Bernoulli product distributions under total variation distance. A Bernoulli product distribution in $d$ dimensions is a distribution over $\{0,1\}^\dims$ parameterized by $\mu\in [0,1]^\dims$, where the $i$th coordinate is distributed $\ber(\mu_i)$, where $\ber(\cdot)$ is a Bernoulli distribution. Let $\Delta_{2,\dims}$ be the class of Bernoulli product distributions in $\dims$ dimensions.} 

\begin{theorem}
\label{thm:ApproximateProductTotalVariation}
The sample complexity of $(\eps, \delta)$-DP  binary product  distribution estimation satisfies 
\[
S_{\tt DE}(\Delta_{2,\dims}, d_{TV}, \alpha, \eps, \delta)  =\Omega \Paren{\frac{\dims}{\alpha^2}+\frac{\dims}{\alpha(\eps+\delta)}}.
\]
\end{theorem}

\newzs{Compared to the upper bound of $O\Paren{\dims \log \Paren{\dims/\dist} \Paren{1/\dist^2+1/\dist\eps}}$ in~\cite{BunKSW2019, KamathLSU18}, our bound is tight up to logarithmic factors when $\delta \le \eps$. \cite{KamathLSU18} also presents a lower bound of $\Omega \Paren{\frac{\dims}{\alpha^2}+\frac{\dims}{\alpha(\eps+\delta)}}$ under $(\eps, \delta)$-DP when $\delta = O(1/\ns)$. Although $\delta = O(1/\ns)$ is the more interesting regime in practice, our bound complements the result by stating that the utility will not improve even if $\delta$ can be as large as $\eps$.}
\begin{proof}	
{Since $\Theta(\dims/\eps^2)$ is an established tight bound for non-private estimation, we only prove the second term.}

{We start by constructing a set of Bernoulli product distributions indexed by $\cE_{\dims} = \{ \pm1\}^{\dims}$. For all $e \in \cE_{\dims}$, let $\p_e = \ber(\mu^e_1) \times \ber(\mu^e_2) \times \cdots \times \ber(\mu^e_\dims)$, where 
\[
\mu^e_i = \frac{1+ e_i\cdot 20\dist}{\dims}.
\]}
{Let $\cV = \{ \p_e^{\ns}, e \in \cE_{\dims}\}$, the set of distributions of 
$n$ i.i.d. samples from $\p_e$, and $\theta\Paren{\p_e^{\ns}} = \p_e$. For 
$u, v\in \cE_{\dims},$, 	$\ell( \theta(\p_u^n),  \theta(\p_v^n)) = d_{TV} 
(p_{u},p_{v})$. We first prove that~\eqref{eqn:assouad-loss} holds under 
total variation distance for an appropriate $\tau$.}


\begin{lemma}
\label{lem:binaryproducttv}
There exists a constant $C_1 > 5 $ such that $ \forall u, v \in \cE_{\dims}$,
\[ 
d_{TV} (\p_{u},\p_{v}) \ge \frac{C_1 \dist}{\dims}\cdot  \sum_{i=1}^{\dims}\mathbb{I}\Paren{u_i \neq v_i}.
\]
\end{lemma}

\begin{proof}
Let $S = \{i \in [\dims] : u_i \neq v_i\}$, and $S^{\prime} = \{i \in S : u_i = 1\} $. WLOG, let $\absv{S^{\prime}} \ge \frac12 \absv{S}$ (or else we can define $S^{\prime} = \{i \in S : u_i = -1\}$). Given a random sample $Z \in \{ \pm1\}^{\dims}$, we define an event $A = \{\forall i \in S^{\prime}, Z_i=0\}$. Now we consider the difference between the following two probabilities, which is a lower bound of the total variation distance between $\p_u$ and $\p_v$.
\begin{align}
d_{TV} (\p_{u},\p_{v}) &\ge \absv{\probofsub{Z \sim \p_{u}} {A} - \probofsub{Z \sim \p_{v}} {A}} \nonumber\\
&= \Paren{1-\frac{1-20\dist}{\dims}}^{\absv{S^{\prime}}} - \Paren{1-\frac{1+20\dist}{\dims}}^{\absv{S^{\prime}}} \nonumber\\
&\ge \frac{40\dist}{\dims} \cdot  \absv{S^{\prime}} \cdot  \Paren{1-\frac{1+20\dist}{\dims}}^{\absv{S^{\prime}}} \nonumber \\
& \ge \frac{40 \dist}{\dims} \cdot \absv{S'} e^{-(1+20\dist)} \ge \frac{C_1\dist}{\dims} \cdot \ham{u}{v}, \nonumber
\end{align}
where in the last two inequalities, we assume $\dims \ge 1000$ and $\dist<0.01$.
\end{proof} 

{Let $D$ be an upper bound on the expected Hamming distance for a coupling between $\dP_{+i}$ and $\dP_{-i}$ over all $i$.  Since $\cV_{\dims} \subset \Delta_{2,\dims}$, applying Theorem~\ref{thm:assouad} with Lemma~\ref{lem:binaryproducttv} we have
\[
R(\cP, d_{TV}, \eps, \delta) \ge \frac{C_1\alpha}{2\dims} \cdot \dims \cdot \Paren{0.9 e^{-10 \eps D} - 10 D \delta} = \frac{C_1\alpha}{2} \cdot \Paren{0.9 e^{-10 \eps D} - 10 D \delta}. 
\]}

{Setting $R(\cP, d_{TV}, \eps, \delta) \le \alpha$, we get $D =\Omega\Paren{\frac1\eps}$ or $D =\Omega\Paren{\frac1\delta}$, or equivalently, $D =\Omega\Paren{\frac1{\eps+\delta}}$. Lemma~\ref{lem:coupling_distance} below shows that we can take $ D =\frac{40\dist\ns}{\dims}$, which proves the result.}
\end{proof}

\begin{lemma} \label{lem:coupling_distance}
There is a coupling between $(X,Y)$ between $\dP_{+i}$ and $\dP_{-i}$, such that
$\expectation{\ham{X}{Y}} \le \frac{40\dist \ns}{\dims}.$
\end{lemma}

\begin{proof}
We generate $Y\sim \dP_{-i}$ from $X\sim\dP_{+i}$ as follows. If the $i$th coordinate of a sample $X$ is $+1$, we independently flip it to $-1$ with probability  $\frac{40\alpha}{1+20\alpha}$ to obtain a sample $Y$. The expected Hamming distance is bounded by $\frac{40\alpha}{1+20\alpha} \cdot \frac{1+20\alpha}{\dims} \cdot \ns = \frac{40\alpha \ns}{\dims}$. 
\end{proof}

\section{Proofs of Existence of Codes (Lemma~\ref{lem:GV} and Lemma~\ref{lem:constantGV2})}
\label{sec:codes}

\begin{proof}[Proof of Lemma~\ref{lem:GV}]
{This proof is a standard argument for Gilbert-Varshamov bound applied to constant weight codes. We use the following version (Theorem 7 in~\cite{GrahamS80}).}
\begin{lemma}
\label{lem:constantGV}
{There exists a length-$k$ constant weight binary code $\cC$ with weight $l$ and minimum Hamming distance $2\delta$, with $$\absv{\cC} \ge \frac{\binom{\ab}{l}}{\sum_{i=0}^{\delta} \binom{l}{i} \binom{\ab-l}{i}}.$$}
\end{lemma}
Applying this Lemma with $2\delta=\frac{l}{4}$, we have
\begin{align}
\absv{\cC} &\ge \frac{\binom{\ab}{l}} {\sum_{j=0}^{l/8} \binom{l}{j} \cdot \binom{\ab-l}{j}}\nonumber  \ge \frac{\binom{\ab}{l}}{\frac{l}{8} \cdot \binom{l}{\frac{l}{8}} \cdot \binom{\ab}{\frac{l}{8}}}
  = \frac1 {\frac{l}{8} \cdot \binom{l}{\frac{l}{8}} } \cdot \prod_{i = 0}^{\frac{7l}{8}- 1} \frac{k - \frac{l}{8}-i}{l - i} \\
& \ge \frac{2\sqrt{7} \pi}{e} \cdot (0.59)^{\frac{7l}{8}} \cdot \Paren{ \frac{\ab-\frac{l}{8}}{l}}^{\frac{7l}{8}} \label{eqn:gv-step}\\
 &\ge \Paren{\frac{\ab}{2^{7/8}l}}^{\frac{7l}{8}}, \nonumber 
\end{align}
{In~\eqref{eqn:gv-step}, we note that $\frac{k - \frac{l}{8}-i}{l - i}$ is monotonically increasing as $i$ increases.} And the first part is obtained by the Stirling's approximation $\sqrt{2\pi}\cdot l^{l+\frac12} \cdot e^{-l} \le l! \le e \cdot l^{l+\frac12} \cdot e^{-l}$ and the fact that $1.1^l \ge \sqrt{l}$ when $l \ge 20$. The last inequality comes from $l \le k/2$ and $15/16 \times 0.59 > 1/2^{7/8}$.
\end{proof}

\begin{proof}[Proof of Lemma~\ref{lem:constantGV2}]
By the Gilbert-Varshamov bound (Lemma~\ref{lem:constantGV}),
\[
	\absv{\cH} \ge \frac{h^\dims}{\sum_{j=0}^{\frac{d}{2}-1} \binom{\dims}{j}(h-1)^j} \ge \frac{h^\dims}{\frac{\dims}{2} \cdot \binom{\dims}{\frac{\dims}{2}} \cdot h^{\frac{\dims}{2}}} \ge
	\frac{h^{\frac{\dims}{2}}} { \dims \cdot 2^{\dims} }  \ge \left(\frac{h}{16}\right)^{\frac{\dims}{2}}.
\]
\end{proof}
\chapter{Privately Learning Markov Random Fields}
\label{cha:MRF}
\section{Introduction}
In this chapter, we continue to study the problem of private distribution estimation. However, we focus on a more complicated class of distributions -- random graphs.

Graphical models are a common structure used to model high-dimensional
data, which find a myriad of applications in diverse research
disciplines, including probability theory, Markov Chain Monte Carlo,
computer vision, theoretical computer science, social network
analysis, game theory, and computational
biology~\cite{LevinPW09,Chatterjee05,Felsenstein04,DaskalakisMR11,GemanG86,Ellison93,MontanariS10}.
While statistical tasks involving general distributions over $\dims$
variables often run into the curse of dimensionality (i.e., an
exponential sample complexity in $\dims$), Markov Random Fields (MRFs)
are a particular family of undirected graphical models which are
parameterized by the ``order'' $t$ of their interactions.  Restricting
the order of interactions allows us to capture most distributions
which may naturally arise, and also avoids this severe dependence on
the dimension (i.e., we often pay an exponential dependence on $t$
instead of $\dims$).  An MRF is defined as follows, see
Section~\ref{sec:preliminaries} for more precise definitions and
notations we will use in this chapter.
\begin{definition}
  Let $\ab, t, \dims \in \mathbb{N}$, $G = (V,E)$ be a graph on $\dims$ nodes, and $C_t(G)$ be the set of cliques of size at most $t$ in $G$.
  A \emph{Markov Random Field} with alphabet size $\ab$ and $t$-order interactions is a distribution $\mathcal{D}$ over $[\ab]^\dims$ such that
  \[
    \Pr_{X \sim \mathcal{D}}[X = x] \propto \exp\left(\sum_{I \in C_t(G)} \psi_I(x) \right),
  \]
  where $\psi_I : [\ab]^{\dims} \rightarrow \mathbb{R}$ depends only on varables in $I$.
\end{definition}

The case when $\ab = t = 2$ corresponds to the prototypical example of an MRF, the Ising model~\cite{Ising25} (Definition~\ref{def:ising}).
More generally, if $t = 2$, we call the model \emph{pairwise} (Definition~\ref{def:pairwise}), and if $\ab = 2$ but $t$ is unrestricted, we call the model a \emph{binary MRF} (Definition~\ref{def:mrf}). In this chapter, we mainly look at these two special cases of MRFs.

Given the wide applicability of these graphical models, there has been
a great deal of work on the problem of graphical model
estimation~\cite{RavikumarWL10, SanthanamW12, Bresler15, VuffrayMLC16,
  KlivansM17, HamiltonKM17, RigolletH17, LokhovVMC18, WuSD19}.  That
is, given a dataset generated from a graphical model, can we infer
properties of the underlying distribution?  Most of the attention has
focused on two learning goals.

\begin{enumerate}
  \item \emph{Structure learning} (Definition~\ref{def:learn-struct}): Recover the set of non-zero edges in $G$.
  \item \emph{Parameter learning} (Definition~\ref{def:learn-params}): Recover the set of non-zero edges in $G$, as well as $\psi_I$ for all cliques $I$ of size at most $t$.
\end{enumerate}

It is clear that structure learning is easier than parameter learning.
Nonetheless, the sample complexity of both learning goals is known to be roughly equivalent.
That is, both can be performed using a number of samples which is only \emph{logarithmic} in the dimension $\dims$ (assuming a model of bounded ``width'' $\lambda$\footnote{This is a common parameterization of the problem, which roughly corresponds to the graph having bounded-degree, see Section~\ref{sec:preliminaries} for more details.}), thus facilitating estimation in very high-dimensional settings.


Our goal is to design algorithms which guarantee both: 

\begin{itemize}
  \item Accuracy: With probability greater than $2/3$, the algorithm learns the underlying graphical model;

  \item Privacy: The algorithm satisfies differential privacy, even when the dataset is not drawn from a graphical model. 
\end{itemize}

Thematically, we investigate the following question: how much additional data is needed to learn Markov Random Fields under the constraint of differential privacy?
As mentioned before, absent privacy constraints, the sample complexity is logarithmic in $\dims$.
Can we guarantee privacy with comparable amounts of data?
Or if more data is needed, how much more?

\subsection{Results and Techniques}
\label{sec:results-techniques}
We proceed to describe our results on privately learning Markov Random Fields.
In this section, we will assume familiarity with some of the most common notions of differential privacy: pure $\ve$-differential privacy, $\rho$-zero-concentrated differential privacy, and approximate $(\ve, \delta)$-differential privacy.
In particular, one should know that these are in (strictly) decreasing order of strength (i.e., an algorithm which satisfies pure DP gives more privacy to the dataset than concentrated DP), formal definitions appear in Section~\ref{sec:preliminaries}.
Furthermore, in order to be precise, some of our theorem statements will use notation which is defined later (Section~\ref{sec:preliminaries}) -- these may be skipped on a first reading, as our prose will not require this knowledge.

\paragraph{Upper Bounds.}
Our first upper bounds are for parameter learning.
First, we have the following theorem, which gives an upper bound for parameter learning pairwise graphical models under concentrated differential privacy, showing that this learning goal can be achieved with $O(\sqrt{\dims})$ samples.
In particular, this includes the special case of the Ising model, which corresponds to an alphabet size $k = 2$.
Note that this implies the same result if one relaxes the learning goal to structure learning, or the privacy notion to approximate DP, as these modifications only make the problem easier. Further details are given in Section~\ref{sec:ub-pair}.
\begin{theorem}
  \label{thm:est-ub}
  There exists an efficient $\rho$-zCDP algorithm which learns the parameters of a pairwise graphical model to accuracy $\dist$ with probability at least $2/3$, which requires a sample complexity of
  $$\ns = O\Paren{\frac{ \lambda^2 k^5 \log(\dims k) e^{O(\lambda)}}{\dist^4}+ \frac{\sqrt{\dims} \lambda^2 k^{5.5} \log^2(\dims k)e^{O(\lambda)}}{\sqrt{\rho} \alpha^3}}$$
\end{theorem}

This result can be seen as a private adaptation of the elegant work of~\cite{WuSD19} (which in turn builds on the structural results of~\cite{KlivansM17}).
Wu, Sanghavi, and Dimakis~\cite{WuSD19} show that $\ell_1$-constrained logistic regression suffices to learn the parameters of all pairwise graphical models. 
We first develop a private analog of this method, based on the private Franke-Wolfe method of Talwar, Thakurta, and Zhang~\cite{TalwarTZ14,TalwarTZ15}, which is of independent interest. This method is studied in Section~\ref{sec:PFW}.
\begin{theorem}
  \label{thm:log-reg}
If we consider the problem of private sparse logistic regression,
there exists an efficient $\rho$-zCDP algorithm that produces a parameter vector $w^{priv}$, such that with probability at least $1-\beta$, the empirical risk
\[ 
\cL(w^{priv}; D) - \cL(w^{erm}; D) = O\Paren{ \frac{\lambda^{\frac{4}{3}}\log(\frac{ \ns\dims}{\beta}) } {(\ns \sqrt{\rho})^{\frac{2}{3}}}}.
\]
\end{theorem}
We note that Theorem~\ref{thm:log-reg} avoids a polynomial dependence on the dimension $\dims$ in favor of a polynomial dependence on the ``sparsity'' parameter $\lambda$.
The greater dependence on $\dims$ which arises in Theorem~\ref{thm:est-ub} is from applying Theorem~\ref{thm:log-reg} and then using composition properties of concentrated DP.

We go on to generalize the results of~\cite{WuSD19}, showing that $\ell_1$-constrained logistic regression can also learn the parameters of binary $t$-wise MRFs.
This result is novel even in the non-private setting. Further details are presented in Section~\ref{sec:bin-mrf}.


The following theorem shows that we can learn the parameters of binary $t$-wise MRFs with $\tilde O(\sqrt{\dims})$ samples.
\begin{theorem}
Let $\cD$ be an unknown binary $t$-wise MRF with associated polynomial $h$. Then there exists an $\rho$-zCDP algorithm which, with probability at least $2/3$, learns the maximal monomials of $h$ to accuracy $\dist$, given $\ns$ i.i.d.\ samples $Z^1,\cdots, Z^{\ns} \sim \cD$, where
$$\ns =O\Paren{ \frac{ e^{5\lambda t} \sqrt{\dims} \log^2(\dims) }{\sqrt{\rho} \dist^{\frac{9}{2}}}+ \frac{  t \lambda^2 \sqrt{\dims} \log{\dims}}{\sqrt{\rho}\dist^2} +  \frac{e^{6\lambda t} \log(\dims)}{\dist^6} } . $$
\end{theorem}

To obtain the rate above, our algorithm uses the Private
  Multiplicative Weights (PMW) method by \cite{HardtR10} to estimate
  all parity queries of all orders no more than $t$. The PMW method
  runs in time exponential in $p$, since it maintains a distribution
  over the data domain. We can also obtain an \emph{oracle-efficient}
  algorithm that runs in polynomial time when given access to an
  empirical risk minimization oracle over the class of parities. By
  replacing PMW with such an oracle-efficient algorithm \textsc{sepFEM} in
  \cite{VietriTBSW20}, we obtain a slightly worse sample complexity
$$\ns =O\Paren{ \frac{ e^{5\lambda t} \sqrt{\dims} \log^2(\dims) }{\sqrt{\rho} \dist^{\frac{9}{2}}}+ \frac{  t \lambda^2 {\dims^{5/4}} \log{\dims}}{\sqrt{\rho}\dist^2} +  \frac{e^{6\lambda t} \log(\dims)}{\dist^6} } . $$

For the special case of structure learning under approximate differential privacy, we provide a significantly better algorithm.
In particular, we can achieve an $O(\log \dims)$ sample complexity, which improves exponentially on the above algorithm's sample complexity of $O(\sqrt{\dims})$.
The following is a representative theorem statement for pairwise graphical models, though we derive similar statements for binary MRFs of higher order.
\begin{theorem}
  \label{thm:struct-ub}
    There exists an efficient $(\varepsilon, \delta)$-differentially private algorithm which, with probability at least $2/3$, learns the structure of a pairwise graphical model, which requires a sample complexity of $$n = O\left(\frac{\lambda^2 \ab^4 \exp(14\lambda) \log(\dims \ab)\log(1/\delta)}{\varepsilon\eta^4}\right).$$
\end{theorem}
This result can be derived using stability properties of non-private algorithms.
In particular, in the non-private setting, the guarantees of algorithms for this problem recover the entire graph \emph{exactly} with constant probability.
This allows us to derive private algorithms at a multiplicative cost of $O(\log(1/\delta)/\varepsilon)$ samples, using either the propose-test-release framework~\cite{DworkL09} or stability-based histograms~\cite{KorolovaKMN09, BunNSV15}.
Further details are given in Section~\ref{sec:struct-ub}.

\paragraph{Lower Bounds.}

We note the significant gap between the aforementioned upper bounds: in particular, our more generally applicable upper bound (Theorem~\ref{thm:est-ub}) has a $O(\sqrt{\dims})$ dependence on the dimension, whereas the best known lower bound is $\Omega(\log \dims)$~\cite{SanthanamW12}.
However, we show that our upper bound is tight.
That is, even if we relax the privacy notion to approximate differential privacy, \emph{or} relax the learning goal to structure learning, the sample complexity is still $\Omega(\sqrt{\dims})$.
Perhaps surprisingly, if we perform both relaxations simultaneously, this falls into the purview of Theorem~\ref{thm:struct-ub}, and the sample complexity drops to $O(\log \dims)$.

First, we show that even under approximate differential privacy, learning the parameters of a graphical model requires $\Omega(\sqrt{\dims})$ samples. The formal statement is given in Section~\ref{sec:est-lb}.
\begin{theorem}[Informal]
  Any algorithm which satisfies approximate differential privacy and learns the parameters of a pairwise graphical model with probability at least $2/3$ requires $\poly(\dims)$ samples.
\end{theorem}
This result is proved by constructing a family of instances of binary pairwise graphical models (i.e., Ising models) which encode product distributions.
Specifically, we consider the set of graphs formed by a perfect matching with edges $(2i, 2i+1)$ for $i \in [\dims/2]$.
In order to estimate the parameter on every edge, one must estimate the correlation between each such pair of nodes, which can be shown to correspond to learning the mean of a particular product distribution in $\ell_\infty$-distance.
This problem is well-known to have a gap between the non-private and private sample complexities, due to methods derived from fingerprinting codes~\cite{BunUV14, DworkSSUV15, SteinkeU17a}, and differentially private Fano's inequality.

Second, we show that learning the structure of a graphical model, under either pure or concentrated differential privacy, requires $\poly(\dims)$ samples. The formal theorem  appears in Section~\ref{sec:struct-lb}.
\begin{theorem}[Informal]
  Any algorithm which satisfies pure or concentrated differential privacy and learns the structure of a pairwise graphical model with probability at least $2/3$ requires $\poly(\dims)$ samples.
\end{theorem}
We derive this result via packing arguments~\cite{HardtT10,BeimelBKN14}, and differentially private Fano's inequality, by showing that there exists a large number (exponential in $\dims$) of different binary pairwise graphical models which must be distinguished.
The construction of a set of size $m$ implies lower bounds of $\Omega(\log m)$ and $\Omega(\sqrt{\log m})$ for learning under pure and concentrated differential privacy, respectively.

\subsubsection{Summary and Discussion}

We summarize our findings on privately learning Markov Random Fields in Table~\ref{tbl:ba-table}, focusing on the specific case of the Ising model. 
We note that qualitatively similar relationships between problems also hold for general pairwise models as well as higher-order binary Markov Random Fields.
Each cell denotes the sample complexity of a learning task, which is a combination of an objective and a privacy constraint.
Problems become harder as we go down (as the privacy requirement is tightened) and to the right (structure learning is easier than parameter learning).

The top row shows that both learning goals require only $\Theta(\log \dims)$ samples to perform absent privacy constraints, and are thus tractable even in very high-dimensional settings or when data is limited.
However, if we additionally wish to guarantee privacy, our results show that this logarithmic sample complexity is only achievable when one considers structure learning under approximate differential privacy.
If one changes the learning goal to parameter learning, \emph{or} tightens the privacy notion to concentrated differential privacy, then the sample complexity jumps to become polynomial in the dimension, in particular $\Omega(\sqrt{\dims})$.
Nonetheless, we provide algorithms which match this dependence, giving a tight $\Theta(\sqrt{\dims})$ bound on the sample complexity.


\begin{table*}[!htb]
\begin{center}
\begin{tabular}{| c | c | c |}
\hline
  \multicolumn{1}{|c|}{} & \multicolumn{1}{c|}{\textbf{Structure Learning}} & \multicolumn{1}{c|}{\textbf{Parameter Learning}} \\
\hline
  \textbf{Non-private}	& $\Theta(\log{\dims})$ (folklore) &  $\Theta(\log{\dims})$ (folklore) \\ \hline
  \textbf{Approximate DP}	& $\Theta(\log{\dims})$ (Theorems~\ref{thm:str-ub-pair})	&  $\Theta(\sqrt{\dims})$ (Theorems~\ref{thm:est-ub-ising} and~\ref{thm:est-lb})	\\ \hline
  \textbf{Zero-concentrated DP}	& $\Theta(\sqrt{\dims})$ (Theorems~\ref{thm:est-ub-ising} and~\ref{thm:str-ising})	&  $\Theta(\sqrt{\dims})$ (Theorems~\ref{thm:est-ub-ising} and~\ref{thm:est-lb})	\\ \hline
  \textbf{Pure DP}		&	$\Omega(\dims)$ (Theorem~\ref{thm:str-ising})&		$\Omega(\dims)$ (Theorem~\ref{thm:str-ising})	\\ \hline
\end{tabular}
\end{center}
  \caption{Sample complexity (dependence on $\dims$) of privately learning an Ising model.}
  \label{tbl:ba-table}
\end{table*}

\subsection{Related Work}
As mentioned before, there has been significant work in learning the structure and parameters of graphical models, see, e.g.,~\cite{ChowL68, CsiszarT06, AbbeelKN06, RavikumarWL10, JalaliJR11, JalaliRVS11, SanthanamW12, BreslerGS14b, Bresler15, VuffrayMLC16, KlivansM17, HamiltonKM17, RigolletH17, LokhovVMC18, WuSD19}.
Perhaps a turning point in this literature is the work of Bresler~\cite{Bresler15}, who showed for the first time that general Ising models of bounded degree can be learned in polynomial time.
Since this result, following works have focused on both generalizing these results to broader settings (including MRFs with higher-order interactions and non-binary alphabets) as well as simplifying existing arguments.
There has also been work on learning, testing, and inferring other statistical properties of graphical models~\cite{BhattacharyaM16, MartindelCampoCU16, DaskalakisDK17, MukherjeeMY18, Bhattacharya19}.
In particular, learning and testing Ising models in statistical distance have also been explored~\cite{DaskalakisDK18,GheissariLP18,DevroyeMR20,DaskalakisDK19, BezakovaBCSV19}, and are interesting questions under the constraint of privacy.

Recent investigations at the intersection of graphical models and differential privacy include~\cite{BernsteinMSSHM17, ChowdhuryRJ20,McKennaSM19}.
Bernstein et al.~\cite{BernsteinMSSHM17} privately learn graphical models by adding noise to the sufficient statistics and use an expectation-maximization based approach to recover the parameters.
However, the focus is somewhat different, as they do not provide finite sample guarantees for the accuracy when performing parameter recovery, nor consider structure learning at all.
Chowdhury, Rekatsinas, and Jha~\cite{ChowdhuryRJ20} study differentially private learning of Bayesian Networks, another popular type of graphical model which is incomparable with Markov Random Fields.
McKenna, Sheldon, and Miklau~\cite{McKennaSM19} apply graphical models in place of full contingency tables to privately perform inference.

Graphical models can be seen as a natural extension of product distributions, which correspond to the case when the order of the MRF $t$ is  $1$.
There has been significant work in differentially private estimation of product distributions~\cite{BlumDMN05, BunUV14, DworkMNS06, SteinkeU17a, KamathLSU18, CaiWZ19, BunKSW2019}.
Recently, this investigation has been broadened into differentially private distribution estimation, including sample-based estimation of properties and parameters, see, e.g.,~\cite{NissimRS07, Smith11, BunNSV15, DiakonikolasHS15, KarwaV18, AcharyaKSZ18, KamathLSU18, BunKSW2019}.
For further coverage of differentially private statistics, see~\cite{KamathU20}.


\section{Preliminaries and Notation}
\label{sec:preliminaries}
In order to distinguish between the vector coordinate and the sample, we use a different notation in this chapter. Given a set of points $\Xon$, we use superscripts, i.e., $X^i$ to denote the $i$-th datapoint. 
Given a vector $X \in \RR^{\dims}$, we use subscripts, i.e., $X_i$ to denote its $i$-th coordinate. We also use $X_{-i}$ to denote the vector after deleting the $i$-th coordinate, i.e. $X_{-i} = [X_1, \cdots, X_{i-1}, X_{i+1}, \cdots, X_{\dims}]$.

\subsection{Markov Random Field Preliminaries}
We first introduce the definition of the Ising model, which is a special case of general MRFs when $\ab=t=2$.
\begin{definition}
  \label{def:ising}
The $\dims$-variable Ising model is a distribution $\cD(A,\theta)$ on $\{-1, 1\}^{\dims}$ that satisfies
\begin{align}
 \proboff{Z = z }{} \propto \exp \Paren{\sum_{1\le i\le j \le \dims} A_{i,j} z_i z_j + \sum_{i \in [\dims]} \theta_i z_i}, \nonumber
\end{align}
where $A \in \RR^{\dims \times \dims}$ is a symmetric weight matrix with $A_{ii} = 0, \forall i \in [\dims]$ and $\theta \in \RR^{\dims}$ is a mean-field vector. 
The dependency graph of $ \cD(A,\theta)$ is an undirected graph $G= (V, E)$, with vertices $V = [\dims]$ and edges $E = \{ (i,j) : A_{i,j} \neq 0 \}$. The width of  $\cD(A,\theta)$ is defined as 
\begin{align}
\lambda (A,\theta)  = \max_{i\in[\dims]} \Paren{\sum_{j \in [\dims]} \absv{A_{i,j}} + \absv{\theta_i} }.\nonumber
\end{align}
Let $\eta (A,\theta)$ be the minimum edge weight in absolute value, i.e., $\eta (A,\theta) = \min_{i,j\in[\dims]: A_{i,j}\neq 0} \absv{A_{i,j}}.$
\end{definition}

We note that the Ising model is supported on $\{-1, 1\}^{\dims}$. A natural generalization is to generalize its support to $[\ab]^{\dims}$, and maintain pairwise correlations.

\begin{definition}
  \label{def:pairwise}
  The $\dims$-variable pairwise graphical model is a distribution $\cD(\cW,\Theta)$ on $[\ab]^{\dims}$ that satisfies
\begin{align}
  \proboff{Z = z }{} \propto \exp \Paren{\sum_{1\le i\le j \le \dims} W_{i,j}(z_i, z_j) + \sum_{i \in [\dims]} \theta_i(z_i)}, \nonumber
\end{align}
 where $\cW = \{ W_{i,j} \in \RR^{\ab \times \ab} : i \neq j \in [\dims]\}$ is a set of  weight matrices satisfying $W_{i,j} = W^T_{j,i}$, and $\Theta = \{\theta_i \in \RR^\ab : i \in [\dims]\}$ is a set of mean-field vectors.
The dependency graph of $ \cD(\cW,\Theta)$ is an undirected graph $G= (V, E)$, with vertices $V = [\dims]$ and edges $E = \{ (i,j) : W_{i,j} \neq 0 \}$. 
  The width of  $\cD(\cW,\Theta)$ is defined as 
\begin{align}
  \lambda (\cW,\Theta)  = \max_{i\in[\dims], a \in [\ab]} \Paren{\sum_{j \in [\dims] \backslash i} \max_{b \in [\ab]} \absv{W_{i,j}(a,b)} + \absv{\theta_i(a)} }.\nonumber
\end{align}
  Define $\eta(\cW, \Theta) = \min_{(i,j) \in E} \max_{a,b} |W_{i,j}(a,b)|$.
\end{definition}

Both the models above only consider pairwise interactions between
nodes.  In order to capture higher-order interactions, we examine the
more general model of Markov Random Fields (MRFs).  In this chapter, we
will restrict our attention to MRFs over a binary alphabet (i.e.,
distributions over $\{\pm1\}^p$).  In order to define binary $t$-wise
MRFs, we first need the following definition of multilinear
polynomials, partial derivatives and maximal monomials.

\begin{definition}
Multilinear polynomial is defined as $h: \RR^{\dims} \rightarrow \RR$ such that $h(x)=\sum_{I} \bar{h}(I) \prod_{i \in I}x_i$ where $\bar{h}(I)$ denotes the coefficient of the monomial $\prod_{i \in I}x_i$ with respect to the variables $(x_i:i \in I)$. Let $\partial_i h(x) = \sum_{J: i \not\in J}\bar{h} (J \cup \{ i\})\prod_{j \in J}x_j$ denote the partial derivative of $h$ with respect to $x_i$. Similarly, for $I \subseteq [\dims]$, let $\partial_I h(x) = \sum_{J: J \cap I = \phi}\bar{h} (J \cup I )\prod_{j \in J}x_j$ denote the partial derivative of $h$ with respect to the variables $(x_i: i \in I)$.
We say $I \subseteq [\dims]$ is a maximal monomial of $h$ if $\bar{h}(J)=0$ for all $J \supset I$.
\end{definition}

Now we are able to formally define binary $t$-wise MRFs.

\begin{definition}
  \label{def:mrf}
For a graph $G= (V, E)$ on $\dims$ vertices, let $C_t(G)$ denotes all cliques of  size at most $t$ in G. A binary $t$-wise Markov random field on $G$ is a distribution $\cD$ on  $\{-1, 1\}^{\dims}$ which satisfies
\begin{align}
 \proboff{Z = z }{Z \sim \cD} \propto \exp \Paren{\sum_{I \in C_t(G) }\varphi_I(z) }, \nonumber
\end{align}
and each $\varphi_I:\RR^{\dims} \rightarrow \RR$ is a multilinear polynomial that depends only on the variables in $I$. 

We call $G$ the dependency graph of the MRF and $h(x) = \sum_{I \in C_t(G)} \varphi_I(x)$ the factorization polynomial of the MRF. The width of $\cD$ is defined as $\lambda = \max_{i \in [\dims]} \normone{\partial_i h}$, where $\normone{h} \coloneqq \sum_{I} \absv{\bar{h}(I)}$.
\end{definition}

Now we introduce the definition of $\delta$-unbiased distribution and its properties. The proof appears in~\cite{KlivansM17}.
\begin{definition}[$\delta$-unbiased]
Let $S$ be the alphabet set, e.g., $S=\{1, -1 \}$ for binary $t$-pairwise MRFs and $S=[\ab]$ for pairwise graphical models. A distribution $\cD$ on $S^{\dims}$ is $\delta$-unbiased if for $Z \sim \cD$, $\forall i \in [\dims]$, and any assignment $x \in S^{\dims-1}$ to $Z_{-i}$, $\min_{z \in S} \probof{Z_i = z |Z_{-i}= x } \ge \delta$.
\end{definition}

The marginal distribution of a $\delta$-unbiased distribution also satisfies $\delta$-unbiasedness.

\begin{lemma}
\label{lem:marginal}
Let $\cD$ be a $\delta$-unbiased on $S^{\dims}$, with alphabet set $S$. For $X \sim \cD$, $\forall i \in [\dims]$, the distribution of $X_{-i}$ is also $\delta$-unbiased.
\end{lemma}

The following lemmas provide $\delta$-unbiased guarantees for various graphical models.
\begin{lemma}
\label{lem:pairwise-unbiased}
Let $\cD(\cW,\Theta)$ be a pairwise graphical model  with alphabet size $\ab$ and width $\lambda(\cW,\Theta)$. Then $\cD(\cW,\Theta)$ is $\delta$-unbiased with $\delta=e^{-2\lambda(\cW,\Theta)} /\ab$. In particular, an Ising model $\cD(A,\theta)$ is $e^{-2\lambda(A,\theta)}/2$-unbiased.
\end{lemma}

\begin{lemma}
Let $\cD$ be a binary $t$-wise MRFs with width $\lambda$. Then $\cD$ is $\delta$-unbiased with $\delta=e^{-2\lambda}/2$. 
\end{lemma}

Finally, we define two possible goals for learning graphical models.
First, the easier goal is \emph{structure learning}, which involves recovering the set of non-zero edges.
\begin{definition}
\label{def:learn-struct}
  An algorithm learns the \emph{structure} of a graphical model if, given samples $Z_1, \dots, Z_n \sim \cD$, it outputs a graph $\hat G = (V,\hat E)$ over $V = [\dims]$ such that $\hat E = E$, the set of edges in the dependency graph of $\cD$. 
\end{definition}

The more difficult goal is \emph{parameter learning}, which requires the algorithm to learn not only the location of the edges, but also their parameter values.

\begin{definition}
\label{def:learn-params}
  An algorithm learns the \emph{parameters} of an Ising model (resp.\ pairwise graphical model) if, given samples $Z_1, \dots, Z_n \sim \cD$, it outputs a matrix $\hat A$ (resp.\ set of matrices $\hat \cW$) such that $\max_{i,j \in [\dims]} |A_{i,j} -\hat A_{i,j}| \leq \alpha$ (resp.\ $|W_{i,j}(a,b) - \widehat W_{i,j}(a,b)| \leq \alpha$, $\forall i\neq j \in [\dims], \forall a,b \in [\ab]$).
\end{definition}

\begin{definition}
An algorithm learns the \emph{parameters} of a binary $t$-wise MRF with associated polynomial $h$ if, given samples $X^1, \dots, X^n \sim \cD$, it outputs another multilinear polynomial $u$ such that
that for all maximal monomial $I \subseteq [\dims]$, $\absv{\bar h(I) -\bar u(I)} \le \dist$.
\end{definition}

\subsection{Privacy Preliminaries}
A \emph{dataset} $X = \Xon \in \cX^n$ is a collection of points from some universe $\cX$. In this chapter we consider a few different variants of differential privacy.  The first is the standard notion of differential privacy, which has been heavily used in the previous chapters. The second is \emph{concentrated differential privacy}~\cite{DworkR16}. In this chapter, we specifically consider its refinement \emph{zero-mean concentrated differential privacy}~\cite{BunS16}.
\begin{definition}[Concentrated Differential Privacy (zCDP)~\cite{BunS16}]
    A randomized algorithm $\cA: \cX^n \rightarrow \cS$
    satisfies \emph{$\rho$-zCDP} if for
    every pair of neighboring datasets $X, X' \in \cX^n$,
    $$\forall \alpha \in (1,\infty)~~~D_\alpha\left(M(X)||M(X')\right) \leq \rho\alpha,$$
    where $D_\alpha\left(M(X)||M(X')\right)$ is the
    $\alpha$-R\'enyi divergence between $M(X)$ and $M(X')$.
\end{definition}


\noindent The following lemma quantifies the relationships between $(\eps,0)$-DP, $\rho$-zCDP and $(\eps,\delta)$-DP.
\begin{lemma}[Relationships Between Variants of DP~\cite{BunS16}] 
\label{lem:dpdefns}
For every $\eps \geq 0$,
\begin{enumerate}
\item If $\cA$ satisfies $(\eps,0)$-DP, then $\cA$ is $\frac{\eps^2}{2}$-zCDP.
\item If $\cA$ satisfies $\frac{\eps^2}{2}$-zCDP, then $\cA$ satisfies $(\frac{\eps^2}{2} + \eps \sqrt{2 \log(\frac{1}{\delta})},\delta)$-DP for every $\delta > 0$.
\end{enumerate}
\end{lemma}
Roughly speaking, pure DP is stronger than zero-concentrated DP, which is stronger than approximate DP.

A crucial property of all the variants of differential privacy is that they can be composed adaptively.  By adaptive composition, we mean a sequence of algorithms $\cA_1(X),\dots,\cA_T(X)$ where the algorithm $\cA_t(X)$ may also depend on the outcomes of the algorithms $\cA_1(X),\dots,\cA_{t-1}(X)$.
\begin{lemma}[Composition of zero-concentrated DP~\cite{BunS16}] \label{lem:dpcomp} If $\cA$ is an adaptive composition of
  differentially private algorithms $\cA_1,\dots, \cA_T$, and $\cA_1,\dots,\cA_T$ are
  $\rho_1,\dots,\rho_T$-zCDP respectively, then $\cA$ is $\rho$-zCDP for
  $\rho = \sum_t \rho_t$.  
\end{lemma}


\section{Parameter Learning of Pairwise Graphical Models}

\subsection{Private Sparse Logistic Regression}
\label{sec:PFW}

As a subroutine of our parameter learning algorithm, we consider the
following problem: given a training data set $D$ consisting of n pairs
of data $D = \{d^j\}_{j=1}^{\ns}= \{ (x^j, y^j)\}_{j=1}^{\ns}$, where
$x^j \in \RR^p$ and $y^j \in \RR$, a constraint set $\cC \in \RR^p$,
and a loss function $\ell: \cC \times \RR^{\dims+1} \rightarrow \RR$, we
want to find
$w^{erm} = \arg\min_{w\in \cC} ~\cL(w;D) = \arg\min_{w\in \cC} ~
\frac1\ns{\sum_{j=1}^{\ns} \ell(w; d^j)}$ with a zCDP constraint. This
problem was previously studied in~\cite{TalwarTZ14}. Before stating
their results, we need the following two definitions.  The first
definition is regarding Lipschitz continuity.
\begin{definition}
A function $\ell : \cC \rightarrow \RR$ is $L_1$-Lipschitz with respect to $\ell_1$ norm, if the following holds.
$$\forall w_1, w_2 \in \cC, \absv{\ell(w_1) - \ell(w_2)} \le L_1 \normone{w_1-w_2}.$$
\end{definition}

The performance of the algorithm also depends on the ``curvature" of the loss function, which is defined below, based on the definition of~\cite{Clarkson10, Jaggi13}. 
A side remark is that this is a strictly weaker constraint than smoothness~\cite{TalwarTZ14}.

\begin{definition}[Curvature constant]
For $\ell:\cC \rightarrow \RR$, $\Gamma_{\ell}$ is defined as 
$$ \Gamma_{\ell} = \sup_{w_1, w_2 \in \cC, \gamma \in (0,1], w_3 = w_1 + \gamma (w_2-w_1)}\frac{2}{\gamma^2} \Paren{\ell(w_3) - \ell(w_1) - \langle w_3-w_1, \nabla \ell(w_1) \rangle}.$$
\end{definition}

Now we are able to introduce the algorithm and its theoretical guarantees.

\begin{algorithm}

\caption{$\cA_{PFW}( D, \cL, \rho, \cC):$ Private Frank-Wolfe Algorithm}
\label{alg:privateFW}

\textbf{Input:}  Data set: $D = \{d^1,\cdots,d^{\ns}\}$, loss function: $\cL(w;D) = \frac1{\ns} \sum_{j=1}^{\ns} \ell(w;d^j)$ (with Lipschitz constant $L_1$), privacy parameters: $\rho$, convex set: $\cC = conv(S)$ with $\normone{\cC} \coloneqq  \max_{s \in S} \normone{s}$, iteration times: $T$

\begin{algorithmic}[1] 
\State Initialize $w$ from an arbitrary point in $\cC$

\For{$t=1$ to $T-1$}
\State $\forall s\in S$, $\alpha_s \leftarrow \langle s, \nabla \cL(w;D)\rangle+ \text{Lap} \Paren{ 0, \frac{L_1 \normone{C} \sqrt{T}}{\ns \sqrt{\rho}}}$

\State $\tilde{w_t} \leftarrow \arg\min_{s \in S} \alpha_s$

\State $w_{t+1} \leftarrow (1-\mu_t)w_{t} + \mu_t \tilde{w_t}$, where $\mu_t = \frac{2}{t+2}$
\EndFor

\end{algorithmic}
\textbf{Output:}  $w^{priv} = w_T$
\end{algorithm}


\begin{lemma}[Theorem 5.5 from~\cite{TalwarTZ14}]
\label{lem:PrivateFW}
Algorithm~\ref{alg:privateFW} satisfies $\rho$-zCDP. Furthermore, let $L_1$, $\normone{\cC}$ be defined as in Algorithm~\ref{alg:privateFW}. Let $\Gamma_{\ell}$ be an upper bound on the curvature constant for the loss function $\ell(\cdot;d)$ for all $d$ and $\absv{S}$ be the number of extreme points in $S$.
If we set $T=\frac{\Gamma_{\ell}^{\frac{2}{3}} (n \sqrt{\rho})^{\frac{2}{3}}}{L_1\normone{\cC}^{\frac{2}{3}}}$, then with probability at least $1-\beta$ over the randomness of the algorithm,

$$ \cL(w^{priv}; D) - \cL(w^{erm}; D) = O\Paren{ \frac{\Gamma_{\ell}^{\frac13} (L_1\normone{\cC}) ^{\frac{2}{3}}\log(\frac{ \ns\absv{S}}{\beta})} {(\ns \sqrt{\rho})^{\frac{2}{3}}}}.$$

\end{lemma}
\begin{proof}
The utility guarantee is proved in~\cite{TalwarTZ14}. Therefore, it is enough to prove the algorithm satisfies $\rho$-zCDP.
According to the definition of the Laplace mechanism, every iteration of the algorithm satisfies $(\sqrt{\frac{\rho}{T}},0)$-DP, which naturally satisfies $\frac{\rho}{T}$-zCDP by Lemma~\ref{lem:dpdefns}. Then, by the composition theorem of zCDP (Lemma~\ref{lem:dpcomp}), the algorithm satisfies $\rho$-zCDP.
\end{proof}

If we consider the specific problem of sparse logistic regression, we will get the following corollary.
\begin{corollary}
\label{cor:EmpiricalError}

If we consider the problem of sparse logistic regression, i.e., $\cL(w;D) = \frac1{n} \sum_{j=1}^n \log(1+e^{-y^j \langle w, x^j \rangle})$, with the constraint that $\cC = \{w: \normone{w} \le \lambda\}$, and we further assume that $\forall j, \norminf{x^j} \le 1, y^j \in \{\pm 1\}$, let $T =  \lambda^{\frac{2}{3}} (\ns \sqrt{\rho})^{\frac{2}{3}}$,  then with probability at least $1-\beta$ over the randomness of the algorithm,
$$ \cL(w^{priv}; D) - \cL(w^{erm}; D) = O\Paren{ \frac{\lambda^{\frac{4}{3}}\log(\frac{ \ns\dims}{\beta}) } {(\ns \sqrt{\rho})^{\frac{2}{3}}}}.$$

Furthermore, the time complexity of the algorithm is $O(T \cdot \Paren{\ns \dims+\dims^2}) =O\Paren{ \ns^{\frac{2}{3}}\cdot \Paren{\ns\dims + \dims^2}}$.
\end{corollary}

\begin{proof}
First let we show $L_1 \le 2$. If we fix sample $d = (x,y)$, then for any $w_1, w_2\in \cC$,
$$\absv{ \ell(w_1;d) -\ell(w_2;d)} \le \max_w \norminf{\nabla_w \Paren{ \ell(w;d)}} \cdot \normone{w_1 - w_2}.$$
Since $\nabla_w \Paren{ \ell(w;d)} = \Paren{\sigma(\langle w, x \rangle) - y} \cdot x$, we have $\norminf{\nabla_w \Paren{ \ell(w;d)}} \le 2$.

  Next, we wish to show $\Gamma_{\ell} \le \lambda^2$.
  We use the following lemma from~\cite{TalwarTZ14}.
\begin{lemma}[Remark 4 in~\cite{TalwarTZ14}]
For any $q,r \ge 1$ such that $\frac1q+\frac1r=1$, $\Gamma_{\ell}$ is upper bounded by $\alpha {\left\lVert \cC \right\rVert_q}^2$, where $\alpha = \max_{w \in \cC, {\left\lVert v \right\rVert_q=1}}  {\left\lVert \nabla^2 \ell(w) \cdot v \right\rVert_q}$.
\end{lemma}

If we take $q=1, r = +\infty$, then $\Gamma_{\ell} \le \alpha \lambda^2$, where 
$$\alpha = \max_{w \in \cC, \normone{v}=1}   \norminf{ \nabla^2 \ell(w;d) \cdot v} \le \max_{i,j \in[\dims]} \Paren{ \nabla^2 \ell(w;d)}_{i,j}.$$
We have $\alpha \le 1$, since $\nabla^2 \ell(w;d) = \sigma(\langle w, x \rangle)  \Paren{1- \sigma(\langle w, x \rangle)} \cdot x x^T$, and $\norminf{x}\le 1$, 

Finally given $\cC = \{w: \normone{w}\le 1 \}$, the number of extreme points of $S$ equals $2\dims$. By replacing all these parameters in Lemma~\ref{lem:PrivateFW}, we have proved the loss guarantee in the corollary. 

With respect to the time complexity, we note that the time complexity of each iteration is $O\Paren{\ns\dims+\dims^2}$ and there are $T$ iterations in total.
\end{proof}

Now if we further assume the data set $D$ is drawn i.i.d.\ from some underlying distribution $P$, the following lemma from learning theory relates the true risk and the empirical risk, which shall be heavily used in the following sections.

\begin{theorem}
\label{thm:generalization_error}
If we consider the same problem setting and assumptions as in Corollary~\ref{cor:EmpiricalError}, and we further assume that the training data set $D$ is drawn i.i.d.\ from some unknown distribution $P$, then with probability at least $1-\beta$ over the randomness of the algorithm and the training data set,
$$ \expectationf{\ell(w^{priv};(X,Y))}{(X,Y)\sim P} - \expectationf{\ell(w^*;(X,Y))}{(X,Y)\sim P}  = O\Paren{ \frac{\lambda^{\frac{4}{3}}\log(\frac{ \ns\dims}{\beta}) } {(\ns \sqrt{\rho})^{\frac{2}{3}}} + \frac{\lambda \log\Paren{\frac1{\beta}}}{\sqrt{n}} },$$
where $w^* = \arg\min_{w\in C}\expectationf{\ell(w;(X,Y))}{(X,Y)\sim P} $.
\end{theorem}

\begin{proof}
 By triangle inequality,
\begin{align*}
&~\expectationf{\ell(w^{priv};(X,Y))}{(X,Y)\sim P} - \expectationf{\ell(w^*;(X,Y))}{(X,Y)\sim P} \nonumber\\
\le& \absv{ \expectationf{\ell(w^{priv};(X,Y))}{(X,Y)\sim P} - \frac1{\ns} \sum_{m=1}^{\ns} \ell(w^{priv};d^m) } + \absv{  \frac1{\ns}\sum_{m=1}^{\ns} \ell(w^{priv};d^m) - \frac1{\ns}\sum_{m=1}^{\ns} \ell(w^{erm};d^m) } \nonumber\\
+& \Paren{  \frac1{\ns}\sum_{m=1}^{\ns} \ell(w^{erm};d^m) - \frac1{\ns}\sum_{m=1}^{\ns} \ell(w^*;d^m) } + \absv{ \expectationf{\ell(w^{*};(X,Y))}{(X,Y)\sim P} - \frac1{\ns}\sum_{m=1}^{\ns} \ell(w^{*};d^m) } \nonumber
\end{align*}

Now we need to bound each term. We firstly bound the first and last term simultaneously. By the generalization error bound (Lemma 7 from~\cite{WuSD19}), they are bounded by $O\Paren{ \frac{\lambda \log\Paren{\frac1{\beta}}}{\sqrt{n}}}$ simultaneously, with probability greater than $1 - \frac{2}{3}\beta$. Then we turn to the second term, by Corollary~\ref{cor:EmpiricalError}, with probability greater than $1 - \frac{1}{3}\beta $, it is bounded by $ O\Paren{ \frac{\lambda^{\frac{4}{3}}\log(\frac{ \ns\dims}{\beta}) } {(\ns \sqrt{\rho})^{\frac{2}{3}}}}$. Finally we bound the third term. According to the definition of $w^{erm}$, the third term should be smaller than 0.
Therefore, by union bound, $ \expectationf{\ell(w^{priv};(X,Y))}{(X,Y)\sim P} - \expectationf{\ell(w^*;(X,Y))}{(X,Y)\sim P}  = O\Paren{ \frac{\lambda^{\frac{4}{3}}\log(\frac{ \ns\dims}{\beta}) } {(\ns \sqrt{\rho})^{\frac{2}{3}}} + \frac{\lambda \log\Paren{\frac1{\beta}}}{\sqrt{n}} }$, with probability greater than $1-\beta$.
\end{proof}


\subsection{Privately Learning Ising Models}

We first consider the problem of estimating the weight matrix of the Ising model. 
To be precise, given $\ns$ i.i.d.\ samples $ \{z^1, \cdots, z^{\ns}\}$ generated from an unknown distribution $\cD(A,\theta)$, our goal is to design an $\rho$-zCDP estimator $\hat{A}$ such that with probability at least $\frac{2}{3}$, $\max_{i,j\in[\dims]}\absv{A_{i,j} - \hat{A}_{i,j}} \le \dist$.

An observation of the Ising model is that for any node $Z_i$, the probability of $Z_i=1$ conditioned on the values of the remaining nodes $Z_{-i}$ follows from a sigmoid function. The next lemma comes from~\cite{KlivansM17}, which formalizes this observation.

\begin{lemma}
\label{lem:IsingToLR}
Let $Z \sim \cD(A,\theta)$ and $Z \in \{-1,1\}^{\dims}$, then $\forall i \in [\dims]$, $\forall x \in \{-1,1\}^{[\dims] \backslash \{i\}}$,
\begin{align}
\probof{Z_i=1|Z_{-i}=x} &= \sigma\Paren{ \sum_{j\neq i}2A_{i,j}x_j + 2\theta_i} = \sigma\Paren{\langle w, x^{\prime} \rangle}.\nonumber
\end{align}
where $w= 2[A_{i,1},\cdots, A_{i,i-1}, A_{i,i+1},\cdots, A_{i,\dims},\theta_i]  \in \RR^{\dims}$, and $x^{\prime} = [x,1] \in \RR^{\dims}$.
\end{lemma}

\begin{proof}
The proof is from~\cite{KlivansM17}, and we include it here for completeness. 
 According to the definition of the Ising model,
\begin{align}
\probof{Z_i=1|Z_{-i}=x} &= \frac{\exp \Paren{\sum\limits_{j\neq i}A_{i,j}x_j +\sum\limits_{j\neq i} \theta_j + \theta_i} }{\exp \Paren{\sum\limits_{j\neq i} A_{i,j}x_j +\sum_{j\neq i} \theta_j + \theta_i}+\exp\Paren{\sum_{j\neq i}-A_{i,j}x_j +\sum_{j\neq i} \theta_j - \theta_i}}\nonumber\\
& =\sigma\Paren{ \sum_{j\neq i}2A_{i,j}x_j + 2\theta_i}.\nonumber
\end{align}
\end{proof}

By Lemma~\ref{lem:IsingToLR}, we can estimate the weight matrix by
solving a logistic regression for each node, which is utilized
in~\cite{WuSD19} to design non-private estimators. Our algorithm uses
the private Frank-Wolfe method to solve the per-node logistic regression
problem, achieving the following theoretical guarantee.

\begin{algorithm}

\caption{Privately Learning Ising Models}
\label{alg:DPIsing}

\textbf{Input:} $\ns$ samples $\{ z^1,\cdots, z^{\ns}\}$, where $z^m \in \{\pm1\}^{\dims}$ for $m\in [\ns]$; an upper bound on $\lambda(A,\theta) \le \lambda$, privacy parameter $\rho$

\begin{algorithmic}[1] 
\For{$i=1$ to $\dims$}

\State $\forall m \in [\ns]$, $x^m \leftarrow [z_{-i}^{m},1]$, $y^m \leftarrow z_i^m$

\State $w^{priv} \leftarrow \cA_{PFW}( D, \cL, \rho^{\prime}, \cC)$, where $\rho^{\prime} = \frac{\rho}{\dims}$,   $D=\{\Paren{x^m, y^m} \}_{m=1}^{\ns}$, $\cL(w;D) =  \frac1{\ns}\sum_{m=1}^{\ns}\log\Paren{1+e^{-y^m \langle w, x^m \rangle}}$, $\cC = \{ \normone{w}\le 2\lambda\}$

\State $\forall j \in \dims$, $\hat{A}_{i,j} \leftarrow \frac1{2} w^{priv}_{\tilde{j}}$, where $\tilde{j} =j $ when $j<i$ and $\widetilde{j}  = j-1$ if $j>i$
\EndFor
\end{algorithmic}
\textbf{Output:}$\hat{A} \in \RR^{\dims \times \dims}$
\end{algorithm}

%
%
%
%
%
%
%

\begin{theorem}
\label{thm:est-ub-ising}
Let $\cD(A,\theta)$ be an unknown $\dims$-variable Ising model with $\lambda(A,\theta)\le \lambda$.
There exists an efficient $\rho$-zCDP algorithm which outputs a weight matrix $\hat{A} \in \RR^{\dims \times \dims}$ such that 
with probability greater than $2/3$, $\max_{i,j \in [\dims]} \absv{A_{i,j} - \hat{A}_{i,j}} \le \dist$ if the number of i.i.d.\ samples satisfies
$$\ns = \Omega\Paren{\frac{ \lambda^2 \log(\dims) e^{12\lambda}}{\dist^4}+ \frac{\sqrt{\dims} \lambda^2 \log^2(\dims)e^{9\lambda}}{\sqrt{\rho} \alpha^3}}.$$
\end{theorem}

\begin{proof}
We first prove that Algorithm~\ref{alg:DPIsing}  satisfies $\rho$-zCDP. Notice that in each iteration, the algorithm solves a private sparse logistic regression under $\frac{\rho}{\dims}$-zCDP. Therefore, Algorithm~\ref{alg:DPIsing} satisfies $\rho$-zCDP  by composition (Lemma~\ref{lem:dpcomp}).

For the accuracy analysis, we start by looking at the first iteration  ($i=1$) and showing that $\absv{A_{1,j} - \hat{A}_{1,j}} \le \dist$, $\forall j \in [\dims]$, with probability greater than $1-\frac1{10\dims}$. 

Given a random sample $Z \sim \cD(A,\theta)$, we let  $X = [Z_{-1},1]$, $Y = Z_1$. From Lemma~\ref{lem:IsingToLR}, $\probof{Y=1|X=x} = \sigma\Paren{\langle w^*, x\rangle}$, where $w^* = 2[A_{1,2},\cdots,A_{1,\dims},\theta_1]$. We also note that $\normone{w^*}\le 2\lambda$, as a consequence of the width constraint of the Ising model.

For any $\ns$ i.i.d.\ samples $\{ z^{m}\}_{m=1}^{\ns}$ drawn from the Ising model, let $x^m = [z_{-1}^{m},1]$ and $y^m = z_1^m$, it is easy to check that each $(x^m,y^m)$ is the realization of $(X,Y)$. Let $w^{priv}$ be the output of  $\cA\Paren{ D, \cL, \frac{\rho}{\dims},  \{w: \normone{w}\le 2\lambda \}}$, where $D = \{(x^m,y^m)\}_{m=1}^{\ns}$.
By Lemma~\ref{thm:generalization_error}, when $\ns =O\Paren{ \frac{ \sqrt{\dims} \lambda^2 \log^2(\dims) }{\sqrt{\rho} \gamma^{\frac{3}{2}}}+ \frac{ \lambda^2 \log(\dims)}{\gamma^2} }$, with probability greater than $1-\frac{1}{10\dims}$, $ \expectationf{\ell(w^{priv};(X,Y))}{Z\sim \cD(A,\theta)} - \expectationf{\ell(w^*;(X,Y))}{Z\sim \cD(A,\theta)} \le \gamma.$

We will use  the following lemma from~\cite{WuSD19}. Roughly speaking, with the assumption that the samples are generated from an Ising model, any estimator $w^{priv}$ which achieves a small error in the loss $\cL$  guarantees an accurate parameter recovery in $\ell_{\infty}$ distance.

\begin{lemma}
\label{lem:parameter-error-ising}
Let $P$ be a distribution on $\{-1,1\}^{\dims-1} \times \{-1,1\}$. Given $u_1\in \RR^{\dims-1}, \theta_1\in \RR$, suppose $\probof{Y=1|X=x} = \sigma\Paren{\langle u_1,x \rangle+\theta_1}$ for $(X,Y) \sim P$. 
If the marginal distribution of $P$ on $X$ is $\delta$-unbiased, 
and $\expectationf{\log\Paren{1+e^{-Y \Paren{ \langle u_1, X \rangle+\theta_1}}} } {(X,Y)\sim P} - \expectationf{ \log\Paren{1+e^{-Y \Paren{ \langle u_2, X \rangle+\theta_2}}} }{(X,Y)\sim P} \le \gamma$ for some $u_2\in \RR^{\dims-1}, \theta_2\in \RR$, and $\gamma \le \delta e^{-2\normone{u_1} -2\normone{\theta_1}-6}$, then $\norminf{u_1-u_2} = O(e^{\normone{u_1} +\normone{\theta_1}} \cdot \sqrt{\gamma/\delta}).$
\end{lemma}

By Lemma~\ref{lem:marginal}, Lemma~\ref{lem:pairwise-unbiased} and Lemma~\ref{lem:parameter-error-ising}, if $ \expectationf{\ell(w^{priv};(X,Y))}{Z\sim \cD(A,\theta)} - \expectationf{\ell(w^*;(X,Y))}{Z\sim \cD(A,\theta)} \le O\Paren{\alpha^2 e^{-6\lambda}}$, we have $\norminf{w^{priv} - w^*} \le \alpha$. By replacing $\gamma = \alpha^2 e^{-6\lambda}$, we prove that $\norminf{A_{1,j} - \hat{A}_{1,j}} \le \dist$ with probability greater than $1-\frac1{10\dims}$. Noting that similar argument works for the other iterations and non-overlapping part of the matrix is recovered in different iterations. By union bound over $\dims$ iterations, we prove that with probability at least $\frac{2}{3}$, $\max_{i,j\in[\dims]}\absv{A_{i,j} - \hat{A}_{i,j}} \le \dist$.

Finally, we note that the time compexity of the algorithm is $poly(\ns, \dims)$ since the private Frank-Wolfe algorithm is time efficient by Corollary~\ref{cor:EmpiricalError}.
\end{proof}

\subsection{Privately Learning Pairwise Graphical Models}
\label{sec:ub-pair}

Next, we study parameter learning for pairwise graphical models over general alphabet. Given $\ns$ i.i.d.\ samples $ \{z^1, \cdots, z^{\ns}\}$ drawn from an unknown distribution $\cD(\cW,\Theta)$, we want to design an $\rho$-zCDP estimator $\hat{\cW}$ such that with probability at least $\frac{2}{3}$, $\forall i \neq j \in [\dims], \forall \posa, \posb \in [\ab], \absv{W_{i,j}(\posa, \posb) - \widehat{W}_{i,j}(\posa, \posb)} \le \dist$. To facilitate our presentation, we assume that $\forall i \neq j \in [\dims]$, every row (and column) vector of $W_{i,j}$ has zero mean.\footnote{The assumption that $W_{i,j}$ is centered is without loss of generality and widely used in the literature~\cite{KlivansM17, WuSD19}. We present the argument here for completeness. Suppose the $a$-th row of $W_{i,j}$ is not centered, i.e., $\sum_{b} W_{i,j} (a,b) \neq 0$, we can define $W^{\prime}_{i,j} (a,b) = W_{i,j} (a,b)  - \frac1k \sum_{b} W_{i,j} (a,b)$ and $\theta^{\prime}_{i}(a)  = \theta_i(a)+\frac1k \sum_{b} W_{i,j} (a,b)$, and the probability distribution remains unchanged.}

Analogous to Lemma~\ref{lem:IsingToLR} for the Ising model, a pairwise graphical model has the following property, which can be utilized to recover its parameters. 
\begin{lemma}[Fact 2 of~\cite{WuSD19}]
\label{lem:PairwiseToLR}
Let $Z \sim \cD(\cW,\Theta)$ and $Z \in [\ab]^\dims$. For any $i \in [\dims]$, any $\posa \neq \posb \in [\ab]$, and any $x \in [\ab]^{\dims-1}$,

$$\probof{Z_i  =\posa | Z_i \in \{\posa, \posb\}, Z_{-i}=x }  = \sigma \Paren{\sum_{j\neq i}\Paren{W_{i,j}(\posa, x_j) -W_{i,j}(\posb, x_j) }+\theta_i(\posa)-\theta_i(\posb)}.
$$
\end{lemma}

Now we introduce our algorithm. Without loss of generality, we consider estimating $W_{1,j}$ for all $j \in [\dims]$ as a running example. We fix a pair of values $(\posa, \posb)$, where $\posa,\posb \in [\ab]$ and $\posa \neq \posb$. Let $S_{\posa,\posb}$ be the samples where $Z_1  \in \{\posa,\posb \}$. In order to utilize Lemma~\ref{lem:PairwiseToLR}, we perform the following transformation on the samples in $S_{\posa,\posb}$: for the $m$-th sample $z^m$, let $y^m =1$ if $z_1^m = \posa$, else $y^m = -1$. And $x^m$ is the one-hot encoding of the vector $[z^m_{-1},1]$, where $\text{OneHotEncode}(s)$ is a mapping from $[\ab]^\dims$  to $\RR^{\dims \times \ab}$, and the $i$-th row is the $t$-th standard basis vector given $s_i = t$. Then we define $w^* \in \RR^{\dims \times \ab}$ as follows:
\begin{align} 
&w^*(j,\cdot) = W_{1,j+1} (\posa, \cdot) -W_{1,j+1}(\posb, \cdot), \forall j \in [\dims-1] ;  \nonumber\\
&w^*(\dims,\cdot) = [\theta_1(\posa)-\theta_1(\posb),0,\cdots,0].\nonumber
\end{align}
Lemma~\ref{lem:PairwiseToLR} implies that $\forall t$, $\probof{Y^t=1} = \sigma\Paren{\langle w^*, X^t \rangle}$, where $\langle \cdot ,\cdot \rangle$ is the element-wise multiplication of matrices. According to the definition of the width of $\cD(\cW,\Theta)$, $\normone{w^*}\le \lambda \ab$. Now we can apply the sparse logistic regression method of Algorithm~\ref{alg:DPPair} to the samples in $S_{\posa,\posb}$.

Suppose $w^{priv}_{\posa,\posb}$ is the output of the private Frank-Wolfe algorithm, we define $U_{\posa, \posb} \in \RR^{\dims \times \ab}$ as follows: $\forall b \in [\ab]$,
\begin{align}
\label{equ:centering}
&U_{\posa, \posb} (j,b) = w^{priv}_{\posa, \posb}(j,b) - \frac1{\ab} \sum_{a \in [\ab]} w^{priv}_{\posa, \posb}(j,a), \forall j \in [\dims-1]; \nonumber\\
& U_{\posa, \posb} (\dims,b) = w^{priv}_{\posa, \posb}(\dims,b) + \frac1{\ab} \sum_{j \in [\dims-1]}\sum_{a \in [\ab]} w^{priv}_{\posa, \posb}(j,a).
\end{align}


$U_{\posa, \posb}$ can be seen as a ``centered" version of $w^{priv}_{\posa, \posb}$ (for the first $\dims-1$ rows).
It is not hard to see that $\langle U_{\posa, \posb}, x \rangle = \langle w^{priv}_{\posa, \posb}, x \rangle$, so $U_{\posa, \posb}$ is also a minimizer of the sparse logistic regression. 

For now, assume that $\forall  j \in [\dims-1], b \in [\ab]$, $U_{\posa, \posb}(j,b)$ is a  ``good'' approximation of $\Paren{ W_{1,j+1} (\posa, b) -W_{1,j+1}(\posb, b)}$, which we will show later. If we sum over $\posb \in [\ab]$, it can be shown that 
$\frac{1}{\ab}\sum_{\posb \in [\ab]}U_{\posa, \posb} ( j,b)$ is also a ``good'' approximation of $W_{1,j+1}(\posa, b)$, for all $ j \in [\dims-1]$, and $ \posa, b \in [\ab]$, because of the centering assumption of $\cW$, i.e., $\forall  j \in [\dims-1], b \in [\ab], \sum_{\posb \in [\ab]} W_{1,j+1}(\posb,b) = 0$. With these considerations in mind, we are able to introduce our algorithm.

\begin{algorithm}

\caption{Privately Learning Pairwise Graphical Model}
\label{alg:DPPair}

\textbf{Input:} alphabet size $k$, $\ns$ i.i.d.\ samples $\{ z^1,\cdots, z^{\ns}\}$, where $z^m \in [\ab]^{\dims}$ for $m\in [\ns]$; an upper bound on $\lambda (\cW,\Theta)  \le \lambda$, privacy parameter $\rho$

\begin{algorithmic}[1] 

\For{$i=1$ to $\dims$}
	\For{each pair $\posa \neq \posb \in [\ab]$}
        \State $S_{\posa, \posb} \leftarrow \{ z^m, m\in [\ns]:  z_i^m \in \{u,v \} \}$
        
        \State $\forall z^m \in S_{\posa, \posb}$, $x^m \leftarrow \text{OneHotEncode}([z^{m}_{-i},1])$, $y^m \leftarrow 1$ if $ z_i^m=\posa$; $y^t \leftarrow -1$ if $ z_i^m=\posb$
   
        \State $w^{priv}_{\posa, \posb} \leftarrow \cA_{PFW}( D, \cL, \rho^{\prime}, \cC)$, where $\rho^{\prime} = \frac{\rho}{\ab^2\dims}$, $D=\{ \Paren{x^m, y^m}:z^m \in S_{\posa, \posb} \}$, $\cL(w;D) = \frac1{|S_{\posa, \posb}|} \sum_{m=1}^{|S_{\posa, \posb}|}  \log\Paren{1+e^{-y^m \langle w, x^m \rangle}}$, $\cC = \{ \normone{w}\le 2\lambda \ab\}$
   
        \State Define $U_{\posa, \posb} \in \RR^{\dims \times \ab}$ by centering the first $\dims-1$ rows of $w^{priv}_{\posa, \posb}$, as in Equation~\ref{equ:centering}
        \EndFor
	
	\For{$j \in [\dims] \backslash i$  and $\posa \in [\ab]$}
 	\State $\widehat{W}_{i,j}(\posa,:) \leftarrow \frac{1}{\ab}\sum_{\posb\in [\ab]}U_{\posa, \posb} (\tilde{j},:)$, where $\tilde{j} = j$ when $j<i$ and $\tilde{j} = j-1$ when $j>i$
	 \EndFor
\EndFor
\end{algorithmic}
\textbf{Output:} $\widehat{W}_{i,j} \in \RR^{\ab \times \ab}$ for all $i \neq j \in[\dims]$
\end{algorithm}

%
%
%
%
%
%
%

The following theorem is the main result of this section.
Its proof is structurally similar to that of Theorem~\ref{thm:est-ub-ising}.

\begin{theorem}
\label{thm:est-ub-pair}
  
Let $\cD(\cW,\Theta)$ be an unknown $\dims$-variable pairwise graphical model distribution, and we suppose that $\cD(\cW,\Theta)$ has width $\lambda(\cW,\Theta)\le \lambda$. There exists an efficient $\rho$-zCDP algorithm which outputs $\widehat{W}$ such that with probability greater than $2/3$,
$\absv{ W_{i,j}(\posa, \posb)- \widehat{W}_{i,j}(\posa, \posb)} \le \dist$, $\forall i \neq j \in [\dims], \forall \posa, \posb \in [\ab]$ if the number of i.i.d.\ samples satisfy
$$\ns = \Omega\Paren{ \frac{ \lambda^2 k^5 \log(\dims k) e^{O(\lambda)}}{\dist^4} + \frac{\sqrt{\dims} \lambda^2 k^{5.5} \log^2(\dims k)e^{O(\lambda)}}{\sqrt{\rho} \alpha^3}}.$$
\end{theorem}

\begin{proof}
We consider estimating $W_{1,j}$ for all $j \in [\dims]$ as an example. 
Fixing one pair $(\posa,\posb)$, let $S_{\posa,\posb}$ be the samples whose first element is either $\posa$ or $\posb$, and $\ns^{\posa,\posb}$ be the number of samples in $S_{\posa,\posb}$. 
We perform the following transformation on the samples in $S_{\posa,\posb}$: for the sample $Z$, let $Y=1$ if $Z_1 = \posa$, else $Y = -1$, and let $X$ be the one-hot encoding of the vector $[Z_{-1},1]$. 

Suppose the underlying joint distribution of $X$ and $Y$ is $P$, i.e., $(X,Y) \sim P$, then
by Theorem~\ref{thm:generalization_error}, when $\ns^{\posa,\posb} = O \Paren{\frac{\lambda^2 \ab^2 \log^2(\dims\ab)}{\gamma^2}+\frac{\sqrt{d} \lambda^2 \ab^3 \log^2(\dims\ab)}{\gamma^{\frac{3}{2}} \sqrt{\rho} } }$, with probability greater than $1-\frac1{10 \dims \ab^2}$, 
$$ \expectationf{\ell(U_{\posa, \posb};(X,Y))}{(X,Y)\sim P} - \expectationf{\ell(w^*;(X,Y))}{(X,Y)\sim P} \le \gamma.$$ 
The following lemma appears in~\cite{WuSD19}, which is analogous to Lemma~\ref{lem:parameter-error-ising} for the Ising model.
\begin{lemma}
\label{lem:parameter-error-kalphabet}
Let $\cD$ be a $\delta$-unbiased distribution on $[\ab]^{\dims-1}$. For $Z \sim \cD$, $X$ denotes the one-hot encoding of $Z$. 
Let $u_1, u_2 \in \RR^{(\dims-1) \times \ab}$ be two matrices where $\sum_{a} u_1 (i,a) =0$ and  $\sum_{a} u_2(i,a) =0$ for all $i \in [\dims-1]$.
Let $P $ be a distribution such that given $u_1, \theta_1\in \RR$, $\probof{Y=1| X= X} = \sigma\Paren{\langle u_1,x \rangle+\theta_1}$ for $(X,Y) \sim P$. Suppose $\expectationf{\log\Paren{1+e^{-Y \Paren{ \langle u_1, X \rangle+\theta_1}}} } {(X,Y)\sim P} - \expectationf{ \log\Paren{1+e^{-Y \Paren{ \langle u_2, X \rangle+\theta_2}}} }{(x,Y)\sim P} \le \gamma$ for $u_2\in \RR^{(\dims-1) \times \ab}, \theta_2\in \RR$, and $\gamma \le \delta e^{-2\norminfone{u_1} -2\normone{\theta_1}-6}$, then $\norminf{u_1-u_2} = O(e^{\norminfone{u_1} +\normone{\theta_1}} \cdot \sqrt{\gamma/\delta}).$
\end{lemma}

By Lemma~\ref{lem:marginal}, Lemma~\ref{lem:pairwise-unbiased} and Lemma~\ref{lem:parameter-error-kalphabet}, if we substitute $\gamma = \frac{e^{-6\lambda} \alpha^2}{\ab}$, when $\ns^{\posa,\posb}  = O\Paren{\frac{ \lambda^2 k^4 \log(\dims k) e^{O(\lambda)}}{\dist^4}+\frac{\sqrt{\dims} \lambda^2 k^{4.5} \log^2(\dims k)e^{O(\lambda)}}{\sqrt{\rho} \alpha^3}}, $
\begin{align}
\label{equation:approxdifference}
\absv{W_{1,j}(\posa, b) - W_{1,j}(\posb, b) - U^{\posa ,\posb} (j ,b)} \le \alpha, \forall j \in [\dims-1], \forall b \in [\ab].
\end{align}
  By a union bound, Equation~(\ref{equation:approxdifference}) holds for all $(\posa, \posb)$ pairs simultaneously with probability greater than $1-\frac1{10\dims}$.
If we sum over $\posb \in [\ab]$ and use the fact that $\forall j,b, \sum_{\posb \in [\ab]} W_{1,j}(\posb,b) = 0$, we have
$$\absv{W_{1,j}(\posa, b) -  \frac{1}{\ab}\sum_{\posb \in [\ab]}U_{\posa, \posb} ( j,b)} \le \alpha, \forall j \in [\dims-1], \forall \posa, b \in [\ab]. $$

Note that we need to guarantee that we obtain $\ns^{\posa, \posb}$ samples for each pair $(\posa, \posb)$. Since $\cD(\cW,\Theta)$ is $\delta$-unbiased, given $Z \sim \cD(\cW, \Theta)$, for all $\posa \neq \posb$,  $\probof{Z \in S_{\posa, \posb}} \ge 2\delta$. By Hoeffding's inequality, when $\ns = O\Paren{\frac{\ns^{\posa,\posb} }{\delta} + \frac{\log (\dims k^2)}{\delta^2}}$, with probability greater than $1 - \frac1{10\dims}$, we have enough samples for all $(\posa, \posb)$ pairs simultaneously. Substituting $\delta = \frac{e^{-6\lambda}}{k}$, we have
$$\ns = O\Paren{\frac{ \lambda^2 k^5 \log(\dims k) e^{O(\lambda)}}{\dist^4} +\frac{\sqrt{\dims} \lambda^2 k^{5.5} \log^2(\dims k)e^{O(\lambda)}}{\sqrt{\rho} \alpha^3}}.$$

The same argument holds for other entries of the matrix.
We conclude the proof by a union bound over $\dims$ iterations.

Finally, we note that the time compexity of the algorithm is $\poly(\ns, \dims)$ since the private Frank-Wolfe algorithm is time efficient by Corollary~\ref{cor:EmpiricalError}.
\end{proof}


\section{Privately Learning Binary $t$-wise MRFs}
\label{sec:bin-mrf}

Let $\cD$ be a $t$-wise MRF on $\{1, -1 \}^{\dims}$ with underlying dependency graph $G$ and factorization polynomial $h(x) = \sum_{I \in C_t(G)} h_I(x)$. 
We assume that the width of $\cD$ is bounded by $\lambda$, i.e., $\max_{i \in [\dims]} \normone{\partial_i h} \le \lambda$, where $\normone{h}\coloneqq \sum_{I \in C_t(G)} \absv{\bar{h}(I)}$. 
Similar to~\cite{KlivansM17}, given $\ns$ i.i.d.\ samples $ \{z^1, \cdots, z^{\ns}\}$ generated from an unknown distribution $\cD$, we consider the following two related learning objectives, under the constraint of $\rho$-zCDP:
\begin{enumerate}
\item[1.] find a multilinear polynomial $u$ such that with probability greater than $\frac{2}{3}$, $\normone{h-u} \coloneqq \sum_{I \in C_t(G)} \absv{\bar{h}(I) -\bar{u}(I)} \le \dist$ ;
\item[2.] find a multilinear polynomial $u$ such that with probability greater than $\frac{2}{3}$, for every maximal monomial $I$ of $h$, $\absv{\bar{h}(I) -\bar{u}(I)} \le \dist$.
\end{enumerate}


We note that our first objective can be viewed as parameter estimation in $\ell_1$ distance, where only an average performance guarantee is provided. 
In the second objective, the algorithm recovers every maximal monomial, which can be viewed as parameter estimation in $\ell_\infty$ distance. 
These two objectives are addressed in Sections~\ref{sec:MRFl_1} and~\ref{sec:MRFl_infty}, respectively.

\subsection{Parameter Estimation in $\ell_1$ Distance}
\label{sec:MRFl_1}
The following property of MRFs, from~\cite{KlivansM17}, plays a critical role in our algorithm. 
The proof is similar to that of Lemma~\ref{lem:IsingToLR}.
\begin{lemma}[Lemma 7.6 of~\cite{KlivansM17}]
\label{lem:sigmoidMRF}
Let $\cD$ be a $t$-wise MRF on $\{1, -1 \}^{\dims}$ with underlying dependency graph $G$ and factorization polynomial $h(x) = \sum_{I \in C_t(G)} h_I(x)$, then
$$\probof{Z_i=1|Z_{-i}=x} = \sigma(2\partial_i h(x)), \forall i \in [\dims], \forall x \in \{1 ,-1\}^{[\dims] \backslash i}.$$
\end{lemma}

 Lemma~\ref{lem:sigmoidMRF} shows that, similar to pairwise graphical models, it also suffices to learn the parameters of binary $t$-wise MRF using sparse logistic regression.

\begin{algorithm}

\caption{Private Learning binary $t$-wise MRF in $\ell_1$ distance}
\label{alg:DPMRF}
\textbf{Input:}  $\ns$ i.i.d.\ samples $\{ z^1,\cdots, z^{\ns}\}$, where $z^m \in \{\pm1\}^{\dims}$ for $m\in [\ns]$; an upper bound $\lambda$ on $\max_{i \in [\dims]} \normone{\partial_i h}$, privacy parameter $\rho$

\begin{algorithmic}[1] 
\For{$i=1$ to $\dims$}
\State $\forall m \in [\ns]$, $x_m \leftarrow \Brack{\prod_{j \in I} z_j^m: I \subset [\dims \backslash i], \absv{I} \le t-1}$, $y_m \leftarrow z_i^m$

\State $w^{priv} \leftarrow \cA( D, \cL, \rho^{\prime}, \cC)$
where $D=\{\Paren{x_m, y_m} \}_{m=1}^{\ns}$, $\ell(w;d) =  \log\Paren{1+e^{-y \langle w, x \rangle}}$, $\cC = \{ \normone{w}\le 2\lambda\}$, and $\rho^{\prime} = \frac{\rho}{\dims}$

\For{$I \subset [\dims \backslash i]$ with $\absv{I} \le t-1$ }
\State $\bar{u}(I \cup \{i\}) = \frac1{2} w^{priv}(I)$,  when $\arg\min (I \cup i) = i$
\EndFor
\EndFor
\end{algorithmic}
\textbf{Output:} $ \bar{u}(I): I \in C_t(K_{\dims})$,  where $K_{\dims}$ is the $\dims$-dimensional complete graph
\end{algorithm}

%
%
%
%
%
%
%

\begin{theorem}
There exists a $\rho$-zCDP algorithm which, with probability at least $2/3$, finds a multilinear polynomial $u$ such that 
$\normone{h-u} \le \dist,$
given $\ns$ i.i.d.\ samples $Z^1,\cdots, Z^{\ns} \sim \cD$, where
$$\ns = O\Paren{ \frac{ (2t)^{O(t)} e^{O(\lambda t)} \cdot \dims^{4t} \cdot \log (\dims) }  {\dist^4} + \frac{ (2t)^{O(t)} e^{O(\lambda t)}\cdot \dims^{3t+\frac12} \cdot \log^2 (\dims)} { \sqrt{\rho}\dist^3}}. $$
\end{theorem}


\begin{proof}
Similar to the previous proof, we start by fixing $i=1$. Given a random sample $Z \sim \cD$,   let $X = \Brack{\prod_{j \in I} Z^j: I \subset [\dims] \backslash 1, \absv{I} \le t-1}$ and $Y = Z_i $. According to Lemma~\ref{lem:sigmoidMRF}, we know that $\expectation{Y|X} = \sigma\Paren{\langle w^*,X \rangle}$, where $w^* = \Paren{ 2\cdot \overline{ \partial_{1}h} (I): I \subset [\dims] \backslash 1, \absv{I}\le t-1}$. Furthermore, $\normone{w^*} \le 2\lambda$ by the width constraint. Now, given $\ns$ i.i.d.\ samples $\{ z^{m}\}_{m=1}^{\ns}$ drawn from $\cD$, it is easy to check that for any given $z^m$, its corresponding $(x^m,y^m)$ is one realization of $(X,Y)$. Let $w^{priv}$ be the output of  $\cA\Paren{ D, \cL, \frac{\rho}{\dims},  \{w: \normone{w}\le 2\lambda \}}$, where $D = \{(x^m,y^m)\}_{m=1}^{\ns}$ and $\ell(w;(x,y)) = \log\Paren{1+e^{-y \langle w, x \rangle}}$. By Lemma~\ref{thm:generalization_error}, $ \expectationf{\ell(w^{priv};(X,Y))}{Z\sim \cD(A,\theta)} - \expectationf{\ell(w^*;(X,Y))}{Z\sim \cD(A,\theta)} \le \gamma$ with probability greater than $1-\frac{1}{10\dims}$, assuming $\ns =\Omega\Paren{ \frac{ \sqrt{\dims} \lambda^2 \log^2(\dims) }{\sqrt{\rho} \gamma^{\frac{3}{2}}}+ \frac{ \lambda^2 \log(\dims)}{\gamma^2} }$.

Now we need the following lemma from~\cite{KlivansM17}, which is analogous to Lemma~\ref{lem:IsingToLR} for the Ising model.

\begin{lemma}[Lemma 6.4 of~\cite{KlivansM17}]
\label{lem:parameter-error-mrf}
Let $P$ be a distribution on $\{-1,1\}^{\dims-1} \times \{-1,1\}$. Given multilinear polynomial $u_1\in \RR^{\dims-1}$, $\probof{Y=1|X=x} = \sigma\Paren{u_1(X)}$ for $(X,Y) \sim P$. 
Suppose the marginal distribution of $P$ on $X$ is $\delta$-unbiased, 
and $\expectationf{\log\Paren{1+e^{-Y \Paren{ u_1(X) }}} } {(X,Y)\sim P} - \expectationf{ \log\Paren{1+e^{-Y \Paren{ u_2(X)}}} }{(X,Y)\sim P} \le \gamma$ for another multilinear polynomial $u_2$, where $\gamma \le \delta^t e^{-2\normone{u_1}-6}$, then $\normone{u_1-u_2} = O\Paren{ (2t)^{t} e^{\normone{u_1}} \cdot \sqrt{\gamma/\delta^t}\cdot {\dims \choose t} }.$
\end{lemma}

By substituting $\gamma = e^{ -O( \lambda t)} \cdot (2t)^{-O(t)} \cdot \dims^{-3t} \cdot \alpha^2 $, we have that with probability greater than $1-\frac1{10\dims}$, $\sum_{I:\arg\min I =1} \absv{\bar{u}(I) - h(I)} \le \frac{\alpha}{\dims}$. We note that the coefficients of different monomials are recovered in each iteration. Therefore, by a union bound over $\dims$ iterations, we prove the desired result.
\end{proof}

\subsection{Parameter Estimation in $\ell_{\infty}$ Distance}
\label{sec:MRFl_infty}
In this section, we introduce a slightly modified version of the algorithm in the last section.

\begin{algorithm}

\caption{Private Learning binary $t$-wise MRF in $\ell_\infty$ distance}
\label{alg:DPMRF-infty}

\textbf{Input:}  $\ns = \ns_1+\ns_2$ i.i.d.\ samples $\{ z^1,\cdots, z^{\ns}\}$, where $z^m \in \{\pm1\}^{\dims}$ for $m\in [\ns]$; an upper bound $\lambda$ on $\max_{i \in [\dims]} \normone{\partial_i h}$, privacy parameter $\rho$

\begin{algorithmic}[1] 
\For{each $I \subset [\dims]$ with $\absv{I} \le t$}
\State Let $Q(I) \coloneqq \frac1{n_2} \sum_{m=n_1}^{n_1+n_2} \prod_{j \in I} z_j^m$

\State Compute $\hat{Q}(I)$, an estimate of  $Q(I)$ through an $\rho/2$-zCDP query release algorithm (PMW~\cite{HardtR10} or sepFEM~\cite{VietriTBSW20})
\EndFor

\For{$i=1$ to $\dims$}
\State $\forall m \in [\ns_1]$, $x_m \leftarrow \Brack{\prod_{j \in I} z^m_j: I \subset [\dims \backslash i], \absv{I} \le t-1}$, $y_m \leftarrow z_i^m$

\State $w^{priv} \leftarrow \cA( D, \cL, \rho^{\prime}, \cC)$, where $D=\{\Paren{x_m, y_m} \}_{m=1}^{\ns}$, $\ell(w;d) =  \log\Paren{1+e^{-y \langle w, x \rangle}}$, $\cC = \{ \normone{w}\le 2\lambda\}$, and $\rho^{\prime} = \frac{\rho}{2\dims}$

\State Define a polynomial $v_i: \RR^{\dims-1} \rightarrow \RR$ by setting $\bar{v_i}(I) = \frac1{2} w^{priv}(I)$ for all $I \subset [\dims \backslash i]$

\For{each $I \subset [\dims \backslash i]$ with $\absv{I} \le t-1$}
\State $\bar{u}(I \cup \{i\}) = \sum_{I^{\prime} \subset [\dims]}  \overline{\partial_I v_i}(I^{\prime}) \cdot \hat{Q}(I^{\prime})$,when $\arg\min (I \cup i) = i$
\EndFor
\EndFor
\end{algorithmic}

\textbf{Output:} $ \bar{u}(I): I \in C_t(K_{\dims})$,  where $K_{\dims}$ is the $\dims$-dimensional complete graph
\end{algorithm}

We first show that if the estimates $\hat Q$ for the parity queries
$Q$ are sufficiently accurate, Algorithm~\ref{alg:DPMRF-infty} solves
the $\ell_\infty$ estimation problem, as long as the sample size $n_1$
is large enough.

\begin{lemma}\label{yo}
  Suppose that the estimates $\hat Q$ satisfies
  $|\hat Q(I) - Q(I)| \leq \alpha/(8\lambda)$ for all $I\subset [p]$
  such that $|I|\leq t$ and $n_2 = \Omega(\lambda^2 t\log(p)/\alpha^2)$. Then
  with probability at least $3/4$, Algorithm~\ref{alg:DPMRF-infty}
  outputs a multilinear polynomial $u$ such that for every maximal
  monomial $I$ of $h$, $\absv{\bar{h}(I) - \bar{u}(I)} \le \dist,$
  given $\ns$ i.i.d.\ samples $Z^1,\cdots, Z^{\ns} \sim \cD$, as long
  as
$$n_1 =\Omega\Paren{ \frac{ e^{5\lambda t} \cdot \sqrt{\dims} \log^2(\dims)
  }{\sqrt{\rho} \dist^{\frac{9}{2}}}+ \frac{ e^{6\lambda t} \cdot
    \log(\dims)}{\dist^6} }. $$
\end{lemma}

\begin{proof}
  We will condition on the event that $\hat{Q}$ is a ``good'' estimate
  of $Q$: $|\hat Q(I) - Q(I)| \leq \alpha/(8\lambda)$ for all
  $I\subset [p]$ such that $|I|\leq
  t$. 
  Let us fix $i=1$. Let
  $X = \Brack{\prod_{j \in I} Z^j: I \subset [\dims] \backslash \{1
    \}, \absv{I} \le t-1}$, $Y = Z_i $, and we know that
  $\expectation{Y|X} = \sigma\Paren{\langle w^*,X \rangle}$, where
  $w^* = \Paren{ 2\cdot \overline{ \partial_{1}h} (I): I \subset
    [\dims] \backslash 1, \absv{I}\le t-1}$. Now given $\ns_1$
  i.i.d.\ samples $\{ z^{m}\}_{m=1}^{\ns_1}$ drawn from $\cD$, let
  $w^{priv}$ be the output of
  $\cA\Paren{ D, \cL, \frac{\rho}{\dims}, \{w: \normone{w}\le 2\lambda
    \}}$, where $D = \{(x^m,y^m)\}_{m=1}^{\ns_1}$ and
  $\ell(w;(x,y)) = \log\Paren{1+e^{-y \langle w, x
      \rangle}}$. Similarly, with probability at least
  $1-\frac{1}{10\dims}$,
  $$ \expectationf{\ell(w^{priv};(X,Y))}{Z\sim \cD(A,\theta)} -
  \expectationf{\ell(w^*;(X,Y))}{Z\sim \cD(A,\theta)} \le \gamma$$ as long as
  $\ns_1 =\Omega\Paren{ \frac{ \sqrt{\dims} \lambda^2 \log^2(\dims)
    }{\sqrt{\rho} \gamma^{\frac{3}{2}}}+ \frac{ \lambda^2
      \log(\dims)}{\gamma^2} }$.

Now we utilize Lemma 6.4 from~\cite{KlivansM17}, which states that if  $ \expectationf{\ell(w^{priv};(X,Y))}{Z\sim \cD(A,\theta)} - \expectationf{\ell(w^*;(X,Y))}{Z\sim \cD(A,\theta)} \le \gamma$, given a random sample $X$, for any maximal monomial $I \subset [\dims]\backslash \{1\}$ of $\partial_1 h$,
$$\probof{  \absv{\overline{\partial_1 h}(I) -  \partial_I v_1(X)} \ge \frac{\dist}{4} } < O\Paren{ \frac{\gamma \cdot e^{3\lambda t}}{\dist^2}}. $$

By replacing
$\gamma = \frac{ e^{-3\lambda t}\cdot \dist^3}{8\lambda}$, we have
$\probof{ \absv{\overline{\partial_1 h}(I) - \partial_I v_1(X)} \ge
  \frac{\dist}{4} } <\frac{\dist}{8\lambda}$, as long as
$n_1 = \Omega\Paren{ \frac{ \sqrt{\dims} e^{5\lambda t} \log^2(\dims)
  }{\sqrt{\rho} \dist^{\frac{9}{2}}}+ \frac{ e^{6\lambda t}
    \log(\dims)}{\dist^6} }$.  Accordingly, for any maximal monomial
$I$,
$ \absv{\expectation{\partial_I v_1(X)} - \overline{\partial_1 h} (I)}
\le \expectation { \absv{\partial_I v_1(X) - \overline{\partial_1
      h}(I) }} \le \frac{\dist}{4} + 2\lambda \cdot
\frac{\dist}{8\lambda} = \frac{\dist}{2}$.  By Hoeffding inequality,
given $n_2 = \Omega\Paren{\frac{\lambda^2t\log {\dims}}{\dist^2}}$, for each maximal
monomial $I$, with probability greater than $1-\frac{1}{\dims^t}$,
$\absv{ \frac1{n_2} \sum_{m=1}^{n_2}\partial_I v_1(X_m) -
  \expectation{\partial_I v_1(X)} }\le \frac{\dist}{4}$. Note that
$\absv{ Q(I) - \hat{Q}(I)} \le \frac{\dist}{8\lambda}$, then
$\absv{\frac1{n_2} \sum_{m=1}^{n_2}\partial_I v_1(X_m)
  -\sum_{I^{\prime} \subset [\dims]} \overline{\partial_I
    v_1}(I^{\prime}) \cdot \hat{Q}(I^{\prime}) } \le \frac{\dist}{8}
$. Therefore,
\begin{align}
&\absv{\sum_{I^{\prime} \subset [\dims]}  \overline{\partial_I v_1}(I^{\prime}) \cdot \hat{Q}(I^{\prime})   -   \overline{\partial_1 h}(I) } \nonumber\\
\le&\absv{ \sum_{I^{\prime} \subset [\dims]}  \overline{\partial_I v_1}(I^{\prime}) \cdot \hat{Q}(I^{\prime})   -\frac1{n_2} \sum_{m=1}^{n_2}\partial_I v_1(X_m) } + \absv{\frac1{n_2} \sum_{m=1}^{n_2}\partial_I v_1(X_m) -   \expectation{\partial_I v_1(X)}  } \nonumber\\
+& \absv{  \expectation{\partial_I v_1(X)} -  \overline{\partial_1 h}(I)   } \nonumber\\
\le&\frac{\dist}{8}+\frac{\dist}{4}+\frac{\dist}{2} = \frac{7\dist}{8}.\nonumber
\end{align}
Finally, by a union bound over $\dims$ iterations and all the maximal monomials, we prove the desired results.\end{proof}

We now consider two private algorithms for releasing the parity
queries. The first algorithm is called Private Multiplicative Weights
(PMW)~\cite{HardtR10}, which provides a better accuracy guarantee but
runs in time exponential in the dimension $\dims$. The following theorem can be viewed as a zCDP version of Theorem 4.3 in~\cite{Vadhan17}, 
by noting that during the analysis, every iteration satisfies $\eps_0$-DP, which naturally satisfies $\eps_0^2$-zCDP, and by replacing the strong composition theorem of $(\eps,\delta)$-DP by the composition theorem of zCDP (Lemma~\ref{lem:dpcomp}).

\begin{lemma}[Sample complexity of PMW, modification of Theorem 4.3 of~\cite{Vadhan17}]\label{hey}
  The PMW algorithm satisfies $\rho$-zCDP and releases $\hat{Q}$ such
  that with probability greater than $\frac{19}{20}$, for all
  $I \subset [\dims]$ with $\absv{I} \le t$,
  $\absv{\hat{Q}(I) -Q(I)} \le \frac{\dist}{8\lambda}$ as long as the
  size of the data set
  $$n_2 = \Omega\Paren{ \frac{ t \lambda^2 \cdot \sqrt{\dims}
      \log{\dims}}{\sqrt{\rho}\dist^2} }.$$
\end{lemma}

The second algorithm sepFEM (Separator-Follow-the-perturbed-leader
with exponential mechanism) has slightly worse sample complexity, but
runs in polynomial time when it has access to an optimization oracle
$\mathcal{O}$ that does the following: given as input a weighted
dataset
$(I_1, w_1), \ldots , (I_m, w_m)\in 2^{[p]} \times \mathbb{R}$, find $x\in \{\pm 1\}^\dims$,
\[
  \max_{x\in \{\pm 1\}^\dims} \sum_{i=1}^m w_i \prod_{j\in I_i} x_{j}.
\]
The oracle $\mathcal{O}$ essentially solves cost-sensitive
classification problems over the set of parity functions~\cite{ZLA03},
and it can be implemented with an integer program
solver~\cite{VietriTBSW20,GaboardiAHRW14}.

\begin{lemma}[Sample complexity of sepFEM,~\cite{VietriTBSW20}]\label{man}
  The sepFEM algorithm satisfies $\rho$-zCDP and releases $\hat{Q}$
  such that with probability greater than $\frac{19}{20}$, for all
  $I \subset [\dims]$ with $\absv{I} \le t$,
  $\absv{\hat{Q}(I) -Q(I)} \le \frac{\dist}{8\lambda}$ as long as the
  size of the data set
  $$n_2 = \Omega\Paren{ \frac{ t \lambda^2 \cdot {\dims^{5/4}}
      \log{\dims}}{\sqrt{\rho}\dist^2} }$$ The algorithm runs in
  polynomial time given access to the optimization oracle
  $\mathcal{O}$ defined above.
\end{lemma}

Now we can combine Lemmas~\ref{yo}, \ref{hey}, and \ref{man} to state
the formal guarantee of Algorithm~\ref{alg:DPMRF-infty}.

\begin{theorem}
  Algorithm~\ref{alg:DPMRF-infty} is a $\rho$-zCDP algorithm which, with probability at least $2/3$, 
  finds a multilinear polynomial $u$ such that for every maximal
  monomial $I$ of $h$, $\absv{\bar{h}(I) - \bar{u}(I)} \le \dist,$
  given $\ns$ i.i.d.\ samples $Z^1,\cdots, Z^{\ns} \sim \cD$, and
  \begin{enumerate}
  \item if it uses PMW for releasing $\hat Q$; it has a sample complexity of
$$\ns =  O\Paren{ \frac{ e^{5\lambda t} \cdot \sqrt{\dims} \log^2(\dims)
  }{\sqrt{\rho} \dist^{\frac{9}{2}}}+ \frac{ t \lambda^2 \cdot
    \sqrt{\dims} \log{\dims}}{\sqrt{\rho}\dist^2}+\frac{ e^{6\lambda t} \cdot
    \log(\dims)}{\dist^6} }$$ 
and a runtime complexity that is exponential in $\dims$;
\item if it uses sepFEM for releasing $\hat Q$, it has a sample
  complexity of
$$\ns =\tilde O\Paren{ \frac{ e^{5\lambda t} \cdot \sqrt{\dims} \log^2(\dims)
  }{\sqrt{\rho} \dist^{\frac{9}{2}}}+ \frac{ t \lambda^2 \cdot
    \dims^{5/4} \log{\dims}}{\sqrt{\rho} \dist^2}+\frac{ e^{6\lambda
      t} \cdot \log(\dims)}{\dist^6} }$$ and runs in polynomial time
whenever $t = O(1)$.
\end{enumerate}
\end{theorem}


\section{Lower Bounds for Parameter Learning}
\label{sec:est-lb}
The lower bound for parameter estimation is based on mean estimation in $\ell_\infty$ distance.

\begin{theorem}
\label{thm:est-lb}
  Suppose $\cA$ is an $(\eps, \delta)$-differentially private
  algorithm that takes $n$ i.i.d.\ samples $Z^1, \ldots , Z^n$ drawn
  from any unknown $\dims$-variable Ising model $\cD(A,\theta)$ and outputs $\hat A$ such that
$\expectation {  \max_{i,j \in [p]} |A_{i, j} - \hat A_{i, j}|} \leq \alpha \leq 1/50.$
  Then $\ns = \Omega\Paren{\frac{\sqrt{\dims}}{\alpha\eps}}$.
\end{theorem}

\begin{proof}
  Consider a Ising model $\cD (A, 0)$ with
  $A \in \RR^{\dims \times \dims}$ defined as follows: for
  $i \in [\frac{\dims}{2}], A_{2i-1,2i} = A_{2i,2i-1} = \eta_i\in
  [-\ln(2), \ln(2)]$, and $A_{ll'} = 0$ for all other pairs of
  $(l, l')$. This construction divides the $\dims$ nodes
  into $\frac{\dims}{2}$ pairs, where there is no correlation between nodes belonging to different pairs.
 It follows that
\begin{align}
  \probof {Z_{2i-1}=1, Z_{2i} =1 } =\probof {Z_{2i-1}=-1,  Z_{2i} = -1} = \frac{1}{2} \frac{e^{\eta_i}}{e^{\eta_i}+1},\nonumber\\
  \probof {Z_{2i-1}=1, Z_{2i} =-1 } =\probof {Z_{2i-1}=-1,  Z_{2i} = 1} =  \frac{1}{2} \frac{1}{e^{\eta_i}+1}.\nonumber
\end{align}
For each observation $Z$, we obtain an observation
$X\in \{\pm 1\}^{\dims/2}$ such that $X_{i} = Z_{2i-1} Z_{2i}$.  Then
each observation $X$ is distributed according to a product
distribution in $\{\pm 1\}^{(\dims/2)}$ such that the mean of each
coordinate $j$ is $(e^{\eta_i} - 1)/(e^{\eta_i} + 1)\in [-1/3, 1/3]$.

Suppose that an $(\eps, \delta)$-differentially private algorithm
takes $n$ observations drawn from any such Ising model distribution
and output a matrix $\hat A$ such that
$\E \left[\max_{i,j\in [p]} |A_{i,j} - \hat A_{i, j}|\right] \leq \alpha$. Let
$\hat \eta_i = \min\{\max\{\hat A_{2i-1, 2i} , -\ln(2)\}, \ln(2)\}$ be
the value of $A_{2i-1, 2i}$ rounded into the range of
$[-\ln(2), \ln(2)]$, and so $|\eta_i - \hat \eta_i| \leq \alpha$. It
follows that
\begin{align*}
  \left| \frac{e^{\eta_i} - 1}{e^{\eta_i} + 1} -  \frac{e^{\hat \eta_i} - 1}{e^{\hat \eta_i} + 1}\right| 
                                                                                                         &= 2  \left| \frac{e^{\eta_i} - e^{\hat \eta_i}}{(e^{\eta_i} + 1)(e^{\hat \eta_i} + 1)}\right|\\
                                                                                                         &< 2 \left| {e^{\eta_i} - e^{\hat \eta_i}}\right| \leq 4 \left(e^{|\eta_i - \hat \eta_i|} - 1\right)\leq 8 |\eta_i - \hat \eta_i| 
\end{align*}
where the last step follows from the fact that $e^a \leq 1 + 2a$ for
any $a\in [0, 1]$. Thus, such private algorithm also can estimate the
mean of the product distribution accurately:
\[
  \E \left[\sum_{i=1}^{p/2} \left| \frac{e^{\eta_i} - 1}{e^{\eta_i} +
        1} - \frac{e^{\hat \eta_i} - 1}{e^{\hat \eta_i} + 1}\right|^2
  \right] \leq 32 p \alpha^2
\]
Now we will use the following sample complexity lower bound on private
mean estimation on product distributions.
\begin{lemma}[Lemma 6.2 of~\cite{KamathLSU18}]\label{prod-low}
  If $M\colon \{\pm 1\}^{n\times d} \rightarrow [-1/3, 1/3]^d$ is
  $(\eps, 3/(64n))$-differentially private, and for every product
  distribution $P$ over $\{\pm 1\}^d$ such that the mean of each
  coordinate $\mu_j$ satisfies $-1/3 \leq \mu_j \leq 1/3$,
  \[
    \E_{X \sim P^n} \left[ \| M(X) - \mu \|_2^2 \right] \leq \gamma^2
    \leq \frac{d}{54},
  \]
  then $n\geq d/(72 \gamma \eps)$.
\end{lemma}
Then our stated bound follows by instantiating $\gamma^2 = 32p\alpha^2$
and $d = p/2$ in Lemma~\ref{prod-low}.\end{proof}


\section{Structure Learning of Graphical Models}
\label{sec:struct-ub}
In this section, we will give an $(\varepsilon,\delta)$-differentially private algorithm for learning the \emph{structure} of a Markov Random Field.
The dependence on the dimension $d$ will be only \emph{logarithmic}, in comparison to the complexity of privately learning the parameters.
As we have shown in Section~\ref{sec:est-lb}, this dependence is necessarily polynomial in $\dims$, even under approximate differential privacy.
Furthermore, as we will show in Section~\ref{sec:struct-lb}, if we wish to learn the structure of an MRF under more restrictive notions of privacy (such as pure or concentrated), the complexity also becomes polynomial in $\dims$.
Thus, in very high-dimensional settings, learning the structure of the MRF under approximate differential privacy is essentially the only notion of private learnability which is tractable.

The following lemma is immediate from stability-based mode arguments (see, e.g., Proposition 3.4 of~\cite{Vadhan17}).
\begin{lemma}
  Suppose there exists a (non-private) algorithm which takes $X = (X^1, \dots, X^n)$ sampled i.i.d.\ from some distribution $\cD$, and outputs some fixed value $Y$ (which may depend on $\cD$) with probability at least $2/3$.
  Then there exists an $(\varepsilon, \delta)$-differentially private algorithm which takes $O\left(\frac{n\log(1/\delta)}{\varepsilon}\right)$ samples and outputs $Y$ with probability at least $1 - \delta$.
\end{lemma}

We can now directly import the following theorem from~\cite{WuSD19}.
\begin{theorem}[\cite{WuSD19}]
There exists an algorithm which, with probability at least $2/3$, learns the structure of a pairwise graphical model.
  It requires $n = O\left(\frac{\lambda^2 \ab^4 e^{14\lambda} \log(\dims \ab)}{\eta^4}\right)$ samples.
\end{theorem}

This gives us the following private learning result as a corollary.
\begin{corollary}
\label{thm:str-ub-pair}
  There exists an $(\varepsilon, \delta)$-differentially private algorithm which, with probability at least $2/3$, learns the structure of a pairwise graphical model.
  It requires $n = O\left(\frac{\lambda^2 \ab^4 e^{14\lambda} \log(\dims \ab)\log(1/\delta)}{\varepsilon\eta^4}\right)$ samples.
\end{corollary}

For binary MRFs of higher-order, we instead import the following theorem from~\cite{KlivansM17}:

\begin{theorem}[\cite{KlivansM17}]
There exists an algorithm which, with probability at least $2/3$, learns the structure of a binary $t$-wise MRF.
It requires $n = O\left(\frac{ e^{O\Paren{ \lambda t}} \log(\frac{\dims}{\eta} )} {\eta^4}\right)$ samples.
\end{theorem}

This gives us the following private learning result as a corollary.
\begin{corollary}
  There exists an $(\varepsilon, \delta)$-differentially private algorithm which, with probability at least $2/3$,  learns the structure of a binary $t$-wise MRF.
  It requires $$n = O\left(\frac{ e^{O\Paren{ \lambda t}} \log(\frac{\dims}{\eta} )\log(1/\delta)} {\varepsilon \eta^4}\right)$$ samples.
\end{corollary}

\section{Lower Bounds for Structure Learning of Graphical Models}
\label{sec:struct-lb}

%
%
%

In this section, we will prove structure learning lower bounds under pure DP or zero-concentrated DP. 
The graphical models we consider are the Ising models and pairwise graphical model. 
However, we note that all the lower bounds for the Ising model also hold for binary $t$-wise MRFs, since the Ising model is a special case of binary $t$-wise MRFs corresponding to $t=2$. 
We will show that under $(\eps,0)$-DP or $\rho$-zCDP, a polynomial dependence on the dimension is unavoidable in the sample complexity.

In Section~\ref{sec:struct-lb-Ising}, we assume that our samples are generated from an Ising model. In Section~\ref{sec:struct-lb-pairwise}, we extend our lower bounds to pairwise graphical models.

\subsection{Lower Bounds for Structure Learning of Ising Models}
\label{sec:struct-lb-Ising}

\begin{theorem}
\label{thm:str-ising}
Any $(\eps,0)$-DP algorithm which learns the structure of an Ising model with minimum edge weight $\eta$ with probability at least $2/3$ requires
$\ns = \Omega\Paren{\frac{\sqrt{\dims}}{\eta\eps} + \frac{\dims}{\eps}}$  samples. 
Furthermore, at least $\ns =\Omega\Paren{ \sqrt {\frac{\dims} {\rho}} }$ samples are required for the same task under $\rho$-zCDP.
\end{theorem}

\begin{proof}
Our lower bound argument is in two steps. 
The first step is to construct a set of distributions, consisting of $2^{\frac{\dims}{2}}$ different Ising models such that any feasible structure learning algorithm should output different answers for different distributions. In the second step, we utilize our Private Fano's inequality, or the packing argument for zCDP~\cite{BunS16} to get the desired lower bound.

To start, we would like to use the following binary code to construct the distribution set. Let $\cC = \{ 0,1\}^{\frac{\dims}{2}}$,
given $c \in C $, we construct the corresponding distribution $\cD (A^c, 0)$ with $A^c \in \RR^{\dims \times \dims}$ defined as follows: for $i \in  [\frac{\dims}{2}], A^c_{2i-1,2i} = A^c_{2i,2i-1} = \eta \cdot c[i]$, and 0 elsewhere. By construction, we divide the $\dims$ nodes into $\frac{\dims}{2}$ different pairs, where there is no correlation between nodes belonging to different pairs. Furthermore, for pair $i$, if $c[i] = 0$, which means the value of node $2i-1$ is independent of node $2i$, it is not hard to show 
\begin{align}
\probof {Z_{2i-1}=1, Z_{2i} =1 } =\probof {Z_{2i-1}=-1,  Z_{2i} = -1} = \frac1{4},\nonumber\\
\probof {Z_{2i-1}=1, Z_{2i} =-1 } =\probof {Z_{2i-1}=-1,  Z_{2i} = 1} = \frac1{4}.\nonumber
\end{align}
On the other hand, if $c[i]=1$,
\begin{align}
\probof {Z_{2i-1}=1, Z_{2i} =1 } =\probof {Z_{2i-1}=-1,  Z_{2i} = -1} = \frac1{2} \cdot \frac{e^{\eta}}{e^{\eta}+1},\nonumber\\
\probof {Z_{2i-1}=1, Z_{2i} =-1 } =\probof {Z_{2i-1}=-1,  Z_{2i} = 1} = \frac1{2} \cdot \frac{1}{e^{\eta}+1}.\nonumber
\end{align}
The Chi-squared distance between these two distributions is 
$$8 \Brack{\Paren{\frac12\cdot\frac{e^\eta}{e^\eta+1} -\frac14}^2+ \Paren{\frac12\cdot\frac{1}{e^\eta+1} -\frac14}^2} = \Paren{1-\frac{2}{e^\eta+1}}^2 \le 4\eta^2.$$

Now we want to upper bound the total variation distance between $\cD (A^{c_1}, 0)$ and $\cD (A^{c_2}, 0)$ for any $c_1 \neq c_2 \in \cC$. Let $P_i$ and $Q_i$ denote the joint distribution of node $2i-1$ and node $2i$ corresponding to $\cD (A^{c_1}, 0)$ and $\cD (A^{c_2}, 0)$.
We have that
\begin{align*}
d_{TV} \Paren{\cD (A^{c_1}, 0),\cD (A^{c_2}, 0)} &\le \sqrt{ 2 d_{KL} \Paren{\cD (A^{c_1}, 0),\cD (A^{c_2}, 0)} }\\
&=  \sqrt{ 2 \sum_{i=1}^{\frac{\dims}{2}}d_{KL} \Paren{P_i, Q_i}} \le \min\Paren{2\eta\sqrt{\dims}, 1},
\end{align*}
where the first inequality is by Pinsker's inequality, and the last inequality comes from the fact that the KL divergence is always upper bounded by the Chi-squared distance.

In order to attain pure DP lower bounds, we utilize the corollary of DP Fano's inequality for estimation (Theorem~\ref{coro:fano}).

For any $c_1, c_2 \in C$, we have $d_{TV} \Paren{\cD (A^{c_1}, 0),\cD (A^{c_2}, 0)} \le \min\Paren{2\eta\sqrt{\dims}, 1}$. By the property of maximal coupling~\cite{Hollander12}, there must exist some coupling between $\cD^{\ns} (A^{c_1}, 0)$ and $\cD^{\ns} (A^{c_2}, 0)$  with expected Hamming distance smaller than $\min\Paren{2\ns\eta\sqrt{\dims} , \ns}$. Therefore, we have $\eps = \Omega \Paren{\frac{\log \absv{\cC}}{\min\Paren{\ns\eta\sqrt{\dims} , \ns}}}$, and  accordingly, $\ns = \Omega\Paren{\frac{\sqrt{\dims}}{\eta\eps} + \frac{\dims}{\eps}}$.

Now we move to zCDP lower bounds. We utilize a different version of the  packing argument~\cite{BunS16}, which works under zCDP.
\begin{lemma}
\label{lem:coupling-zCDP}
Let $\cV=\{P_1, P_2,...,P_M\}$ be a set of $M$ distributions over $\cX^{\ns}$. Let $ \{S_{i}\}_{i \in [M]}$ be a collection of disjoint subsets of $\cS$. If there exists an $\rho$-zCDP algorithm $\cA : \cX^{\ns} \to \cS$ such that for every $i \in [M]$, given $\Zon \sim P_i$, $\probof{\cA(\Zon) \in S_{i}}\ge \frac{9}{10}$, then
$$\rho = \Omega \Paren{\frac{\log M}{\ns^2}}.$$
\end{lemma}
By Lemma~\ref{lem:coupling-zCDP}, we derive $\rho = \Omega\Paren{\frac{\dims}{\ns^2}}$ and $\ns =\Omega\Paren{ \sqrt {\frac{\dims} {\rho}} }$ accordingly.
\end{proof}

\subsection{Lower Bounds for Structure Learning of Pairwise Graphical Models}
\label{sec:struct-lb-pairwise}

Similar techniques can be used to derive lower bounds for pairwise graphical models.

\begin{theorem}
Any $(\eps,0)$-DP algorithm which learns the structure of the $\dims$-variable pairwise graphical models with minimum edge weight $\eta$ with probability at least $2/3$ requires
$\ns = \Omega\Paren{\frac{\sqrt{\dims}}{\eta\eps} + \frac{\ab^2\dims}{\eps}}$ samples. Furthermore, at least $\ns =\Omega\Paren{  \sqrt {\frac{\ab^2\dims} {\rho}} }$ samples are required for the same task under $\rho$-zCDP.
\end{theorem}

\begin{proof}
Similar to before, we start with constructing a distribution set consisting of $2^{O\Paren{\ab\dims}}$ different pairwise graphical models such that any accurate structure learning algorithm must output different answers for different distributions.

Let $\cC$ be the real symmetric matrix with each value constrained to either $0$ or $\eta$, i.e., $\cC = \{W \in \{0, \eta \}^{\ab \times \ab}: W = W^T\}$. 
Without loss of generality, we assume $\dims$ is even. 
Given $c = [ c_1, c_2,\cdots, c_{\dims}]$, where $ c_1, c_2,\cdots, c_{\dims} \in C $, we construct the corresponding distribution $\cD (\cW^c, 0)$ with $\cW^c$ defined as follows: for $l \in  [\frac{\dims}{2}], W^c_{2l-1,2l} = c_l$, and for other pairs $(i,j)$, $W^c_{i,j} =  0$. Similarly, by this construction we divide $\dims$ nodes into $\frac{\dims}{2}$ different pairs, and there is no correlation between nodes belonging to different pairs.

We first prove lower bounds under $(\eps,0)$-DP.  By Theorem~\ref{coro:fano}, $\eps = \Omega\Paren{\frac{\log \absv{\cC}}{\ns}}$, since for any two $\ns$-sample distributions, the expected coupling distance can be always upper bounded by $\ns$. We also note that $\absv{\cC} = \Paren{ 2^{\frac{\ab (\ab+1)}{2}}}^\dims$. Therefore, we have $\ns = \Omega\Paren{\frac{\ab^2\dims}{\eps}}$. At the same time, $\ns = \Omega\Paren{\frac{\sqrt{\dims}}{\eta\eps}}$ is another lower bound, inherited from the easier task of learning Ising models.

With respect to zCDP, we utilize Lemma~\ref{lem:coupling-zCDP} and obtain $\rho = \Omega\Paren{\frac{\ab^2 \dims}{\ns^2}}$. 
Therefore, we have $\ns =\Omega\Paren{  \sqrt {\frac{\ab^2\dims} {\rho}} }$.
\end{proof}
%
%
%


%


\chapter{Private Hypothesis Selection}
\label{cha:phs}
\section{Introduction}
Perhaps the most fundamental question in statistics is that of simple hypothesis testing.
Given two known distributions $p$ and $q$, and a dataset generated according to one of these distributions, the goal is to determine which distribution the data came from.
This problem can be generalized in two ways that we consider in this paper.
First, rather than just two distributions, one can consider a setting where the goal is to select from a set of $k$ distributions.
We refer to this setting as \emph{$k$-wise simple hypothesis testing}.
Furthermore, the data may not have been generated according to \emph{any} distribution from the set of known distributions -- instead, the goal is to just select a distribution from the set which is competitive with the best possible (in an appropriate distance measure).
This problem is the core object of this chapter, and we denote it as \emph{hypothesis selection}.

The hypothesis selection problem appears naturally in a number of settings.
For instance, we may have a collection of distribution learning algorithms that are effective under different assumptions on the data, but it is unknown which ones hold in advance.
Hypothesis selection allows us to simply run all of these algorithms in parallel and pick a good output from these candidate distributions afterwards.
More generally, a learning algorithm may first ``guess'' various parameters of the unknown distribution during and for each guess produce a candidate output distribution. Hypothesis selection allows us to pick a final result from this set of candidates.
Finally, near-optimal sample complexity bounds can often be derived by enumerating all possibilities within some parametric class of distributions (i.e.,  a cover) and then applying hypothesis selection with this enumeration as the set of hypotheses~\cite{DevroyeL01}.

Classical work (e.g.,~\cite{Yatracos85, DevroyeL96, DevroyeL97, DevroyeL01}) on these problems has shown that, even in the most general setting of hypothesis selection, there are effective algorithms with sample complexity scaling only \emph{logarithmically} in the number of candidate hypotheses.
Building on this, there has been significant study into hypothesis selection with additional desiderata, including computational efficiency, robustness, weaker access to hypotheses, and more (e.g.,~\cite{MahalanabisS08, DaskalakisDS12b, DaskalakisK14, SureshOAJ14, AcharyaJOS14b, DiakonikolasKKLMS16, AcharyaFJOS18, BousquetKM19, BunKSW2019}).

One consideration which has not received significant attention in this setting is that of \emph{data privacy}, which we explore in this chapter.

We first distinguish between two common definitions of differential privacy.
The first is \emph{central differential privacy} (also known as the \emph{trusted curator} setting), which has been frequently used in previous chapters. In central differential privacy, users transmit their data to a central server without any obfuscation, and the algorithm operates on this dataset with the restriction that its final output must be appropriately privatized.

The second is \emph{local differential privacy} (LDP)~\cite{Warner65, EvfimievskiGS03, KasiviswanathanLNRS11}, in which users trust no one: each individual privatizes their own data before sending it to the central server.
In some sense, LDP places the privacy barrier closer to the users, and as a result, has seen adoption in practice by a number of companies that analyze sensitive user data, including Google~\cite{ErlingssonPK14}, Microsoft~\cite{DingKY17}, and Apple~\cite{AppleDP17}.

Recently, Bun, Kamath, Steinke, and Wu~\cite{BunKSW2019} showed that under the constraint of central differential privacy, one can still perform hypothesis selection with sample complexity which scales logarithmically in the number of hypotheses.
A priori, it was not clear that this would be possible.
Non-privately, one can apply methods which essentially ask ``Which of these two distributions fits the data better?'' for all $O(k^2)$ pairs of hypotheses.
Crucially, one can reuse the same set of $O(\log k)$ samples for all such comparisons (rather than drawing fresh samples for each one), and accuracy can be proved by a Chernoff and union bound style argument.
A naive privatization of this method would result in a polynomial dependence on $k$, due to issues arising from sample reuse and the composition of privacy losses.
\cite{BunKSW2019} avoid this issue by a careful application of tools from the differential privacy literature (i.e., the exponential mechanism~\cite{McSherryT07}), achieving an $O(\log k)$ sample complexity.
However, their method relies upon techniques which are not available in the local model of differential privacy.
Indeed, at first glance, it may not be clear how to improve upon an $\tilde O(k^2)$ sample complexity in the local model, achieved by simply using a fresh set of samples for each comparison, and using randomized response to privately perform the comparison.
This raises the question: what is the sample complexity of hypothesis selection under local differential privacy?
Can the problem be solved with a logarithmic dependence of the number of samples on the number of candidate hypotheses?
Or do we require a polynomial number of samples?

\subsection{Results, Techniques, and Discussion}
To describe our results, we more formally define the problems of $k$-wise simple hypothesis testing and hypothesis selection. 

\begin{definition}
  Suppose we are given a set of $n$ data points $\Xon$, which are sampled i.i.d.\ from some (unknown) distribution $p$, and a set of $k$ distributions $\mathcal{Q} = \{q_1, \dots, q_k\}$.
  The goal is to output a distribution $\hat q \in \mathcal{Q}$ such
  that $\dtv{p}{\hat q} \leq c \min_{q^* \in \mathcal{Q}} \dtv{p}{q^*}
  + \alpha$, for some $c =c(\alpha, k)$.
  
  We refer to the value of $c(\alpha,k)$ as the \emph{agnostic approximation factor}.
  If $c(\alpha,k)$ is an absolute constant, then we denote this problem as \emph{hypothesis selection}.
  If $c(\alpha, k)$ grows with $k$ and $\frac1\alpha$, we refer to
  this problem as \emph{weak hypothesis selection}. If we require that
  $p \in \mathcal{Q}$, that $\min_{i\neq j} \dtv{q_i}{q_j} \ge \alpha$, and that the algorithm must correctly identify $p$, then we denote this problem as \emph{$k$-wise simple hypothesis testing}.
\end{definition}

We introduce a formal definition of $\ve$-local differential privacy ($\ve$-LDP) in Section~\ref{sec:pre}.

Our first result shows that $k$-wise simple hypothesis testing (and thus, hypothesis selection) requires $\Omega(k)$ samples.
\begin{restatable}{theorem}{lb}\label{thm:fully-interactive-lb}
    Let $\ve \in (0,1)$.
  Suppose $M$ is an $\ve$-LDP protocol that solves the $k$-wise
  simple hypothesis testing problem  with probability at least $1/3$
  when given $n$ samples from some distribution $p \in \mathcal{Q}$,
  for any set $\mathcal{Q} = \{q_1, \ldots, q_k\}$ such
  that $\min_{i \neq j} \dtv{q_i}{q_j} \ge \alpha$. Then
  $n = \Omega\left( \frac{k}{ \alpha^2 \ve^2}\right).$
\end{restatable}
The theorem above shows that the cost of hypothesis testing is exponentially larger under local differential privacy than under central differential privacy (i.e., $\Omega(k)$ versus $O(\log k)$), and it holds even when the LDP protocol is allowed the power of full interactivity.
The construction used to prove this lower bound is the problem of $1$-sparse mean estimation, previously identified as a problem of interest by Duchi, Jordan, and Wainwright~\cite{DuchiJW13, DuchiJW17}. The lower bound follows from results in \cite{DuchiR19}.
Given the construction, our result can be seen as a translation of existing results. Further details are given in Section~\ref{sec:lb}.

With a lower bound of $\Omega(k)$ samples, and the aforementioned naive upper bound of $\tilde O(k^2)$ samples, the problem remains to identify the correct sample complexity.
We provide two different algorithms which require $\tilde O(k)$ samples, nearly matching this lower bound.
The first is for the special case of $k$-wise simple hypothesis testing, and is a non-interactive protocol -- each user only sends a message to the curator once, independently of the messages sent by other users. 
The second solves the more general problem of hypothesis selection, but requires sequential interactivity (albeit only $O(\log \log k)$ rounds of interaction): users still only send a message to the curator once, but the curator may request different types of messages from later users based on the messages sent by earlier users.
Less interaction in a protocol is generally preferred, and the role and power of interactivity in local differential privacy is one of the most significant questions in the area (see, e.g.~\cite{KasiviswanathanLNRS11,JosephMNR19,DanielyF19,DuchiR19,JosephMR20}).

Our first algorithmic result gives a non-interactive mechanism with $\tilde O(k)$ sample complexity for sufficiently well separated instances. Define $\beta :=  \min_{q \in \mathcal{Q}} \dtv{p}{q}$.

\begin{theorem}
  \label{thm:1Round}
  For every $\varepsilon \in [0,1)$, there is a non-interactive $\epsilon$-LDP algorithm that with
  probability at least $1-1/k^2$ outputs a distribution
  $\hat q \in \mathcal{Q}$ such that $\dtv{p}{\hat q} \leq \alpha$, if the number of samples
  $n \gg k (\log k)^3/(\alpha^4\epsilon^2)$ and $\beta \ll \alpha^2/\log k$.\footnote{We use $A\ll B$ to denote that $A\le c B$ for some sufficiently small constant $c>0$. Similarly we use $A \gg B$ to denote that $A \ge C B$ for some sufficiently large constant $C>0.$ $A \lesssim B$ is used interchangeably with $A=O(B)$. Similarly $A\gtrsim B$ is used interchangeably with $A=\Omega(B).$}
\end{theorem}

We prove the theorem in Section \ref{sec:non-adaptive}. While somewhat more general, the above theorem immediately gives a non-interactive $\tilde O(k)$-sample algorithm for the important special case of LDP $k$-wise simple hypothesis testing.

\begin{corollary}
  \label{cor:1rRealizable}
Suppose our instance of hypothesis testing is such that $p \in \mathcal{Q}$ and all distributions in $\mathcal{Q}$ are $\Omega(\alpha)$-far from each other in total variation distance.
For $\varepsilon \in [0,1)$, there exists a non-interactive $\varepsilon$-LDP algorithm which identifies $p$ with high probability, given $n = O\left(\frac{k \log^3 k}{\alpha^4 \varepsilon^2}\right)$ samples.
\end{corollary}

Our algorithm is based on a noised log-likelihood test, though significant massaging and manipulation of the problem instance is required to achieve an acceptable sample complexity. 
In our algorithm, the users are divided into $k$ groups. Each user in the $i^{th}$ group sends the log-likelihood (with some Laplace noise added for privacy) of observing the sample given to the user if the true distribution was $q_i$. The log-likelihoods from all the users in the $i^{th}$ group are aggregated and the most likely distribution is output. Alternatively, we can also think of our algorithm as using the samples from the $i^{th}$ group to estimate KL-divergences between the unknown distribution and $q_i$ and finally outputting the closest distribution.
For this approach to work, we need all the log-likelihoods to be
bounded. We achieve this by a {\em flattening lemma} which makes all
the distributions close to uniform, while preserving their total
variation distances. Moreover, this flattening can be implemented
locally by the users transforming their samples from the original
distribution. We believe that our flattening lemma may have
applications in other DP problems.

\smallskip
Our second algorithmic result is a $O(\log \log k)$-round sequentially interactive $\tilde O(k)$-sample algorithm for LDP hypothesis selection.

\begin{corollary}[Informal version of Corollary~\ref{cor:set-params-ldp}]
  \label{cor:informal-sequential}
  Suppose we are given $n$ samples from an unknown distribution $p$ and a set of descriptions of $k$ distributions $\mathcal{Q}$.
  There exists an algorithm which identifies a distribution $\hat q \in \mathcal{Q}$, such that $\dtv{p}{\hat q} \leq 27  \min_{q^* \in \mathcal{Q}} \dtv{p}{q^*} + O(\alpha)$ with probability $9/10$.
  The algorithm is $\ve$-LDP, requires $O(\log \log k)$ rounds of sequential interactivity, and $n = O\left(\frac{k \log k \log \log k}{\alpha^2 \varepsilon^2}\right)$ samples.
\end{corollary}

The $k$-wise simple hypothesis testing and hypothesis selection
problems can also be studied in the Statistical Queries (SQ) model of~\cite{Kearns98}. In this model, rather than being given samples
from a distribution $p$, the algorithm can ask queries specified by
bounded functions $\phi$, and get a (possibly adversarial) additive
$\tau$-approximation to the expectation of $\phi$ under $p$, where the
parameter $\tau$ is usually called the tolerance. For distributional
problems, \cite{KasiviswanathanLNRS11} showed that sample complexity in the LDP model is equivalent up to
polynomial factors to complexity in the SQ model, measured in terms of
the number of queries and the inverse tolerance
$\frac1\tau$. In particular, this
connection and our lower bound in
Theorem~\ref{thm:fully-interactive-lb} imply that $k$-wise simple
hypothesis testing in the SQ model requires that either the number of
queries or $\frac1\tau$ be polynomial in $k$. Because of the
polynomial loss, however, our precise study of the sample complexity
of these problems does not immediately translate to the SQ model. We
remark that both the 1-round algorithm in
Corollary~\ref{cor:1rRealizable}, and the algorithm in
Corollary~\ref{cor:informal-sequential} can be implemented in the SQ
model, and require, respectively, $1$ round and $O(\log \log k)$ rounds
of adaptive queries. Understanding the precise relationship between
the number of queries, the tolerance parameter, and the number of
rounds of adaptivity for solving hypothesis selection in the SQ model
is an interesting direction for future work.

Interestingly, Corollary~\ref{cor:informal-sequential} is derived as a consequence of a connection to maximum selection with adversarial comparators, a problem of independent interest.
This connection was previously established in works by Acharya, Falahatgar, Jafarpour, Orlitsky, and Suresh~\cite{AcharyaJOS14b, AcharyaFJOS18}. 
Prior work, however, has not exploited this connection under LDP constraints. Given the aforementioned importance of interactivity in the LDP setting, we initiate a study of the maximum selection with adversarial comparators problem from the perspective of understanding the trade-off between the number of rounds of parallel comparisons, and the total number of comparisons. 
The problem is as follows: we are given a set of items of unknown value, and we can perform comparisons between pairs of items.
If the value of the items is significantly different, the comparison will correctly report the item with the larger value.
If the values are similar, then the result of the comparison may be arbitrary.
The goal is to output an item with value close to the maximum.
We wish to minimize the total number of comparisons performed, as well as the number of rounds of interactivity.

Our main result for this setting gives a family of algorithms and lower bounds, parameterized by the number of rounds used (denoted by $t$).
Setting $t = O(\log \log k)$ yields Corollary~\ref{cor:informal-sequential}.
\begin{theorem}[Restatement of Theorems~\ref{thm:better-t-round} and~\ref{thm:lower-t-round}]
  \label{thm:informal-comparisons}
  For every $t \in \mathbb{Z}^+$, there exists a $t$-round protocol which, with probability $9/10$, approximately solves the problem of parallel approximate maximum selection with adversarial comparators from a set of $k$ items.
  The algorithm requires $O(k^{1 + \frac{1}{2^t-1}}t)$ comparison queries.
  Furthermore, any algorithm which provides these guarantees requires $\Omega(\frac{ k^{1 + \frac{1}{2^t-1}}}{3^t})$ comparison queries.
\end{theorem}
For each number of rounds $t$, we prove an upper bound and an almost-matching lower bound. 
In order to get down to a near-linear number of comparisons, we require $O(\log \log k)$ rounds, which is exponentially better than the $O(\log k)$ rounds required by previous algorithms.
Interestingly, in this setting, while maximum selection (with standard comparisons) with $\tilde O(k)$ queries is achievable in only 3 rounds, we show that $\Theta(\log \log k)$ rounds are both necessary and sufficient to achieve a near-linear number of comparisons when the results might be adversarial.

Our upper bounds follow by carefully applying a recursive tournament structure: in each round, we partition the input into appropriately-sized smaller groups, perform all pairwise-comparisons within each group, and send only the winners to the next round.
Additional work is needed to prevent the quality of approximation from decaying as the number of rounds increases. For the lower bound, we restate the problem as a game, in which the adversary constructs a random complete directed graph with a unique sink, and the algorithm queries the directions of edges, and tries to identify the sink in the smallest number of queries and rounds. We give a strategy in which the adversary constructs a layered graph with $t+1$ layers, where $t$ is the number of rounds in the game. We can guarantee that, if the algorithm does not make enough queries, then even after conditioning on the answers to the queries in the first $q$ rounds, the last $t+1 - q$ layers of the graph remain sufficiently random, so that the algorithm cannot guess the sink with reasonable probability. In particular, after $t$ rounds, there is still enough randomness in the $(t+1)$-st layer to make sure that algorithm cannot guess the sink correctly with high probability.

A self-contained description of the connection between hypothesis selection and maximum selection with adversarial comparators, as well as our upper and lower bounds, appear in Section~\ref{sec:comparators}.

\subsection{Related Work}
\label{sec:phs_related}
As mentioned before, our work builds on a long line of investigation on hypothesis selection.
This style of approach was pioneered by Yatracos~\cite{Yatracos85}, and refined in subsequent work by Devroye and Lugosi~\cite{DevroyeL96, DevroyeL97, DevroyeL01}.
After this, additional considerations have been taken into account, such as computation, approximation factor, robustness, and more~\cite{MahalanabisS08, DaskalakisDS12b, DaskalakisK14, SureshOAJ14, AcharyaJOS14b, DiakonikolasKKLMS16, AcharyaFJOS18, BousquetKM19, BunKSW2019}.
Most relevant is the recent work of Bun, Kamath, Steinke, and Wu~\cite{BunKSW2019}, which studies hypothesis selection under central differential privacy.
Our results are for the stronger constraint of local differential privacy.

Versions of our problem have been studied under both central and local differential privacy.
In the local model, the most pertinent result is that of Duchi, Jordan, and Wainwright~\cite{DuchiJW13, DuchiJW17}, showing a lower bound on the sample complexity for simple hypothesis testing between two known distributions. 
This matches folklore upper bounds for the same problem.
However, the straightforward way of extending said protocol to $k$-wise simple hypothesis testing would incur a cost of $\tilde O(k^2)$ samples.
Other works on hypothesis testing under local privacy include~\cite{GaboardiR18,Sheffet18,AcharyaCFT18,AcharyaCT19,JosephMNR19}.
In the central model, some of the early work was done by the Statistics community~\cite{VuS09,UhlerSF13}.
More recent work can roughly be divided into two lines -- one attempts to provide private analogues of classical statistical tests~\cite{WangLK15,GaboardiLRV16,KiferR17,KakizakiSF17,CampbellBRG18,SwanbergGGRGB19,CouchKSBG19}, while the other focuses more on achieving minimax sample complexities for testing problems~\cite{CaiDK17, AcharyaSZ18, AliakbarpourDR18, AcharyaKSZ18, CanonneKMUZ19, AliakbarpourDKR19, AminJM20}. 
While most of these focus on composite hypothesis testing, we highlight~\cite{CanonneKMSU19} which studies simple hypothesis testing.
Work of Awan and Slavkovic~\cite{AwanS18} gives a universally optimal test for binomial data, however Brenner and Nissim~\cite{BrennerN14} give an impossibility result for distributions with domain larger than $2$.
For further coverage of differentially private statistics, see~\cite{KamathU20}.

We are the first to study parallel maximum selection with adversarial comparators.
Prior work has investigated (non-parallel) maximum selection and sorting with adversarial comparators~\cite{AjtaiFHN09, AcharyaJOS14b, AcharyaFJOS18}.
Works by Acharya, Falahatgar, Jafarpour, Orlitsky, and Suresh established the connection with hypothesis selection~\cite{AcharyaJOS14b, AcharyaFJOS18}.
The parallelism model we study here was introduced by Valiant~\cite{Valiant75}, for parallel comparison-based problems with non-adversarial comparators.
Also, note that the noisy comparison models considered in some of these papers (where comparisons are incorrect with a certain probability) is different from the adversarial comparator model we study.
Thematically similar investigations on round complexity exist in the context of best arm identification for multi-armed bandits~\cite{AgarwalAAK17, TaoZZ19}.


\section{Preliminaries}
\label{sec:pre}

  

In the local setting of differential privacy, we imagine that each user has a single datapoint. 
We require that each individual's output is differentially private.
\begin{definition}[\cite{Warner65, EvfimievskiGS03, KasiviswanathanLNRS11}]
  Suppose there are $n$ individuals, where the $i$th individual has
  datapoint $X_i$. In each round $q$ of the protocol, there is a set
  $U_q\subseteq [n]$ of active individuals, and  each individual
  $i$ in $U_q$  computes some (randomized) function of their datapoint $X_i$, and of
  all messages $\{m_{r,j}: r\le q, j \in U_r\}$ output by all individuals in previous rounds, and outputs a message $m_{q,i}$. 
  A protocol is \emph{$\varepsilon$-locally differentially private}
  (LDP) if the set $\{m_{q,i}: q \in [t], i \in U_q\}$ of all messages
  output during the $t$ rounds of the protocol is $\ve$-differentially
  private with respect to the inputs $(X_1, \ldots, X_n)$.
\end{definition}

We note that there are many notions of interactivity in LDP, and we
cover the two primary definitions which we will be concerned with:
non-interactive and sequentially interactive protocols.
\begin{definition}
  An $\ve$-LDP protocol is \emph{non-interactive} if the number of
  rounds is $t=1$, and $U_1 = [n]$, i.e., every individual $i$ outputs a
  single message $m_i$, dependent only on their datapoint $X_i$.

  An $\ve$-LDP protocol is \emph{sequentially interactive} with $t$
  rounds of interaction if the sets $U_1, \ldots, U_t$ of active individuals
  in each round are disjoint.
  
\end{definition}

We recall the canonical $\ve$-LDP algorithm, randomized response.
\begin{lemma}
  \label{lem:rr}
  \emph{Randomized response} is the protocol when each user has a bit $X_i \in \{0,1\}$ and outputs $X_i$ with probability $\frac{e^\ve}{1 + e^\ve}$ and $1 - X_i$ with probability $\frac{1}{1 + e^\ve}$.
  It satisfies $\ve$-local differential privacy.
\end{lemma}

There exists a simple folklore algorithm for $\ve$-LDP $2$-wise simple hypothesis testing: use randomized response to privately count the number of samples which fall into the region where one distribution places more mass, and output the distribution which is more consistent with the resulting estimate.
This gives the following guarantees.
\begin{lemma}
  There exists a non-interactive $\ve$-LDP algorithm which solves $2$-wise simple hypothesis testing with probability $1- \beta$, which requires $n = O(\log(1/\beta)/\alpha^2\varepsilon^2)$ samples.
\end{lemma}

This can be extended to $k$-wise simple hypothesis testing by simply running said algorithm on pairs of distributions and picking the one which never loses a hypothesis test.
This gives us an $\tilde O(k^2)$ baseline algorithm for locally private hypothesis selection.
\begin{corollary}
  There exists a non-interactive $\ve$-LDP algorithm which solves $k$-wise simple hypothesis testing with high probability, which requires $n = O(k^2\log k/\alpha^2\varepsilon^2)$ samples.
\end{corollary}
We note that the same algorithm also solves the more general problem of $\ve$-LDP hypothesis selection, see Section~\ref{sec:comparators} and particularly Section~\ref{sec:comp:easy-k2}.

\section{Lower Bounds for Locally Private Hypothesis Selection}
\label{sec:lb}

In this section we state sample complexity lower bound results on
locally private hypothesis selection. We will first focus on the lower
bound for non-interactive protocols, and leverage a known lower
bound on locally private selection due to \cite{Ullman18} (a similar statement appears in~\cite{DuchiJW17}), which also
follows from the lower bound for sparse estimation in~\cite{DuchiJW17}. Let
$d\in \mathbb{N}$, $\alpha \in [0,1]$, and let $U_d$ be a uniform
distribution over $\{\pm 1\}^d$. For every $b \in \{\pm 1\}$ and
$j\in [d]$, we define distribution
$p_{b,j} = (1 - \alpha) U_d + \alpha \left( U_d \mid x_j = b\right)$,
that is, the distribution that is uniform over $\{\pm 1\}^d$ except
that $X_j = b$ with probability $1/2 + \alpha$.

\begin{theorem}[Theorem 3.1 of \cite{Ullman18}]\label{goodhombre}
  Let $\ve \in (0,1)$.  Let $d > 32$, $B$ be distbuted uniformly
  over $\{\pm 1\}$, and let $J$ be distributed uniformly over $[d]$.
  Suppose $M$ is an non-interactive $\ve$-LDP protocol and $n$ is such that
  \[
    \Pr_{B, J, X_1, \ldots , X_n \sim (p_{B,J}| B, J)} [ M(X_1, \ldots
    , X_n) = (B, J)] \geq 1/3.
  \]
  Then
  $$n = \Omega\left( \frac{d \log d}{ \alpha^2 \ve^2}\right).$$
\end{theorem}

To obtain a lower bound on hypothesis selection, we will rely on the
following fact that bounds the total variation distance between the
distributions $p_{b,j}$ (see e.g., Lemma 6.4 in \cite{KamathLSU18}).

\begin{fact}\label{prodtv}
  Let $q$ and $q'$ be two product distributions over $\{\pm 1\}^d$
  with mean vectors $\mu$ and $\mu'$ respectively, such that
  $\mu_i \in [-1/3, 1/3]$ for all $j\in [d]$. Suppose that
  $\|\mu - \mu'\|_2 \geq \alpha$ for any $\alpha \leq \alpha_0$ with
  some absolute constant $0 < \alpha_0 \leq 1$. Then
  $\dtv{q}{ q'} \geq C \alpha$, for some absolute constant $C$.
\end{fact}

\begin{theorem}[Non-interactive lower bound]
  Let $\ve \in (0,1)$.
  Suppose $M$ is a non-interactive an $\ve$-LDP protocol that solves the $k$-wise
  simple hypothesis testing problem  with probability at least $1/3$
  when given $n$ samples from some distribution $p \in \mathcal{Q}$,
  where $\mathcal{Q} = \{q_1, \ldots, q_k\}$ are distributions such
  that $\min_{i \neq j} \dtv{q_i}{q_j} \ge \alpha$. Then
  $$n \geq \Omega\left( \frac{k \log k}{ \alpha^2 \ve^2}\right).$$
\end{theorem}

\begin{proof}
  Let $\mathcal{Q} = \{p_{b,j} \mid b\in \{\pm 1\}, j\in [d]\}$ be a
  set of $k=2d$ probability distributions. For any pair of
  distributions $q, q'\in \mathcal{Q}$, we know from Fact~\ref{prodtv}
  that $\dtv{q}{q'} \geq \alpha/C$ for some absolute constant $C$.
  Thus, our stated bound follows from
  Theorem~\ref{goodhombre}.
\end{proof}

Next we will derive a sample complexity lower bound for general
locally private protocols. We will build on a result due to
\cite{DuchiR19} and consider the set of 1-sparse Gaussian
distributions $\{\mathcal{N}(\theta, I_d) \mid \theta\in \Theta\}$,
where
$\Theta = \{\theta \in \mathbb{R}^d\mid \|\theta\|_2 = \alpha,
\|\theta\|_0 = 1\}$ is the set of vectors that have a single non-zero
coordinate, equal to $-\alpha$ or $+\alpha$.

Following the result of \cite{DuchiR19} (and the framework of
\cite{BravermenGMNW}), we can obtain a general lower bound analogous
to Theorem~\ref{goodhombre}.

\begin{theorem}[Corollary~6 of \cite{DuchiR19}, Theorem~4.5 of \cite{BravermenGMNW}] \label{goodhombre2}
  Let $\ve \in (0,1)$. Let $U$ be a uniform distbution over
  $\Theta$. Suppose $M$ is an $\ve$-LDP protocol, and $n$ is such that
  \[
    \Pr_{\theta \sim U, X_1, \ldots , X_n \sim \mathcal{N}(\theta, I)}
    [ M(X_1, \ldots , X_n) = \theta] \geq 1/3.
  \]
  Then
  $$n \geq \Omega\left( \frac{d}{ \alpha^2 \ve^2}\right).$$
\end{theorem}


\lb*
\begin{proof}
  For any two 1-sparse vectors $\theta, \theta'\in \Theta$ such that
  $\theta\neq \theta'$, the total variation distance between their
  Gaussian distributions is given by
  $\|\theta - \theta'\|_2 = \sqrt{2} \alpha$ (see, e.g.,
  \cite{DevroyeMR18b}). Thus, our stated bound follows from
  Theorem~\ref{goodhombre2}.
\end{proof}


\section{Non-Interactive Locally Private Hypothesis Selection}
\label{sec:non-adaptive}

In this section, we prove Theorem \ref{thm:1Round}.  For simplicity of notation, we assume without loss of generality that $q_1,q_2,\dots,q_k$ and $p$ are discrete probability distributions 
on domain $[N]$, where $[N]:= \{1,2,\ldots, N\}.$ 
See the discussion in Remark~\ref{rmk:cont} on how to deal with continuous distributions.
Here we propose an algorithm which uses $n\lesssim k\,\polylog(k)/(\alpha^4\eps^2)$ samples, and outputs a distribution $\hat q \in \mathcal{Q}$ which has TV distance of at most $O(\alpha)$ with $p$, when $\beta \ll  \alpha^2/\log k.$ Recall that $\beta :=  \min_{q \in \mathcal{Q}} \dtv{p}{q}$.
In this mechanism, the users are divided into $k$ groups $G_1,G_2,\dots,G_k$ of size $n/k$ each. Let $X_{ij} \sim p$ denote the sample with the $j^{th}$ user in the group $G_i$. Our non-interactive mechanism is described in Algorithm \ref{alg:LikelyhoodTest}.

\begin{algorithm}
 \caption{Non-interactive $\eps$-DP mechanism for LPHS}
 \label{alg:LikelyhoodTest}
\hspace*{\algorithmicindent} \textbf{Input:} Distributions $\mathcal{Q} = \{q_1,\dots,q_k\}$, Samples $(X_{ij})_{i\in [k],j \in [n/k]}$ from unknown distribution $p$, sensitivity parameter for Laplace noise $L$, privacy parameter $\ve$, function $\gamma:[N]\to \R^+$ such that $|\log(\gamma(a)/q_i(a))|\le L$ for all $a\in [N],i\in [k].$\footnotemark \\
\hspace*{\algorithmicindent} \textbf{Output:} $\hat q \in \mathcal{Q}$ such that $\dtv{p}{\hat q}\le \alpha$ with high probability.
\begin{algorithmic}[1]
 \For {$i\in [k]$}
  \For {$j \in [n/k]$}
    \State The $j^{th}$ user in group $G_i$ sends $Z_{ij} := \log(\gamma(X_{ij})/q_i(X_{ij}))+ \Lap(L/\eps)$ to the central server
  \EndFor
  \State The central server computes $C_i= \frac{1}{(n/k)} \cdot \sum_{j\in [n/k]} Z_{ij}.$
 \EndFor
 \State \textbf{return} $\argmin{i} C_i$.
\end{algorithmic}
\end{algorithm}
\footnotetext{In other words, we require $\Dinfty(\gamma||q_i),\Dinfty(q_i||\gamma)\le L$ for all $i\in [k]$, i.e., all the distributions $q_1,q_2,\dots,q_k$ are close to some distribution $\gamma$. To prove Theorem~\ref{thm:1Round}, we will instantiate Algorithm~\ref{alg:LikelyhoodTest} with $\gamma$ being the uniform distribution on $[N]$, but we state Algorithm~\ref{alg:LikelyhoodTest} with arbitrary $\gamma$ for generality.}

\begin{lemma}
\label{thm:LikelyhoodTest}
Let $\eps\in (0,1]$ be some fixed privacy parameter. Suppose
$\beta\ll \alpha^2/L$ and $ n\gg \frac{ k(\log k) L^2} {\alpha^4 \eps^2}$. Then Algorithm~\ref{alg:LikelyhoodTest} is $\eps$-LDP and outputs
$\hat q \in \mathcal{Q}$ with probability at least $1-1/k^2$ such that
$\dtv{p}{\hat q} \leq \alpha$.
\end{lemma}
\begin{proof}
  From our assumption, $|\log(\gamma(a)/q_i(a))|\le L$ for
  $a\in [N],i\in [k]$. The algorithm adds noise sampled from
  $\Lap(L/\eps)$, hence $\eps$-LDP guarantee follows easily from the
  properties of the Laplace mechanism \cite{DworkR14}. We will now
  prove correctness.  Let $i^*\in [k]$ be such that
  $\dtv{p}{q_{i^*}} = \beta$.  Fix a group $G_i$ and consider,

  \begin{eqnarray*}
  \E_{X_{ij}\sim p}[C_i] &=& \frac{1}{(n/k)} \cdot \E \left[\sum_{j\in [n/k]} Z_{ij} \right]\\
  &=& \E_{a\sim p}\left[\log\left(\frac{\gamma(a)}{q_i(a)}\right)\right] \quad ({\text{By the linearity of expectation and}} \hspace{1mm} \E[\text{Lap}(L/\epsilon)] = 0)\\
  &=&\sum_{a\in [N]} p(a)\log\left(\frac{\gamma(a)}{q_i(a)}\right) \\
  &=& \sum_{a \in [N]} q_{i^*}(a)\log\left(\frac{q_{i^*}(a)}{q_i(a)}\right) + \sum_{a \in [N]} (p(a)-q_{i^*}(a))\log\left(\frac{\gamma(a)}{q_i(a)}\right) \\
  &+& \sum_{a \in [N]} q_{i^*}(a)\log\left(\frac{\gamma(a)}{q_{i^*}(a)}\right) \\
  &=& \DKL(q_{i^*}||q_i) + \sum_{a \in [N]} (p(a)-q_{i^*}(a))\log\left(\frac{\gamma(a)}{q_i(a)}\right) - \DKL(q_{i^*}||\gamma).
  \end{eqnarray*}

Let $B=-\DKL(q_{i^*}||\gamma)$. By re-arranging the above term we get 
\begin{eqnarray*}
\left|\E_{X_{ij}\sim q}[C_i]-\DKL(q_{i^*}||q_i)-B \right| &\le&  \sum_{a \in [N]} |p(a)-q_{i^*}(a)| \cdot \left|\log\left({\gamma(a)}/{q_i(a)}\right)\right| \\
&\le&  \sup_{a \in [N]}\left|\log\left(\frac{\gamma(a)}{q_i(a)}\right)\right| \cdot \left( \sum_{a \in [N]} |p(a)-q_{i^*}(a)|\right) \\
&\le& L\cdot 2 \dtv{p}{q_{i^*}} \\
&\le& 2L\beta \le 0.1\alpha^2.
\end{eqnarray*}
Now observe that each $Z_{ij}$ can be expressed as $W_{ij} + Y_{ij}$, where $Y_{ij} \sim \text{Lap}(L/\epsilon)$, and the support of random variable $W_{ij}$ is in the interval $[-L, L]$ from our assumption.
Therefore, we can apply the standard Hoeffding's inequality and concentration of Laplace random variables (see \cite{Hoeffding94, ChanSS11} for example) to obtain
$\Pr\left[|C_i-\E[C_i]|\ge 0.1 \alpha^2\right]\le \exp\left(-\Omega(1)\cdot \frac{(n/k)\alpha^4}{(L/\eps)^2}\right) \le \frac{1}{k^3}$.

By taking the union bound, with probability at least $1-1/k^2$, $|C_i-\DKL(q_{i^*}||q_i)-B|\le 0.2\alpha^2$  for all $i\in [k]$.
In particular, $C_{i^*}\le B+0.2\alpha^2$. This implies that if $i'=\argmin{i} C_i$, then $C_{i'}\le B+0.2\alpha^2$. It remains to argue that $\dtv{p}{q_{i'}} < \alpha$. Suppose not. Consider any $q_i$ such that $\dtv{p}{q_i} > \alpha$. This implies that $\dtv{q_{i*}}{q_i} > \alpha/2$ based on our assumption. Now consider
$C_i \ge B+\DKL(q_{i^*}||q_i) - 0.2 \alpha^2 \ge B+ 2\dtv{q_{i^*}}{ q_i}^2 - 0.2\alpha^2 \ge B+0.3\alpha^2$, where we used Pinsker's inequality.
\end{proof}

We will now prove that we can take $L=O(\log k)$ in Algorithm~\ref{alg:LikelyhoodTest} and Lemma~\ref{thm:LikelyhoodTest}. For this we will need the following lemma. Given a randomized map\footnote{i.e., $\phi(a)$ has a distribution over $[N']$ for each $a\in [N]$.} $\phi:[N]\to [N']$ and a distribution $q$ on $[N]$, the distribution $\phi \circ q$ on $[N']$ is defined as the distribution of $\phi(a)$ when $a$ is sampled from $q$. (In other words, $\phi \circ q$ is the pushforward of $q$.) For the remaining part of this section, let $U_{N'}$ denote the uniform distribution on $[N'].$

\begin{lemma}[Flattening Lemma]
\label{lem:splitting}
Let $q_1,q_2,\dots,q_k$ be distributions over $[N]$. There exists a randomized map $\phi:[N]\to [N']$ (depending on $q_1,\dots,q_k$) for some $N\le N'\le (k+1)N$ s.t.
\begin{enumerate}
    \item for every $a\in [N'], i\in [k],$ $\frac{1}{2N'}\le (\phi \circ q_i)(a) \le \frac{1}{N}$ and
    \item $\dtv{\phi\circ q_i}{ \phi \circ q_{i'}} =\frac{1}{2} \cdot \dtv{q_i}{q_{i'}}$ for any two distributions $q_i, q_{i'}$.
\end{enumerate}
\end{lemma}
\begin{proof}
Let $M(a)=\max_{i\in [k]} q_i(a)$ for $a\in [N]$. Let $N'=\sum_{a\in [N]} \lceil M(a)\cdot N\rceil$ and let $[N']=\cup_{a\in [N']} S_a$ be a partition of $[N']$ with $|S_a|=\lceil M(a) \cdot N\rceil$. Define $\phi':[N]\to [N']$ as follows: $\phi'(a)$ is uniformly distributed over $S_a$. Now it is clear that for every for every $b\in [N']$, $(\phi'\circ q_i)(b)\le \frac{1}{N}$. It is also clear that $\ellone{\phi'\circ q_{i}- \phi'\circ q_{i'}}=\ellone{q_{i}-q_{i'}}$ for any two distributions $q_{i},q_{i'}$.
We now mix in the uniform distribution $U_{N'}$ into $\phi'$, i.e., we define $\phi:[N]\to [N']$ as follows: $\phi(a)$ is distributed as $\phi'(a)$ with probability $1/2$ and distributed as $U_{N'}$ with probability $1/2$. Now for every $b\in [N'],$ $\frac{1}{2N'}\le (\phi\circ p_i)(b) \le \frac{1}{N}$. And $\ellone{\phi'\circ q_i- \phi'\circ q_{i'}}=\frac{1}{2}\ellone{q_i-q_{i'}}$ for any two distributions $q_i,q_{i'}$. 
We are now left with showing the upper bound on $N.$
\begin{eqnarray*}
N' &=& \sum_{a\in [N]}\lceil M(a) \cdot N\rceil \\
&\le& \sum_{a\in [N]} (M(a) \cdot N +1) \\
&=& N +\sum_{a\in [N]} \left(\max_{i\in [k]} q_i(a)\right) N \\
&\le& N+\sum_{a\in [N]} \left(\sum_{i\in [k]} q_i(a)\right)N =(k+1)N.
\end{eqnarray*}
\end{proof}
Now we have all the ingredients to finish the proof of Theorem \ref{thm:1Round}.
\begin{proof}
By using the randomized map $\phi$ as constructed in Lemma~\ref{lem:splitting}, the users first map their sample $a\sim p$ to a sample $\phi(a)$. Note that $\phi(a) \sim \phi\circ p$. Next we run the Algorithm~\ref{alg:LikelyhoodTest} with distributions $\phi\circ q_1,\dots, \phi\circ q_k$ and $\gamma:[N']\to \R^+$ given by $\gamma(b)=1/N'$ for all $b \in [N']$. From the first property mentioned in Lemma~\ref{lem:splitting}, we get $L=\log(k)+O(1).$ From the second property in Lemma~\ref{lem:splitting}, we know the TV distances are preserved by $\phi$.
\end{proof}

\begin{remark} 
If we are able to get $L=O(\alpha)$, then we get the nearly optimal sample complexity of $n=O\left(\frac{k\polylog(k)}{\alpha^2\eps^2}\right)$, formalized in the following question.
\end{remark}

\begin{question}
Given distributions $p_1,p_2,\dots,p_k$ which are $\alpha$-far to each other in $\ell_1$-distance, is there a randomized map $\phi:[n]\to [N]$ (which can depend on $p_1,\dots,p_k$) s.t.
\begin{enumerate}
    \item For all $i\in [n], a\in [N]$,  $\frac{1-\alpha}{N}\le (\phi \circ p_i)(a) \le \frac{1+\alpha}{N}$ and
    \item $\ellone{\phi\circ p- \phi \circ q} = \Theta(\ellone{p-q})$ for any two distributions $p,q$ with high probability.
\end{enumerate}
Note that $N$ can be arbitrarily large.
\end{question}

\begin{remark}
\label{rmk:cont}
The arguments in our proof can be easily generalized to continuous probability distributions. However, as our results do not depend on the domain size, it is intuitive to think of the following simple mapping from continuous distributions to discrete distributions on the domain $[N]$. First, we can approximate (to any precision) a set of continuous distributions by a set of discrete distributions on a finite support such that TV distances are preserved. We can then map any set of discrete distributions on a finite support to a set of discrete distributions on the domain $[N]$, where $N$ will depend on the desired precision. 
\end{remark}
 
\section{Hypothesis Selection via Adversarial Comparators}
\label{sec:comparators}
In this section, we give upper bounds for locally private hypothesis selection via a reduction to adversarial comparators, as introduced by~\cite{AcharyaJOS14b,AcharyaFJOS18}.
We begin by describing the reduction and how it can be implemented in the LDP setting in Section~\ref{sec:comp:reduction}.
This allows us to immediately obtain a non-interactive private algorithm which takes $\tilde O(k^2)$ samples and a sequentially-interactive algorithm which takes $\tilde O(k)$ samples (Section~\ref{sec:comp:easy-k2}).
However, this sequentially-interactive algorithm requires $O(\log k)$ rounds -- we give an algorithm which improves upon this round-complexity by an exponential factor.
We start in Section~\ref{sec:comp:2-rounds} by giving a simple $\tilde O(k^{4/3})$-sample algorithm which takes $2$ rounds: with the addition of only a single additional round, the sample complexity becomes significantly subquadratic.
This illustrates one of the main ideas behind our full upper bound, an $\tilde O(k)$-sample algorithm which takes only $O(\log \log k)$ rounds.
This is acheived by generalizing our $2$ round algorithm to general $t$: we give $t$-round algorithms for $1 \leq t \leq O(\log \log k)$, with sample complexities which interpolate between $\tilde O(k^2)$ and $\tilde O(k)$.
Other ideas are required to achieve an approximation which does not increase with $t$, which are described in Section~\ref{sec:comp:ub}.
We complement these upper bounds with lower bounds which show that these algorithms in the adversarial comparator setting are essentially tight (for \emph{every} choice of $t$) (Section~\ref{sec:comp:lb}).

\subsection{Adversarial Comparators and Connections to Locally Private Hypothesis Selection}
\label{sec:comp:reduction}
We describe the adversarial comparator setting of~\cite{AcharyaJOS14b,AcharyaFJOS18}, as well as their reduction to this model for the hypothesis selection problem.
The input is a set of $k$ items, with unknown values $x_1, \dots, x_k \in \mathbb{R}$.
An adversarial comparator is a function $C$, which takes two items $x_i$ and $x_j$,\footnote{In a slight abuse of notation, we use $x_i$ to refer to the item as well as its value.}, and outputs $\max\{x_i, x_j\}$ if $|x_i - x_j| > 1$ and $x_i$ or $x_j$ (adversarially) if $|x_i - x_j| \leq 1$.

We note that such a comparator can be either non-adaptive or adaptive.
In the former case, the results of all comparisons must be fixed ahead of time, whereas in the latter case, results of comparisons may depend on previous comparisons.
All of the mentioned algorithms will work in the (harder) adaptive case, and our lower bounds are for the (easier) non-adaptive case, and thus both have the same implications in the alternate setting for adaptivity.

We sometimes denote a comparison as a \emph{query}.
The goal is to output an item with value as close to the maximum as possible, with probability at least $2/3$.\footnote{Usual arguments allow us to boost this success probability to $1-\beta$ at a cost of $O(\log (1/\beta))$ repetitions, which can be done in parallel.}
More precisely, let $x^* = \max \{x_1, \dots, x_k\}$.
A number $x$ is a $\tau$-approximation of $x^*$ if $x \geq x^* - \tau$.
Simple examples (e.g., Lemma 2 of~\cite{AcharyaFJOS18}) show that it is impossible to output a $\tau$-approximation with probability $\geq 2/3$ for any $\tau < 2$ when we have $k \geq 3$ items.

We initiate study of \emph{parallel} approximate maximum selection under adversarial comparators.
Parallel maximum selection has recently been studied in other settings (including the standard comparison setting and with noisy (but not adversarial) comparisons, see, e.g.,~\cite{BravermanMW16}).
In this setting, the algorithm has $t$ rounds: in round $i$, the algorithm simultaneously submits $m_i$ pairs of items, and then simultaneously receives the results of the adversarial comparator applied to all $m_i$ pairs.
The total query complexity is $\sum_{i = 1}^t m_i$.

We now discuss the connection between this problem and hypothesis selection, as presented in Section 6 of~\cite{AcharyaFJOS18}.
We will then show how this connection still applies when considering the same problem under LDP.
First, we recall the Scheff\'e test of Devroye and Lugosi~\cite{DevroyeL01}, as described in Algorithm~\ref{alg:scheffe}.\footnote{We comment that this can be implemented in near-linear time, and $q_1(S)$ and $q_2(S)$ can be estimated to sufficient accuracy using Monte Carlo techniques.}
Given $n$ samples from $p$, with probability at least $1 - \beta$, it will output a distribution $\hat q$ such that $\dtv{p}{\hat q} \leq 3 \min\{\dtv{p}{q_1}, \dtv{p}{q_2}\} + \sqrt{\frac{2.5\log (1/\beta)}{n}}$.
In other words, if $ \min \Paren{ \dtv{p}{q_1},  \dtv{p}{q_2}} \leq \alpha$, then $n = O\left(\frac{\log(1/\beta)}{\alpha^2}\right)$ samples suffice to output a $\hat q \in \{q_1, q_2\}$ such that $\dtv{p}{\hat q} \leq (3 + \gamma)\alpha$, where $\gamma$ can be taken to be an arbitrarily small constant. Another way to phrase this is that the test returns $q_1$ if $\dtv{p}{q_1} < \frac{1}{3 + \gamma}\dtv{p}{q_2}$, it returns $q_2$ if $\dtv{p}{q_2} < \frac{1}{3 + \gamma}\dtv{p}{q_1}$, and it may return arbitrarily otherwise.
If we let $x_i = -\log_{3+\gamma}\dtv{p}{q_i}$, then the test will output $\max\{x_i, x_j\}$ if $|x_i - x_j| > 1$, or arbitrarily otherwise.
Note that this is precisely an implementation of the adversarial comparator function $C$ as described above, and thus the hypothesis selection problem can be reduced to (approximate) maximum selection with adversarial comparators.
In particular, a $\tau$-approximation for the maximum selection problem becomes a $(3+\gamma)^\tau$ agnostic approximation factor for hypothesis selection, which becomes $3^\tau + \gamma'$ if $\tau$ is a constant, for some other constant $\gamma' >0$ which can be taken to be arbitrarily small.
Each comparison is implemented using $O\left(\frac{\log (1/\beta)}{\alpha^2}\right)$ samples from $p$-- in fact, by a union bound argument, if we wish to perform $m$ comparisons and require the total failure probability under $1/3$, all of them can be done with the same set of $O\left(\frac{\log m}{\alpha^2}\right)$ samples.

\begin{algorithm}
    \caption{Scheff\'e Test}
    \label{alg:scheffe}
    \hspace*{\algorithmicindent} \textbf{Input:} $n$ samples $X_1, \dots, X_n$ from \emph{unknown} $p$, distributions $q_1$ and $q_2$ \\
    \hspace*{\algorithmicindent} \textbf{Output:} Distribution $q_1$ or $q_2$
    \begin{algorithmic}[1] 
        \Procedure{Scheff\'e}{$X, q_1, q_2$} 
        \State Let $S = \{x\ :\ q_1(x) > q_2(x)\}$.
        \State Let $q_1(S)$ and $q_2(S)$ be the probability mass that $q_1$ and $q_2$ assign to $S$ .
        \State Let $\hat p(S) = \frac{1}{n}\sum_{i=1}^n \mathbbm{1}_{X_i \in S} $ be the empirical mass assigned by $X_1, \dots, X_n$ to $S$. \label{ln:emp-mass}
        \If {$|q_1(S) - \hat p(S)| < |q_2(S) - \hat p(S)|$}
        \State \textbf{return} $q_1$.
        \Else
        \State \textbf{return} $q_2$.
        \EndIf
        \EndProcedure
    \end{algorithmic}
\end{algorithm}

It remains to justify that a similar reduction still holds under LDP constraints.
Recall that each individual possesses a single $X_i$, and they wish for their messages sent to the curator to be $\ve$-DP.
Only Line~\ref{ln:emp-mass} of Algorithm~\ref{alg:scheffe} depends on the private data, which is a statistical query, easily implemented under LDP.
More precisely, rather than sending the bit $\mathbbm{1}_{X_i \in S}$ to the curator, the user can send $Y_i$, which is a version of it privatized by Randomized response (Lemma~\ref{lem:rr}).
The curator can then form an $\ve$-LDP estimate of $p(S)$ by computing $\hat p(S) = \frac{e^\ve + 1}{e^\ve - 1}\left(\frac{1}{n}\sum Y_i - \frac1{e^\ve+1}\right)$.
Plugging this estimate into Line~\ref{ln:emp-mass}, it is not hard to show the modified procedure satisfies the following accuracy guarantee:
if $\min \Paren{ \dtv{p}{q_1},  \dtv{p}{q_2}} \leq \alpha$, then $n = O\left(\frac{\log(1/\beta)}{\ve^2\alpha^2}\right)$ samples suffice to output an $\ve$-LDP $\hat q \in \{q_1, q_2\}$ such that $\dtv{p}{\hat q} \leq (3 + \gamma)\alpha$, where $\gamma$ can be taken to be an arbitrarily small constant.

The above addresses the case of a single comparison.
If we wish to make $m$ comparisons (which are all correct with high probability), we partition users into $m$ sets of size $O\left(\frac{\log m}{\ve^2\alpha^2}\right)$ and use the data from each part to privately perform the appropriate comparison.
This takes a total of $O\left(\frac{m \log m}{\ve^2\alpha^2}\right)$ samples.
In particular, we can not reuse the same set of $O(\log m)$ samples for all comparisons (as in the non-private case), since it violate the privacy constraint, and doing so would give rise to algorithms which violate our main lower bound for locally private hypothesis selection (Theorem~\ref{thm:fully-interactive-lb}).
Finally, we note that a $t$-round algorithm in the maximum selection setting corresponds to a $t$-round sequentially interactive $\ve$-LDP algorithm for hypothesis selection, as we never query the same individual twice.

To conclude this section, we state the guarantees of the (trivial) algorithm which performs maximum selection from a set of $2$ elements, and the corollary for LDP hypothesis selection implied by the above reduction.

\begin{claim}
  There exists a $1$-round algorithm which achieves a $1$-approximation in the problem of parallel approximate maximum selection with adversarial comparators, in the special case where $k= 2$.
  The algorithm requires $1$ query.
\end{claim}

\begin{corollary}
  There exists a $1$-round algorithm which achieves a $(3+\gamma)$-agnostic approximation factor for locally private hypothesis selection with probability $1 -\beta$, in the special case where $k=2$, where $\gamma > 0$ is an arbitrarily small constant.
  The sample complexity of the algorithm is $O\left(\frac{\log (1/\beta)}{\ve^2 \alpha^2}\right)$.
\end{corollary}

For the following subsections, we will focus on the problem of parallel approximate maximum selection with adversarial comparators, stating corollaries to locally private hypothesis selection as appropriate.
Our primary concerns will be to simultaneously minimize the query/sample complexity and the round complexity, while minimizing the approximation/agnostic approximation factor is a secondary concern.
Nevertheless, our new algorithms for maximum selection will have an approximation constant of at most $3$, very close to the information-theoretic optimum of $2$.

\subsection{Baseline Algorithms}
\label{sec:comp:easy-k2}
In this section, we state some baseline results in this model, based on previously known algorithms.
This includes a $O(k^2)$-query non-interactive algorithm, and a $O(k)$-query $O(\log k)$-round algorithm.

The first method is a ``round-robin'' tournament method, which, in a single round, performs all pairwise comparisons and outputs the item which is declared to be the maximum the largest number of times (Algorithm~\ref{alg:round-robin}). 
This straightforward method is stated and analyzed in~\cite{AcharyaJOS14b,AcharyaFJOS18}, and the equivalent procedure for hypothesis selection (absent privacy constraints) was known prior~\cite{DevroyeL01}.
\begin{algorithm}
    \caption{1-Round Algorithm for Maximum Selection}
    \label{alg:round-robin}
    \hspace*{\algorithmicindent} \textbf{Input:} $k$ items $x_1, \dots, x_k$ \\
    \hspace*{\algorithmicindent} \textbf{Output:} Approximate maximum $x_i$ 
    \begin{algorithmic}[1] 
        \Procedure{Round-Robin}{$x_1, \dots, x_k$} 
      \For {all pairs $x_i, x_j$}
        \State Compare $x_i$ and $x_j$, record which one is reported to be the winner.
      \EndFor
      \State \textbf{return} the $x_i$ which is reported to be the winner the most times.
        \EndProcedure
    \end{algorithmic}
\end{algorithm}

\begin{claim}
  \label{clm:round-robin}
  There exists a $1$-round algorithm which achieves a $2$-approximation in the problem of parallel approximate maximum selection with adversarial comparators.
  The algorithm requires $O(k^2)$ queries.
\end{claim}

\begin{corollary}
  There exists a $1$-round algorithm which achieves a $(9+\gamma)$-agnostic approximation factor for locally private hypothesis selection with high probability, where $\gamma > 0$ is an arbitrarily small constant.
  The sample complexity of the algorithm is $O\left(\frac{k^2 \log k}{\ve^2 \alpha^2}\right)$.
\end{corollary}

The clear drawback of this method is that the complexity of the resulting algorithms is quadratic in $k$.
Unfortunately, a simple argument shows that this is tight for any $1$-round protocol: roughly, if we do not compare the smallest and second smallest items, we do not know which is smaller, and thus any algorithm which doesn't perform all $\binom{k}{2}$ comparisons in its $1$ round will be wrong with probability $1/2$ (more formal lower bounds for more general settings appear in Section~\ref{sec:comp:lb}).
The natural questions are, if we expend more rounds, can we reduce the sample complexity? 
And how many rounds are needed to achieve the information-theoretic optimum of a linear query complexity?
Many recent works have focused on this question without concern for the number of rounds expended~\cite{AjtaiFHN09, DaskalakisK14, SureshOAJ14, AcharyaJOS14b, AcharyaFJOS18}, culminating in algorithms with linear complexity.
When the round complexity is analyzed, it can be shown that all these methods take $O(\log k)$ rounds.
We state the implied results for our setting in the following claim and corollary, omitting details as we will shortly improve on the round complexity to be $O(\log \log k)$.

\begin{claim}[\cite{AcharyaJOS14b, AcharyaFJOS18}]
  There exists an $O(\log k)$-round algorithm which achieves a $2$-approximation in the problem of parallel approximate maximum selection with adversarial comparators.
  The algorithm requires $O(k)$ queries.
\end{claim}

\begin{corollary}
  There exists an $O(\log k)$-round algorithm which achieves a $(9+\gamma)$-agnostic approximation factor for locally private hypothesis selection with high probability, where $\gamma > 0$ is an arbitrarily small constant.
  The sample complexity of the algorithm is $O\left(\frac{k \log k}{\ve^2\alpha^2}\right)$.
\end{corollary}

\subsection{A Sub-Quadratic Algorithm with $2$ Rounds}
\label{sec:comp:2-rounds}
In this section, we give a simple $2$-round algorithm which results in a significantly better query complexity of $O(k^{4/3})$.
In Section~\ref{sec:comp:ub}, we generalize this to $t$-round protocols, but provide this as a warm-up and to convey one of the main ideas.

\begin{algorithm}
    \caption{2-Round Algorithm for Maximum Selection}
    \label{alg:2-round}
    \hspace*{\algorithmicindent} \textbf{Input:} $k$ items $x_1, \dots, x_k$ \\
    \hspace*{\algorithmicindent} \textbf{Output:} Approximate maximum $x_i$ 
    \begin{algorithmic}[1] 
        \Procedure{2-Round}{$x_1, \dots, x_k$} 
        \State Partition $x_1$ through $x_k$ into $k^{2/3}$ sets of size $k^{1/3}$. \label{ln:2-round-1}
        \State Run \textsc{Round-Robin} on each set to obtain $k^{2/3}$ winners. \label{ln:2-round-2}
        \State \textbf{return} the winner of \textsc{Round-Robin} on the set of $k^{2/3}$ winners. \label{ln:2-round-3}
        \EndProcedure
    \end{algorithmic}
\end{algorithm}

\begin{figure}[h!]
  \includegraphics[scale=0.5]{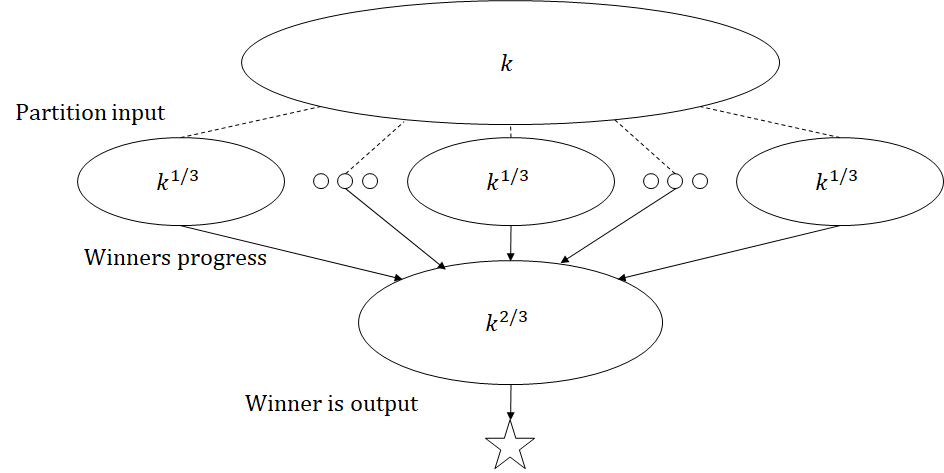}
  \centering
  \caption{An illustration of Algorithm~\ref{alg:2-round}. In the first round, the input is partitioned into sets of size $k^{1/3}$ and a round-robin tournament is performed on each.
  In the second round, a single round-robin tournament is performed on the winners from the previous round.}
\end{figure}

Algorithm~\ref{alg:2-round} describes the procedure, whose guarantees are summarized in the following theorem.
\begin{theorem}
  \label{thm:2-round}
  There exists a $2$-round algorithm which achieves a $4$-approximation in the problem of parallel approximate maximum selection with adversarial comparators.
  The algorithm requires $O(k^{4/3})$ queries.
\end{theorem}

The resulting corollary for LDP hypothesis selection is the following.
\begin{corollary}
  There exists an $2$-round algorithm which achieves a $(81+\gamma)$-agnostic approximation factor for locally private hypothesis selection with high probability, where $\gamma > 0$ is an arbitrarily small constant.
  The sample complexity of the algorithm is $O\left(\frac{k^{4/3} \log k}{\ve^2\alpha^2}\right)$.
\end{corollary}

We proceed to prove the guarantees stated in Theorem~\ref{thm:2-round}.

\begin{proof}
  The number of rounds is easily seen to be 2: Lines~\ref{ln:2-round-1} and~\ref{ln:2-round-2} can be performed in one round, and Line~\ref{ln:2-round-3}, which depends on the results of the previous round, is performed in the second round.

  We next analyze the number of queries. 
  Line~\ref{ln:2-round-2} performs the quadratic round-robin tournament of Claim~\ref{clm:round-robin} on sets of size $k^{1/3}$.
  The resulting number of queries for each set is $O(k^{2/3})$, and since there are $k^{2/3}$ sets, the total number of queries here is $O(k^{4/3})$.
  Line~\ref{ln:2-round-3} performs the same quadratic round-robin tournament on one set of size $k^{2/3}$, which takes $O(k^{4/3})$ queries.
  Therefore, the total number of queries is $O(k^{4/3})$.

  Finally, we justify that this achieves a $4$-approximation to the maximum.
  Consider the first round: a maximum element is placed into one of the $k^{2/3}$ sets, and by the guarantees of Claim~\ref{clm:round-robin}, the winner for this set will be a $2$-approximation to the maximum.
  Therefore, the maximum among the winners is a $2$-approximation to the overall maximum, and again by the guarantees of Claim~\ref{clm:round-robin}, the winner of this round will be a $4$-approximation to the maximum, as desired.
\end{proof}

\subsection{A Near-Linear-Sample Algorithm with $O(\log \log k)$ Rounds}
\label{sec:comp:ub}
In this section, we describe our main result in this setting, a family of algorithms for approximate maximum selection parameterized by $t$, which is the allowed number of rounds.
By setting $t = O(\log \log k)$, we will get an $O(k \log \log k)$-query algorithm which requires only $O(\log \log k)$ rounds, improving exponentially on the round complexity of previous approaches.
In particular, the following corollaries are obtained from Theorem~\ref{thm:better-t-round} and Corollary~\ref{cor:better-t-round-ldp} with an optimized setting of parameters.

\begin{corollary}
  \label{cor:set-params-select}
  There exists an $O(\log \log k)$-round algorithm which, with probability $9/10$, achieves a $3$-approximation in the problem of parallel approximate maximum selection with adversarial comparators.
  The algorithm requires $O(k \log \log k)$ queries.
\end{corollary}

\begin{corollary}
  \label{cor:set-params-ldp}
  There exists an $O(\log \log k)$-round algorithm which achieves a $(27+\gamma)$-agnostic factor for locally private hypothesis selection with probability $9/10$, where $\gamma > 0$ is an arbitrarily small constant.
  The sample complexity of the algorithm is $O\left(\frac{k\log k \log \log k}{\ve^2\alpha^2}\right)$.
\end{corollary}

The method is a careful recursive application of the approach described in Algorithm~\ref{alg:2-round}. 
Specifically, given $t$ allowed rounds of adaptivity, we partition the items into several smaller sets, perform the round-robin algorithm on each, and then feed the winners into the algorithm which is allowed $t-1$ rounds of adaptivity.
A judicious setting of parameters will allow the number of comparisons to decay quite rapidly as the number of rounds is increased.
This construction is described and analyzed in Section~\ref{sec:comp:ub:recursive}.
One challenge is that each round of the algorithm will potentially lose an additive $2$ in the approximation, resulting in an overall $2t$-approximation. 
To avoid this, we employ ideas from~\cite{DaskalakisK14}: we simultaneously apply two algorithms, at least one of which will be effective depending on whether the density of elements close to the maximum is high or low.
We describe the necessary modification and analyze the resulting approach in Section~\ref{sec:comp:ub:bound}.

\subsubsection{A Recursive Application of the $2$-Round Method}
\label{sec:comp:ub:recursive}
Our main result of this section will be the following lemma.
While the round and query complexity are essentially optimal (see Section~\ref{sec:comp:lb}), the quality of approximation is unsatisfactory -- our approach to improving this approximation is described in~\ref{sec:comp:ub:bound}.
\begin{lemma}
  \label{lem:t-round}
  There exists a $t$-round algorithm which achieves a $2t$-approximation in the problem of parallel approximate maximum selection with adversarial comparators.
  The algorithm requires $O(k^{1 + \frac{1}{2^t-1}}t)$ queries.
\end{lemma}

The method is described in Algorithm~\ref{alg:t-round}.
Note that for $t=1$ or $t=2$, this simplifies to Algorithms~\ref{alg:round-robin} and~\ref{alg:2-round}, respectively.
\begin{algorithm}
    \caption{$t$-Round Algorithm for Maximum Selection}
    \label{alg:t-round}
    \hspace*{\algorithmicindent} \textbf{Input:} $k$ items $x_1, \dots, x_k$, number of rounds $t$\\
    \hspace*{\algorithmicindent} \textbf{Output:} Approximate maximum $x_i$ 
    \begin{algorithmic}[1] 
        \Procedure{Multi-Round}{$x_1, \dots, x_k,t$} 
        \If {$t = 1$}
        \State \textbf{return} the winner of \textsc{Round-Robin} on $x_1, \dots, x_k$.
        \EndIf
        \State Set $\eta_t = \frac{1}{2^t - 1}$.
        \State Partition $x_1$ through $x_k$ into $k^{1-\eta_t}$ sets of size $k^{\eta_t}$. \label{ln:partition}
        \State Run \textsc{Round-Robin} on each set to obtain $k^{1-\eta_t}$ winners. \label{ln:round-robin}
        \State \textbf{return} the winner of \textsc{Multi-Round} on the set of $k^{1-\eta_t}$ winners with $t-1$ rounds. \label{ln:recursive}
        \EndProcedure
    \end{algorithmic}
\end{algorithm}

\begin{figure}[h!]
  \includegraphics[scale=0.5]{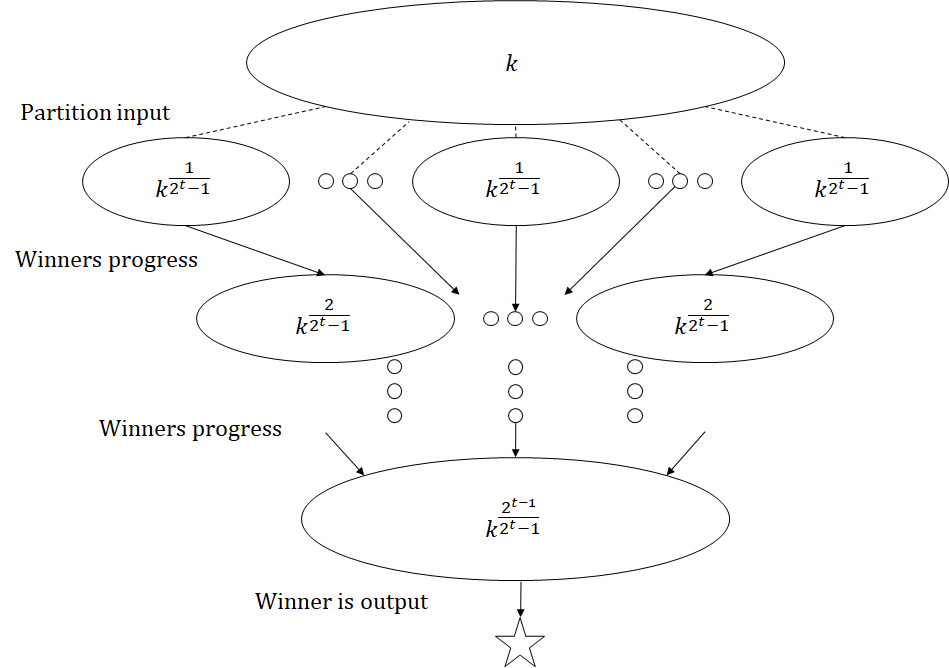}
  \centering
  \caption{An illustration of Algorithm~\ref{alg:t-round}. The input is partitioned into several sets and a round-robin tournament is performed on each. In subsequent rounds, winners are merged into fewer but larger sets, until we have only a single winner.}
\end{figure}

We proceed with proving that this algorithm satisfies the guarantees stated in Lemma~\ref{lem:t-round}.

\begin{proof}
  We prove the guarantees by induction.
  The base case corresponds to $t = 1$.
  As mentioned before, this is exactly equal to Algorithm~\ref{alg:round-robin}, and thus by Claim~\ref{clm:round-robin}, the lemma holds.

  Now, we prove the lemma for a general $t>1$, assuming it holds for $t-1$.
  The number of rounds is trivial: $1$ round is spent performing Lines~\ref{ln:partition} and~\ref{ln:round-robin}, and $t-1$ rounds are spent on the recursive call in Line~\ref{ln:recursive}.
  The approximation is also easy to reason about: the maximum element in the input appears in one of the sets in the partition in Line~\ref{ln:partition}, and therefore the winner of the corresponding set will be a $2$-approximation of the maximum.
  Thus, the set of winners which are fed into the recursive call in Line~\ref{ln:recursive} will have a $2$-approximation of the maximum.
  The inductive hypothesis guarantees that the winner of the recursive call will be a $2(t-1)$-approximation to \emph{this} item, making it a $2t$-approximation to the maximum.
  
  Finally, it remains to reason about the query complexity.
  Comparisons are only performed in Lines~\ref{ln:round-robin} and~\ref{ln:recursive}.
  In the former, we perform the round-robin tournament on $k^{1-\eta_t}$ sets of size $k^{\eta_t}$, so the total number of comparisons is $k^{1-\eta_t} \cdot O(k^{2\eta_t}) = O(k^{1+\eta_t})$.
  In the latter, the recursive call has an input of size $k^{1-\eta_t}$, so by the inductive hypothesis, the number of comparisons done in the recursive call is $O\left(\left(k^{1-\eta_t}\right)^{1 + \frac{1}{2^{t-1}-1}}(t-1)\right)$.
  Substituting in the value $\eta_t = \frac{1}{2^t-1}$, these two terms sum to $O(k^{1+\frac{1}{2^t -1}}t)$, as desired.
\end{proof}

\subsubsection{Bounding the Approximation Factor}
\label{sec:comp:ub:bound}
While the guarantees of Lemma~\ref{lem:t-round} are strong in terms of the round and query complexity, the approximation leaves something to be desired.
We alleviate this issue in a similar way as~\cite{DaskalakisK14}, by running a very simple strategy in parallel to the main method of Algorithm~\ref{alg:t-round}.
The intuition is as follows: if an item with maximum value $x^*$ is never compared with an item with value $x'$ such that $x^* > x' \geq 1$ (i.e., numbers which are $1$-approximations to the maximum), it will never lose a comparison.
If the fraction of such elements is low, then an item with value $x^*$ will make it to the final round, thus guaranteeing that the overall winner will be a $2$-approximation to the maximum.
On the other hand, if the fraction of such elements is high, then we can sample a small number of items such that we select at least one $1$-approximation to $x^*$, and running the round-robin algorithm on this set will guarantee a $3$-approximation to the maximum.

Our method is described more precisely in Algorithm~\ref{alg:better-t-round}, and the guarantees are described in Theorem~\ref{thm:better-t-round}.
\begin{algorithm}
    \caption{Better $t$-Round Algorithm for Maximum Selection}
    \label{alg:better-t-round}
    \hspace*{\algorithmicindent} \textbf{Input:} $K$ items $x_1, \dots, x_k$, number of rounds $t$\\
    \hspace*{\algorithmicindent} \textbf{Output:} Approximate maximum $x_i$ 
    \begin{algorithmic}[1] 
        \Procedure{Better-Multi-Round}{$x_1, \dots, x_k,t$} 
        \State Run \textsc{Multi-Round} on a random permutation of $x_1, \dots x_k$ with $t$ rounds, but halt when $t =1$ and let $L$ be the set of all remaining items. \label{ln:better:mr-call}
        \State Let $H$ be a random subset of $\{x_1, \dots, x_k\}$ of size $O\left(k^\frac{2^{t-1}}{2^{t}-1}\right)$. \label{ln:better:subset}
        \State Run \textsc{Round-Robin} on $L \cup H$ and \textbf{return} the winner. \label{ln:better:rr}
        \EndProcedure
    \end{algorithmic}
\end{algorithm}

\begin{figure}[h!]
  \includegraphics[scale=0.5]{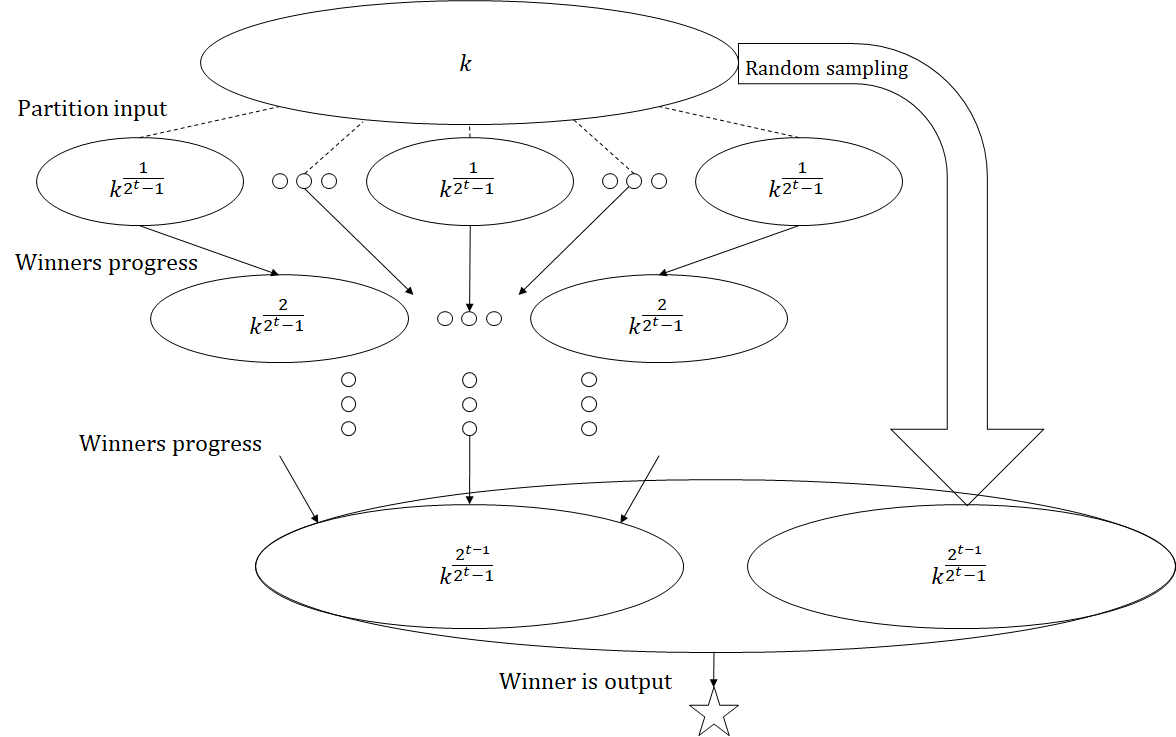}
  \centering
  \caption{An illustration of Algorithm~\ref{alg:better-t-round}. Similar to Algorithm~\ref{alg:t-round}, but in the last round, we perform a round-robin tournament additionally involving a random sample of items from the input.}
\end{figure}

\begin{theorem}
  \label{thm:better-t-round}
  There exists a $t$-round algorithm which, with probability $9/10$, achieves a $3$-approximation in the problem of parallel approximate maximum selection with adversarial comparators.
  The algorithm requires $O(k^{1 + \frac{1}{2^t-1}}t)$ queries.
\end{theorem}

This gives the following corollary for LDP hypothesis selection. 
\begin{corollary}
  \label{cor:better-t-round-ldp}
  There exists a $t$-round algorithm which achieves a $(27+\gamma)$-agnostic approximation factor for locally private hypothesis selection with probability $9/10$, where $\gamma > 0$ is an arbitrarily small constant.
  The sample complexity of the algorithm is $O\left(\frac{k^{1 + \frac{1}{2^t-1}}t \log k}{\ve^2\alpha^2}\right)$.
\end{corollary}
Corollaries~\ref{cor:set-params-select} and~\ref{cor:set-params-ldp} follow from these statements with an appropriate setting of $t$.

To conclude, we prove Theorem~\ref{thm:better-t-round}.

\begin{proof}
  The number of rounds is straightforward to analyze: Line~\ref{ln:better:mr-call} takes $t-1$ rounds (since we stop one round early), and Lines~\ref{ln:better:subset} and~\ref{ln:better:rr} can be done in $1$ last round.

  To analyze the number of comparisons, we require the following claim, which quantifies the number of items that make it to the last round of \textsc{Multi-Round}.
  \begin{claim}
    \label{clm:bound-round}
    $|L| = k^\frac{2^{t-1}}{2^t -1}$.
  \end{claim}
  \begin{proof}
    We recall the $\eta_t$ notation of Algorithm~\ref{alg:t-round}.
    The number of items which begin the first round of the algorithm is clearly $k$.
    Since these are partitioned into $k^{1-\eta_t}$ groups, each producing a single winner which progresses to the next round, we have $k^{1-\eta_t}$ items which begin the second round of the algorithm.
    A similar reasoning implies that the number of items entering the third round of the algorithm is $\left(k^{1-\eta_t}\right)^{1-\eta_{t-1}}$.
    Noting that $|L|$ is the number of items entering the $t$-th (i.e., final) round of the algorithm, the same logic shows that 
    \[
      \log_k |L| = \prod_{i=0}^{t-2} \left(1 - \frac1{2^{t-i}-1}\right) = \prod_{i=0}^{t-2} \left(\frac{2(2^{t-i-1}-1)}{2^{t-i}-1}\right) = \frac{2^{t-1}}{2^t - 1},
    \]
    as desired.
    The latter equality can be seen by a telescoping argument, as the numerators cancel the subsequent denominators.
  \end{proof}
  With this in hand, the number of comparisons is the number of comparisons due to Line~\ref{ln:better:mr-call} (which is $O\left(k^{1 + \frac{1}{2^t-1}}(t-1)\right)$ by the same argument as in the proof of Lemma~\ref{lem:t-round}) plus the number of comparisons due to Line~\ref{ln:better:rr}, which is $O\left((|H| + |L|)^2\right) = O\left(\left(k^\frac{2^{t-1}}{2^{t}-1}\right)^2\right) =  O\left(k^{1 + \frac{1}{2^t-1}}\right)$.
  Combining both of these gives the desired number of comparisons.

  Finally, we justify the accuracy guarantee.
  We split the analysis into two cases, based on the density of items which have value comparable to the maximum.
  Let $\zeta = \max_i x_i$ be the maximum value, let $B = \{ i\ :\ \zeta > x_i \geq \zeta - 1\}$ be the set of items which are $1$-approximations to (but not strictly equal to) the maximum value, and let $\gamma = \frac{|B|}{k}$ be their density.

  First, suppose that $\gamma \geq 1/10k^\frac{2^{t-1}}{2^{t}-1}$.
  If we let the hidden constant in the size of $H$ be $100$ (i.e., $|H| = 100k^\frac{2^{t-1}}{2^{t}-1}$), then Markov's inequality says that at least one item in $H$ will be a $1$-approximation to the maximum value with probability at least $9/10$.
  By the guarantees of \textsc{Round-Robin} (quantified in Claim~\ref{clm:round-robin}), the result of Line~\ref{ln:better:rr} will be a $3$-approximation to the maximum value, as desired.

  On the other hand, suppose that $\gamma \leq 1/10k^\frac{2^{t-1}}{2^{t}-1}$.
  We argue that an item with value $\zeta$ makes it to the final round of \textsc{Multi-Round} and is included in $L$ -- if this happens, then by the guarantees of \textsc{Round-Robin}, the result of Line~\ref{ln:better:rr} will be a $2$-approximation to $\zeta$ and the proof is complete.
  This happens if an item with value $\zeta$ is never compared to any element from $B$ within the first $t-1$ rounds.
  Fix some such item: the probability it is compared with some element from $B$ is upper bounded by the probability that any element of $B$ appears in the same subtree of depth $t-1$ leading up to the final round.
  The number of elements contained in this subtree is $k/|L| = k^\frac{2^{t-1}-1}{2^t -1}$, by Claim~\ref{clm:bound-round}.
  The expected number of items from $B$ in this subtree is bounded as $\gamma \cdot k/|L| \leq \frac{1}{10k^\frac{1}{2^t-1}} \leq 1/10$, and the result follows again from Markov's inequality.
\end{proof}

\subsection{A Lower Bound for Selection via Adversarial Comparators}
\label{sec:comp:lb}

\newcommand{\fa}{\tau}

In this section, we provide a lower bound for adversarial maximum selection with constrained interactivity. In Section~\ref{sec:two-round-lower-bound}, we consider a special case when $t=2$ and prove that any 2-round algorithm requires $\Omega\Paren{k^{\frac{4}{3}}}$ comparisons to find an approximate maximum. In Section~\ref{sec:t-round-lower-bound}, we generalize our result and technique to $t$ rounds and prove that any $t$-round algorithm requires $\Omega\Paren{ \frac { k^{1+ \frac{1}{2^{t}-1}}}{3^t} }$ comparisons. These lower bounds hold even for a non-adaptive adversary.

\subsubsection{A Lower Bound for $2$-round Algorithms}
\label{sec:two-round-lower-bound}

We warm up with a simpler case which illustrates the main ideas, namely, a lower bound for $2$-round algorithms. Specifically, we show that no matter how large the approximation factor $\fa$ is, any 2-round algorithm which solves the parallel approximate maximum selection problem requires $\Omega(k^{\frac{4}{3}})$ comparisons. 

\begin{theorem}
 \label{thm:lower-2-round}
For  any  $\fa > 1$, any $2$-round algorithm which achieves a $\fa$-approximation in the problem of parallel approximate maximum selection with non-adaptive adversarial comparators requires $\Omega(k^{\frac{4}{3}})$ queries.
\end{theorem}

We remark that, since our result is proved in the setting of non-adaptive adversarial comparators, it also automatically holds for adaptive comparators as well.

In our lower bound constructions, we reformulate the parallel approximate maximum selection problem as a game between an adversary and the algorithm. Before the game starts, the adversary commits to a random tournament (i.e., a complete directed graph)\footnote{Note that in this section we use ``tournament'' in the graph theoretic sense.} on $k$ nodes, each identified with one of the $k$ items. We will require that the tournament has, with probability $1$, a single sink node. Then, the algorithm player asks $m_1$ queries to the adversary, each query corresponding to a comparison between items $x_i$ and $x_j$. If the corresponding edge between $x_i$ and $x_j$ in the tournament is directed from $x_i$ to $x_j$, then the adversary answers that $x_j > x_i$, and, otherwise, the adversary answers that $x_i > x_j$. Equivalently, the algorithm asks for the directions of $m_1$ edges, which are revealed by the adversary. Afterwards, the player asks $m_2$ additional queries, based on the information gained from the initial $m_1$ queries, and the adversary answers them according to the directions of edges in the tournament. The game continues in this manner for $t$ rounds, where in round $q$ the algorithm asks $m_q$ queries, possibly dependent on all the query answers so far. After the $t$-th round, the algorithm must declare the ``winner'', i.e., the sink in the tournament.

Note that we can always produce item values so that the query answers are valid for the adversarial comparators model, and the sink node is the unique $\fa$-approximate maximum. Let $C_1, \ldots C_\ell$ be the strongly connected components of the tournament, ordered so that, if $i < j$, then all edges between $C_i$ and $C_j$ are directed from $C_i$ to $C_j$. Then we can set, for example, $x_j = 2i\fa$ for all $x_j$ in $C_i$. This way all queries to two items in the same strongly connected component can be answered arbitrarily, and all queries to items in two different components can be answered according to the direction of edges in the tournament. Moreover, we want to mention two special components. First, since there is a unique sink node $x_{i^*}$, $C_\ell$ must be equal to $\{x_{i^*}\}$, and therefore, $x_{i^*}$ is the unique $\fa$-approximate maximum. Second, in order to ``fool'' the player, the adversary sets $C_{\ell-1} = \{ x_{i^{\prime}}\}$, where all edges incident on $x_{i^\prime}$ are directed towards $x_{i^\prime}$, except the edge from $x_{i^\prime}$ to $x_{i^*}$. Thus, if the algorithm can achieve a $\fa$-approximation in the parallel approximate maximum selection problem, it can identify the sink node in the game above, and especially, distinguish it from $ x_{i^{\prime}}$.

We are now ready to prove Theorem~\ref{thm:lower-2-round}.

\begin{proof}
  We model the problem as the game described above, with $t =2$. By Yao's minimax principle, we can assume, without loss of generality, that the algorithm player makes deterministic choices.  
We start with the construction of the random tournament. From now on, to make the notation more convenient, we will denote nodes/items by their indices, i.e., we will write $i$ rather than $x_i$. Let $U_0$ denote the complete set of the nodes. Firstly, the adversary picks a uniformly random subset $U_1$ of $k^{\frac{2}{3}}$ nodes from $U_0$. Then from the adversary picks two nodes $i^*$ and $i^\prime$ uniformly at random from $U_1$.


Now we describe the directions of the edges of the tournament. For convenience, we define $V_0 \coloneqq U_0 \backslash U_1$ and $V_1 \coloneqq U_1 \backslash \{i^*,i'\}$. All edges incident on $i^*$ are directed towards $i^*$, i.e., $i^*$ is our sink node. All edges incident on $i^\prime$ are directed towards $i^\prime$, except the edge from $i^\prime$ to $i^*$. All edges from $V_0$ to $V_1$ are directed towards the node in $V_1$. Finally, the direction of any edge between two nodes in $V_0$ or two nodes in $V_1$ is chosen uniformly and independently from all other random choices.

Now we switch to the side of the player. As noted above, any algorithm which achieves $\fa$-approximation must correctly identify $i^{*}$ as the sink, with probability higher than $\frac{2}{3}$. Given $m_1 = m_2 = \frac{1}{100} k^\frac{4}{3}$, we want to show that any algorithm which asks $m = m_1 + m_2$ queries can not find $i^*$ with this probability. In the first round, the player asks $m_1 =  \frac{1}{100} k^\frac{4}{3} $ number of queries. We use $e_j = \{\alpha_j, \beta_j\}$ to denote the $j$-th query, where $j \in [m_1]$. Let $S$ denote the set of the nodes in $U_1$ which have ever competed with some other nodes from $U_1$, i.e., $S = \{ i_1 \in U_1 : \exists i_2 \in U_1, \exists j \in [m_1], e_j = \{i_1,i_2\}\}$.
Now we want to show that the following two ``bad'' events happen with a small probability:

$$ A_1 = \{ i^{\prime} \in S ~ \cup ~ i^{*} \in S \}, ~~~~ A_2 = \{  | V_1 \cap S |  \ge \frac1{2} \cdot k^\frac{2}{3}\}.$$

We bound the probability of event $A_1$ and $A_2$, respectively. For the rest of the proof, we will assume that $k$ is a large enough constant. 
By a union bound, 
$$\pr{A_1} \le \pr{ i^{\prime} \in S } + \pr{i^{*} \in S}  \le 2 \cdot m_1 \cdot  \frac{k^{\frac{2}{3}}-1}{\binom{k}{2}} \le 0.05.$$

With respect to $A_2$, let $e_j = \{\alpha_j, \beta_j\}$, where $j \in [m_1]$. We note that $| V_1 \cap S | \le 2\cdot \sum_{j\in [m_1]} \id (\alpha_j \in V_1, \beta_j \in V_1)$, where $\forall j$, $ \E \Paren{ \id (\alpha_j \in V_1, \beta_j \in V_1)} = \binom {k^{\frac{2}{3} }} {2} \backslash \binom{k}{2}  = \frac{k^{\frac{2}{3}}-1 }{k-1} \cdot k^{-\frac13}$. Furthermore, $\forall j_1 \neq j_2$, $\id (\alpha_{j_1} \in V_1, \beta_{j_1} \in V_1)$ and $\id (\alpha_{j_2} \in V_1, \beta_{j_2} \in V_1)$ are negatively correlated.

Therefore,
$$ \E \Paren{\sum_{j\in [m_1]} \id (\alpha_j \in V_1, \beta_j \in V_1)}  = m_1 \cdot \frac{k^{\frac{2}{3}}-1 }{k-1} \cdot k^{-\frac13} = \frac1{100}\cdot \frac{k}{k-1}\cdot (k^{\frac{2}{3}}-1),$$
$$ \Var \Paren{\sum_{j\in [m_1]} \id (\alpha_j \in V_1, \beta_j \in V_1)}  \le m_1 \cdot \frac{k^{\frac{2}{3}}-1 }{k-1} \cdot k^{-\frac13}\cdot \Paren{1 - \frac{k^{\frac{2}{3}}-1 }{k-1} \cdot k^{-\frac13} } < \frac1{100}\cdot \frac{k}{k-1}\cdot (k^{\frac{2}{3}}-1).$$

By Chebyshev's inequality,  
$$ \pr {A_2} \le \pr {\sum_{j\in [m_1]} \id (\alpha_j \in V_1, \beta_j \in V_1) \ge \frac1{4} \cdot k^\frac{2}{3} } \le k^{-\frac{2}{3}} \le 0.05, $$
where in the last inequality, we assume $k\ge 100$.

Now we move to the second round. From now on, we condition on neither $A_1$ nor $A_2$ holding, which happens with probability at least $0.9$. Then, conditional on $A_1$ and on the answers to the first $m_1$ queries, the pair $\{i^*, i^\prime\}$ is distributed uniformly in the set $R = \{i^*, i^\prime\} \cup (V_1 \setminus S)$. Moreover, if the algorithm does not query $\{i^*, i^\prime\}$ in the second round, then $i^*$ and $i^\prime$ will have the same distribution conditional on all $m$ queries, and the algorithm will not be able to identify $i^*$ with probability higher than $0.5$. Then, conditional on $A_1$, $A_2$, and the queries from the first round, the probability that the algorithm queries $\{i^*, i^\prime \}$ in the second round is at most
\[
  m_2 \cdot \frac1{\binom{|R|}{2}} \le \frac{k^{4/3}}{50 \cdot \frac1{2}k^{2/3}\cdot (\frac1{2} k^{2/3}-1)} \le 0.1.
\]
Therefore, the success rate of any deterministic $2$-round algorithm making at most $\frac{k^{2/3}}{100}$ queries is at most $0.1 + 0.1 + 0.5 < \frac{9}{10}$. As already noted, by Yao's minimax principle this also implies the result for randomized algorithms.


\end{proof}

\subsubsection{A Lower Bound for $t$-round Algorithms}
\label{sec:t-round-lower-bound}

In this section, we extend our 2-round lower bound to $t$ rounds. Specifically, we want to prove the following theorem.

\begin{theorem}
\label{thm:lower-t-round}
For  any  $\fa > 1$, any $t$-round algorithm which achieves  $\fa$-approximation in the problem of parallel approximate maximum selection with non-adaptive adversarial comparators requires $\Omega\Paren{ \frac { k^{1+ \frac{1}{2^{t}-1}}}{3^t} }$ queries.
\end{theorem}

We continue to model the problem as the game described in the previous subsection, but now with general $t$. We start with the construction of the random tournament, where a similar hierarchical structure to the 2-round construction is adopted. In the structure in Section~\ref{sec:two-round-lower-bound}, we can view node $i^*$ and $i^{\prime}$ as layer $2$, nodes in set $V_1$ as layer $1$, and all the other nodes as layer 0. We have thus designed a $3$-layer hierarchical structure in the proof of the 2-round lower bound, where edges are directed from lower to higher layers, and edges in the same layer are directed randomly. In this section, we generalize this construction to the following $(t+1)$-layer hierarchical structure, which we denote as $(k,t)$-construction.

Let $U_0$ denote the complete set of the nodes. In the first round, the adversary uniformly at random picks $k^{\frac{2^t-2}{2^t-1}}$ different nodes from $U_0$, which are denoted as $U_1$; etc.; in the $q$-th round, the adversary uniformly randomly picks $k^{\frac{2^t -2 ^q}{2^t-1} }$ from $U_{q-1}$, denoted as $U_q$, where $q \in [t-1]$. Finally, the adversary uniformly at random picks two nodes from $U_{t-1}$, denoted as $i^*$ and $i^{\prime}$, respectively, and we let $U_t = \{ i^*, i^{\prime}\}$ for the purpose of consistency.
For convenience, we define $V_0 = U_0 \backslash U_1$, $\cdots$, $V_q = U_q \backslash U_{q+1}$, where $0\le q \le t-1$, and $V_t = U_t = \{ i^*, i^{\prime}\}$. For $i < j$, we direct all edges from $V_i$ to $V_j$; for $q \neq t$,  edges between two nodes in $V_q$ are given a uniformly random direction; finally, the edge between $i^*$ and $i^\prime$ is directed towards $i^*$. Thus, $i^*$ is the unique sink in the graph.

The following is the core lemma in this section.

\begin{lemma}
\label{lem:main-t-lower}
Given a $(k,t)$-construction, and $ \forall \gamma<1$, every
deterministic $t$-round algorithm which finds $i^*$ with
probability higher than $\Paren{ \frac12+ \frac{\gamma}{100} \cdot
  3^{t}}$ requires $\Omega\left(\gamma\Paren{k^{1+ \frac{1}{2^{t}-1}}}\right)$ queries.
\end{lemma}

It is not hard to show that Theorem~\ref{thm:lower-t-round} can be
viewed as a corollary of the lemma, since given the random
$(k,t)$-construction, by setting $\gamma=\frac1{3^t}$, the lemma tells
that every $t$-round algorithm which finds $i^{*}$ with constant
probability makes at least
$\Omega\Paren{ \frac { k^{1+ \frac{1}{2^{t}-1}}}{3^t} }$ queries, and
any algorithm which achieves $\tau$-approximation should find $i^{*}$
with constant probability. Finally, by Yao's minimax principle, this
also holds for randomized algorithms. Therefore, our remaining task is to prove
Lemma~\ref{lem:main-t-lower}.

\begin{proof}
We prove the lemma by induction. Throughout the proof we assume that
$k$ is large enough with respect to $t$ and $\frac1\gamma$. We will
assume that the algorithm makes at most
$\frac{\gamma}{100}\Paren{k^{1+ \frac{1}{2^{t}-1}}}$ queries, and show
inductively that it succeeds in identifying $i^*$ with probability at
most $\frac12+ \frac{\gamma}{100} \cdot 3^{t}$.

For the base case when $t=2$, the lemma holds from the argument in the previous section.
For the inductive step, let $t$ be any integer where $t\ge 3$. Recall
that the number of queries asked by the algorithm in the first round
is $m_1 \le \frac{\gamma}{100} k^{1+ \frac{1}{2^{t}-1 }} $. We use
$e_j = (\alpha_j, \beta_j)$ to denote the $j$-th query, where $j \in
[m_1]$. By analogy with the $2$-round proof, let $S$ denote the set of nodes in $U_1$ which have ever competed with some other nodes from $U_1$, i.e., $S = \{ i_1 \in U_1 : \exists i_2 \in U_1, \exists j \in [m_1], e_j = (i_1,i_2)\}$. Now we want to show that the following $t$ ``bad'' events happen with a small probability:

$$ \forall q \in [t-1], A_q = \{  | V_q \cap S |  \ge \frac{1}{10} \cdot k^{\frac{2^t -2 ^q}{2^t-1} }\}, ~~~~ A_{t} = \{ i^{\prime} \in S ~ \cup ~ i^{*} \in S \}.$$

We bound the probability of event $A_t$ first. By a union bound, 
$$\pr{A_t} \le \pr{ i^{\prime} \in S } + \pr{i^{*} \in S}  \le 2 \cdot m_1 \cdot  \frac{k^{\frac{2^t -2}{2^t-1} }-1}{\binom{k}{2}} \le 0.05 \gamma.$$

With respect to $A_q, q \in [t-1]$, let $e_j = (\alpha_j, \beta_j)$, where $j \in [m_1]$. We note that $| V_q \cap S | \le 2\cdot \sum_{j\in [m_1]} \id (\alpha_j \in V_1, \beta_j \in V_q)$, where $\forall j$, $ \E \Paren{ \id (\alpha_j \in V_1, \beta_j \in V_q)} =  \frac{ k^{\frac{2^t-2}{2^t-1}} \cdot  k^{\frac{2^t-2^q}{2^t-1}}} { \binom{k}{2} }$, which is roughly $k^{-\frac{2^q}{2^t-1}}$. Furthermore, $\forall j_1 \neq j_2$, $\id (\alpha_{j_1} \in V_1, \beta_{j_1} \in V_q)$ and $\id (\alpha_{j_2} \in V_1, \beta_{j_2} \in V_q)$ are negatively correlated.
Therefore,
$$ \E \Paren{\sum_{j\in [m_1]} \id (\alpha_j \in V_1, \beta_j \in V_q)}  = m_1 \cdot \frac{ k^{\frac{2^t-2}{2^t-1}} \cdot  k^{\frac{2^t-2^q}{2^t-1}}} { \binom{k}{2} }\le \frac{1}{50}\cdot k^{\frac{2^t -2 ^q}{2^t-1} },$$
$$ \Var \Paren{\sum_{j\in [m_1]} \id (\alpha_j \in V_1, \beta_j \in V_q)} \le m_1\cdot \frac{ k^{\frac{2^t-2}{2^t-1}} \cdot  k^{\frac{2^t-2^q}{2^t-1}}} { \binom{k}{2} }\cdot \Paren{1-\frac{ k^{\frac{2^t-2}{2^t-1}} \cdot  k^{\frac{2^t-2^q}{2^t-1}}} { \binom{k}{2} }}  \le\frac{1}{50}\cdot k^{\frac{2^t -2 ^q}{2^t-1} }.$$

By Chebyshev's inequality,  
$$ \pr {A_q} \le \pr {\sum_{j\in [m_1]} \id (\alpha_j \in V_1, \beta_j
  \in V_q) \ge \frac{1}{20} \cdot k^{\frac{2^t -2 ^q}{2^t-1}} } \le 25
k^{- \frac{2^{t-1}}{2^t-1}} \le \frac{0.05\gamma}{t}, $$
where in the last inequality, we assume $k \ge \frac{Ct^2}{\gamma^2}$
for a large enough constant $C$.

From now on, we condition on none of the bad events $A_1, \ldots, A_t$
holding, which happens with probability at least $1 - 0.1\gamma$. We
also condition on the answers to the first $m_1$ queries. We would
like to say that the conditional distribution on the graph induced on
$U_1 \setminus S$ is identical to that of a
$(k^\prime, t-1)$-construction for $k^\prime = U_1 \setminus
S$. However, because of the random choice of $S$, the sizes of
$U_q\setminus S$ are not exactly as prescribed in the definition of a
$(k^\prime, t-1)$ construction. In order to finish the induction, we
consider the following process. For
$ k^{\prime} = \frac12 k^{\frac{2^t -2}{2^t-1} }$, we first 
denote $V^\prime_t = \{i^*, i^{\prime}\}$; then, we uniformly at
random draw $(k^{\prime})^{\frac{2^{t-2}}{2^{t-1}-1} } -2 $ nodes from
from $V_{t-1} \backslash S$,  and denote them as $V^{\prime}_{t-1}$; from
$V_q \backslash S$, $q \in [t-1]$, we uniformly at random draw
$(k^{\prime})^{\frac{2^{t-1} -2^{q-1}}{2^{t-1}-1} } -
|V^{\prime}_{q+1}| < \Paren{\frac12}^{\frac{2^{t-1}
    -2^{q-1}}{2^{t-1}-1}} \cdot k^{\frac{2^t -2^q}{2^t-1}} <
\frac{3}{4}k^{\frac{2^t -2^q}{2^t-1}}$ nodes, and denote them as
$V^{\prime}_q$. Conditonal on the bad events not holding, and on the
query answers from the first round, the subgraph induced on the nodes from $V^{\prime}_1$,
$V^{\prime}_2$, $\cdots$, and $V^{\prime}_t$, is distributed
identically to a $(k^{\prime},t-1)$ construction. Clearly, for the
algorithm to determine the sink $i^*$ in the full tournament, it must
also determine it in this subgraph. Ignoring queries in rounds
$2, \ldots, t$ to edges not in the subgraph, the algorithm is allowed to
ask at most 
$m = \frac{\gamma}{100}\Paren{k^{1+ \frac{1}{2^{t}-1}}} \le
\frac{\gamma}{100} \cdot 2.7 \cdot \Paren{ \Paren{k^{\prime}}^{1+
    \frac{1}{2^{t-1}-1}}}$ queries, and, by the inductive assumption, any
$(t-1)$-round algorithm can find $i^{*}$ with probability at most
$\frac12+ \frac{3^{t-1}}{100}\cdot 2.7\gamma $. Finally, by a union bound, the
probability of success of the $t$-round algorithm is at most
$\frac12+ \frac{3^{t-1}}{100}\cdot 2.7\gamma + 0.1 \gamma \le \frac12+ \frac{ 3^t}{100}
\gamma$. This finishes the inductive step.
\end{proof}

\chapter{Future Directions}
Compared with the classical topics in statistical inference, statistical inference with the presence of malicious users is far less understood. In the this section, we will mention some new directions and future work.

\noindent \textbf{Maturing Private Algorithms}

\noindent Up to this point, we have established many results in the area of private learning. However,  this is far from enough to bring differential privacy to real practice. First, there are still many fundamental problems unsolved in this area. For example, efficient algorithms are missing, even for a simple task of estimating product distributions with pure DP constraint. Besides, DP-SGD (a privatized version of SGD algorithm) is the most popular and successful algorithm in solving empirical risk minimization. However, its performance is always unsatisfactory under complex models, which brings significant difficulty to building private neural networks in practice. Second, differential privacy is designed to protect against membership inference attack, while there are many other privacy notions which are more suitable if the attack is different. By developing algorithms for these new models, I will facilitate private machine learning in settings where differential privacy was previously considered untenable. It is also interesting to explore the connections between different privacy notions.

\noindent \textbf{Robust Machine Learning in Graphical Models}

\noindent  Machine learning algorithms are always built under the assumption that the data is clean and well-behaved. However, this is not the case in the real world, where data is always inaccurate, or even malicious. For example, recent studies have shown that autonomous cars can be fooled by toxic signs. In such settings, standard statistical methods may give meaningless results without careful design. 

Graphical models are very useful in high dimensional statistical inference tasks. For example, Ising models are central in statistical physics, and phylogenetic trees are prevalent in biological applications. An interesting direction is to study robustness when dealing with graphical models, which is highly under-explored. 
This can model settings where we want to carry out graphical statistical tasks with the existence of a small number of malicious nodes,  and we wish to prevent them from significantly manipulating our results.

\noindent \textbf{The Interplay between Different Resource Limitations}

\noindent Another interesting direction is to study how different  constraints interplay with each other in machine learning tasks. For example, our model may suffer from attacks of multiple types simultaneously. A recent work has shown that poisoning LDP messages can be far more destructive than poisoning the raw data itself. In other words, privacy guarantees amplify the risk of the algorithms attacked by data manipulation. We believe such a phenomenon also exists in the other resource constraints,  such as the interplay between the communication constraint and data manipulation.  Understanding them is a critical step to building real-world machine learning systems.

\appendix

\bibliography{abr,masterref}

\end{document}